\documentclass[twoside,11pt]{article}

\usepackage{amsfonts}
\usepackage{amsmath}
\usepackage{amssymb} 
\usepackage{amsthm}
\usepackage{authblk} 
\usepackage{bbm}
\usepackage{bm}
\usepackage{cancel}
\usepackage{color}
\usepackage{dsfont}
\usepackage{extarrows}
\usepackage{float}
\usepackage{mathdots}
\usepackage{graphicx}
\usepackage{grffile}
\usepackage{mathrsfs}
\usepackage{mdframed}
\usepackage[round]{natbib}
\usepackage{subfig} 
\usepackage{url}
\usepackage[top=1in, bottom=1in, left=1in, right=1in]{geometry}

\renewcommand{\d}{\mathop{}\!\mathrm{d}}

\newcommand{\mnorm}[1]{{\left\vert\kern-0.25ex\left\vert\kern-0.25ex\left\vert #1 
    \right\vert\kern-0.25ex\right\vert\kern-0.25ex\right\vert}}

\newcommand\newday[1]{\leavevmode\xleaders\hbox{*}\hfill\kern0pt\\ \centerline{\Large \textbf{#1}}}

\DeclareMathOperator*{\Uu}{\cup}
\DeclareMathOperator{\var}{var}

\usepackage[shortlabels]{enumitem}

\newtheorem{theorem}{Theorem}

\newtheorem{lemma}[theorem]{Lemma}
\newenvironment{lem}[1]
  {\renewcommand\thelemma{#1}\lemma}
  {\endlemma}

  \newtheorem{corollary}[theorem]{Corollary}

    \newtheorem{proposition}[theorem]{Proposition}

\newtheorem{definition}{Definition}

\newtheorem{example}{Example}

\allowdisplaybreaks

\makeatletter
\newcommand{\pushright}[1]{\ifmeasuring@#1\else\omit\hfill$\displaystyle#1$\fi\ignorespaces}
\newcommand{\pushleft}[1]{\ifmeasuring@#1\else\omit$\displaystyle#1$\hfill\fi\ignorespaces}
\makeatother

\title{Generalized Score Matching for General Domains}
\author[1]{Shiqing Yu}
\author[2]{Mathias Drton}
\author[3]{Ali Shojaie}
\affil[1]{Department of Statistics, University of Washington, Seattle, Washington, 98195, U.S.A.}
\affil[2]{Department of Mathematics, Technical University of Munich, 85748 Garching bei M\"{u}nchen, Germany}
\affil[3]{Department of Biostatistics, University of Washington, Seattle, Washington, 98195, U.S.A.}

\begin{document}

\maketitle

\begin{abstract}
{Estimation of density functions supported on general domains arises when the data is naturally restricted to a proper subset of the real space.
  This problem is complicated by typically intractable normalizing constants.
  Score matching provides a powerful tool for estimating densities with such intractable normalizing constants, but as originally proposed is limited to densities on $\mathbb{R}^m$ and $\mathbb{R}_+^m$. In this paper, we offer a natural generalization
  of score matching that accommodates 
  densities supported on a very general class of domains. We apply the framework to truncated graphical and pairwise interaction models, 
  and provide theoretical guarantees for the resulting estimators. We also generalize a recently proposed method from bounded to unbounded domains, and empirically demonstrate the advantages of our method.
\vspace{.15cm}}
\end{abstract}

\section{Introduction}
Probability density functions, especially in multivariate
graphical models, are often defined only up to a normalizing
constant. In higher dimensions, computation of the normalizing
constant is typically an intractable problem that becomes worse when
the 
distributions are defined only on a proper subset of
the real space $\mathbb{R}^m$. For example, even truncated multivariate Gaussian densities have intractable normalizing constants except for special situations, e.g.,~with diagonal covariance matrices. This inability to calculate normalizing constants makes density estimation for general domains very challenging. 


 \emph{Score matching}  \citet{hyv05} is a 
computational efficient
solution to density estimation that bypasses the calculation of
normalizing constants and
has enabled, in particular, large-scale applications of non-Gaussian
graphical models
\cite{MR3338335,lin16,NIPS2016_6530,NIPS2015_6006,yum19}.  Its original
formulation targets distributions supported on
$\mathbb{R}^m$.  It was extended to treat the non-negative orthant
$\mathbb{R}_+^m$ in \citep{hyv07}, with more recent generalizations in
\citet{yu18,yus19}.  An extension to products of intervals like
$[0,1]^m$ was given in \citep{jan15,jan18,tan19}, and more general
bounded domains were considered in \citep{liu19}. Despite this
progress, the existing
approaches have important limitations: The method in
\citet{liu19} only allows for bounded support, and
earlier methods for $\mathbb{R}_+^m$ and $[0,1]^m$ offer ad-hoc
solutions that cannot be directly extended to more general
domains. This paper addresses these limitations by developing a
unifying framework that 
encompasses the existing methods and
applies to
unbounded domains.  The framework enables
new applications for more complicated domains yet retains the
computational efficiency of the original score matching.
The remainder of
this introduction provides a more detailed review of the score
matching estimator and the contributions of this paper.

\subsection{Score Matching and its Generalizations}

The original \emph{score matching} estimator introduced in
\citet{hyv05} is based on the idea of minimizing the Fisher distance
given by the expected $\ell_2$ distance between the gradients of the
true log density $\log p_0$ on $\mathbb{R}^m$ and a proposed log
density $\log p$, that is, 
\begin{equation}
  \label{eq:score:hyv:original}
  \int_{\mathbb{R}^m}p_0(\boldsymbol{x})\|\nabla_{\boldsymbol{x}}\log
  p(\boldsymbol{x})-\nabla_{\boldsymbol{x}}\log
  p_0(\boldsymbol{x})\|_2^2\d\boldsymbol{x}. 
\end{equation}
Integration by parts leads to an associated empirical loss, in which
an additive constant term depending only on $p_0$ is
ignored.  This loss avoids calculations of the normalizing constant
through dealing with the derivatives of the log-densities only.
Minimizing the empirical loss to derive an estimator is particularly
convenient if $p$ belongs to an exponential family because the loss is
then a quadratic function of the family's canonical parameters.  The
latter property holds, in particular, for the
Gaussian case, where  the methods
proposed by \cite{MR3306432,MR3180660} constitute a special case of
score matching.  

In \citep{hyv07}, the approach was generalized to densities on $\mathbb{R}_+^m=[0,\infty)^m$ by minimizing
instead 
\begin{equation}
    \label{eq:score:hyv:nonneg}
	\int_{\mathbb{R}_+^m}p_0(\boldsymbol{x})\|\nabla_{\boldsymbol{x}}\log p(\boldsymbol{x})\odot\boldsymbol{x}-\nabla_{\boldsymbol{x}}\log p_0(\boldsymbol{x})\odot\boldsymbol{x}\|_2^2\d\boldsymbol{x}.
\end{equation}
The element-wise multiplication (``$\odot$'') with $\boldsymbol{x}$
dampens discontinuities at the boundary of $\mathbb{R}_+^m$ and
facilitates integration by parts for deriving an
empirical loss that does not depend on the true $p_0$.

In recent work, we proposed a generalized score
matching approach for densities on $\mathbb{R}_+^m$ by using
the square-root of slowly growing and preferably bounded functions
$\boldsymbol{h}(\boldsymbol{x})$ in place of $\boldsymbol{x}$ in the
element-wise multiplication \citep{yu18,yus19}.  This modification improves performance (theoretically and
empirically) as it avoids higher moments in the empirical loss.  Another
recent work extended score matching to supports given by
bounded open subsets $\mathfrak{D}\subset\mathbb{R}^m$ with
piecewise smooth boundaries \citep{liu19}.  The idea there
is to minimize
\begin{equation}
  \label{eq:score:liu}
\sup_{g\in\mathcal{G}}\int_{\mathfrak{D}}p_0(\boldsymbol{x})g(\boldsymbol{x})\|\nabla_{\boldsymbol{x}}\log p(\boldsymbol{x})-\nabla_{\boldsymbol{x}}\log p_0(\boldsymbol{x})\|_2^2\d\boldsymbol{x},
\end{equation}
where
$\mathcal{G}\equiv\{g|g(\boldsymbol{x})=0\,\forall\boldsymbol{x}\in\partial
\mathfrak{D}\text{ and }g\text{ is 1-Lipschitz continuous}\}$, and
$\partial\mathfrak{D}$ is the boundary of $\mathfrak{D}$. The supremum
in the loss is achieved at
$g_0(\boldsymbol{x})\equiv\min_{\boldsymbol{x}'\in\partial\mathfrak{D}}\|\boldsymbol{x}-\boldsymbol{x}'\|$,
the distance of $\boldsymbol{x}$ to $\partial\mathfrak{D}$.

\subsection{A Unifying Framework for General Domains}

In this paper, we further extend generalized score matching
with the aim of avoiding the limitations of existing work and
allowing for general and possibly unbounded domains $\mathfrak{D}$
with positive Lebesgue measure.
We require merely that all sections of
$\mathfrak{D}\subseteq\mathbb{R}^m$, i.e., the sets of values of any
component $x_j$ fixing all other components $\boldsymbol{x}_{-j}$, are
countable disjoint unions of intervals in $\mathbb{R}$.  This level of
generality ought to cover all practical cases.  To handle
such domains, we compose the function $\boldsymbol{h}$ in the
generalized score matching loss of \citet{yu18,yus19} with a
component-wise distance function
$\boldsymbol{\varphi}=(\varphi_1,\ldots,\varphi_m):\mathfrak{D}\to\mathbb{R}_+^m$.
To define $\varphi_j(\boldsymbol{x})$, we consider the interval in the
section given by $\boldsymbol{x}_{-j}$ that contains $x_j$ and 
compute the distance between $x_j$ and the boundary of this interval.
The function
$\varphi_j(\boldsymbol{x})$ is then defined as
the minimum of the distance and a user-selected constant $C_j$.
The loss resulting from this extension, with the composition $\boldsymbol{h}\circ\boldsymbol{\varphi}$ in place of
$\boldsymbol{h}$,  can again be
approximated by an empirical loss that is quadratic in the canonical
parameters of exponential families.

As an application of the proposed framework, we study a class of
pairwise interaction models for an $m$-dimensional random vector
$\boldsymbol{X}=(X_i)_{i=1}^m$ that was considered in \citet{yu18,yus19} and in
special cases in earlier literature.  These \emph{$a$-$b$ models} postulate a
probability density function proportional to 
\begin{equation}\label{eq_interaction_density}
\exp\left\{-\frac{1}{2a}\boldsymbol{x}^a\mathbf{K}\boldsymbol{x}^a+\frac{1}{b}\boldsymbol{\eta}^{\top}\boldsymbol{x}^b\right\},
\qquad \boldsymbol{x}\in\mathfrak{D}.
\end{equation}
Where past work assumes $\mathfrak{D}=\mathbb{R}^m$ or
$\mathfrak{D}=\mathbb{R}_+^m$, we here allow a general domain
$\mathfrak{D}\subset\mathbb{R}^m$.  In~\eqref{eq_interaction_density},
$a\geq 0$ and $b\geq 0$ are known constants, and
$\mathbf{K}\in\mathbb{R}^{m\times m}$ and
$\boldsymbol{\eta}\in\mathbb{R}^m$ are unknown parameters to be
estimated. For $a=0$ we define
${\boldsymbol{x}^a}^{\top}\mathbf{K}\boldsymbol{x}^a/a\equiv(\log
\boldsymbol{x})^{\top}\mathbf{K}(\log \boldsymbol{x})$ and for $b=0$
we define $\boldsymbol{\eta}^{\top}\boldsymbol{x}^b/b\equiv
\boldsymbol{\eta}^{\top}(\log\boldsymbol{x})$. The case where $a=0$
was not considered in \citet{yu18,yus19}. This model class provides a
simple yet rich framework for pairwise interaction models.  In
particular, if
$\mathfrak{D}=\mathfrak{D}_1\times\cdots\times\mathfrak{D}_m$ is a
product set, then $X_i$ and $X_j$ are conditionally independent given
all others if and only if $\kappa_{ij}=\kappa_{ji}=0$ in the
interaction matrix $\mathbf{K}$; i.e., the $a$-$b$ models become
graphical models \citep{maa18}. When $a=b=1$, model
(\ref{eq_interaction_density}) is  a (truncated) Gaussian
graphical model, with $\boldsymbol{\Sigma}\equiv\mathbf{K}^{-1}$ the
covariance matrix and $\boldsymbol{\Sigma}^{-1}\boldsymbol{\eta}$ the
mean parameter. The case where $a=b=1/2$ with
$\mathfrak{D}=\mathbb{R}_+^m$ is the exponential square
root graphical model from \citet{ino16}. 

For estimation of a sparse interaction matrix $\mathbf{K}$ in 
high-dimensional $a$-$b$ models, we take up an $\ell_1$ regularization
approach considered in \citet{lin16} and improved in
\citet{yu18,yus19}.
In \citet{yu18,yus19}, we showed that this approach permits recovery
of the support of $\mathbf{K}$ under
sample complexity $n=\Omega(\log m)$ for Gaussians  truncated to
$\mathfrak{D}=\mathbb{R}_+^m$.  Here, we prove that the same sample
complexity is achieved for Gaussians truncated to any domain $\mathfrak{D}$ that is a
finite disjoint union of convex sets with $n=\Omega(\log m)$
samples. In addition, we derive similar results for general
$a$-$b$ models on bounded subsets of $\mathbb{R}_+^m$ with positive
measure for $a>0$, or if $\log\mathfrak{D}$ is bounded for $a=0$.
On unbounded domains for $a>0$ or for unbounded $\log\mathfrak{D}$ and
$a=0$, we require $n$ to be $\Omega(\log m)$ times a factor that may weakly depend on $m$.

\subsection{Organization of the Paper}

The rest of the paper is structured as follows. We provide the
necessary background on score matching in Section~\ref{Preliminaries}. In Section~\ref{Generalized Score Matching for General Domains}, we introduce and detail our new methodology, 
along with the regularized generalized 
estimator for exponential families. In Section~\ref{$a$-$b$ Models on Domains with Positive Measure}, we define the $a$-$b$ interaction models and focus on application of our method to these models on domains with positive Lebesgue measure. Theoretical results and numerical experiments are given in Sections~\ref{Theory} and \ref{Numerical Experiments}, respectively. We apply our method to a DNA methylation dataset in Section~\ref{DNA Methylation Data}. Longer proofs are included in the Appendix. An implementation that incorporates various types of domain $\mathfrak{D}$ is available in the \texttt{genscore} R package.

\subsection{Notation}
We use lower-case letters for constant scalars, vectors and functions
and upper-case letters for random scalars and vectors (except some
special cases). We reserve regular font for scalars (e.g.~$a$, $X$)
and boldface for vectors (e.g.~$\boldsymbol{a}$, $\boldsymbol{X}$),
and $\mathbf{1}_m=(1,\dots,1)\in\mathbb{R}^m$.  For two vectors
$\boldsymbol{u},\boldsymbol{v}\in\mathbb{R}^m$, we write
$\boldsymbol{u}\succ\boldsymbol{v}$ if $u_j>v_j$ for $j=1,\dots,m$.
Matrices are in upright bold, with constant matrices in upper-case ($\mathbf{K}$, $\mathbf{M}$) and random data matrices in lower-case ($\mathbf{x}$, $\mathbf{y}$). Superscripts index rows and subscripts index columns in a data matrix $\mathbf{x}$, so, $\boldsymbol{X}^{(i)}$ is the $i$-th row, and $X_j^{(i)}$ is its $j$-th feature.

For vectors $\boldsymbol{u},\boldsymbol{v}\in\mathbb{R}^m$, $\boldsymbol{u}\odot\boldsymbol{v}\equiv(u_1v_1,\ldots,u_mv_m)$ denotes the Hadamard product  (element-wise multiplication), and the $\ell_a$-norm for $a\geq 1$ is denoted $\|\boldsymbol{u}\|_a=(\sum_{j=1}^m|u_j|^a)^{1/a}$, with $\|\boldsymbol{u}\|_{\infty}=\max_{j=1,\ldots,m}|u_j|$. For $a\in\mathbb{R}$, let $\boldsymbol{v}^a\equiv (v_1^a,\ldots,v_m^a)$. Similarly, for function $\boldsymbol{f}:\mathbb{R}^m\to\mathbb{R}^m$, $\boldsymbol{x}\mapsto (f_1(\boldsymbol{x}),\ldots, f_m(\boldsymbol{x}))$, we write $\boldsymbol{f}^a(\boldsymbol{x})\equiv(f_1^a(\boldsymbol{x}),\ldots,f_m^a(\boldsymbol{x}))$. Similarly, we also write $\boldsymbol{f}'(\boldsymbol{x})\equiv(\partial f_1(\boldsymbol{x})/\partial x_1,\ldots,\partial f_m(\boldsymbol{x})/\partial x_m)$.

For a matrix $\mathbf{K}=[\kappa_{ij}]_{i,j}\in\mathbb{R}^{n\times m}$, its vectorization is obtained by stacking its columns into an $\mathbb{R}^{nm}$ vector. Its Frobenius norm is $\mnorm{\mathbf{K}}_{F}=\|\mathrm{vec}(\mathbf{K})\|_2$, its max norm is $\|\mathbf{K}\|_{\infty}\equiv\|\mathrm{vec}(\mathbf{K})\|_{\infty}\equiv\max_{i,j}|\kappa_{ij}|$, and its $\ell_a$--$\ell_b$ operator norm is $\mnorm{\mathbf{K}}_{a,b}\equiv\max_{\boldsymbol{x}\neq\boldsymbol{0}}\|\mathbf{K}\boldsymbol{x}\|_b/\|\boldsymbol{x}\|_a$, with $\mnorm{\mathbf{K}}_a\equiv\mnorm{\mathbf{K}}_{a,a}$.

For a vector $\boldsymbol{x}\in\mathbb{R}^m$ and an index
$j\in\{1,\dots,m\}$, we write $\boldsymbol{x}_{-j}$ for the subvector
that has the $j$th component removed.  For a function $f$ of a vector
$\boldsymbol{x}$, we may also write $f(x_j;\boldsymbol{x}_{-j})$ to
stress the dependency on $x_j$, especially when $\boldsymbol{x}_{-j}$
is fixed and only $x_j$ is varied, and write $\partial_j
f(\boldsymbol{x})=\left.\partial_j f(y;\boldsymbol{x}_{-j})/\partial
  y\right|_{y\equiv x_j}$. For two compatible functions $f$ and $g$,
$f\circ g$ denotes their function composition. Unless otherwise noted, the considered probability density functions
are densities with respect to the Lebesgue measure on $\mathbb{R}^m$.

\section{Preliminaries}\label{Preliminaries}

Suppose $\boldsymbol{X}\in\mathbb{R}^m$ is a random vector with distribution function $P_0$ supported on domain $\mathfrak{D}\subseteq\mathbb{R}^m$ and a twice continuously differentiable probability density function $p_0$ with respect to the Lebesgue measure restricted to $\mathfrak{D}$. Let $\mathcal{P}(\mathfrak{D})$ be a family of distributions of interest with twice continuously differentiable densities on $\mathfrak{D}$. 
The goal is to estimate $p_0$ by picking the distribution $P$ from $\mathcal{P}(\mathfrak{D})$ with density $p$ minimizing an empirical loss that measures the distance between $p$ and $p_0$.

\subsection{Original Score Matching on $\mathbb{R}^m$} 
The original \emph{score matching} loss proposed by \citet{hyv05} for $\mathfrak{D}\equiv\mathbb{R}^m$ is given by
\begin{equation*}
J_{\mathbb{R}^m}(P)\equiv\frac{1}{2}\int_{\mathbb{R}^m}p_0(\boldsymbol{x})\|\nabla\log p(\boldsymbol{x})-\nabla\log p_0(\boldsymbol{x})\|_2^2\d\boldsymbol{x},
\end{equation*}
in which the gradients can be thought of as gradients with respect to a hypothetical location parameter and evaluated at the origin \citep{hyv05}. The log densities enable estimation without calculating the normalizing constants of $p$ and $p_0$. Under mild conditions, using integration by parts, the loss can be rewritten as
\[J_{\mathbb{R}^m}(P)\equiv\int_{\mathbb{R}^m}p_0(\boldsymbol{x})\sum_{j=1}^m\left[\partial_{jj}\log p(\boldsymbol{x})+\frac{1}{2}\left(\partial_j\log p(\boldsymbol{x})\right)^2\right]\d\boldsymbol{x}\]
plus a constant independent of $p$. One can thus use a sample average to approximate the loss without knowing the true density $p_0$.

\subsection{Score Matching on $\mathbb{R}_+^m$} 
\label{Score Matching on R+m}

Consider $\mathfrak{D}\equiv\mathbb{R}_+^m$. Let $\boldsymbol{h}:\mathbb{R}_+^m\to\mathbb{R}_+^m,\,\boldsymbol{x}\mapsto(h_1(x_1),\ldots,h_m(x_m))^{\top}$, where $h_1,\ldots,h_m:\mathbb{R}_+\to\mathbb{R}_+$ are almost surely positive functions that are absolutely continuous in every bounded sub-interval of $\mathbb{R}_+$. The \emph{generalized $\boldsymbol{h}$-score matching loss} proposed by \citet{yu18,yus19} is
\begin{equation}\label{eq_loss_yu18}
J_{\boldsymbol{h},\mathbb{R}_+^m}(P)\equiv\frac{1}{2}\int_{\mathbb{R}_+^m}p_0(\boldsymbol{x})\left\|\nabla\log p(\boldsymbol{x})\odot  \boldsymbol{h}^{1/2}(\boldsymbol{x})-\nabla\log p_0(\boldsymbol{x})\odot \boldsymbol{h}^{1/2}(\boldsymbol{x})\right\|_2^2\d\boldsymbol{x}.
\end{equation}
The score matching loss for $\mathbb{R}_+^m$ originally proposed by \citet{hyv07} is a special case of (\ref{eq_loss_yu18}) with $\boldsymbol{h}(\boldsymbol{x})=\boldsymbol{x}^2$. In \citet{yu18,yus19} we proved that by choosing slowly growing and preferably bounded $h_1,\ldots,h_m$, the estimation efficiency can be significantly improved. Under assumptions that for all $P\in\mathcal{P}(\mathbb{R}_+^m)$ with density $p$,
\begin{enumerate}[label=(A0.\arabic*)]
\item $p_0(x_j;\boldsymbol{x}_{-j})h_j(x_j)\partial_j\log p(x_j;\boldsymbol{x}_{-j})\left|_{x_j\searrow 0^+}^{x_j\nearrow+\infty}\right.=0,\quad\forall \boldsymbol{x}_{-j}\in\mathbb{R}_+^{m-1}$ $\forall j$;
\item $\mathbb{E}_{p_0}\left\|\nabla\log p(\boldsymbol{X})\odot \boldsymbol{h}^{1/2}(\boldsymbol{X})\right\|_2^2<+\infty$,\quad $\mathbb{E}_{p_0}\left\|\left(\nabla\log p(\boldsymbol{X})\odot\boldsymbol{h}(\boldsymbol{X})\right)'\right\|_1<+\infty$,
\end{enumerate}
where $f(\boldsymbol{x})\left|^{x_j\nearrow+\infty}_{x_j\searrow 0^+}\right.\equiv\lim_{x_j\nearrow +\infty}f(\boldsymbol{x})-\lim_{x_j\searrow 0^+}f(\boldsymbol{x})$, the loss (\ref{eq_loss_yu18}) can be rewritten as
\begin{multline*}
J_{\boldsymbol{h},\mathbb{R}_+^m}(P)\equiv\int_{\mathbb{R}_+^m}p_0(\boldsymbol{x})\sum_{j=1}^{m}\left[h_j'(x_j)\partial_j(\log p(\boldsymbol{x}))+h_j(x_j)\partial_{jj}(\log p(\boldsymbol{x}))+\frac{1}{2}h_j(x_j)[\partial_j(\log p(\boldsymbol{x}))]^2\right]\d\boldsymbol{x}
\end{multline*}
plus a constant independent of $p$. One can thus estimate $p_0$ by minimizing the empirical loss 
$J_{\boldsymbol{h},\mathbb{R}_+^m}(P)$.

\subsection{Score Matching on Bounded Open Subsets of
  $\mathbb{R}^m$}
\label{sec_intro_liu19}

The method proposed in \citet{liu19} 
estimates a density $p_0$ on a bounded open subset $\mathfrak{D}\subset\mathbb{R}^m$ with a piecewise smooth boundary $\partial\mathfrak{D}$ by minimizing the following ``maximally weighted score matching'' loss
\begin{equation}\label{eq_loss_g0}
J_{g_0,\mathfrak{D}}(P)\equiv\sup_{g\in\mathcal{G}}\frac{1}{2}\int_{\mathbb{R}_+^m}g(\boldsymbol{x})p_0(\boldsymbol{x})\left\|\nabla\log p(\boldsymbol{x})-\nabla\log p_0(\boldsymbol{x})\right\|_2^2\d\boldsymbol{x}.
\end{equation}
with $\mathcal{G}\equiv\{g|g(\boldsymbol{x})=0,\forall\boldsymbol{x}\in\partial\mathfrak{D}\text{ and }g\text{ is }L\text{-Lipschitz continuous}\}$ for some constant $L>0$. The authors show that the maximum is obtained with $g_0(\boldsymbol{x})\equiv L\cdot\inf_{\boldsymbol{x}'\in\partial\mathfrak{D}}\|\boldsymbol{x}-\boldsymbol{x}'\|_2$, i.e.~the $\ell_2$ distance of $\boldsymbol{x}$ to the boundary of $\mathfrak{D}$; using integration by parts similar to the previous methods, (\ref{eq_loss_g0}) can be estimated using the empirical loss which can be calculated with a closed form.

\section{Generalized Score Matching for General Domains}\label{Generalized Score Matching for General Domains}

\subsection{Assumption on the Domain}

For $\boldsymbol{x}\in\mathbb{R}^m$ and any index $j=1,\ldots,m$,
write
$\mathfrak{C}_{j,\mathfrak{D}}\left(\boldsymbol{x}_{-j}\right)\equiv\{y\in\mathbb{R}:(y;\boldsymbol{x}_{-j})\in\mathfrak{D}\}$
for the section of $\mathfrak{D}$ obtained by fixing the coordinates
in $\boldsymbol{x}_{-j}$.  This $j$th section is the projection of the
intersection between $\mathfrak{D}$ and the line
$\{(y;\boldsymbol{x}_{-j}):y\in\mathbb{R}\}$.  A non-empty $j$th
section is obtained from the vectors $\boldsymbol{x}_{-j}$ in the set
$\mathfrak{S}_{-j,\mathfrak{D}}\equiv\left\{\boldsymbol{x}_{-j}:\mathfrak{C}_{j,\mathfrak{D}}\left(\boldsymbol{x}_{-j}\right)\neq\varnothing\right\}\subset\mathbb{R}^{m-1}$.
For notational simplicity, we 
drop their dependency 
on $\mathfrak{D}$.

\begin{definition}\label{def_V}
We say that a domain $\mathfrak{D}\subseteq\mathbb{R}^m$ is a
\emph{component-wise countable union of intervals} if it is measurable, and for any index $j=1,\ldots,m$ and any
  $\boldsymbol{x}_{-j}\in\mathfrak{S}_{-j,\mathfrak{D}}$, the section
  $\mathfrak{C}_{j,\mathfrak{D}}\left(\boldsymbol{x}_{-j}\right)$ is a
  countable union of disjoint intervals, meaning that
  \begin{equation}\label{eq_C_union}
    \mathfrak{C}_{j,\mathfrak{D}}\left(\boldsymbol{x}_{-j}\right)\equiv\Uu_{k=1}^{K_j(\boldsymbol{x}_{-j})} I_k(\boldsymbol{x}_{-j}),
  \end{equation}
  where $K_j(\boldsymbol{x}_{-j})\in\mathbb{N}\cup\{\infty\}$, and
  each set $I_k(\boldsymbol{x}_{-j})$ is an interval (closed, open, or half-open) with endpoints
  $-\infty\leq a_{k,j}(\boldsymbol{x}_{-j})\le
  b_{k,j}(\boldsymbol{x}_{-j})\leq +\infty$, with the
    $I_k(\boldsymbol{x}_{-j})$'s being the connected components of
    $\mathfrak{C}_{j,\mathfrak{D}}\left(\boldsymbol{x}_{-j}\right)$.
    The last point rules out constructions like $I_1=(0,1]$ and
    $I_2=(1,2]$ but allows $I_1=(0,1)$ and $I_2=(1,2]$.
  We define the \emph{component-wise} boundary set of such a component-wise countable union
  of intervals as
\begin{equation}\label{def_V_boundary}
\underline{\partial}\mathfrak{D}\equiv\left\{\boldsymbol{x}\in\mathbb{R}^m:\exists j=1,\ldots,m,\,\boldsymbol{x}_{-j}\in\mathfrak{S}_{-j,\mathfrak{D}},\,x_j\in\Uu_{k=1}^{K_j(\boldsymbol{x}_{-j})}\{a_{k,j}(\boldsymbol{x}_{-j}),b_{k,j}(\boldsymbol{x}_{-j})\}\backslash\{\pm\infty\}\right\}.
\end{equation}
\end{definition}


\subsection{Generalized Score Matching Loss for General Domains}\label{Methodology}
We first define a \emph{truncated component-wise distance}, which is
based on distances within connected components of sections.
\begin{definition}
  Let $\boldsymbol{C}=(C_1,\dots,C_m)$ be comprised of positive
  constants, so $\boldsymbol{C}\succ\boldsymbol{0}$.  Let
  $\mathfrak{D}\subseteq\mathbb{R}^m$ be a non-empty component-wise
  countable union of intervals whose sections are presented as in~\eqref{eq_C_union}.  For any
  vector $\boldsymbol{x}\in\mathfrak{D}$, define the \emph{truncated
    component-wise distance of $\boldsymbol{x}$ to the boundary of
    $\mathfrak{D}$} as
\begin{align}
\boldsymbol{\varphi}_{\boldsymbol{C},\mathfrak{D}}(\boldsymbol{x})&\equiv\left(\varphi_{C_1,\mathfrak{D},1}(\boldsymbol{x}),\ldots,\varphi_{C_m,\mathfrak{D},m}(\boldsymbol{x})\right)\in\mathbb{R}_+^m,\label{def_dist}\\
\varphi_{C_j,\mathfrak{D},j}(\boldsymbol{x})&\equiv \begin{cases}
C_j, & a_{k,j}=-\infty,\,b_{k,j}=+\infty,\\
\min(C_j,b_{k,j}-x_j),&a_{k,j}=-\infty,\,x_j\leq b_{k,j}<+\infty,\\
\min(C_j,x_j-a_{k,j},b_{k,j}-x_j),&-\infty<a_{k,j}\leq x_j\leq b_{k,j}<+\infty,\\
\min(C_j,x_j-a_{k,j}),&-\infty<a_{k,j}\leq x_j,\,b_{k,j}=+\infty,
\end{cases}
\end{align}
where $k$ is the index for which $x_j\in I_k(\boldsymbol{x}_{-j})$ and
$a_{k,j}\leq b_{k,j}$ are the endpoints of $I_k(\boldsymbol{x}_{-j})$.
\end{definition}

Our idea for defining a score matching loss suitable for general
domains is now to use the generalized score matching framework from
(\ref{eq_loss_yu18}) but apply the function $\boldsymbol{h}$ to $\boldsymbol{\varphi}_{\boldsymbol{C},\mathfrak{D}}(\boldsymbol{x})$ instead of to $\boldsymbol{x}$.


\begin{definition}
  \label{def:loss}
Suppose the true distribution $P_0$ has a twice continuously
differentiable density $p_0$ supported on
$\mathfrak{D}\subseteq\mathbb{R}^m$, a non-empty component-wise
countable union of intervals. 
Given positive constants $\boldsymbol{C}\succ\boldsymbol{0}$, and
$\boldsymbol{h}:\mathbb{R}_+^m\to\mathbb{R}_+^m$,
$\boldsymbol{y}\mapsto(h_1(y_1),\ldots,h_m(y_m))$ with
$h_1,\ldots,h_m:\mathbb{R}_+\to\mathbb{R}_+$, the \emph{generalized
  ($\boldsymbol{h},\boldsymbol{C},\mathfrak{D}$)-score matching loss} for $P\in\mathcal{P}(\mathfrak{D})$ with density $p$ is defined as
\begin{multline}\label{def_loss}
  J_{\boldsymbol{h},\boldsymbol{C},\mathfrak{D}}(P)\equiv\\
  \frac{1}{2}\int_{\mathfrak{D}}p_0(\boldsymbol{x})\left\|\nabla\log p(\boldsymbol{x})\odot \left(\boldsymbol{h}\circ\boldsymbol{\varphi}_{\boldsymbol{C},\mathfrak{D}}\right)^{1/2}(\boldsymbol{x})-\nabla\log p_0(\boldsymbol{x})\odot \left(\boldsymbol{h}\circ\boldsymbol{\varphi}_{\boldsymbol{C},\mathfrak{D}}\right)^{1/2}(\boldsymbol{x})\right\|_2^2\d\boldsymbol{x}.
\end{multline}
\end{definition}


In (\ref{def_loss}), we apply the loss from (\ref{eq_loss_yu18}) with
the choice
$(\boldsymbol{h}\circ\boldsymbol{\varphi}_{\boldsymbol{C},\mathfrak{D}}$)
in place of $\boldsymbol{h}$.  The function
$\boldsymbol{\varphi}_{\boldsymbol{C},\mathfrak{D}}$ transforms a
point $\boldsymbol{x}\in\mathfrak{D}$ into the component-wise distance
vector in $\mathbb{R}_+^m$. The loss from Definition~\ref{def:loss} is thus a natural extension of our work in
\citet{yu18,yus19}, with the appeal that
$\boldsymbol{\varphi}_{\boldsymbol{C},\mathfrak{D}}$ usually has a
closed-form solution and can be computed efficiently.  For
$\mathfrak{D}=\mathbb{R}_+^m$ and
$\boldsymbol{C}=(+\infty,\dots,+\infty)$ it holds that
$\boldsymbol{\varphi}_{\boldsymbol{C},\mathbb{R}_+^m}(\boldsymbol{x})=\boldsymbol{x}$,
and the generalized score matching loss from (\ref{eq_loss_yu18})
becomes a special case of~\eqref{def_loss}. In
\citet{yu18,yus19}, we suggested taking the components of
$\boldsymbol{h}$ as bounded functions, which may now also be
incorporated via finite truncation points $\boldsymbol{C}$ for
$\boldsymbol{\varphi}$.  If
$\boldsymbol{h}(\boldsymbol{x})=\boldsymbol{x}^2$,
$\boldsymbol{C}=(+\infty,\dots,+\infty)$ and $\mathfrak{D}\equiv\mathbb{R}_+^m$,
then
$(\boldsymbol{h}\circ\boldsymbol{\varphi}_{\boldsymbol{C},\mathfrak{D}})^{1/2}(\boldsymbol{x})\equiv
\boldsymbol{x}$ gives the estimator in \citet{hyv07,lin16};
see~\eqref{eq:score:hyv:nonneg}. When choosing
$\boldsymbol{h}(\boldsymbol{x})=\mathbf{1}_m$, i.e., constant one, we
have
$(\boldsymbol{h}\circ\boldsymbol{\varphi}_{\boldsymbol{C},\mathfrak{D}})^{1/2}(\boldsymbol{x})\equiv
\mathbf{1}_m$ and recover the original score-matching for
$\mathbb{R}^m$ from \citet{hyv05}; see~\eqref{eq:score:hyv:original}.

For a bounded domain $\mathfrak{D}$, our approach is different from
directly using a distance in $\mathbb{R}^m$ as proposed in
\citet{liu19}; recall~\eqref{eq:score:liu} with the optimizer
$g_0(\boldsymbol{x})$ being the $\ell_2$ distance of $\boldsymbol{x}$ to
the usual boundary of $\mathfrak{D}$.  Instead, we decompose the
distance for each component and apply an extra transformation via the
function $\boldsymbol{h}$.

Figure \ref{plot_phi_l2} illustrates the case of the 2-d unit disk given by $x_1^2+x_2^2<
1$.  
While the function from
\citet{liu19} is
$g_0(\boldsymbol{x})=1-\sqrt{x_1^2+x_2^2}$, our method uses
$\varphi_1(\boldsymbol{x})=\sqrt{1-x_2^2}-|x_1|$ and
$\varphi_2(\boldsymbol{x})=\sqrt{1-x_1^2}-|x_2|$, assuming that
$\boldsymbol{C}=(C_1,C_2)$ has $C_1,C_2\ge 1$.   In Figure
\ref{plot_phi_l2_nn} we consider the 2-d unit disk
restricted to $\mathbb{R}_+^2$, 
where $\varphi_1(\boldsymbol{x})=\min\Big\{x_1,\sqrt{1-x_2^2}-x_1\Big\}$, $\varphi_2(\boldsymbol{x})=\min\Big\{x_2,\sqrt{1-x_1^2}-x_2\Big\}$, $g_0(\boldsymbol{x})=\min\Big\{x_1,x_2,1-\sqrt{x_1^2+x_2^2}\Big\}$.

\begin{figure}[t]
\centering
\subfloat[$g_0$]
{\includegraphics[width=0.24\textwidth]{{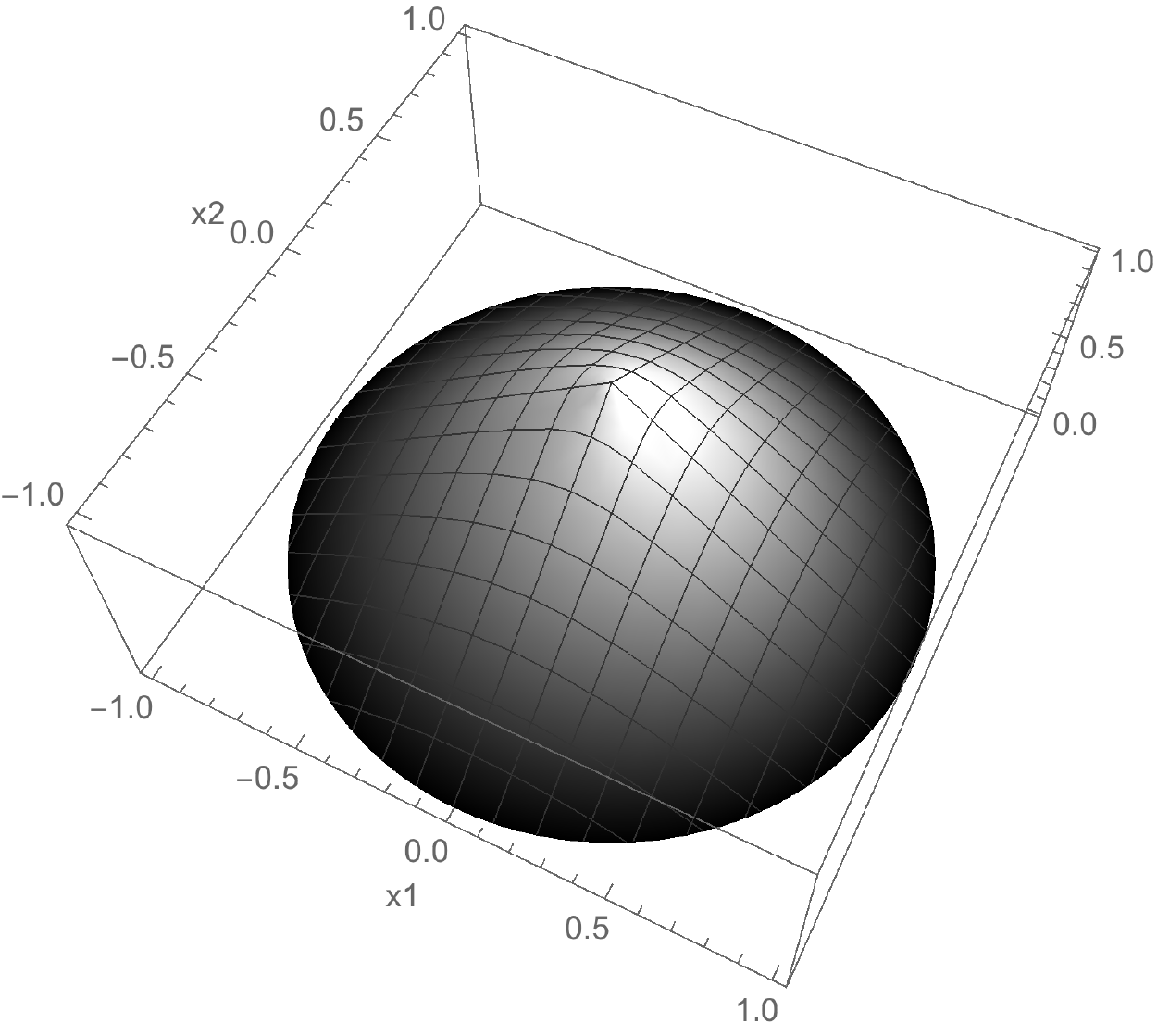}}\hspace{-0.02in}}
\subfloat[$\varphi_{\mathbf{1}_2,\mathfrak{D},1} \vee \varphi_{\mathbf{1}_2,\mathfrak{D},2}$]
{\includegraphics[width=0.24\textwidth]{{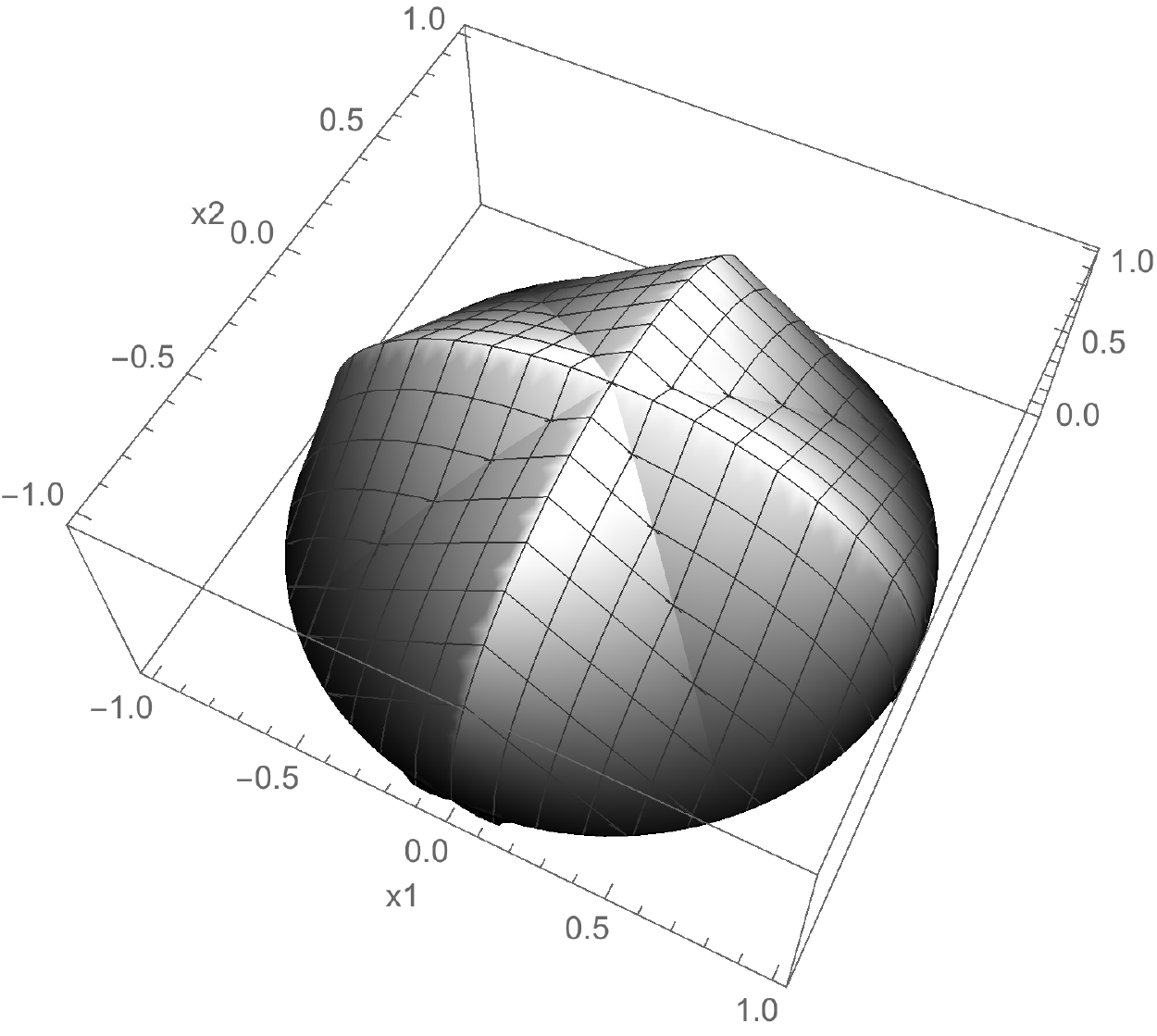}}\hspace{-0.02in}}
\subfloat[$\varphi_{\mathbf{1}_2,\mathfrak{D},1}$]
{\includegraphics[width=0.24\textwidth]{{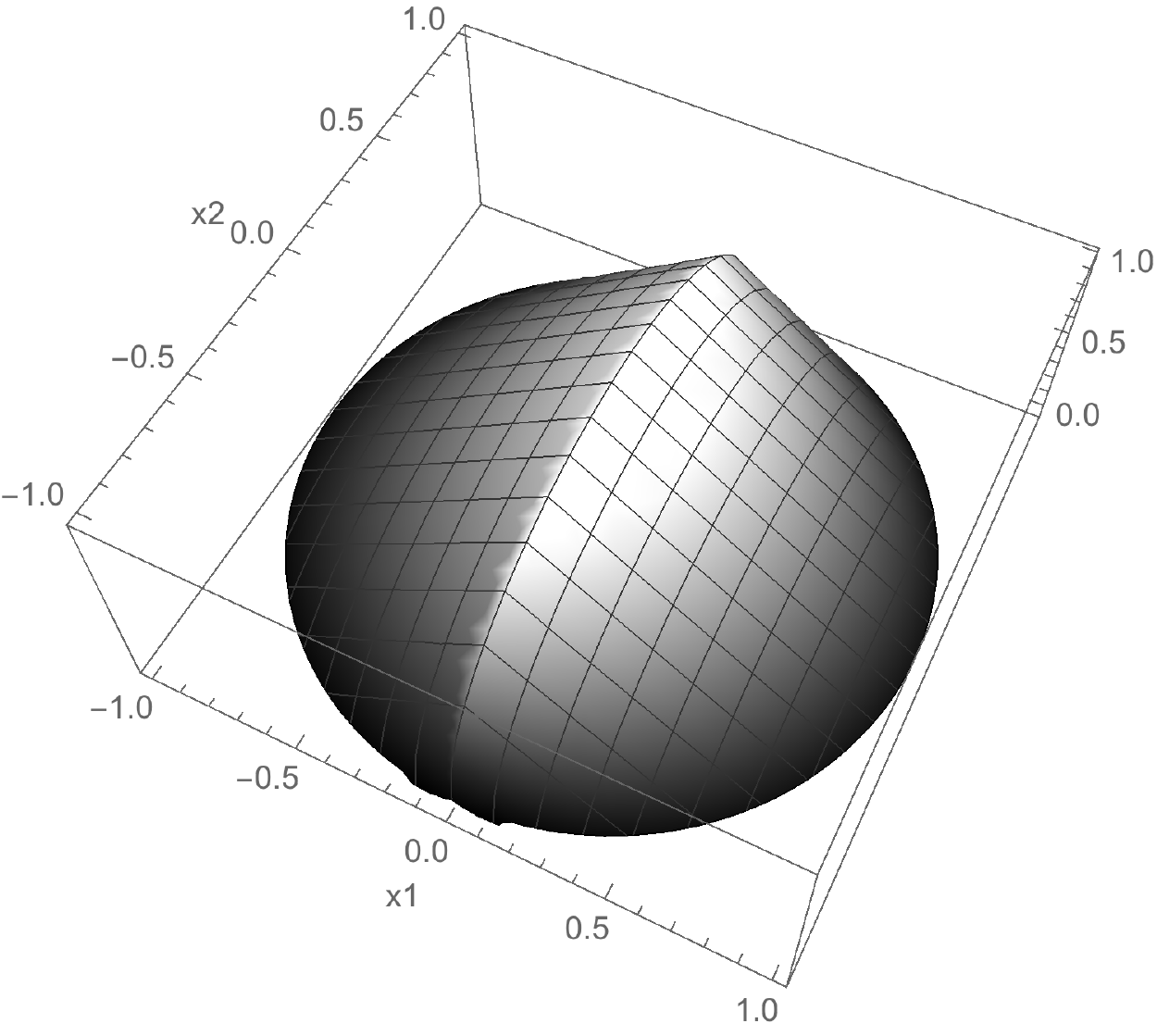}}\hspace{-0.02in}}
\subfloat[$\varphi_{\mathbf{1}_2,\mathfrak{D},2}$]
{\includegraphics[width=0.24\textwidth]{{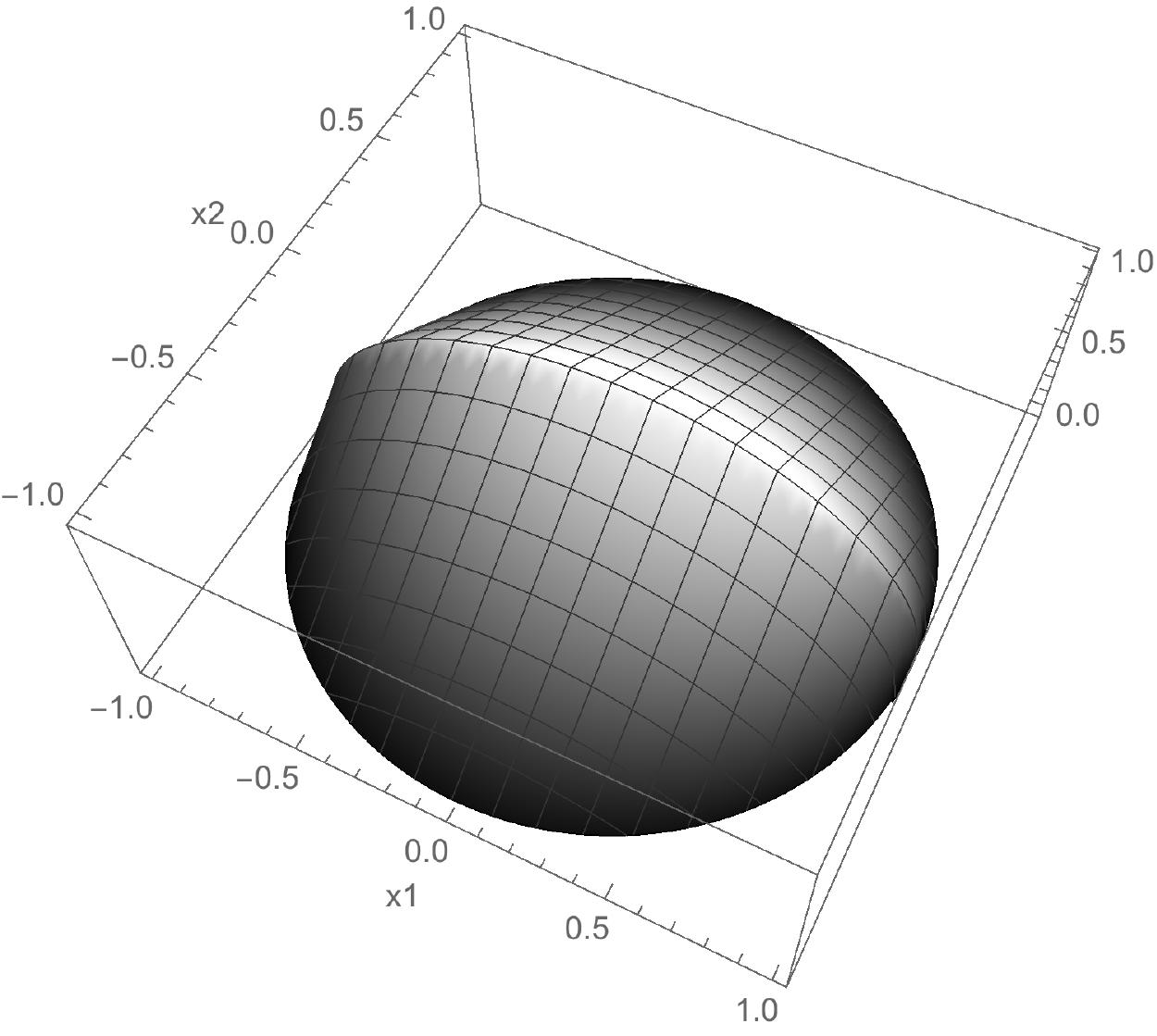}}\hspace{-0.02in}}
\caption{Comparison of $g_0$, $\varphi_{\mathbf{1}_2,\mathfrak{D},1}$ and $\varphi_{\mathbf{1}_2,\mathfrak{D},2}$ on $\mathfrak{D}\equiv\{\boldsymbol{x}\in\mathbb{R}^2:\|\boldsymbol{x}\|_2<1\}$.}\label{plot_phi_l2}
\end{figure}

\begin{figure}[t]
\centering
\subfloat[$g_0$]
{\includegraphics[width=0.24\textwidth]{{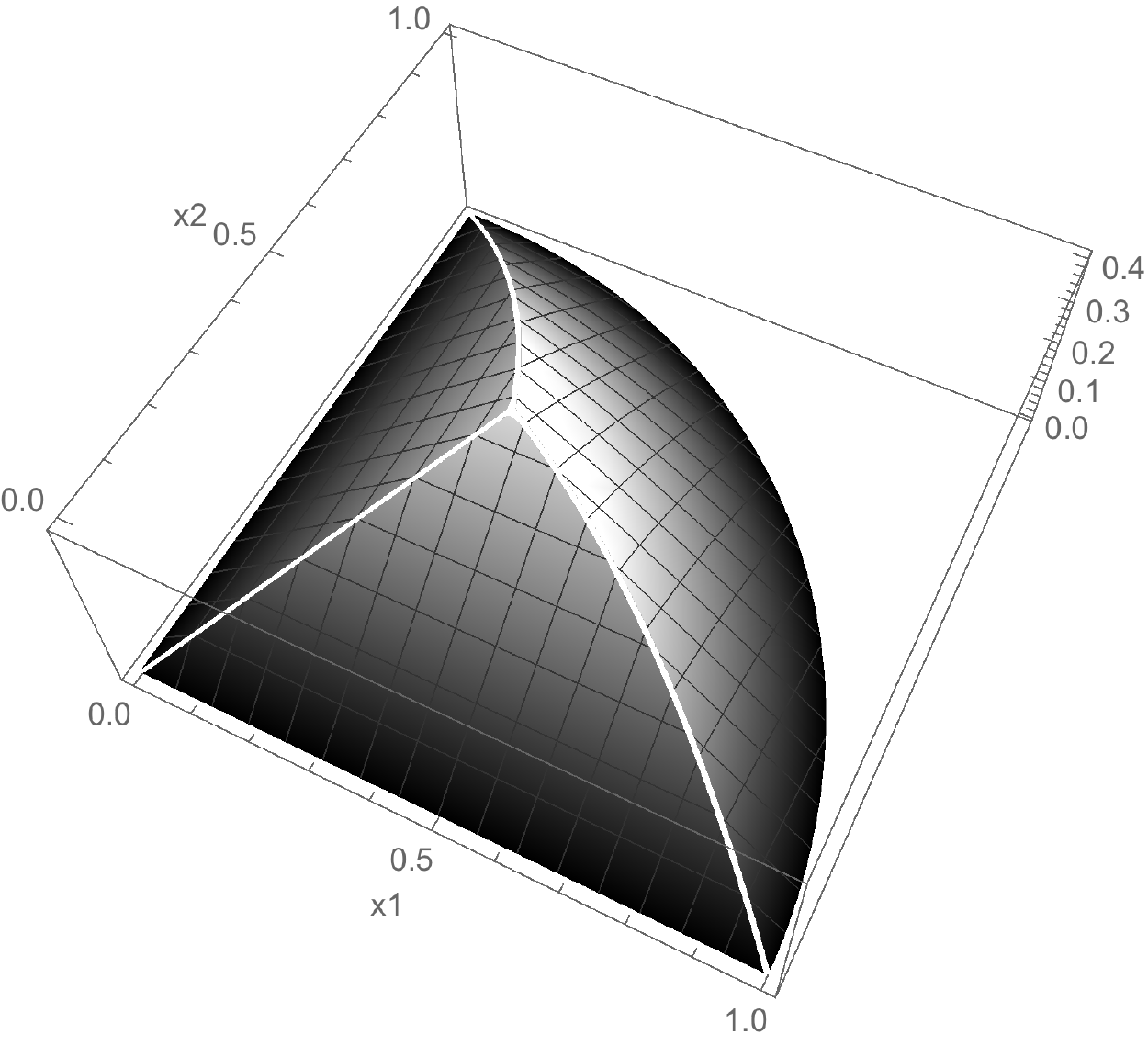}}\hspace{-0.02in}}
\subfloat[$\varphi_{\mathbf{1}_2,\mathfrak{D},1} \vee \varphi_{\mathbf{1}_2,\mathfrak{D},2}$]
{\includegraphics[width=0.24\textwidth]{{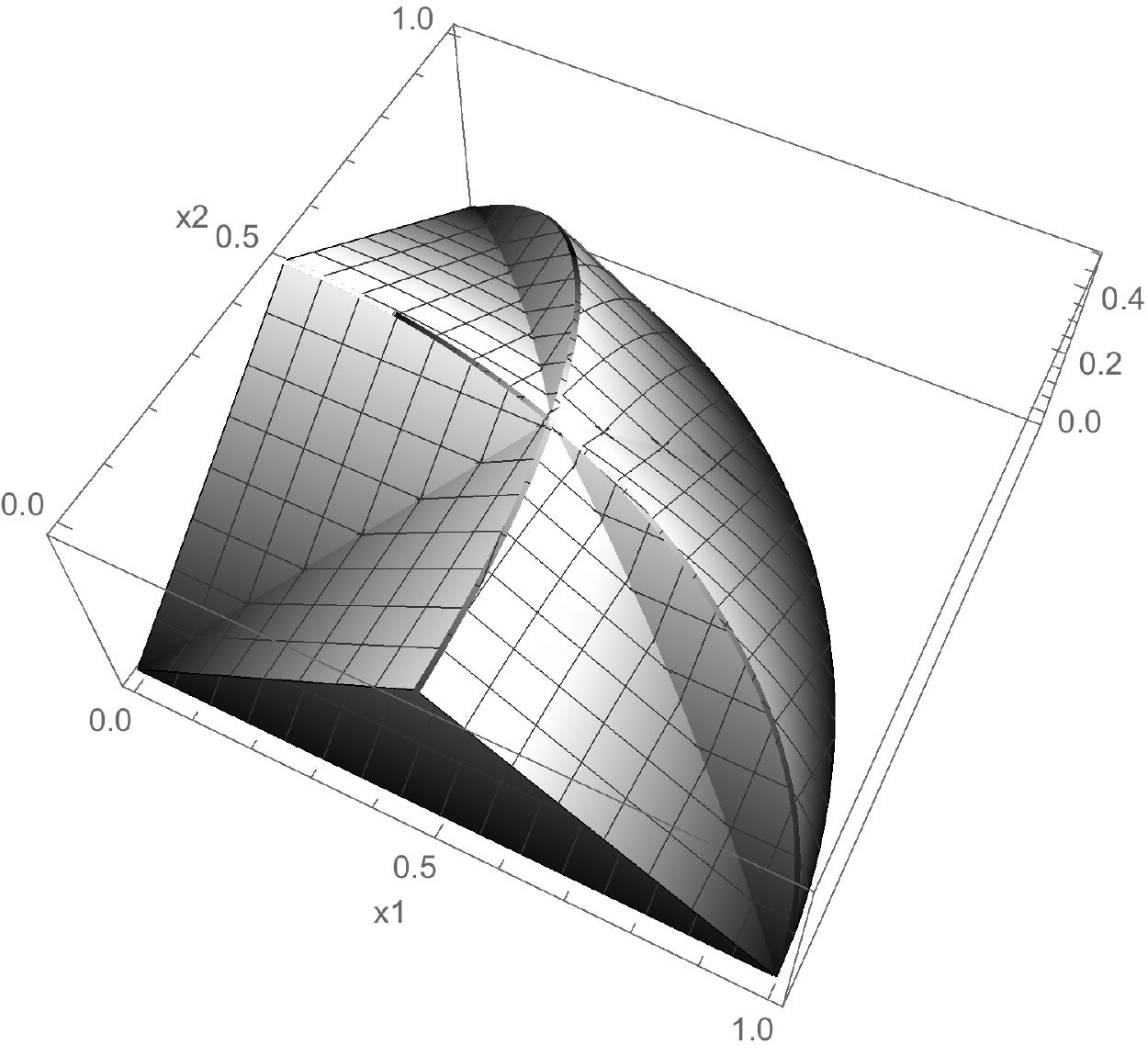}}\hspace{-0.02in}}
\subfloat[$\varphi_{\mathbf{1}_2,\mathfrak{D},1}$]
{\includegraphics[width=0.24\textwidth]{{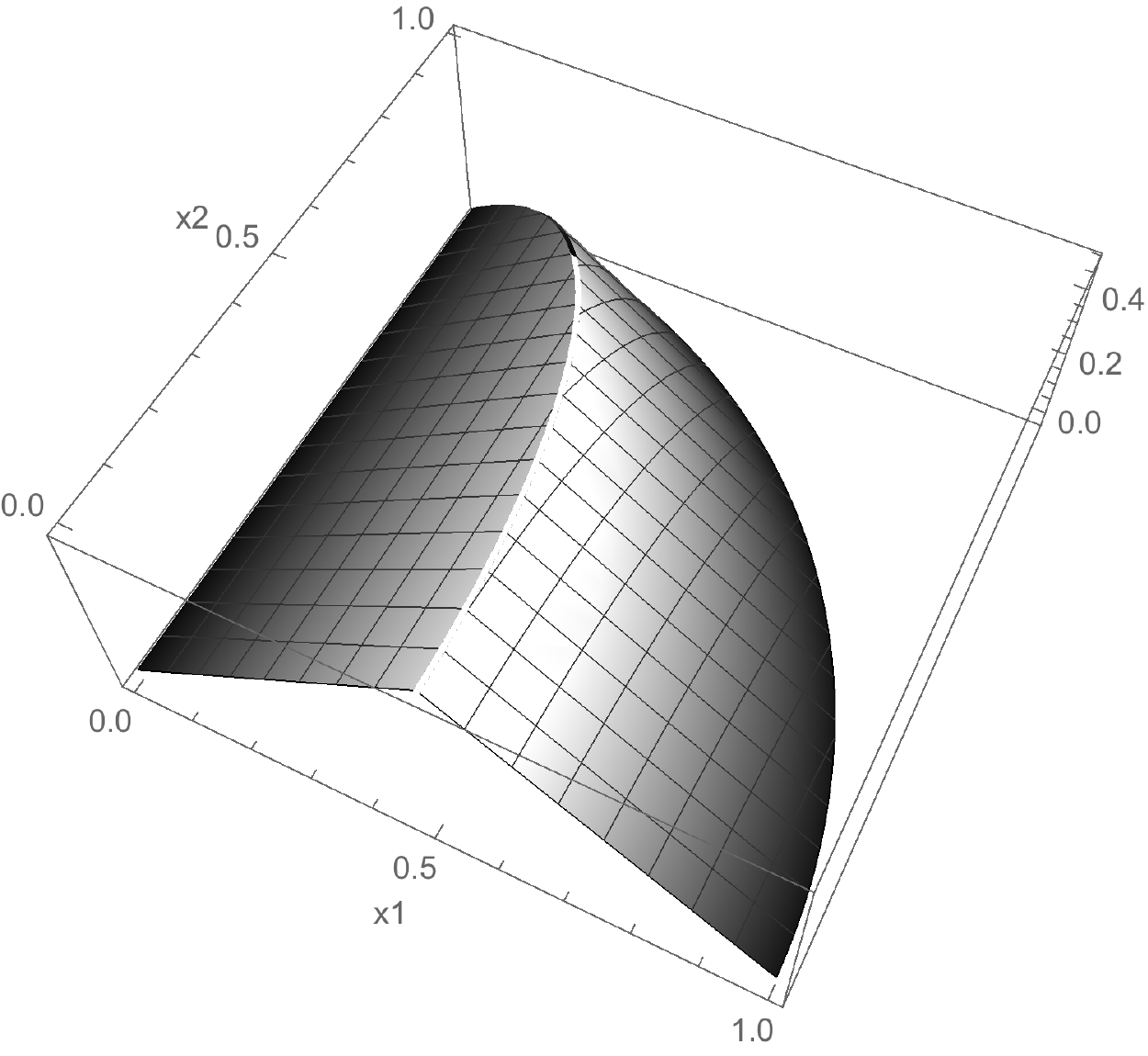}}\hspace{-0.02in}}
\subfloat[$\varphi_{\mathbf{1}_2,\mathfrak{D},2}$]
{\includegraphics[width=0.24\textwidth]{{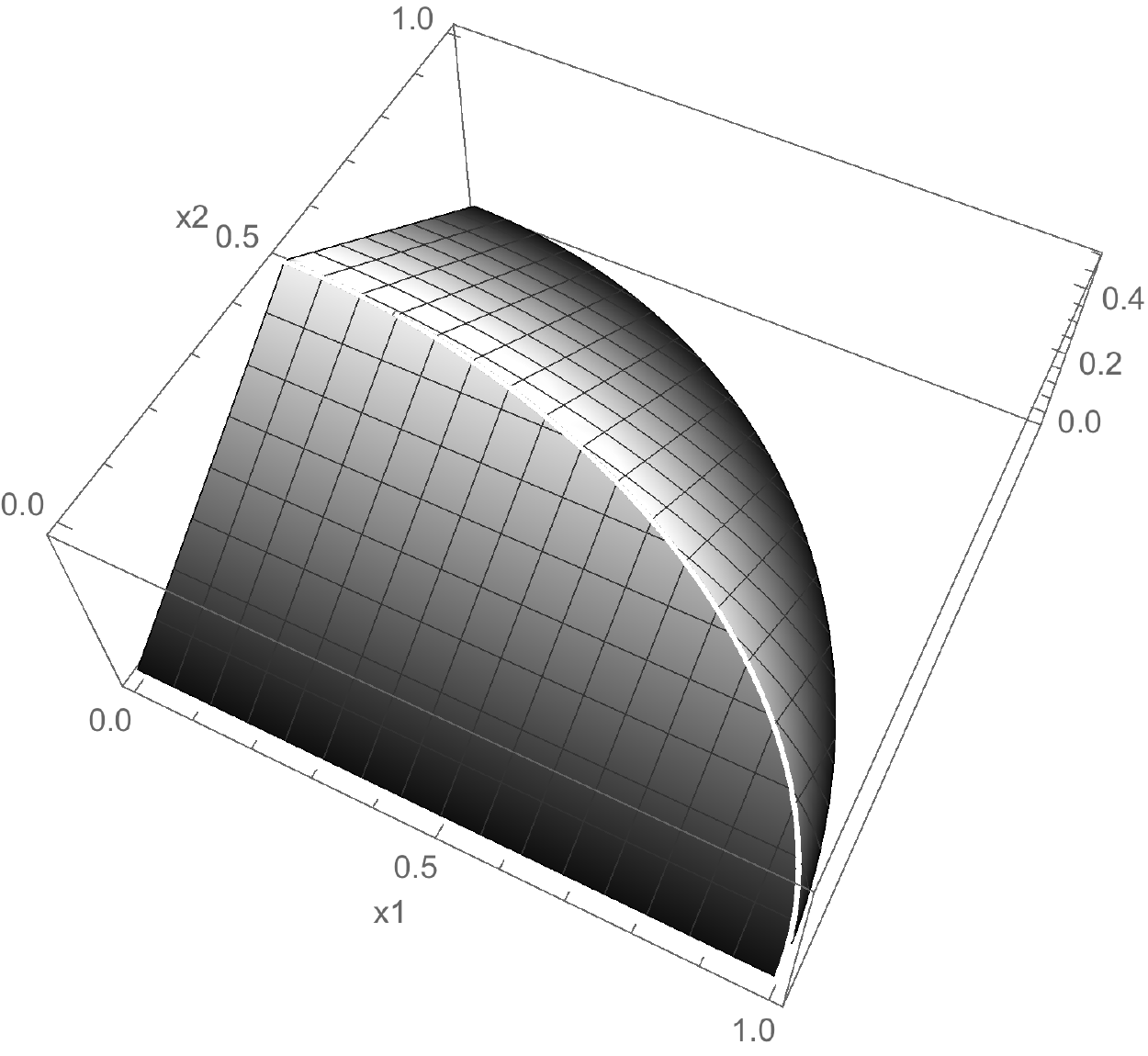}}\hspace{-0.02in}}
\caption{Comparison of $g_0$, $\varphi_{\mathbf{1}_2,\mathfrak{D},1}$ and $\varphi_{\mathbf{1}_2,\mathfrak{D},2}$ on $\mathfrak{D}\equiv\{\boldsymbol{x}\in\mathbb{R}_+^2:\|\boldsymbol{x}\|_2<1\}$.}\label{plot_phi_l2_nn}
\end{figure}

\subsection{Examples of component-wise distances} 
We give the form of the component-wise distance $\boldsymbol{\varphi}_{\boldsymbol{C},\mathfrak{D}}$ for
different examples of domains.

\begin{example}[$\mathbb{R}^m$ and $\mathbb{R}_+^m$]
Two frequently encountered domains are the real space and its
nonnegative orthant.  These are the original settings considered in
\citet{hyv05,hyv07,lin16,yu18,yus19}. We have $\boldsymbol{\varphi}_{\boldsymbol{C},\mathbb{R}^m}(\boldsymbol{x})=\boldsymbol{C}$ and $\boldsymbol{\varphi}_{\boldsymbol{C},\mathbb{R}_+^m}(\boldsymbol{x})=(\min(C_1,x_1),\ldots,\min(C_m,x_m))$.
\end{example}

\begin{example}[Unit hypercube]
  Now consider the unit hypercube $\mathfrak{D}=[-1/2,1/2]^m$
  as an example of a compact set and encountered  in applications like
  \citep{jan15,jan18,tan19}. Every non-empty section
  $\mathfrak{C}_{j,\mathfrak{D}}\left(\boldsymbol{x}_{-j}\right)$
  equals $[-1/2,1/2]$, 
  and so $\varphi_{C_j,\mathfrak{D},j}(\boldsymbol{x})=\min\{C_j,1/2-|x_j|\}$. Since the hypercube is bounded  by nature, it is natural to drop the truncation by $C_j$ and simply use $\boldsymbol{\varphi}_{\mathfrak{D}}(\boldsymbol{x})=\mathbf{1}_m/2-|\boldsymbol{x}|$.
\end{example}

\begin{example}[$\mathcal{L}^q$ ball]\label{Lq Balls}
As a compact domain that is difficult for previously
proposed approaches, consider the $\mathcal{L}^q$ ball with radius
$r>0$ and $q\geq 1$, so
$\mathfrak{D}\equiv\{\boldsymbol{x}\in\mathbb{R}^m:\|\boldsymbol{x}\|_q\leq
r\}$. Given a point $\boldsymbol{x}\in\mathfrak{D}$ and an index $j\in\{1,\ldots,m\}$, the section
$\mathfrak{C}_{j,\mathfrak{D}}\left(\boldsymbol{x}_{-j}\right)$
is the interval $[-(r^q-\mathbf{1}^{\top}|\boldsymbol{x}_{-j}|^q)^{1/q},(r^q-\mathbf{1}^{\top}|\boldsymbol{x}_{-j}|^q)^{1/q}]$, and so $\varphi_{C_j,\mathfrak{D},j}(\boldsymbol{x})=\min\left\{C_j,(r^q-\mathbf{1}^{\top}|\boldsymbol{x}_{-j}|^q)^{1/q}-|x_j|\right\}$.
\end{example}

\begin{example}[$\mathcal{L}^q$ ball restricted to $\mathbb{R}_+^m$]\label{Lq Balls on R+}
Further restricted, consider $\mathfrak{D}\equiv\{\boldsymbol{x}\in\mathbb{R}_+^m:\|\boldsymbol{x}\|_q\leq
r\}$, the nonnegative part of the $\mathcal{L}^q$ ball with radius $r>0$ and
$q\geq 1$. Given a point $\boldsymbol{x}\in\mathfrak{D}$ and an index $j$,
the section
$\mathfrak{C}_{j,\mathfrak{D}}\left(\boldsymbol{x}_{-j}\right)$
is $[0,(r^q-\mathbf{1}^{\top}|\boldsymbol{x}_{-j}|^q)^{1/q}]$, so $\varphi_{C_j,\mathfrak{D},j}(\boldsymbol{x})=\min\left\{C_j,x_j,(r^q-\mathbf{1}^{\top}|\boldsymbol{x}_{-j}|^q)^{1/q}-x_j\right\}$.
\end{example}

\begin{example}[Complement of $\mathcal{L}^q$ ball]
Now consider
$\mathfrak{D}\equiv\{\boldsymbol{x}\in\mathbb{R}^m:\|\boldsymbol{x}\|_q>r\}$,
the complement of the $\mathcal{L}^q$ ball with radius $r>0$ and $q\geq  1$. Given $\boldsymbol{x}\in\mathfrak{D}$ and $j$, we now have 
\[
  \mathfrak{C}_{j,\mathfrak{D}}\left(\boldsymbol{x}_{-j}\right)=
  \begin{cases}
   \mathbb{R} &\text{ if
   }\mathbf{1}_{m-1}^{\top}|\boldsymbol{x}_{-j}|^q> r^q,\\
   \left(-\infty,-(r^q-\mathbf{1}_{m-1}^{\top}|\boldsymbol{x}_{-j}|^q)^{1/q}\right)\cup\left((r^q-\mathbf{1}_{m-1}^{\top}|\boldsymbol{x}_{-j}|^q)^{1/q},+\infty\right)&\text{
     otherwise};
  \end{cases}
\]
\[
  \varphi_{C_j,\mathfrak{D},j}(\boldsymbol{x})=\begin{cases}C_j & \text{if }\mathbf{1}_{m-1}^{\top}|\boldsymbol{x}_{-j}|^q>r^q,\\ 
\min\left\{C_j,|x_j|-(r^q-\mathbf{1}_{m-1}^{\top}|\boldsymbol{x}_{-j}|^q)^{1/q}\right\}&
\text{otherwise}.\end{cases}
\]
\end{example}

\begin{example}[Complement of $\mathcal{L}^q$ ball restricted to
  $\mathbb{R}_+^m$] 
Next consider
$\mathfrak{D}\equiv\{\boldsymbol{x}\in\mathbb{R}_+^m:\|\boldsymbol{x}\|_q>r\}$,
the complement of the nonnegative part of the $\mathcal{L}^q$ ball
with radius $r>0$ and $q\geq 1$.  Given
$\boldsymbol{x}\in\mathfrak{D}$ and 
$j$, 
\[
  \mathfrak{C}_{j,\mathfrak{D}}\left(\boldsymbol{x}_{-j}\right)=
  \begin{cases}
   \mathbb{R}_+ &\text{ if
   }\mathbf{1}_{m-1}^{\top}|\boldsymbol{x}_{-j}|^q> r^q,\\
   \left(-\infty,-(r^q-\mathbf{1}_{m-1}^{\top}|\boldsymbol{x}_{-j}|^q)^{1/q}\right)\cup\left((r^q-\mathbf{1}_{m-1}^{\top}|\boldsymbol{x}_{-j}|^q)^{1/q},+\infty\right)&\text{
     otherwise};
  \end{cases}
\]
\[
  \varphi_{C_j,\mathfrak{D},j}(\boldsymbol{x})=\begin{cases}\min\{C_j,x_j\} & \text{if }\mathbf{1}_{m-1}^{\top}|\boldsymbol{x}_{-j}|^q>r^q,\\ 
\min\left\{C_j,x_j-(r^q-\mathbf{1}_{m-1}^{\top}|\boldsymbol{x}_{-j}|^q)^{1/q}\right\}&
\text{otherwise}.\end{cases}
\]
\end{example}

\begin{example}[Complicated domains defined by inequality constraints]
More generally, a domain $\mathfrak{D}$ may be determined by a series of intersections/unions of regions determined by inequality constraints, e.g., $\mathfrak{D}=\{\boldsymbol{x}\in\mathbb{R}^m:(f_1(\boldsymbol{x})\leq c_1\wedge f_2(\boldsymbol{x})\leq c_2)\vee  f_3(\boldsymbol{x})\geq c_3\}$. In this case, to calculate $\boldsymbol{\varphi}_{\boldsymbol{C},\mathfrak{D}}$ we may plug in $\boldsymbol{x}_{-j}$ as given and solve numerically $f_i(x_j;\boldsymbol{x}_{-j})=c_i$ for $i=1,2,3$, and obtain the boundary points for $x_j$ using simple algorithms for interval unions/intersections. This is implemented in the package \texttt{genscore} for some types of polynomials $f_i$ and arbitrary intersections/unions.
\end{example}

\subsection{The Empirical Generalized Score Matching Loss}\label{Lemmas and Empirical Loss}

From this section on, we simplify notation by dropping the dependence of $\boldsymbol{\varphi}$ on $\boldsymbol{C}$ and $\mathfrak{D}$.

\begin{lemma}\label{lem_identifiability}
Suppose $\boldsymbol{C}\succ\boldsymbol{0}$, $p_0(\boldsymbol{x})>0$ for almost every $\boldsymbol{x}\in\mathfrak{D}$ and $h_1(y),\ldots,h_m(y)>0$ for all $y>0$. Then $J_{\boldsymbol{h},\boldsymbol{C},\mathfrak{D}}(P)=0$ if and only if $p_0=p$ for a.e.~$\boldsymbol{x}\in\mathfrak{D}$.
\end{lemma}

\begin{proof}[Proof of Lemma \ref{lem_identifiability}]
By the measurability and definition of $\mathfrak{D}$, and using the
Fubini-Tonelli theorem, the component-wise boundary set $\underline{\partial}\mathfrak{D}$ from
Definition~\ref{def_V} is a Lebesgue-null set. Thus,
$\boldsymbol{\varphi}(\boldsymbol{x})\succ\boldsymbol{0}$ for almost
every $\boldsymbol{x}\in\mathfrak{D}$, so that
$(\boldsymbol{h}\circ\boldsymbol{\varphi})(\boldsymbol{x})\succ
\boldsymbol{0}$ for a.e.~$\boldsymbol{x}\in\mathfrak{D}$. So
$J_{\boldsymbol{h},\boldsymbol{C},\mathfrak{D}}(P)=0$ if and only if
$\nabla\log p_0(\boldsymbol{x})=\nabla\log p(\boldsymbol{x})$ for
a.e.~$\boldsymbol{x}\in\mathfrak{D}$, i.e., $\log p_0(\boldsymbol{x})=\log p(\boldsymbol{x})+c_0$ for a.e.~$\boldsymbol{x}\in\mathfrak{D}$ for some constant $c_0$, or $p_0(\boldsymbol{x})=c_1\cdot p(\boldsymbol{x})$ for a.e.~$\boldsymbol{x}\in\mathfrak{D}$ for some non-zero constant $c_1\equiv\exp(c_0)$. Since $p_0$ and $p$ both integrate to $1$ over $\mathfrak{D}$, we have $c_1=1$ and $p_0=p$ for a.e.~$\boldsymbol{x}\in\mathfrak{D}$.
\end{proof}

According to Lemma \ref{lem_identifiability} the proposed score
matching method requires the domain $\mathfrak{D}$ to have
positive Lebesgue measure in $\mathbb{R}^m$.  For some
null sets, e.g, a probability simplex, an appropriate transformation
to a lower-dimensional set with positive measure can be given but
we defer discussion of such domains to future work, and assume $\mathfrak{D}$ has positive Lebesgue measure for the rest of this paper.

\begin{lemma}\label{lem_loss} Similar to (A0.1) and (A0.2) in Section
  \ref{Score Matching on R+m}, assume the following to hold for all $p\in\mathcal{P}(\mathfrak{D})$,
\begin{enumerate}[label=(A.\arabic*),leftmargin=30pt]
\item $p_0(x_j;\boldsymbol{x}_{-j})h_j(\varphi_j(\boldsymbol{x}))\partial_j\log p(x_j;\boldsymbol{x}_{-j})\left|_{x_j \searrow a_{k}(\boldsymbol{x}_{-j})^+}^{x_j\nearrow b_{k}(\boldsymbol{x}_{-j})^-}\right.=0$\\ for all $k=1,\ldots,K_j(\boldsymbol{x}_{-j})$ and  $\boldsymbol{x}_{-j}\in\mathfrak{S}_{-j}$ for all $j$;
\item $\int_{\mathfrak{D}}p_0(\boldsymbol{x})\left\|\nabla\log
    p(\boldsymbol{x})\odot
    (\boldsymbol{h}\circ\boldsymbol{\varphi})^{1/2}(\boldsymbol{x})\right\|_2^2\d\boldsymbol{x}<+\infty$,
  \\
$\int_{\mathfrak{D}}p_0(\boldsymbol{x})\left\|\left[\nabla\log p(\boldsymbol{x})\odot(\boldsymbol{h}\circ\boldsymbol{\varphi})(\boldsymbol{x})\right]'\right\|_1\d\boldsymbol{x}<+\infty$.
\end{enumerate}
Also assume that 
\begin{enumerate}[label=(A.\arabic*),leftmargin=30pt]
\setcounter{enumi}{2}
\item $\forall j=1,\ldots,m$ and
  a.e.~$\boldsymbol{x}_{-j}\in\mathfrak{S}_{-j}$, the component
  function $h_j$ of $\boldsymbol{h}$ is absolutely continuous in any
  bounded sub-interval of the section
  $\mathfrak{C}_j(\boldsymbol{x}_{-j})$.
  This implies the same for $(h_j\circ\varphi_j)$ and also that $\partial_j (h_j\circ\varphi_j)$ exists a.e.
\end{enumerate}
 Then the loss $J_{\boldsymbol{h},\boldsymbol{C},\mathfrak{D}}(P)$ is equal to a constant depending on $p_0$ only (i.e., independent of $p$) plus
\begin{multline}\label{eq_equivalent_loss}
\frac{1}{2}\sum_{j=1}^m\int_{\mathfrak{D}}p_0(\boldsymbol{x})\cdot(h_j\circ\varphi_j)(\boldsymbol{x})\cdot\left[\partial_j\log p(\boldsymbol{x})\right]^2\d\boldsymbol{x}
+\sum_{j=1}^m\int_{\mathfrak{D}}p_0(\boldsymbol{x})\cdot\partial_j \left[(h_j\circ\varphi_j)(\boldsymbol{x})\cdot\partial_j\log p(\boldsymbol{x})\right]\d\boldsymbol{x}.
\end{multline}
\end{lemma}
The proof of the lemma is given in Appendix \ref{Proofs}. The lemma
enables us to estimate the population loss, or rather those parts that are
relevant for estimation of $P$,  using the empirical loss 
\begin{equation}\label{eq_empirical_loss}
\hat{J}_{\boldsymbol{h},\boldsymbol{C},\mathfrak{D}}(P)=\frac{1}{2}\sum_{j=1}^m\sum_{i=1}^n\frac{1}{2}(h_j\circ\varphi_j)\left(\boldsymbol{x}^{(i)}\right)\cdot\left[\partial_j\log p\left(\boldsymbol{x}^{(i)}\right)\right]^2 
+\partial_j \left[(h_j\circ\varphi_j)\left(\boldsymbol{x}^{(i)}\right)\cdot\partial_j\log p\left(\boldsymbol{x}^{(i)}\right)\right].
\end{equation}
As the canonical choices of $\boldsymbol{h}$ are power functions in $\boldsymbol{x}$, we give the following sufficient conditions for the assumptions in the lemma.
\begin{proposition}\label{prop_assumptions_power}
Suppose for all $j=1,\ldots,m$, $h_j(x_j)=x_j^{\alpha_j}$ for some $\alpha_j>0$. Suppose in addition that for all $j$ and $\boldsymbol{x}_{-j}\in\mathfrak{S}_{-j}$ and all $p\in\mathcal{P}$ we have
\begin{enumerate}[(1)]
\item $p_0(x_j;\boldsymbol{x}_{-j})\partial_j\log p(x_j;\boldsymbol{x}_{-j})=o\left(1/(x_j-c_{k,j})^{\alpha_j}\right)$ as $x_j\nearrow c_{k,j}\equiv b_{k,j}(\boldsymbol{x}_{-j})<+\infty$ or as $x_j\searrow c_{k,j}\equiv a_{k,j}(\boldsymbol{x}_{-j})>-\infty$ for all $k$, and
\item $p_0(x_j;\boldsymbol{x}_{-j})\partial_j\log p(x_j;\boldsymbol{x}_{-j})\to 0$ as $x_j\nearrow+\infty$ if $\mathfrak{C}_{j}(\boldsymbol{x}_{-j})$ is unbounded from above, and as $x_j\searrow-\infty$ if $\mathfrak{C}_{j}(\boldsymbol{x}_{-j})$ is unbounded from below.
\end{enumerate}
Then (A.1) and (A.3) are satisfied.
\end{proposition}
\begin{proof}[Proof of Proposition \ref{prop_assumptions_power}]
Condition (A.3) is satisfied by the property of $h_j$. (A.1) is also satisfied since by construction $(h_j\circ\varphi_j)(\boldsymbol{x})$ becomes $|x_j-c_{k,j}|^{\alpha_j}$ as $x_j\to c_{k,j}\in\cup_{k=1}^{K_j}\{a_{k,j}(\boldsymbol{x}_{-j}),b_{k,j}(\boldsymbol{x}_{-j})\}$ and also $(h_j\circ\varphi_j)$ is bounded by $C^{\alpha_j}$ as $x_j\nearrow +\infty$ or $x_j\searrow -\infty$, if applicable.
\end{proof}

\subsection{Exponential Families and Regularized Score Matching}\label{Exponential Families and Regularized Score Matching}
Consider the case where
$\mathcal{P}(\mathfrak{D})\equiv\{p_{\boldsymbol{\theta}}:\boldsymbol{\theta}\in\boldsymbol{\Theta}\subset\mathbb{R}^r\}$
for some $r=1,2,\ldots$ is an exponential family comprised of continuous distributions supported on $\mathfrak{D}$ with densities of the form
\begin{equation}\label{eq_exp_density}
\log p_{\boldsymbol{\theta}}(\boldsymbol{x})=\boldsymbol{\theta}^{\top}\boldsymbol{t}(\boldsymbol{x})-\psi(\boldsymbol{\theta})+b(\boldsymbol{x}),\quad\boldsymbol{x}\in\mathfrak{D},
\end{equation}
where $\boldsymbol{\theta}\in\mathbb{R}^r$ is the unknown canonical parameter of interest, $\boldsymbol{t}(\boldsymbol{x})\in\mathbb{R}^r$ are the sufficient statistics, $\psi(\boldsymbol{\theta})$ is the normalizing constant, and $b(\boldsymbol{x})$ is the base measure. The empirical loss $\hat{J}_{\boldsymbol{h},\boldsymbol{C},\mathfrak{D}}$ (\ref{eq_empirical_loss}) can then be written as a quadratic function in the canonical parameter:
\begin{align}\label{eq_loss_exponential}
\hspace{-0.02in}\hat{J}_{\boldsymbol{h},\boldsymbol{C},\mathfrak{D}}(p_{\boldsymbol{\theta}})&=\frac{1}{2}\boldsymbol{\theta}^{\top}\boldsymbol{\Gamma}(\mathbf{x})\boldsymbol{\theta}-\boldsymbol{g}(\mathbf{x})^{\top}\boldsymbol{\theta}+\mathrm{const},\quad\text{with}\\
\boldsymbol{\Gamma}(\mathbf{x})&=\frac{1}{n}\sum_{i=1}^{n}\sum_{j=1}^m(h_j\circ\varphi_j)\left(\boldsymbol{X}^{(i)}\right)\partial_j\boldsymbol{t}(\boldsymbol{X}^{(i)})
\left(\partial_j\boldsymbol{t}\left(\boldsymbol{X}^{(i)}\right)\right)^{\top}\quad\text{and}\label{def_Gamma}\\
\boldsymbol{g}(\mathbf{x})&=-\frac{1}{n}\sum_{i=1}^n\sum_{j=1}^m\left[(h_j\circ\varphi_j)\left(\boldsymbol{X}^{(i)}\right)\partial_j b\left(\boldsymbol{X}^{(i)}\right)
\partial_j\boldsymbol{t}\left(\boldsymbol{X}^{(i)}\right)\nonumber \right.\\
&\quad\left.+(h_j\circ\varphi_j)\left(\boldsymbol{X}^{(i)}\right)\partial_{jj}\boldsymbol{t}\left(\boldsymbol{X}^{(i)}\right) 
+\partial_j(h_j\circ\varphi_j)\left(\boldsymbol{X}^{(i)}\right)\partial_j\boldsymbol{t}\left(\boldsymbol{X}^{(i)}\right)\right]\label{def_g},
\end{align}
where $\partial_j\boldsymbol{t}(\boldsymbol{x})=(\partial_j t_1(\boldsymbol{x}),\ldots,\partial_j t_r(\boldsymbol{x}))^{\top}\in\mathbb{R}^r$. Note that  (\ref{def_Gamma}) and (\ref{def_g}) are sample averages of functions in the data matrix $\mathbf{x}$ only. These forms are an exact analog of those in Theorem 5 in \citet{yus19}. As expected, we can thus obtain the following consistency result similar to Theorem 6 in \citet{yus19}:

\begin{theorem}[Theorem 6 of \citet{yus19}]
Suppose the true density is $p_0\equiv p_{\boldsymbol{\theta}_0}$ and that
\begin{enumerate}[label=(C\arabic*)]
\item $\boldsymbol{\Gamma}$ is almost surely invertible, and 
\item $\boldsymbol{\Sigma}_0\equiv\mathbb{E}_{p_0}\left[\left(\boldsymbol{\Gamma}(\mathbf{x})\boldsymbol{\theta}_0-\boldsymbol{g}(\mathbf{x})\right)\left(\boldsymbol{\Gamma}(\mathbf{x})\boldsymbol{\theta}_0-\boldsymbol{g}(\mathbf{x})\right)^{\top}\right]$, $\boldsymbol{\Gamma}_0\equiv\mathbb{E}_{p_0}\boldsymbol{\Gamma}(\mathbf{x})$, $\boldsymbol{\Gamma}_0^{-1}$, and $\boldsymbol{g}_0\equiv\mathbb{E}_{p_0}\boldsymbol{g}(\mathbf{x})$ exist and are component-wise finite.
\end{enumerate}
Then the minimizer of (\ref{eq_loss_exponential}) is almost surely unique with closed form solution $\hat{\boldsymbol{\theta}}\equiv\boldsymbol{\Gamma}(\mathbf{x})^{-1}\boldsymbol{g}(\mathbf{x})$ with
\[\hat{\boldsymbol{\theta}}\to_{\text{a.s.}}\boldsymbol{\theta}_0
\quad\text{and}\quad\sqrt{n}(\hat{\boldsymbol{\theta}}-\boldsymbol{\theta}_0)\to_d\mathcal{N}_{r}\left(\boldsymbol{0},\boldsymbol{\Gamma}_0^{-1}\boldsymbol{\Sigma}_0\boldsymbol{\Gamma}_0^{-1}\right)\quad\text{as }n\to\infty.\]
\end{theorem}

Estimation in high-dimensional settings, in which the number of parameters $r$ may
exceed 
the sample size $n$,  usually benefits from regularization. For exponential families, as in \citet{yus19}, we add an $\ell_1$ penalty on $\boldsymbol{\theta}$ to the loss in (\ref{eq_loss_exponential}), while multiplying the diagonals of the $\boldsymbol{\Gamma}$ by a \emph{diagonal multiplier} $\delta>1$:
\begin{definition}
Let $\delta>1$, and $\boldsymbol{\Gamma}_{\delta}(\mathbf{x})$ be $\boldsymbol{\Gamma}(\mathbf{x})$ with diagonal entries multiplied by $\delta$. For exponential family distributions in (\ref{eq_exp_density}), the \emph{regularized generalized $(\boldsymbol{h},\boldsymbol{C},\mathfrak{D})$-score matching loss} is defined as
\begin{equation}\label{eq_loss_regularized}
\hat{J}_{\boldsymbol{h},\boldsymbol{C},\mathfrak{D},\lambda,\delta}(p_{\boldsymbol{\theta}})\equiv\frac{1}{2}\boldsymbol{\theta}^{\top}\boldsymbol{\Gamma}_{\delta}(\mathbf{x})\boldsymbol{\theta}-\boldsymbol{g}(\mathbf{x})^{\top}\boldsymbol{\theta}+\lambda\|\boldsymbol{\theta}\|_1.
\end{equation}
\end{definition}
The multiplier $\delta>1$, together with the $\ell_1$ penalty,
resembles an elastic net penalty and prevents the loss in
(\ref{eq_loss_regularized}) from being unbounded from below for
smaller $\lambda$, in which case there can be infinitely many
minimizers.  This is discussed 
in Section 4 in
\citet{yus19}, where a default for $\delta$ is given, so that no
tuning for this parameter is necessary. 
Minimization of
(\ref{eq_loss_regularized}) can be efficiently done using
coordinate-descent with warm starts, which along with other
computational details is discussed in Section 5.3 of \citet{yus19}.

\subsection{Extension of the Method from  \citet{liu19} to Unbounded
  Domains}\label{Extension of g0}

A key ingredient to our treatment of unbounded domains is truncation
of distances.  This idea can also be applied to the method proposed
for general bounded domains in \citet{liu19}; recall Section
\ref{sec_intro_liu19}.  For any component-wise countable union of
intervals $\mathfrak{D}$ as 
in Definition \ref{def_V}, we may
modify the loss in (\ref{eq_loss_g0}) to
\begin{equation*}
J_{g_0,C,\mathfrak{D}}(P)\equiv\sup_{g\in\mathcal{G}}\frac{1}{2}\int_{\mathbb{R}_+^m}g(\boldsymbol{x})p_0(\boldsymbol{x})\left\|\nabla\log p(\boldsymbol{x})-\nabla\log p_0(\boldsymbol{x})\right\|_2^2\d\boldsymbol{x},
\end{equation*}
but with
$\mathcal{G}\equiv\{g|g(\boldsymbol{x})=0,\forall\boldsymbol{x}\in\underline{\partial}\mathfrak{D},\,g\text{  is }L\text{-Lipschitz continuous}\text{ and }g\leq C\}$
instead. Here, we use the same Lipschitz constant $L>0$ but add an
extra truncation constant $C>0$.  Moreover, we use the component-wise boundary set 
$\underline{\partial}\mathfrak{D}$ (\ref{def_V_boundary}) instead of the usual boundary set $\partial\mathfrak{D}$ used in \citet{liu19}. Following the same proof as for their Proposition 1 and dropping the Lipschitz constant $L$ by replacing $C$ with $C/L$ (or equivalently choosing $L=1$), it is easy to see that the maximum is obtained at
\begin{equation}\label{eq_extended_g0}
g_0(\boldsymbol{x})\equiv\min\left\{C,\inf_{\boldsymbol{x}'\in\underline{\partial} \mathfrak{D}}\|\boldsymbol{x}-\boldsymbol{x}'\|_2\right\},
\end{equation}
the $\ell_2$ distance of $\boldsymbol{x}$ to $\underline{\partial}\mathfrak{D}$ truncated above by $C$, which naturally extends the method in \citet{liu19} to unbounded domains. In the special case where $\underline{\partial}\mathfrak{D}=\varnothing$, we must have $\mathfrak{D}\equiv\mathbb{R}^m$ and $g_0(\boldsymbol{x})\equiv C$ by the expression above (with the convention of $\inf\varnothing=+\infty$), which coincides with the original score matching in \citet{hyv05}.

Assuming that (A.1) and (A.2) from Lemma \ref{lem_loss} hold when replacing $(\boldsymbol{h}\circ\boldsymbol{\varphi})(\boldsymbol{x})$ by $g_0(\boldsymbol{x})\mathbf{1}_m$ and $h_j(\varphi_j(\boldsymbol{x}))$ by $g_0(\boldsymbol{x})$, the same conclusion there applies, i.e.
\begin{equation*}
J_{g_0,C,\mathfrak{D}}(P)\equiv\frac{1}{2}\sum_{j=1}^m\int_{\mathfrak{D}}p_0(\boldsymbol{x})g_0(\boldsymbol{x})\left[\partial_j\log p(\boldsymbol{x})\right]^2\d\boldsymbol{x}+\sum_{j=1}^m\int_{\mathfrak{D}}p_0(\boldsymbol{x})\partial_j \left[g_0(\boldsymbol{x})\partial_j\log p(\boldsymbol{x})\right]\d\boldsymbol{x}
\end{equation*}
plus a constant depending on $p_0$ but not on $p$; this is the same loss as in Equation (13) in \citet{liu19} with a truncation by $C$ applied to $g_0$. The proof is in the same spirit of that for Lemma \ref{lem_loss} and is thus omitted.

\section{Pairwise Interaction Models on Domains with Positive Measure}\label{$a$-$b$ Models on Domains with Positive Measure}


\subsection{Pairwise Interaction Power $a$-$b$ Models}\label{Pairwise Interaction Power $a$-$b$ Models}

As one realm  of application of the proposed estimation method, we
consider exponential family models that postulate pairwise
interactions between power-transformations of the observed variables,
as
in (\ref{eq_interaction_density}):
\begin{equation}\label{eq_interaction_density2}
p_{\boldsymbol{\eta},\mathbf{K}}(\boldsymbol{x})\propto\exp\left(-\frac{1}{2a}{\boldsymbol{x}^a}^{\top}\mathbf{K}\boldsymbol{x}^a+\frac{1}{b}\boldsymbol{\eta}^{\top}\boldsymbol{x}^b\right)\mathds{1}_{\mathfrak{D}}(\boldsymbol{x})
\end{equation}
for which we treat $\boldsymbol{x}^0\equiv\log\boldsymbol{x}$ and
$1/0\equiv1$, on a domain $\mathfrak{D}\subseteq\mathbb{R}^m$ with a
positive measure. Here $a\geq 0$ and $b\geq 0$ are known constants and
the interaction matrix $\mathbf{K}\in\mathbb{R}^{m\times m}$ and the
linear vector $\boldsymbol{\eta}\in\mathbb{R}^m$ are unknown
parameters of interest. As in \citet{yus19}, our focus will be on the
support of $\mathbf{K}$, $S(\mathbf{K})=\{(i,j):\kappa_{ij}\neq 0\}$,
that for product domains defines the conditional independence graph of
$\boldsymbol{X}\sim p_{\boldsymbol{\eta},\mathbf{K}}$.  However, we
simultaneously estimate the nuisance parameter $\boldsymbol{\eta}$
unless it is explicitly assumed to be $\mathbf{0}$.

When $a=b=1$, model~(\ref{eq_interaction_density2}) is a truncated
Gaussian model. From $a=b=1/2$, we may obtain the exponential square-root graphical model in \citet{ino16}. The gamma model 
in \citet{yus19} has $a=1/2$ and $b=0$. 

\subsection{Finite Normalizing Constant and Validity of Score Matching}

The following theorem gives detailed sufficient conditions for the
$a$-$b$ density $p_{\boldsymbol{\eta},\mathbf{K}}$ in
(\ref{eq_interaction_density2}) to be a proper density on a domain $\mathfrak{D}\subseteq\mathbb{R}^m$ with positive Lebesgue measure.

\begin{theorem}[Sufficient conditions for finite normalizing constant]\label{thm_norm_const}
Denote
$\rho_j(\mathfrak{D})\equiv\overline{\{x_j:\boldsymbol{x}\in\mathfrak{D}\}}$
the closure of the range of $x_j$ in the domain $\mathfrak{D}$. If any
of the following conditions holds, the density in
(\ref{eq_interaction_density2}) is a proper density, i.e., the right-hand of (\ref{eq_interaction_density2}) is integrable over $\mathfrak{D}$:

\begin{enumerate}[label=(CC\arabic*),leftmargin=35pt]
\item $a>0$, $b>0$, $\mathfrak{D}$ is bounded;
\item $a>0$, $b>0$, ${\boldsymbol{v}^a}^{\top}\mathbf{K}\boldsymbol{v}^a> 0\,\,\forall\boldsymbol{v}\in\mathfrak{D}\backslash\{\boldsymbol{0}\}$, and either $2a>b$ or $\boldsymbol{\eta}^{\top}\boldsymbol{v}^b\leq 0$ $\forall\boldsymbol{v}\in\mathfrak{D}$;
\item $a>0$, $b=0$, $\eta_j>-1$ for all $j$ s.t.~$0\in\rho_j(\mathfrak{D})$, and one of the following holds:
\begin{enumerate}[(i),leftmargin=0pt]
\item $\mathfrak{D}$ is bounded;
\item $\mathfrak{D}$ is unbounded and ${\boldsymbol{v}^a}^{\top}\mathbf{K}\boldsymbol{v}^a> 0\,\,\forall\boldsymbol{v}\in\mathfrak{D}\backslash\{\boldsymbol{0}\}$;
\item $\mathfrak{D}$ is unbounded, ${\boldsymbol{v}^a}^{\top}\mathbf{K}\boldsymbol{v}^a\geq 0\,\,\forall\boldsymbol{v}\in\mathfrak{D}$ and  $\eta_j<-1$ for all $j$ s.t.~$\rho_j(\mathfrak{D})$ is unbounded (which implies that $\rho_j(\mathfrak{D})=[0,+\infty)$ is not allowed for any $j$);
\end{enumerate}
\item $a=0$, $\mathfrak{D}$ is bounded and $0\not\in\rho_j(\mathfrak{D})$ for all $j$;
\item $a=0$, $b=0$, $\log(\boldsymbol{x})^{\top}\mathbf{K}\log(\boldsymbol{x})>0$ $\forall \boldsymbol{x}\in\mathfrak{D}$;
\item $a=0$, $b>0$, $\log(\boldsymbol{x})^{\top}\mathbf{K}\log(\boldsymbol{x})> 0$ $\forall\boldsymbol{x}\in\mathfrak{D}$ and $\eta_j\leq 0$ for all $j$ s.t.~$\rho_j(\mathfrak{D})$ is unbounded;
\item $a=0$, $b>0$, $\log(\boldsymbol{x})^{\top}\mathbf{K}\log(\boldsymbol{x})\geq 0$ $\forall\boldsymbol{x}\in\mathfrak{D}$ and $\eta_j< 0$ for all $j$ s.t.~$\rho_j(\mathfrak{D})$ is unbounded.
\end{enumerate}
In the centered case where $\boldsymbol{\eta}=\boldsymbol{0}$ is known, any condition in terms of $b$ and $\boldsymbol{\eta}$ can be ignored.
\end{theorem}


To simplify our discussion, the following corollary gives a simpler set of sufficient conditions for integrability of the density.
\begin{corollary}[Sufficient conditions for finite normalizing constant; simplified]\label{cor_norm_const}
Suppose 
\begin{enumerate}[label=(CC\arabic*{*}),leftmargin=40pt]
\setcounter{enumi}{-1}
\item $\mathbf{K}$ is positive definite 
\end{enumerate}
and one of the following conditions holds:
\begin{enumerate}[label=(CC\arabic*{*}),leftmargin=40pt]
\item $2a>b>0$ or $a=b=0$;
\item $a>0$, $b=0$, $\boldsymbol{\eta}\succ-\mathbf{1}_m$;
\item $a=0$, $b>0$, $\eta_j\leq 0$ for any $j$ such that $x_j$ is unbounded in $\mathfrak{D}$.
\end{enumerate}
Then the right-hand side of (\ref{eq_interaction_density2}) is integrable over $\mathfrak{D}$.
In the 
case where $\boldsymbol{\eta}\equiv\boldsymbol{0}$, (CC0*) is sufficient.
\end{corollary}
\noindent For simplicity, we use 
conditions (CC0*)--(CC3*) throughout the 
paper. The following theorem gives sufficient conditions on $\boldsymbol{h}$ that satisfy 
conditions (A.1)--(A.3) in Lemma \ref{lem_loss} for score matching.

\begin{theorem}[Sufficient conditions that satisfy assumptions for score matching]\label{thm_A}
Suppose (CC0*) and one of (CC1*) through (CC3*) holds, and $\boldsymbol{h}(\boldsymbol{x})=\left(x_1^{\alpha_1},\ldots,x_m^{\alpha_m}\right)$, where
\begin{enumerate}[(1)]
\item if $a>0$ and $b>0$, $\alpha_j>\max\{0,1-a,1-b\}$;
\item if $a>0$ and $b=0$, $\alpha_j>1-\eta_{0,j}$;
\item if $a=0$, $\alpha_j\geq 0$.
\end{enumerate}
Then conditions (A.1), (A.2) and (A.3) in Lemma \ref{lem_loss} are satisfied and the equivalent form of the generalized score matching loss (\ref{eq_equivalent_loss}) holds, and the empirical loss (\ref{eq_empirical_loss}) is valid. In the centered case with $\boldsymbol{\eta}\equiv\boldsymbol{0}$, it suffices to have $a>0$ and $\alpha_j>\max\{0,1-a\}$ or $a=0$ and $\alpha_j\geq 0$.
\end{theorem}

\subsection{Estimation}\label{Estimation_general}

Let $\boldsymbol{\Psi}\equiv \left[\mathbf{K}^{\top}\,\, \boldsymbol{\eta}\right]^{\top}\in\mathbb{R}^{(m+1)\times m}$. In this section, we give the form of $\boldsymbol{\Gamma}\in\mathbb{R}^{(m+1)m\times (m+1)m}$ and $\boldsymbol{g}\in\mathbb{R}^{(m+1)m}$ in the unpenalized loss $\frac{1}{2}\mathrm{vec}(\boldsymbol{\Psi})^{\top}\boldsymbol{\Gamma}\mathrm{vec}(\boldsymbol{\Psi})-\boldsymbol{g}^{\top}\mathrm{vec}(\boldsymbol{\Psi})$ following (\ref{def_Gamma})--(\ref{def_g}). $\boldsymbol{\Gamma}$ is block-diagonal, with the $j$-th $\mathbb{R}^{(m+1)\times (m+1)}$ block
\begin{align*}
\boldsymbol{\Gamma}_j(\mathbf{x})&
\equiv\begin{bmatrix}
\boldsymbol{\Gamma}_{\mathbf{K},j} & \boldsymbol{\gamma}_{\mathbf{K},\boldsymbol{\eta},j} \\
\boldsymbol{\gamma}_{\mathbf{K},\boldsymbol{\eta},j}^{\top} & {\gamma}_{\boldsymbol{\eta},j}
\end{bmatrix}\\
&\equiv\frac{1}{n}\sum\limits_{i=1}^n
\begin{bmatrix}
(h_j\circ\varphi_j)\left(\boldsymbol{X}^{(i)}\right){X_j^{(i)}}^{2a-2}{\boldsymbol{X}^{(i)}}^a{{\boldsymbol{X}^{(i)}}^a}^{\top} & -(h_j\circ\varphi_j)\left(\boldsymbol{X}^{(i)}\right){X_j^{(i)}}^{a+b-2}{\boldsymbol{X}^{(i)}}^a\\
-(h_j\circ\varphi_j)\left(\boldsymbol{X}^{(i)}\right){X_j^{(i)}}^{a+b-2}{{\boldsymbol{X}^{(i)}}^a}^{\top} & (h_j\circ\varphi_j)\left(\boldsymbol{X}^{(i)}\right){X_j^{(i)}}^{2b-2}
\end{bmatrix}.
\end{align*}
On the other hand, $\boldsymbol{g}\equiv\mathrm{vec}\left(\left[ \mathbf{g}_{\mathbf{K}}^{\top}\,\,\,\, \boldsymbol{g}_{\boldsymbol{\eta}}\right]^{\top}\right)\in\mathbb{R}^{(m+1)m}$, where $\mathbf{g}_{\mathbf{K}}\in\mathbb{R}^{m\times m}$ and $\boldsymbol{g}_{\boldsymbol{\eta}}\in\mathbb{R}^m$ correspond to $\mathbf{K}$ and $\boldsymbol{\eta}$, respectively. The $j$-th column of $\mathbf{g}_{\mathbf{K}}$ is 
\begin{multline}\label{eq_gK}
\frac{1}{n}\sum_{i=1}^n\left(\partial_j(h_j\circ\varphi_j)\left(\boldsymbol{X}^{(i)}\right){X_j^{(i)}}^{a-1}+(a-1)(h_j\circ\varphi_j)\left(\boldsymbol{X}^{(i)}\right){X_j^{(i)}}^{a-2}\right){\boldsymbol{X}^{(i)}}^a\\
+a(h_j\circ\varphi_j)\left(\boldsymbol{X}^{(i)}\right){X_j^{(i)}}^{2a-2}\boldsymbol{e}_{j,m},
\end{multline}
where $\boldsymbol{e}_{j,m}\in\mathbb{R}^m$ has $1$ at the $j$-th component and 0 elsewhere, and the $j$-th entry of $\boldsymbol{g}_{\boldsymbol{\eta}}$ is
\begin{equation}\label{eq_geta}
\frac{1}{n}\sum_{i=1}^n-\partial_j(h_j\circ\varphi_j)\left(\boldsymbol{X}^{(i)}\right){X_j^{(i)}}^{b-1}-(b-1)(h_j\circ\varphi_j)\left(\boldsymbol{X}^{(i)}\right){X_j^{(i)}}^{b-2}.
\end{equation}
If $a=0$, set the coefficients $(a-1)$ to $-1$ and $a$ to $1$ in
(\ref{eq_gK}); for $b=0$ set $(b-1)$ to $-1$ in the second term of (\ref{eq_geta}).

As in \citet{yus19} we only apply the diagonal multiplier $\delta$ to the diagonals of $\boldsymbol{\Gamma}_{\mathbf{K},j}\in\mathbb{R}^{m\times m}$, not $\gamma_{\boldsymbol{\eta},j}\in\mathbb{R}$. Note that each block of $\boldsymbol{\Gamma}$ and $\boldsymbol{g}$ correspond to each column of $\boldsymbol{\Psi}$, i.e.~$(\boldsymbol{\kappa}_{,j},\eta_j)\in\mathbb{R}^{m+1}$. In the penalized generalized score-matching loss (\ref{eq_loss_regularized}),  the penalty $\lambda$ for $\mathbf{K}$ and $\boldsymbol{\eta}$ can be different, $\lambda_{\mathbf{K}}$ and $\lambda_{\boldsymbol{\eta}}$, respectively, as long as the ratio $\lambda_{\boldsymbol{\eta}}/\lambda_{\mathbf{K}}$ is fixed. For $\mathbf{K}$ we follow the convention that we penalize its off-diagonal entries only. That is,
\begin{align}\label{eq_loss_ab}
\frac{1}{2}\mathrm{vec}(\boldsymbol{\Psi})^{\top}{\boldsymbol{\Gamma}}_{\delta}(\mathbf{x})\mathrm{vec}(\boldsymbol{\Psi})-{\boldsymbol{g}}(\mathbf{x})^{\top}\mathrm{vec}(\boldsymbol{\Psi})+\lambda_{\mathbf{K}}\|\mathbf{K}_{\mathrm{off}}\|_1+\lambda_{\boldsymbol{\eta}}\|\boldsymbol{\eta}\|_1.
\end{align}
In the case where we do not penalize $\boldsymbol{\eta}$, i.e.~$\lambda_{\boldsymbol{\eta}}=0$, 
we can simply profile out $\boldsymbol{\eta}$, solve for $\hat{\boldsymbol{\eta}}=\boldsymbol{\Gamma}_{\boldsymbol{\eta}}^{-1}\left(\boldsymbol{g}_{\boldsymbol{\eta}}-\boldsymbol{\Gamma}_{\mathbf{K},\boldsymbol{\eta}}^{\top}\mathrm{vec}(\hat{\mathbf{K}})\right)$, plug this back in and rewrite the loss in $\mathbf{K}$ only. Let $\boldsymbol{\Gamma}_{\delta,\mathbf{K}}\in\mathbb{R}^{m^2\times m^2}$ be the block-diagonal matrix with blocks $\boldsymbol{\Gamma}_{\mathbf{K},j}$ and diagonal multiplier $\delta$, and let $\boldsymbol{\Gamma}_{\mathbf{K},\boldsymbol{\eta}}\in\mathbb{R}^{m^2\times m}$ and $\boldsymbol{\Gamma}_{\boldsymbol{\eta}}\in\mathbb{R}^{m\times m}$ be the (block-)diagonal matrices with blocks $\boldsymbol{\gamma}_{\mathbf{K},\boldsymbol{\eta},j}$ and $\gamma_{\boldsymbol{\eta},j}$, respectively. Denote $\boldsymbol{\Gamma}_{\delta,\mathrm{profiled}}$ as the Schur complement of  $\boldsymbol{\Gamma}_{\delta,\mathbf{K}}$ of $\left[\begin{smallmatrix} 
\boldsymbol{\Gamma}_{\delta,\mathbf{K}} & \boldsymbol{\Gamma}_{\mathbf{K},\boldsymbol{\eta}} \\ \boldsymbol{\Gamma}_{\mathbf{K},\boldsymbol{\eta}}^{\top} & \boldsymbol{\Gamma}_{\boldsymbol{\eta}}
\end{smallmatrix}\right]$,  i.e.~$\boldsymbol{\Gamma}_{\delta,\mathbf{K}}-\boldsymbol{\Gamma}_{\mathbf{K},\boldsymbol{\eta}}\boldsymbol{\Gamma}_{\boldsymbol{\eta}}^{-1}\boldsymbol{\Gamma}_{\mathbf{K},\boldsymbol{\eta}}^{\top}$, which is guaranteed to be positive definite for $\delta>1$. Then the profiled loss is
\begin{multline}\label{eq_loss_regularized_profiled}
\hat{J}_{\boldsymbol{h},\boldsymbol{C},\mathfrak{D},\lambda,\delta,\mathrm{profiled}}(p_{\mathbf{K}})\equiv
\frac{1}{2}\mathrm{vec}(\mathbf{K})^{\top}\boldsymbol{\Gamma}_{\delta,\mathrm{profiled}}\mathrm{vec}(\mathbf{K})-\left(\boldsymbol{g}_{\mathbf{K}}-\boldsymbol{\Gamma}_{\mathbf{K},\boldsymbol{\eta}}\boldsymbol{\Gamma}_{\boldsymbol{\eta}}^{-1}\boldsymbol{g}_{\boldsymbol{\eta}}\right)^{\top}\mathrm{vec}(\mathbf{K})+\lambda_{\mathbf{K}}\|\mathbf{K}\|_1.
\end{multline}

\subsection{Univariate Examples}
To illustrate our generalized score matching framework, we first
present univariate Gaussian models on a general domain
$\mathfrak{D}\subset\mathbb{R}$ that is a countable union of intervals
and has  positive Lebesgue measure.  In particular, we estimate one of $\mu_0$  and $\sigma_0^2$ assuming the other is known, given that the true density is
\[p_{\mu_0,\sigma_0^2}(x)\propto\exp\left\{-\frac{(x-\mu_0)^2}{2\sigma_0^2}\right\},\quad
  x\in\mathfrak{D},\] with $\mu_0\in\mathbb{R}$ and $\sigma_0^2>0$.
Let $X^{(1)},\ldots,X^{(n)}\sim p_{\mu_0,\sigma_0^2}$ be
i.i.d.~samples.   Without any regularization by an $\ell_1$ or $\ell_2$
penalty and assuming the true $\sigma_0^2$ is known, we have
similar to Example 3.1 in \citet{yus19} that our generalized score
matching estimator for
$\mu_0$ is
\[\hat{\mu}\equiv\frac{\sum_{i=1}^n(h\circ\varphi_C)\left(X^{(i)}\right)\cdot X^{(i)}-\sigma_0^2(h\circ\varphi_C)'\left(X^{(i)}\right)}{\sum_{i=1}^n (h\circ\varphi_C)\left(X^{(i)}\right)}.\]
By Theorem \ref{thm_A}, it suffices to choose $h(x)=x^{\alpha}$ with $\alpha>0$. Similar to \citet{yus19}, we have
\begin{equation}\label{eq_univariate_mu_var}
\sqrt{n}\left(\hat{\mu}-\mu_0\right)\to_d\mathcal{N}\left(0,\frac{\mathbb{E}_0\left[\sigma_0^2(h\circ\varphi_C)^2(X)+\sigma_0^4{(h\circ\varphi_C)'}^2(X)\right]}{\mathbb{E}_0^2\left[(h\circ\varphi_C)(X)\right]}\right),
\end{equation}
if the expectations exist. On the other hand, assuming the true $\mu_0$ is known, the estimator for $\sigma^2$ is
\[\hat{\sigma}^2\equiv\frac{\sum_{i=1}^n(h\circ\varphi_C)\left(X^{(i)}\right)\cdot\left(X^{(i)}-\mu_0\right)^2}{\sum_{i=1}^n(h\circ\varphi_C)\left(X^{(i)}\right)+(h\circ\varphi_C)'\left(X^{(i)}\right)\cdot\left(X^{(i)}-\mu_0\right)}, \text{ with limiting distribution}\]
\begin{multline}\label{eq_univariate_sigma_var}
\sqrt{n}\left(\hat{\sigma}^2-\sigma_0^2\right)\to_d\mathcal{N}\left(0, 
\frac{\mathbb{E}_0\left[2\sigma_0^6(h\circ\varphi_C)^2(X)\cdot(X-\mu_0)^2+\sigma^8_0{(h\circ\varphi_C)'}^2(X)\cdot(X-\mu_0)^2\right]}{\mathbb{E}_0^2\left[(h\circ\varphi_C)(X)\cdot(X-\mu_0)^2\right]}\right).
\end{multline}

Figure \ref{plot_univariate} shows a standard
normal distribution $\mathcal{N}(0,1)$ restricted to three univariate
domains: $\mathfrak{D}_2\equiv(-\infty,
-3/2]\cup[3/2,+\infty)$, $\mathfrak{D}_3\equiv[-1,-3/4]\cup[3/4,1]$,
and their union $\mathfrak{D}_1\equiv(-\infty,
-3/2]\cup[-1,-3/4]\cup[3/4,1]\cup[3/2,+\infty)$.
The endpoints are chosen so that the probability of the variable lying
in each interval is roughly the same: $\mathcal{N}(0,1)\left((-\infty,-3/2]\right)\approx
0.0668$ and $\mathcal{N}(0,1)\left([-1,-3/4]\right)\approx
0.0680$.  To pick the truncation point $C$ for the distance $\varphi_{C,\mathfrak{D}}$, we choose $\pi\in(0,1]$
and let $C$ be the $\pi$ quantile of the distribution of
$\varphi_{+\infty,\mathfrak{D}}(X)$, where the random variable $X$
follows the truncation of $\mathcal{N}(0,1)$ to the domain
$\mathfrak{D}$.  So, $C$ is such that $\mathbb{P}\left(\varphi_{+\infty,\mathfrak{D}}(X)\leq C\right)=\pi$. Here, $\varphi_{+\infty,\mathfrak{D}_1}(X)=|X|-3/2$ if $|X|>3/2$, or $\min(|X|-3/4,1-|X|)$ otherwise, $\varphi_{+\infty,\mathfrak{D}_2}(X)=|X|-3/2$, and $\varphi_{+\infty,\mathfrak{D}_3}(X)=\min(|X|-3/4,1-|X|)$.

The first subfigure in each row of Figure \ref{plot_univariate} shows
the density on each domain, along with the corresponding
$\varphi_{+\infty}(X)$ in red, whose $y$ axis is on the right. The
second plot in each row shows the log asymptotic variance for the corresponding $\hat{\mu}$, as on the right-hand side of (\ref{eq_univariate_mu_var}), and the third shows that for $\hat{\sigma}^2$ as in (\ref{eq_univariate_sigma_var}). Each curve represents a different $\alpha$ in $h(x)=x^\alpha$, and the $x$ axis represents the quantiles $\pi$ associated with the truncation point $C$ as above. Finally, the red dotted curve shows $C$ versus $\pi$ for each domain. The ``bumps'' in the variance for $x^{0.5}$ are due to numerical instability in integration.
\begin{figure}
\begin{tabular}{|c|c|c|}
\hline
\subfloat[Density on $\mathfrak{D}_1$]
{\hspace{-0.0in}\includegraphics[width=0.27\textwidth]{{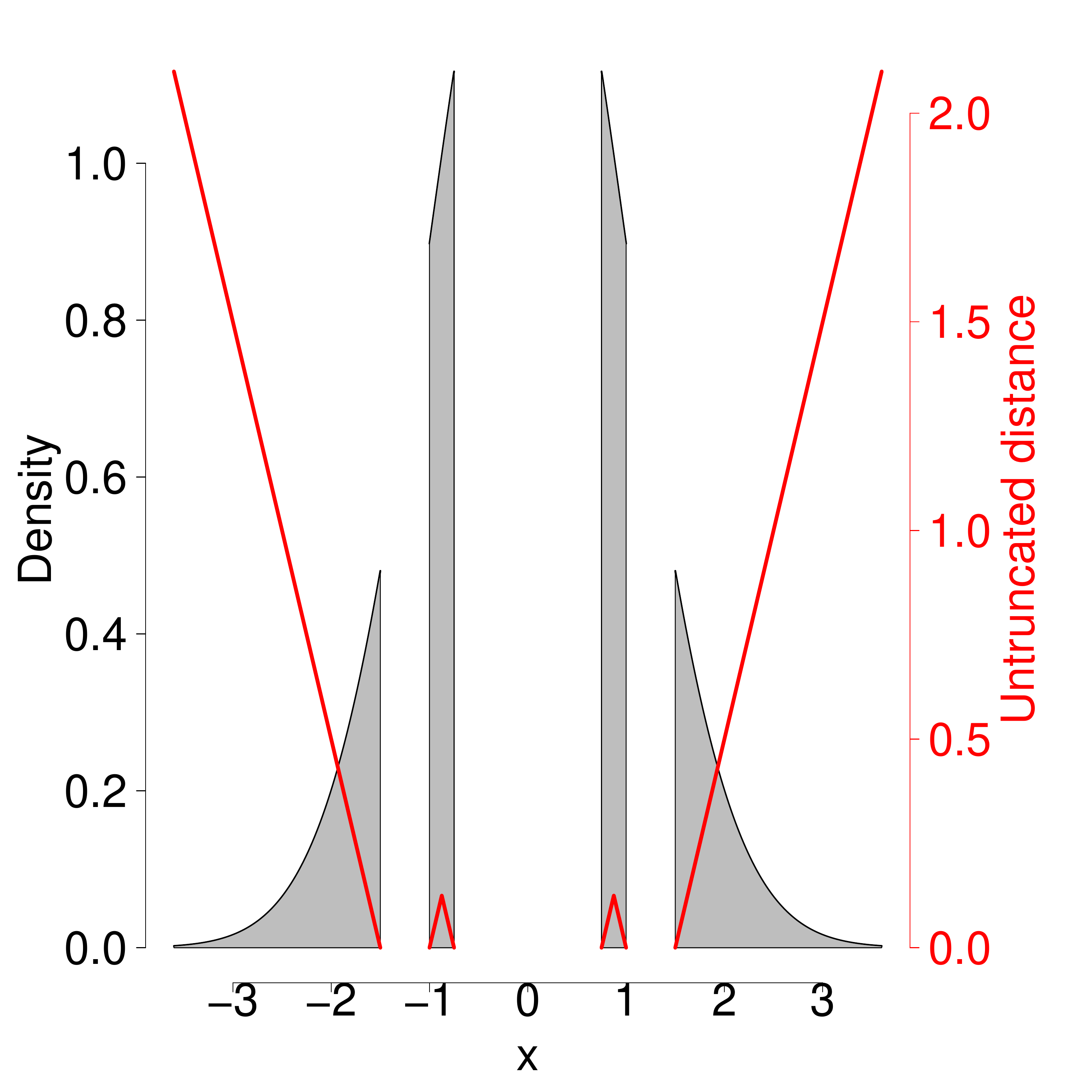}}\hspace{-0.1in}} &
\subfloat[Asymptotic $\log\mathrm{var}\left(\hat{\mu}\right)$ for  $\mathfrak{D}_1$]
{\hspace{-0.01in}\includegraphics[width=0.27\textwidth]{{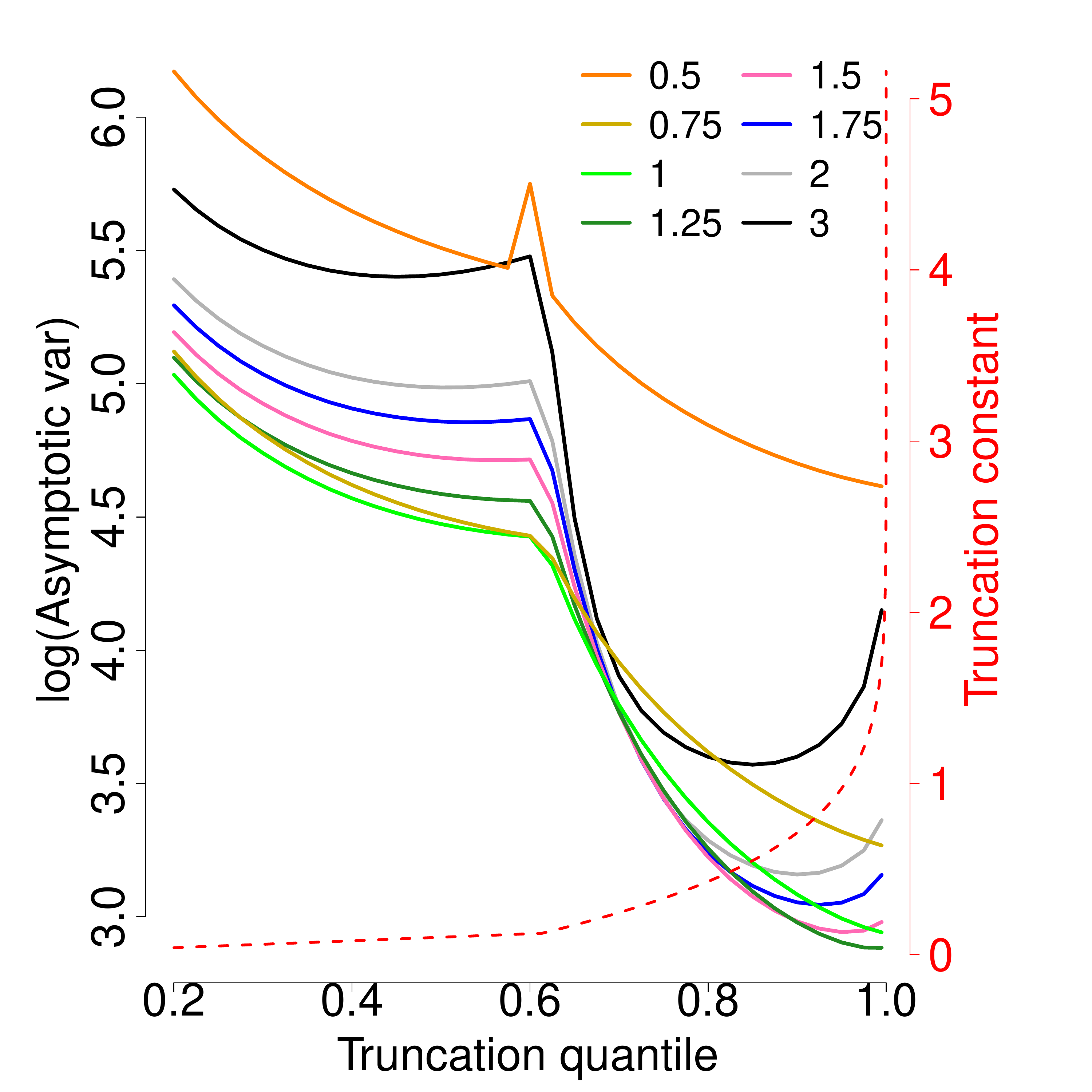}}\hspace{-0.01in}} &
\subfloat[Asymptotic $\log\mathrm{var}\left(\hat{\sigma}^2\right)$ for  $\mathfrak{D}_1$]
{\hspace{-0.01in}\includegraphics[width=0.27\textwidth]{{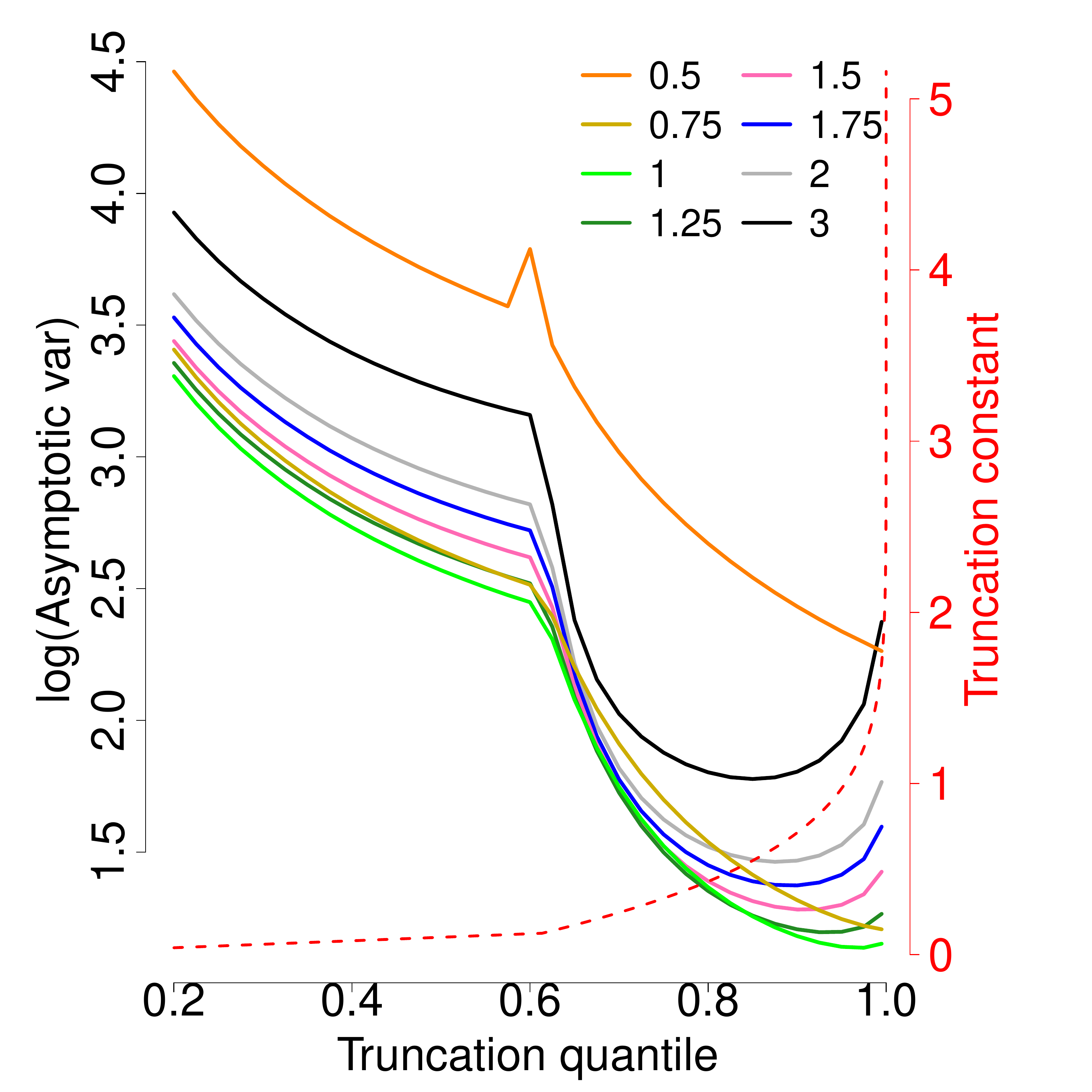}}\hspace{-0.01in}} \\ \hline
\subfloat[Density on $\mathfrak{D}_2$]
{\hspace{-0.1in}\includegraphics[width=0.27\textwidth]{{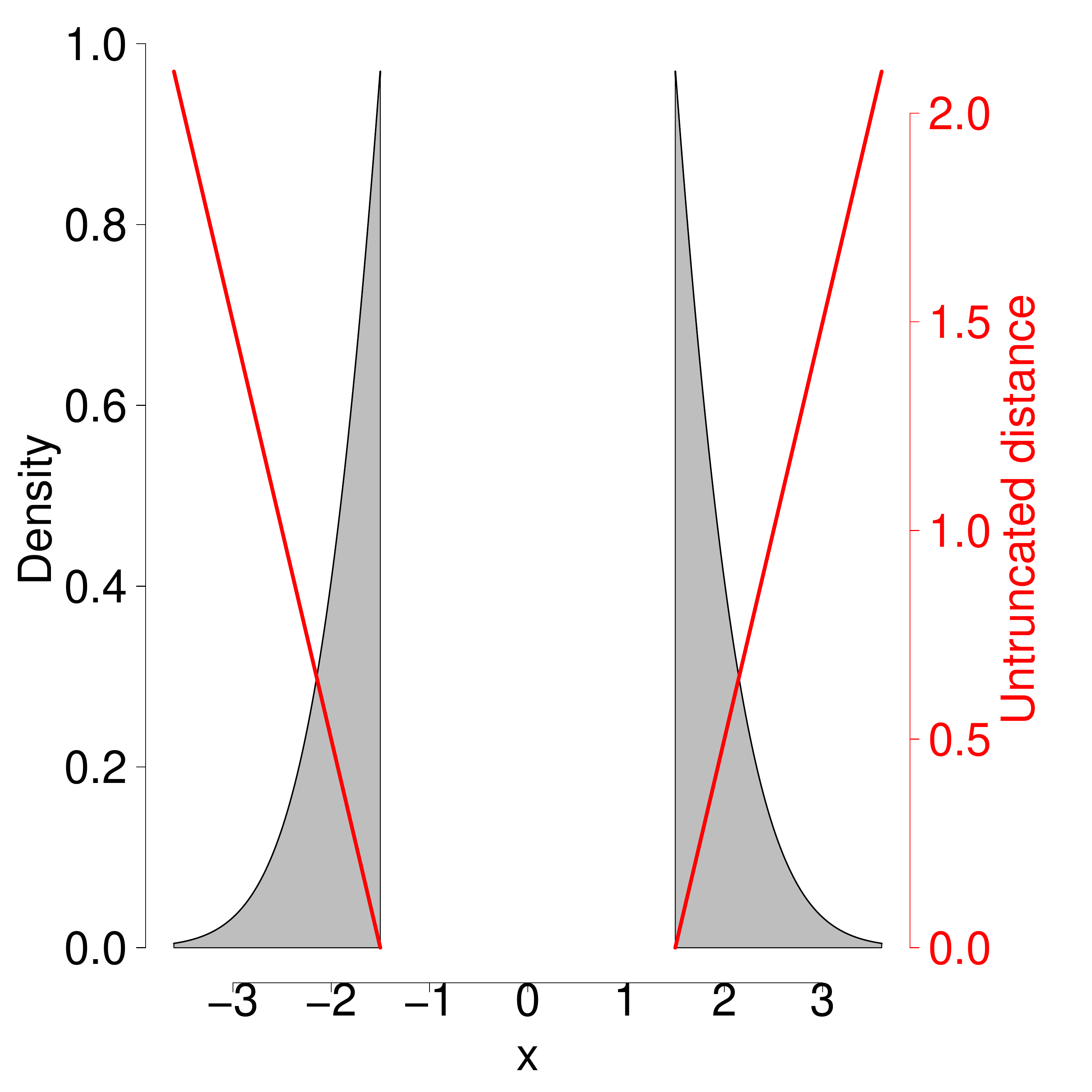}}\hspace{-0.1in}} &
\subfloat[Asymptotic $\log\var\left(\hat{\mu}\right)$ for  $\mathfrak{D}_2$]
{\hspace{-0.01in}\includegraphics[width=0.27\textwidth]{{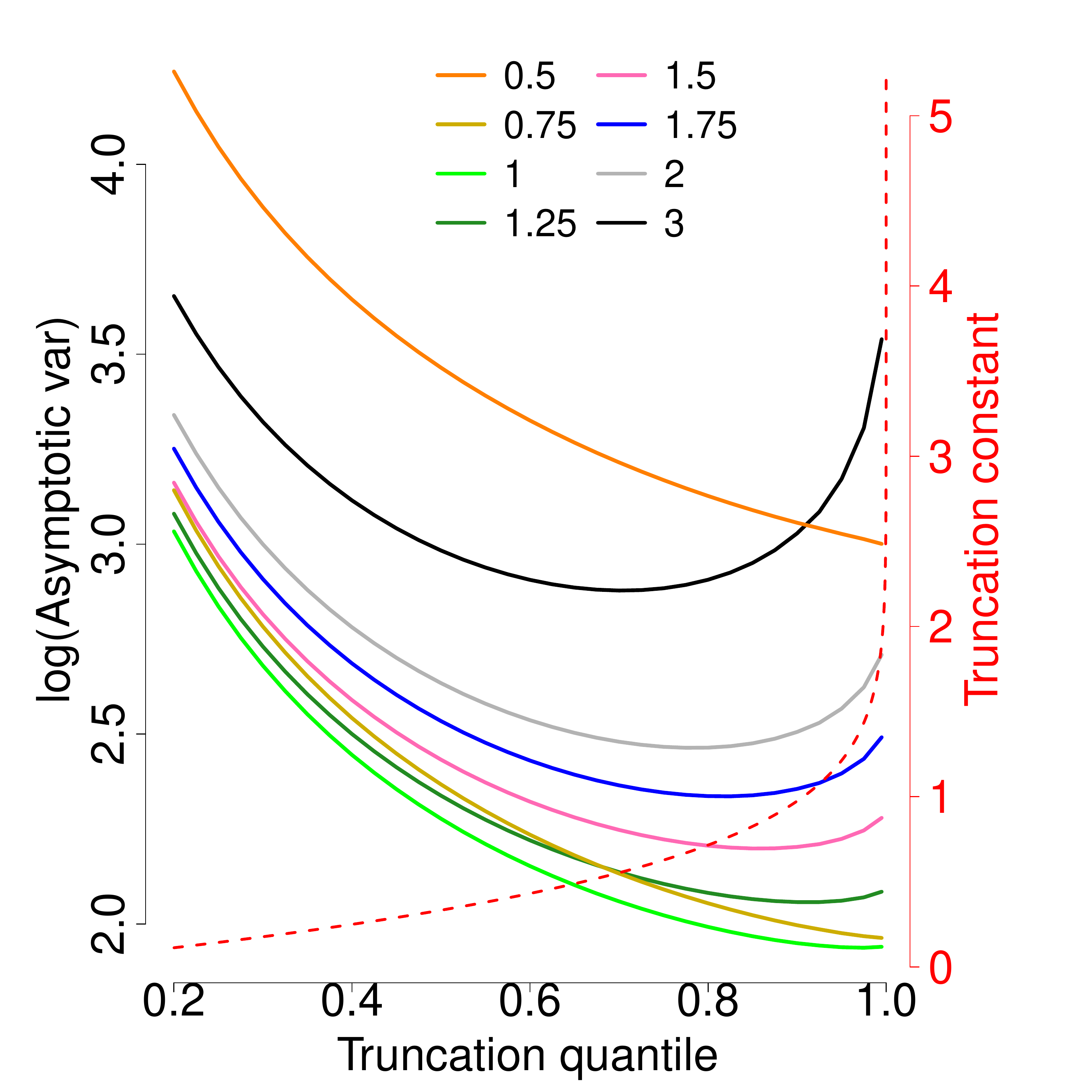}}\hspace{-0.01in}} & 
\subfloat[Asymptotic $\log\var\left(\hat{\sigma}^2\right)$ for $\mathfrak{D}_2$]
{\hspace{-0.01in}\includegraphics[width=0.27\textwidth]{{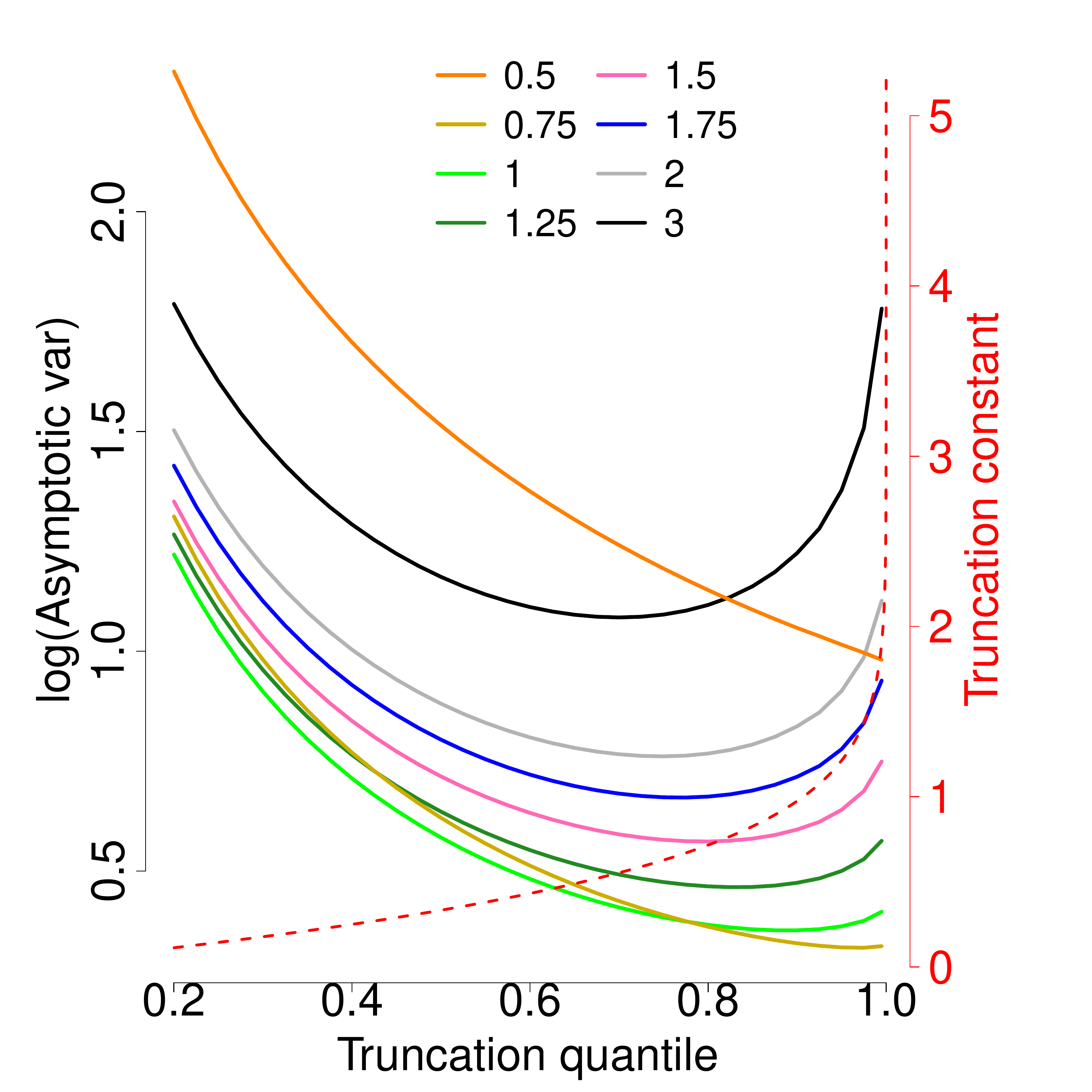}}\hspace{-0.01in}} \\ \hline
\subfloat[Density on $\mathfrak{D}_3$]
{\hspace{-0.1in}\includegraphics[width=0.27\textwidth]{{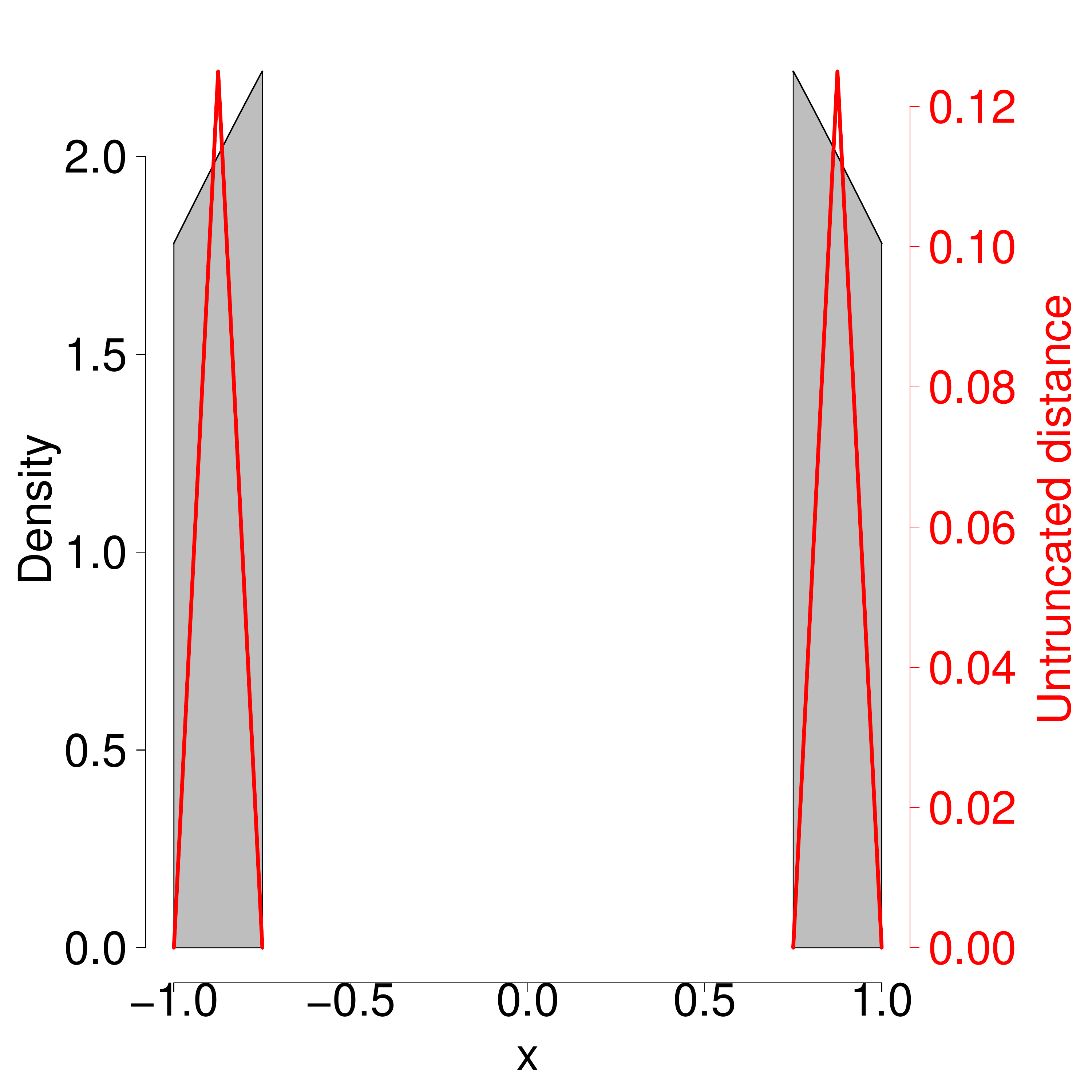}}\hspace{-0.1in}} &
\subfloat[Asymptotic $\log\var\left(\hat{\mu}\right)$ for $\mathfrak{D}_3$]
{\hspace{-0.01in}\includegraphics[width=0.27\textwidth]{{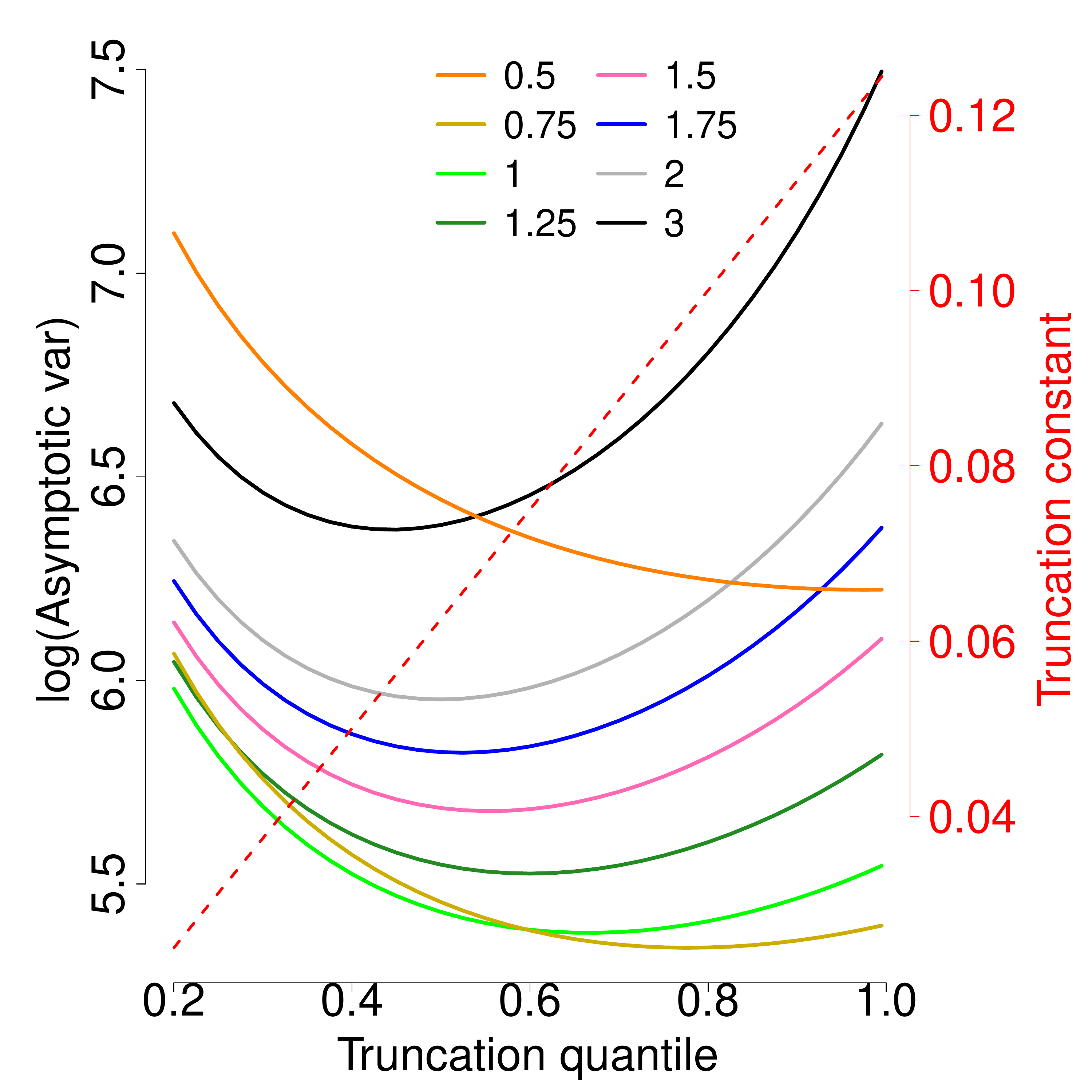}}\hspace{-0.01in}} &
\subfloat[Asymptotic $\log\var\left(\hat{\sigma}^2\right)$ for $\mathfrak{D}_3$]
{\hspace{-0.01in}\includegraphics[width=0.27\textwidth]{{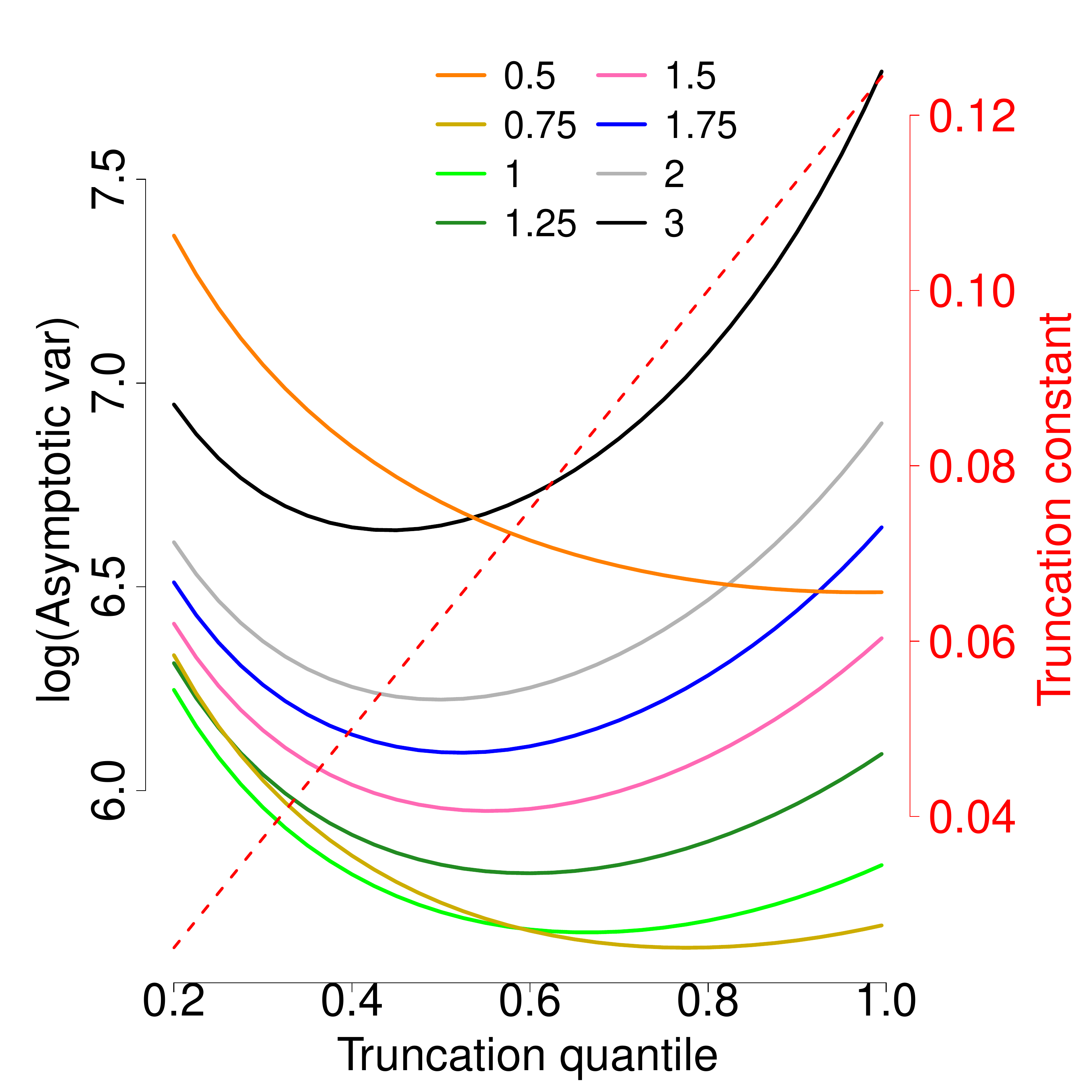}}\hspace{-0.01in}} \\ \hline
\end{tabular}
\caption{Univariate Gaussian example. Each row represents a domain,
  with the first subfigure plotting the density, the second the
  asymptotic log variance of $\hat{\mu}$, the third that of
  $\hat{\sigma}^2$. In each first subfigure the red lines show the
  untruncated distances $\varphi_{+\infty}$ for each domain, and the
  red dotted lines in the second and third show the truncation point
  $C$ versus the $\pi$ quantile.
}\label{plot_univariate}
\end{figure}

As we show in Section \ref{Theory}, for the purpose of edge recovery for graphical models, we recommend using $\boldsymbol{h}(\boldsymbol{x})=\left(x_1^{\alpha_1},\ldots,x_m^{\alpha_m}\right)$ with $\alpha\geq 1$ for $\mathfrak{D}$ that is a finite disjoint union of convex subsets of $\mathbb{R}^m$. Although minimizing the asymptotic variance in the univariate case is a different task, $\alpha=1$ also seems to be consistently the best performing choice.

For $\mathfrak{D}_2$ and $\mathfrak{D}_3$, all variance curves are
U-shaped, while for $\mathfrak{D}_1=\mathfrak{D}_2\cup\mathfrak{D}_3$
we see two such curves piecewise connected at
$C_0=\max\varphi_{+\infty,\mathfrak{D}_3}(x)=(1-3/4)/2=0.125$. To the
right of $C_0$, the truncation is applied to the two unbounded
intervals (i.e.~$\mathfrak{D}_2$) only. The first segments of most
$\mathrm{var}(\hat{\mu})$ curves for $\mathfrak{D}_1$ as well as most
curves for $\mathfrak{D}_2$ indicate there might still be benefit from
truncating the distances $\varphi$ within the bounded intervals,
although the $\mathrm{var}\left(\hat{\sigma}^2\right)$ curves for
$\mathfrak{D}_1$ as well as both curves for $x^{0.75}$ on
$\mathfrak{D}_2$ suggest otherwise. On the other hand, the curves for
$\mathfrak{D}_1$ and $\mathfrak{D}_2$ imply that a truncation constant
larger than $C_0$ is favorable; the ticks on the right-hand side
indicate that the curves for $\mathfrak{D}_2$ reach their minimum at
$C\geq 0.5$.  Hence,  a separate truncation point $C$ for each
connected component of $\mathfrak{D}$ could be beneficial, especially
for unbounded sets.  However, 
the necessary tuning of multiple parameters becomes infeasible for
$m\gg 1$ and we do not further examine it in this paper.

\section{Theoretical Properties}\label{Theory}
This section presents theoretical guarantees for our generalized
method applied to the pairwise interaction power $a$-$b$ models. We
first state a result analogous to \citet{yus19} for truncated Gaussian
densities on a general domain $\mathfrak{D}$,
and then present a general result for $a$-$b$ models. In particular,
similar to and as a generalization of \citet{yus19}, we give a high
probability bound on the deviation of our estimates $\hat{\mathbf{K}}$
and $\hat{\boldsymbol{\eta}}$ from their true values $\mathbf{K}_0$
and $\boldsymbol{\eta}_0$. The main challenge in deriving the results
lies in obtaining marginal probability tail bounds of each observed
value in $\boldsymbol{\Gamma}$ and $\boldsymbol{g}$. We first restate
Definition 12 from \citet{yus19}.

\begin{definition}\label{def_constants_centered}
Let
$\boldsymbol{\Gamma}_0\equiv\mathbb{E}_0\boldsymbol{\Gamma}(\mathbf{x})$
and $\boldsymbol{g}_0\equiv\mathbb{E}_0\boldsymbol{g}(\mathbf{x})$ be
the population versions of $\boldsymbol{\Gamma}(\mathbf{x})$ and
$\boldsymbol{g}(\mathbf{x})$ under the distribution given by a true
parameter matrix
$\mathbf{\Psi}_0\equiv\left[\mathbf{K}_0,\boldsymbol{\eta}_0\right]^{\top}\in\mathbb{R}^{m(m+1)}$,
or $\mathbf{\Psi}_0\equiv\mathbf{K}_0\in\mathbb{R}^{m^2}$ in the
centered case with $\boldsymbol{\eta}_0\equiv\boldsymbol{0}$.  The
support of a matrix $\mathbf{\Psi}$ is
$S(\mathbf{\Psi})\equiv\{(i,j):\psi_{ij}\neq 0\}$, and we let
$S_0=S(\mathbf{\Psi}_0)$.   We define
$d_{\mathbf{\Psi}_0}$ to be the maximum number of non-zero entries in
any column of  $\mathbf{\Psi}_0$, and $c_{\mathbf{\Psi}_0}\equiv\mnorm{\mathbf{\Psi}_0}_{\infty,\infty}$. Writing $\mathbf{\Gamma}_{0,AB}$ for the $A\times B$ submatrix of $\boldsymbol{\Gamma}_0$, we define
$c_{\boldsymbol{\Gamma}_0}\equiv\mnorm{(\boldsymbol{\Gamma}_{0,S_0S_0})^{-1}}_{\infty,\infty}.$
Finally, $\boldsymbol{\Gamma}_0$ satisfies \emph{the irrepresentability condition with incoherence parameter $\omega\in(0,1]$ and edge set $S_0$} if
\begin{equation}\label{irrepresentability}
\mnorm{\boldsymbol{\Gamma}_{0,S_0^cS_0}(\boldsymbol{\Gamma}_{0,S_0S_0})^{-1}}_{\infty,\infty}\leq (1-\omega).
\end{equation}
\end{definition}

\subsection{Truncated Gaussian Models on A Finite Disjoint Union of Convex Sets} 
Truncated Gaussian models are covered by our $a$-$b$ models described in Section \ref{Pairwise Interaction Power $a$-$b$ Models} with $a=b=1$. When the domain $\mathfrak{D}$ is a finite disjoint union of convex sets with a positive Lebesgue measure, we have the following theorem similar to Theorem 17 in \citet{yus19}, which bounds the errors as long as one uses finite truncation points $\boldsymbol{C}$ for $\boldsymbol{\varphi}_{\boldsymbol{C},\mathfrak{D}}$ and each component in $\boldsymbol{h}(\boldsymbol{x})$ is a power function with a positive exponent. 

Specifically, we consider the truncated Gaussian distribution on $\mathfrak{D}$ with inverse covariance parameter $\mathbf{K}_0\in\mathbb{R}^{m\times m}$ and mean parameter $\boldsymbol{\mu}_0$, namely with density
\[p_{\boldsymbol{\eta}_0,\mathbf{K}_0}(\boldsymbol{x})\propto\exp\left(-\frac{1}{2}{\boldsymbol{x}}^{\top}\mathbf{K}_0\boldsymbol{x}+\boldsymbol{\eta}_0^{\top}\boldsymbol{x}\right)\mathds{1}_{\mathfrak{D}}(\boldsymbol{x})\]
with $\mathbf{K}_0$ positive definite and $\boldsymbol{\eta}_0\equiv\mathbf{K}_0\boldsymbol{\mu}_0$. We assume $\mathfrak{D}\subset\mathbb{R}^m$ to be a component-wise countable union of intervals (Def \ref{def_V}) with positive Lebesgue measure, and assume it is a finite disjoint union of convex sets $\boldsymbol{\Delta}\equiv\left\{\mathfrak{D}_1,\ldots,\mathfrak{D}_{|\boldsymbol{\Delta}|}\right\}$, i.e.~$\mathfrak{D}\equiv\mathfrak{D}_1\sqcup\cdots\sqcup\mathfrak{D}_{|\boldsymbol{\Delta}|}$. 

\begin{theorem}\label{theorem_convex_ggm}
Suppose the data matrix contains $n$ i.i.d.~copies of $\boldsymbol{X}$ following a truncated Gaussian distribution on $\mathfrak{D}$ as above with parameters $\mathbf{K}_0\in\mathbb{R}^{m\times m}$ and $\boldsymbol{\mu}_0$,  Let $\mathbf{\Psi}_0\equiv\left[\mathbf{K}_0,\boldsymbol{\eta}_0\right]^{\top}\equiv \left[\mathbf{K}_0,\mathbf{K}_0\boldsymbol{\mu}_0\right]^{\top}$. Assume that (A.1)--(A.3) in Lemma \ref{lem_loss} hold, and in addition that $\boldsymbol{h}$ and the truncation points $\boldsymbol{C}$ in the truncated component-wise distance $\boldsymbol{\varphi}_{\boldsymbol{C},\mathfrak{D}}$ satisfy $0\leq (h_j\circ\varphi_{C_j,\mathfrak{D},j})(\boldsymbol{x})\leq M$ and $0\leq \partial_j(h_j\circ\varphi_{C_j,\mathfrak{D},j})(\boldsymbol{x})\leq M'$ almost surely for all $j=1,\ldots,m$ for some constants $0<M,M'<+\infty$. Note that $\boldsymbol{h}(\boldsymbol{x})=(x_1^{\alpha_1},\ldots,x_m^{\alpha_m})$ with $\alpha_1,\ldots,\alpha_m\geq 1$ satisfies all these assumptions, according to Theorem \ref{thm_A}. Let the diagonal multiplier $\delta$ introduced in Section \ref{Exponential Families and Regularized Score Matching} satisfy
\[1<\delta<C(n,m)\equiv 2-\left(1+4e\max\left\{\left(6\log m+2\log|\boldsymbol{\Delta}|\right)/n,\sqrt{\left(6\log m+2\log|\boldsymbol{\Delta}|\right)/n}\right\}\right)^{-1}\]
and suppose further that $\boldsymbol{\Gamma}_{0,S_0S_0}$ is invertible and satisfies the irrepresentability condition (\ref{irrepresentability}) with $\omega\in(0,1]$. Define $c_{\boldsymbol{X}}\equiv2\max_{\mathfrak{D}'\in\boldsymbol{\Delta}}\max_j\Big|2\sqrt{(\mathbf{K}_0^{-1})_{jj}}+\sqrt{e}\,\mathbb{E}_{0}X_j\mathds{1}_{\mathfrak{D}'}(\boldsymbol{X})\Big|$. Suppose for $\tau>3$ the sample size and the regularization
 parameter satisfy
\begin{align}
n&>\mathcal{O}\left(\left(\tau\log m+\log|\boldsymbol{\Delta}|\right)\max\left\{\frac{M^2c_{\boldsymbol{\Gamma}_0}^2c_{\boldsymbol{X}}^4d_{\boldsymbol{\Psi}_0}^2}{\omega^2},\frac{Mc_{\boldsymbol{\Gamma}_0}c_{\boldsymbol{X}}^2d_{\boldsymbol{\Psi}_0}}{\omega}\right\}\right),\\
\lambda&>\mathcal{O}\left[(Mc_{\boldsymbol{\Psi}_0}c_{\boldsymbol{X}}^2+M'c_{\boldsymbol{X}}+M)\left(
\sqrt{\frac{\tau\log m+\log|\boldsymbol{\Delta}|}{n}}+\frac{\tau\log m+\log|\boldsymbol{\Delta}|}{n}\right)\right].
\end{align}
Then the following statements hold with probability $1-m^{3-\tau}$:
\begin{enumerate}[(a)]
\item The regularized generalized $\boldsymbol{h}$-score matching estimator $\hat{\boldsymbol{\Psi}}$ that minimizes (\ref{eq_loss_regularized}) is unique, has its support included in the true support, $\hat{S}\equiv S(\hat{\boldsymbol{\Psi}})\subseteq S_0$, and satisfies
\begin{alignat*}{3}
\|\hat{\mathbf{K}}-\mathbf{K}_0\|_{\infty}&\leq\frac{c_{\boldsymbol{\Gamma}_0}}{2-\omega}\lambda,&&\|\hat{\boldsymbol{\eta}}-\boldsymbol{\eta}_0\|_{\infty}\leq\frac{c_{\boldsymbol{\Gamma}_0}}{2-\omega}\lambda,\\
\mnorm{\hat{\mathbf{K}}-\mathbf{K}_0}_{F}&\leq\frac{c_{\boldsymbol{\Gamma}_0}}{2-\omega}\lambda\sqrt{|S_0|},&&\mnorm{\hat{\boldsymbol{\eta}}-\boldsymbol{\eta}_0}_{F}\leq\frac{c_{\boldsymbol{\Gamma}_0}}{2-\omega}\lambda\sqrt{|S_0|},\\
\mnorm{\hat{\mathbf{K}}-\mathbf{K}_0}_{2}&\leq\frac{c_{\boldsymbol{\Gamma}_0}}{2-\omega}\lambda\min\left(\sqrt{|S_0|},d_{\boldsymbol{\Psi}_0}\right),
\qquad&& 
\mnorm{\hat{\boldsymbol{\eta}}-\boldsymbol{\eta}_0}_{2}\leq\frac{c_{\boldsymbol{\Gamma}_0}}{2-\omega}\lambda\min\left(\sqrt{|S_0|},d_{\boldsymbol{\Psi}_0}\right).
\end{alignat*}
\item If in addition $\min_{j,k:(j,k)\in S_0}|\kappa_{0,jk}|>\frac{c_{\boldsymbol{\Gamma}_0}}{2-\omega}$ and $\min_{j:(m+1,j)\in S_0}|\eta_{0,j}|>\frac{c_{\boldsymbol{\Gamma}_0}}{2-\omega}\lambda$,
then we have $\hat{S}=S_0$, $\mathrm{sign}(\hat{\kappa}_{jk})=\mathrm{sign}(\kappa_{0,jk})$ for all $(j,k)\in S_0$ and $\mathrm{sign}(\hat{\eta}_j)=\mathrm{sign}(\eta_{0j})$ for $(m+1,j)\in S_0$.
\end{enumerate}
In the centered setting, the same bounds hold by removing the dependencies on $\hat{\boldsymbol{\eta}}$ and $\boldsymbol{\eta}_0$.
\end{theorem}

The proposed method naturally extends our previous work, and the above
results follow by applying the proof for Theorem 17 of
\citet{yus19} with two modifications: (i) using the triangle
inequality, split the concentration bounds in (39), (43), and (44) in
\citet{yus19} into one for each set
$\mathfrak{D}_1,\ldots,\mathfrak{D}_{|\boldsymbol{\Delta}|}$ and
combine the results with a union bound; (ii) in the proof of Lemma
22.1 of \citet{yus19}, replace $\mathfrak{D}\equiv\mathbb{R}_+^m$ by
each 
$\mathfrak{D}'=\mathfrak{D}_1,\ldots,\mathfrak{D}_{|\boldsymbol{\Delta}|}$
and replace $X_1$ by $X_1\mathds{1}_{\mathfrak{D}'}(\boldsymbol{X})$,
as the proof there only uses the convexity of the
domain.

\subsection{Bounded Domains in $\mathbb{R}_+^m$ with Positive Measure}
In this section we present results for general $a$-$b$ models on bounded domains with positive measure.

\begin{theorem}\label{theorem_bounded_nonlog_ggm}
(1) Suppose $a>0$ and $b\geq 0$. Let $\mathfrak{D}$ be a bounded
subset of $\mathbb{R}_+^m$ with positive Lebesgue measure with
$\mathfrak{D}\subseteq[u_1,v_1]\times\cdots\times [u_m,v_m]$ for
finite nonnegative constants $u_1,v_1,\ldots,u_m,v_m$, and suppose
that the true parameters $\mathbf{K}_0$ and $\boldsymbol{\eta}_0$
satisfy the conditions in Corollary \ref{cor_norm_const} (for a well-defined density). Assume $\boldsymbol{h}(\boldsymbol{x})\equiv(x_1^{\alpha_1},\ldots,x_m^{\alpha_m})$ with $\alpha_1,\ldots,\alpha_m\geq \max\{1,2-a,2-b\}$, and suppose $\boldsymbol{\varphi}_{\boldsymbol{C},\mathfrak{D}}$ has truncation points $\boldsymbol{C}=(C_1,\ldots,C_m)$ with $0<C_j< +\infty$ for $j=1,\ldots,m$. Define
\begin{align*}
\zeta_j(\alpha_j,p_j)&\equiv\begin{cases}\min\left\{C_j,(v_j-u_j)/2\right\}^{\alpha_j}(u_j+v_j)^{p_j}/2^{p_j}, & p_j<0, v_j-u_j\leq 2C_j, \\ 
\min\left\{C_j,(v_j-u_j)/2\right\}^{\alpha_j}(u_j+C_j)^{p_j}, & p_j<0, v_j-u_j> 2C_j, \\
\min\left\{C_j,(v_j-u_j)/2\right\}^{\alpha_j}v_j^{p_j}, & p_j\geq 0,\end{cases}\\
\varsigma_{\boldsymbol{\Gamma}}&\,\equiv\max\limits_{j,k =1,\ldots,m}\max\left\{\zeta_j(\alpha_j,2a-2)v_k^{2a},\,\zeta_j(\alpha_j,2b-2)\right\},\\
\varsigma_{\boldsymbol{g}}&\,\equiv\max\limits_{j,k =1,\ldots,m}\max\big\{\alpha_j\zeta_j(\alpha_j-1,a-1)v_k^{a}+|a-1|\zeta_j(\alpha_j,a-2)v_k^{a}+a\zeta_j(\alpha_j,2a-2),\\
&\quad\quad\quad\quad\quad\quad\quad\,\,\alpha_j\zeta_j(\alpha_j-1,b-1)+|b-1|\zeta_j(\alpha_j,b-2)\big\}.
\end{align*}
Suppose that $\boldsymbol{\Gamma}_{0,S_0S_0}$ is invertible and satisfies the irrepresentability condition (\ref{irrepresentability}) with $\omega\in(0,1]$. Suppose that for $\tau>0$ the sample size, the regularization parameter and the diagonal multiplier satisfy
\begin{align}
n&>72c_{\boldsymbol{\Gamma}_0}^2d_{\boldsymbol{\Psi}_0}^2\varsigma_{\boldsymbol{\Gamma}}^2(\tau\log m+\log 4)/\omega^2,\label{eq_bounded_nonlog_n}\\
\lambda&>\frac{3(2-\omega)}{\omega}\max\left\{c_{\boldsymbol{\Psi}_0}\varsigma_{\boldsymbol{\Gamma}}\sqrt{2(\tau\log m+\log 4)/n},\varsigma_{\boldsymbol{g}}\sqrt{(\tau\log m+\log 4)/(2n)}\right\},\\
1&<\delta<C_{\text{bounded}}(n,m,\tau)\equiv 1+\sqrt{(\tau\log m+\log 4)/(2n)}.\label{eq_bounded_nonlog_delta}
\end{align}
Then the statements (a) and (b) in Theorem \ref{theorem_convex_ggm} hold with probability at least $1-m^{-\tau}$.

(2) For $a=0$ and $b\geq 0$, if $u_1,\ldots,u_m>0$, letting $w_j\equiv\max\{|\log u_j|,|\log v_j|\}$, the above holds with 
\begin{align*}
\varsigma_{\boldsymbol{\Gamma}}&\equiv\max\limits_{j,k =1,\ldots,m}\max\{\zeta_j(\alpha_j,-2)w_k^2,\,\zeta_j(\alpha_j,2b-2)\},\\
\varsigma_{\boldsymbol{g}}&\equiv\max\limits_{j,k =1,\ldots,m}\max\{\alpha_j\zeta_j(\alpha_j-1,-1)w_k+|a-1|\zeta_j(\alpha_j,-2)w_k+a\zeta_j(\alpha_j,-2),\\
&\phantom{\max\limits_{j,k =1,\ldots,m}\max\{}\alpha_j\zeta_j(\alpha_j-1,b-1)+|b-1|\zeta_j(\alpha_j,b-2)\}.
\end{align*}
\end{theorem}
We note that the requirement on $\alpha_j\geq 1$ is only for bounding the two $\partial_j(h_j\circ\varphi_j)$ terms in $\boldsymbol{g}(\mathbf{x})$ and might not be necessary in practice as we see in the simulation studies. 

\subsection{Unbounded Domains in $\mathbb{R}_+^m$ with Positive Measure}

For unbounded domains $\mathfrak{D}\subset\mathbb{R}_+^m$ in the
non-negative orthant, we are able to give consistency results but only
with a sample complexity that includes
an additional
unknown constant factor that may depend on $m$. For simplicity we only show the results for $a>0$. 
The following lemma enables us to bound each row of the data matrix
$\mathbf{x}$ by a finite cube with high probability and then proceed
as for Theorem \ref{theorem_bounded_nonlog_ggm}.

\begin{lemma}\label{theorem_subexp_positive_measure}
Suppose $\mathfrak{D}$ has positive measure, and the true parameters $\mathbf{K}_0$ and $\boldsymbol{\eta}_0$ satisfy the conditions in Corollary \ref{cor_norm_const}. Then for all $j=1,\ldots,m$, $X_j^{2a}$ is sub-exponential if $a>0$ and $\log X_j$ is sub-exponential if $a=0$.
\end{lemma}

We have the following corollary of Theorem
\ref{theorem_bounded_nonlog_ggm}. The result involves an unknown
constant, namely the sub-exponential norm $\|\cdot\|_{\psi_1}$ of
$X_j^{2a}$, and
$$
n=\Omega\left(\log
  m\right)\cdot\max_j\mathcal{O}\left(\left\|X_j^{2a}\right\|_{\psi_1}\right)^{(\alpha_j+\max\{4a,2b\}-2)/a}
$$
becomes the required sample complexity. We conjecture that the
sub-exponential norm scales like $\Omega\left(\left(\log
    m\right)^c\right)$ for $c$ small, but leave the exact dependency on $m$  for further research.
\begin{corollary}\label{corollary_apos_unbounded}
Suppose $a>0$ and $\mathfrak{D}$ is a subset of $\mathbb{R}_+^m$ with positive measure and suppose that the true parameters $\mathbf{K}_0$ and $\boldsymbol{\eta}_0$ satisfy the conditions in Corollary \ref{cor_norm_const}. Let $\rho_j(\mathfrak{D})\equiv\overline{\{x_j:\boldsymbol{x}\in\mathfrak{D}\}}$ and $\rho_{\mathfrak{D}}^*\equiv\{j=1,\ldots,m:\sup\rho_j(\mathfrak{D})<+\infty\}$, and suppose $\rho_j(\mathfrak{D})\subseteq[u_j,v_j]$ for $j\in\rho_{\mathfrak{D}}^*$. Then Theorem \ref{theorem_bounded_nonlog_ggm} holds with $\log 4$ replaced by $\log 6$ in (\ref{eq_bounded_nonlog_n})--(\ref{eq_bounded_nonlog_delta}), and $u_j=\max\left\{\mathbb{E}_0 X_j^{2a}-\epsilon_{3,j},0\right\}^{1/(2a)}$ and $v_j=\left(\mathbb{E}_0 X_j^{2a}+\epsilon_{3,j}\right)^{1/(2a)}$ for $j\not\in\rho_{\mathfrak{D}}^*$, where
\begin{align*}
\epsilon_{3,j}&\equiv\max\left\{2\sqrt{2}e\left\| X_j^{2a}\right\|_{\psi_1}\sqrt{\log 3+\log n+\tau\log m+\log\left(m-|\rho_{\mathfrak{D}}^*|\right)},\right.\\
&\quad\quad\quad\quad\quad\left. 4e\left\| X_j^{2a}\right\|_{\psi_1}\left(\log 3+\log n+\tau\log m+\log\left(m-|\rho_{\mathfrak{D}}^*|\right)\right)\right\},\\
\left\| X_j^{2a}\right\|_{\psi_1}&\equiv\sup_{q\geq 1}\left(\mathbb{E}_0| X_j|^{2aq}\right)^{1/q}/q\geq \mathbb{E}_0 X_j^{2a}.
\end{align*}
\end{corollary}


\section{Numerical Experiments}\label{Numerical Experiments}

In this section we present results of numerical experiments using our
method from Sections \ref{Methodology} and \ref{Lemmas and Empirical
  Loss}, as well as our extension of \citet{liu19} from Section
\ref{Extension of g0}.

\subsection{Estimation --- Choice of $\boldsymbol{h}$ and $\boldsymbol{C}$}\label{sec_simulation_h_and_C}
Multiplying the gradient $\nabla\log p(\boldsymbol{x})$
with functions $(\boldsymbol{h}\circ\boldsymbol{\varphi}_{\boldsymbol{C},\mathfrak{D}})^{1/2}(\boldsymbol{x})$ is key to our method, where the $j$-th component of $\boldsymbol{\varphi}_{\boldsymbol{C},\mathfrak{D}}(\boldsymbol{x})=(\varphi_{C_1,\mathfrak{D},1}(\boldsymbol{x}),\ldots,\varphi_{C_m,\mathfrak{D},m}(\boldsymbol{x}))$ is the distance of $x_j$ to the boundary of its domain holding $\boldsymbol{x}_{-j}$ fixed, with this distance truncated from above by some constant $C_j>0$.
We use a single function $h$ for all components
(so, $\boldsymbol{h}(\boldsymbol{x})=(h(x_1),\ldots,h(x_m))$),
which we choose as $h(x)=x^{c}$ with exponent
$c=i/4$ for $i=0,1,\ldots,8$.  Instead of
pre-specifying  truncation points in $\boldsymbol{C}$,
we select
$0<\pi\leq 1$ and set each $C_j$ to be the $\pi$th sample quantile of
$\{\varphi_{+\infty,\mathfrak{D},j}(\boldsymbol{x}^{(1)}),\ldots,\varphi_{+\infty,\mathfrak{D},j}(\boldsymbol{x}^{(n)})\}$,
where $\boldsymbol{x}^{(i)}$ is the $i$th row of the data
matrix $\mathbf{x}$.
Infinite values of $\varphi_{+\infty,\mathfrak{D},j}$ are ignored,
and $\varphi_j\equiv 1$ if
$\varphi_{+\infty,\mathfrak{D},j}(\boldsymbol{x}^{(1)})=\cdots=\varphi_{+\infty,\mathfrak{D},j}(\boldsymbol{x}^{(n)})=+\infty$. This
empirical choice of the truncation points automatically adapts to the
scale of data, and we found it to be more effective than fixing the
constant to a grid from 0.5 to 3 as done in \citet{yus19}. Our
experiments consider $\pi=0.2,0.4,0.6,0.8,1$, where $\pi=1$ means
no truncation of finite distances.

Note that for $c=0$,
$(\boldsymbol{h}\circ\boldsymbol{\varphi}_{\boldsymbol{C},\mathfrak{D}})(\boldsymbol{x})\equiv
1$ and the method reduces to the original score-matching for
$\mathbb{R}^m$ of \citet{hyv05}.  With $c=2$, $\boldsymbol{C}=(+\infty,\dots,+\infty)$ and $\mathfrak{D}\equiv\mathbb{R}_+^m$, $(\boldsymbol{h}\circ\boldsymbol{\varphi}_{\boldsymbol{C},\mathfrak{D}})(\boldsymbol{x})\equiv \boldsymbol{x}^2$ corresponds to the estimator of \citet{hyv07,lin16}. The case where $\mathfrak{D}\equiv\mathbb{R}_+^m$ corresponds to \citet{yu18,yus19}.

We also consider  our extension of the method from \citet{liu19}, for
which we use $g_0(\boldsymbol{x})$ from
\eqref{eq_extended_g0} as opposed to
$(\boldsymbol{h}\circ\boldsymbol{\varphi}_{\boldsymbol{C},\mathfrak{D}})^{1/2}(\boldsymbol{x})$;
see Section \ref{Extension of g0}. The constant $C$ in this case is
also determined using quantiles, but now of the untruncated $\ell_2$
distances of the given data sample to
$\underline{\partial}\mathfrak{D}$. For $C=+\infty$ ($\pi=1$) there is
no truncation and the estimator corresponds to \citet{liu19}.

\subsection{Numerical Experiments for Domains with Positive Measure}\label{Numerical Experiments for Domains with Positive Measure}
We present results for general $a$-$b$ models restricted to domains with positive Lebesgue measure.

\subsubsection{Experimental Setup}

Throughout, we choose dimension $m=100$ and sample sizes $n=80$ and
$1000$. For brevity, we only present results for the centered case
(assuming $\boldsymbol{\eta}\equiv \boldsymbol{0}$) where the $b$
power does not come into play, i.e., the density is
proportional to
$\exp\left\{-{\boldsymbol{x}^a}^{\top}\mathbf{K}\boldsymbol{x}^a/(2a)\right\}$
for $a>0$ or
$\exp\left(- {\log
    \boldsymbol{x}}^{\top}\mathbf{K}\log\boldsymbol{x}/2\right)$ for
$a=0$.  Indeed, the experiments in \citet{yus19} suggest that the
results for non-centered settings are similar with the best
choice of $h$ mainly depending on $a$ but not $b$.

We consider six settings: 
(1) $a=0$ (log), (2) $a=1/2$ (exponential square root; \citet{ino16}),
(3) $a=1$ (Gaussian), (4) $a=3/2$ as well as some more extreme cases (5) $a=2$ and (6) $a=3$. For all settings, we consider the following subsets of $\mathbb{R}^m_+$ as our domain $\mathfrak{D}$:
\begin{enumerate}[i)]
\item non-negative $\ell_2$ ball $\{\boldsymbol{x}\in\mathbb{R}_+^m:\|\boldsymbol{x}\|_2\leq c_1\}$, which we call $\ell_2$-nn (``non-negative''),
\item complement of $\ell_2$ ball in $\mathbb{R}_+^m$: $\{\boldsymbol{x}\in\mathbb{R}_+^m:\|\boldsymbol{x}\|_2\geq c_1\}$, which we call $\ell_2^{\complement}$-nn, and 
\item $[c_1,+\infty)^m$, which we call \emph{unif-nn},
\end{enumerate}
for some $c_1>0$. For the Gaussian ($a=1$) case consider in addition the following subsets of $\mathbb{R}^m$:
\begin{enumerate}[i)]
\setcounter{enumi}{3}
\item the entire $\ell_2$ ball $\{\boldsymbol{x}\in\mathbb{R}^m:\|\boldsymbol{x}\|_2\leq c_1\}$, which we call $\ell_2$,
\item the complement of $\ell_2$ ball in $\mathbb{R}^m$: $\{\boldsymbol{x}\in\mathbb{R}^m:\|\boldsymbol{x}\|_2\geq c_1\}$, which we call $\ell_2^{\complement}$, and
\item $\left((-\infty,c_1]\cup[c_1,+\infty)\right)^m$, which we call \emph{unif}.
\end{enumerate}

The constant $c_1$ in each setting above is determined in the
following way. We first generate $n$ samples from the corresponding
untruncated distribution on $\mathbb{R}^m_+$ for i)--iii) or
$\mathbb{R}^m$ for iv)--vi), then determine the $c_1$ so that exactly
half of the samples would fall inside the truncated boundary.


The true interaction matrices $\mathbf{K}_0$ are taken
block-diagonal as in \citet{lin16} and \citet{yus19}, with 10 blocks
of size $m/10=10$. In each block, each lower-triangular element is set
to $0$ with probability $1-\rho$ for some $\rho\in(0,1)$, and is
otherwise drawn from $\text{Uniform}[0.5,1]$. Symmetry determines the upper triangular
elements. The diagonal elements are chosen as a common positive value such that
$\mathbf{K}_0$ has minimum eigenvalue $0.1$. We generate 5 different
$\mathbf{K}_0$, and run 10 trials for each of them. We choose $(\rho,n)=(0.8,1000)$ and $(0.2,80)$ so that $n/(d_{\mathbf{K}_0}^2\log m)$ is roughly constant, recall our theory in Section \ref{Theory}.

\subsubsection{Results}

Our focus is on recovery of the support of
$\mathbf{K}_0=(\kappa_{0,i,j})$, i.e., the set
$S_{0,\text{off}}\equiv\{(i,j):i\neq j\wedge\kappa_{0,i,j}\neq 0\}$
which corresponds to an undirected graph with
$S_{0,\text{off}}$ as edge set. We use the area under the ROC curve
(AUC) as the measure of performance.  Let $\hat{\mathbf{K}}$ be an
estimate with support
$\hat{S}_{\text{off}}\equiv\{(i,j):i\neq j\wedge\hat{\kappa}_{i,j}\neq
0\}$.  Then the ROC curve plots the true positive rate (TPR) against
the false positive rate (FPR), with
\[\mathrm{FPR}\equiv\frac{|\hat{S}_{\text{off}}\backslash S_{0,\text{off}}|}{m(m-1)-|S_{0,\text{off}}|}\quad\quad\text{and}\quad\quad \mathrm{TPR}\equiv\frac{|\hat{S}_{\text{off}}\cap S_{0,\text{off}}|}{|S_{0,\text{off}}|}.\]
We plot the AUC averaged over all 50 trials in each setting against
the probability $\pi$ used to set the truncation
points $\boldsymbol{C}$.
Each plotted curve is for one choice of the function
$h(x)$, or for $g_0(\boldsymbol{x})$.
The $y$-ticks on the right-hand side are the original AUC values,
whereas those on the left are the AUCs divided by the AUC for
$h(x)=1$, measuring the relative performance of each method compared
to the original score matching in \citet{hyv05}; $h(x)=1$
does not depend on the truncation and is constant in each plot.

Plots for $a\leq 1$ are shown below, and for $a>1$ in Appendix
\ref{Additional Plots}.   We conclude that in most settings our method using
$h(x)=x^c$ with $c\approx\max\{2-a,0\}$ works the best, as we
also observed in \citet{yus19}.  In most settings the truncated
$g_0$ function does not work well (\citet{liu19} corresponds to
$\pi=1$). The only notable exceptions are the domains iv)--vi), i.e., Gaussian models on subsets of $\mathbb{R}^m$ not restricted to
$\mathbb{R}_+^m$, see Figure \ref{plot_gaussian}. The original
score matching in \citet{hyv05}
seems to work the best in these settings, suggesting
that estimation of Gaussians on such domains might not be
challenging enough to warrant switching to the more complex generalized
methods. However, 
Figure \ref{plot_gaussian}, for the iv) $\ell_2$ and v)
$\ell_2^{\complement}$ domains, shows only insignificant differences
in the performance of all estimators. 

\begin{figure}[htp]
\centering
\vspace{-0.3in}
\subfloat[$n=80$, $\ell_2$-nn domain]
{\includegraphics[width=0.45\textwidth]{{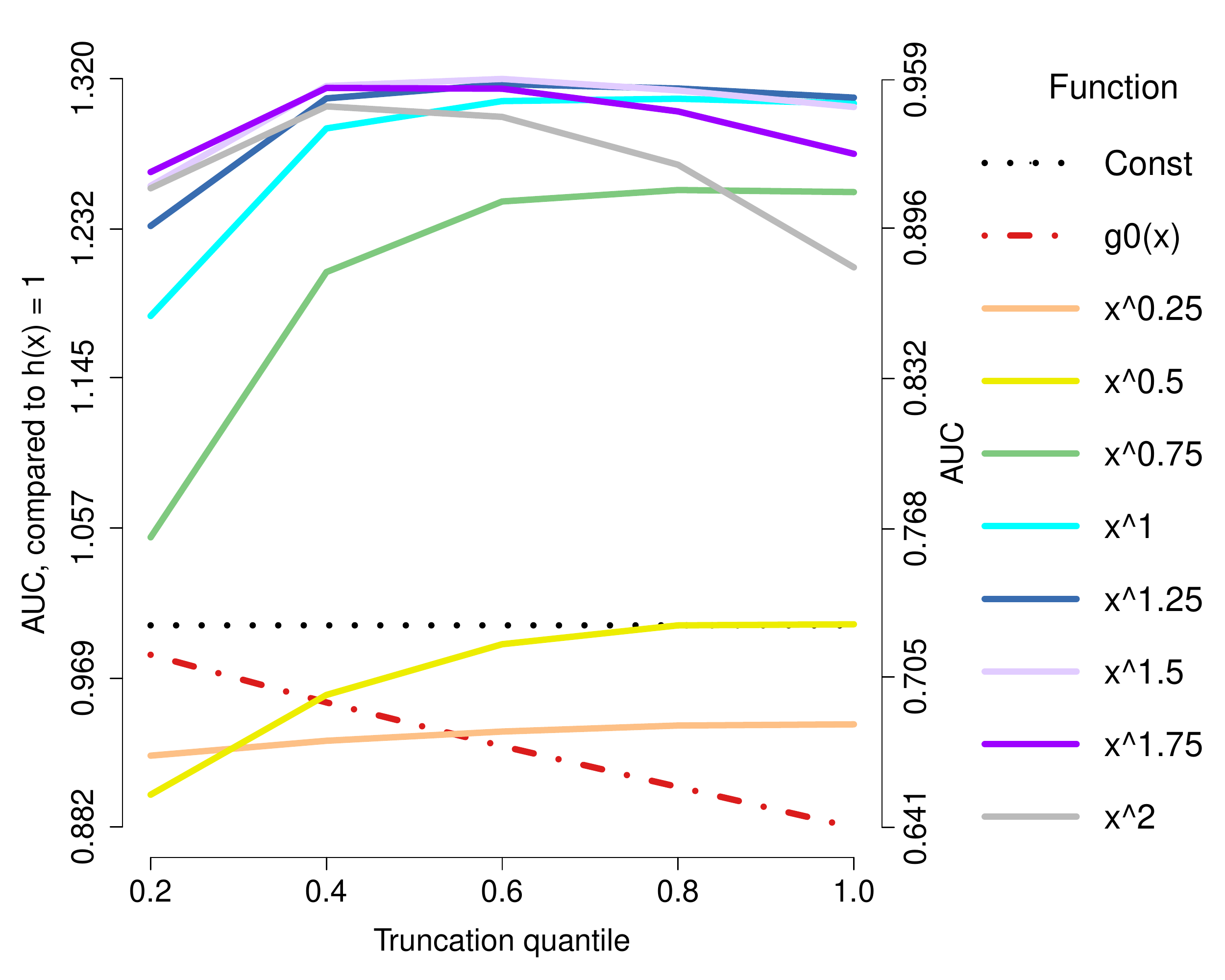}}\hspace{0.1in}}
\subfloat[$n=1000$, $\ell_2$-nn domain]
{\includegraphics[width=0.45\textwidth]{{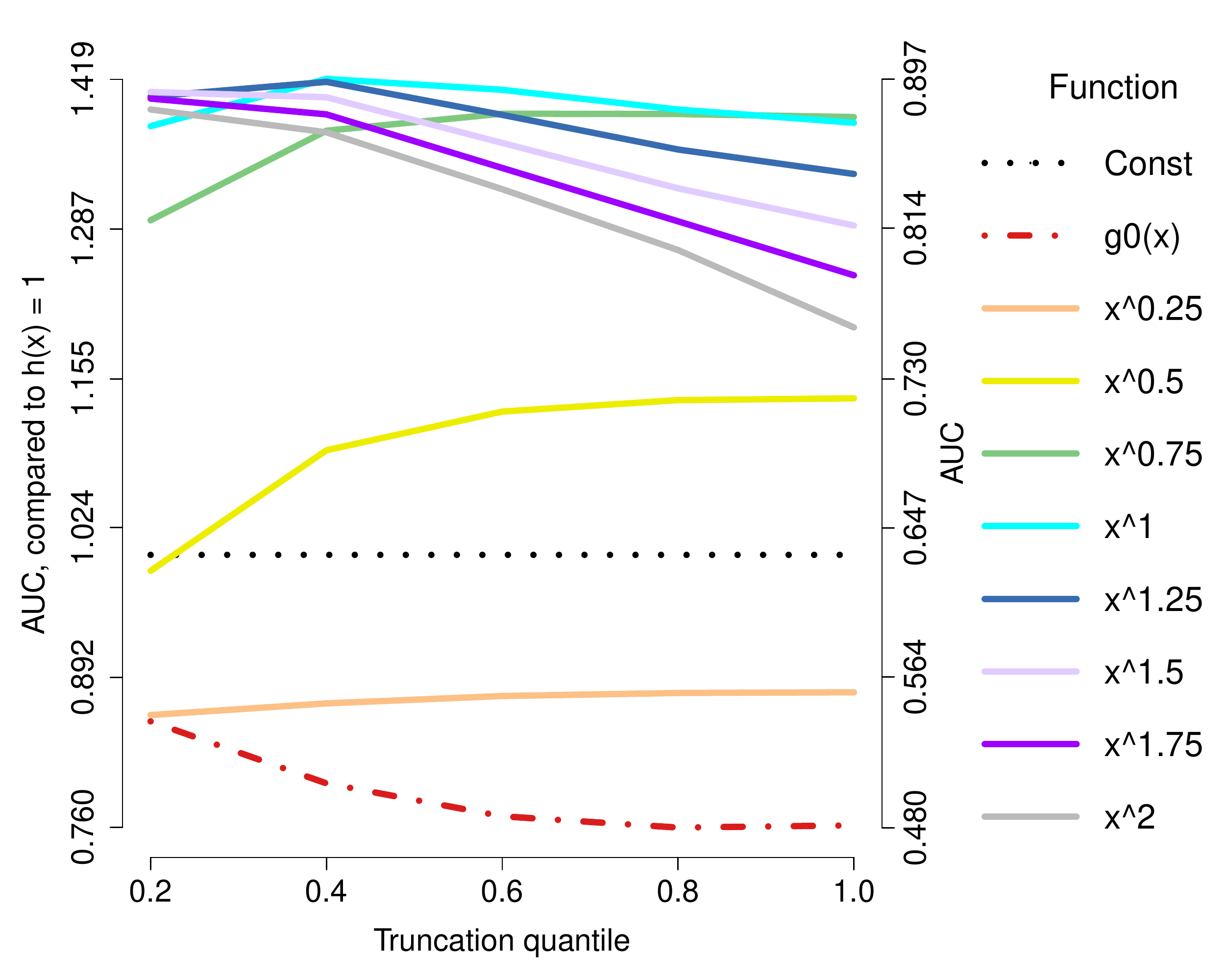}}\hspace{-0.02in}}
\\ \vspace{-0.1in}
\subfloat[$n=80$, $\ell_2^{\complement}$-nn domain]
{\includegraphics[width=0.45\textwidth]{{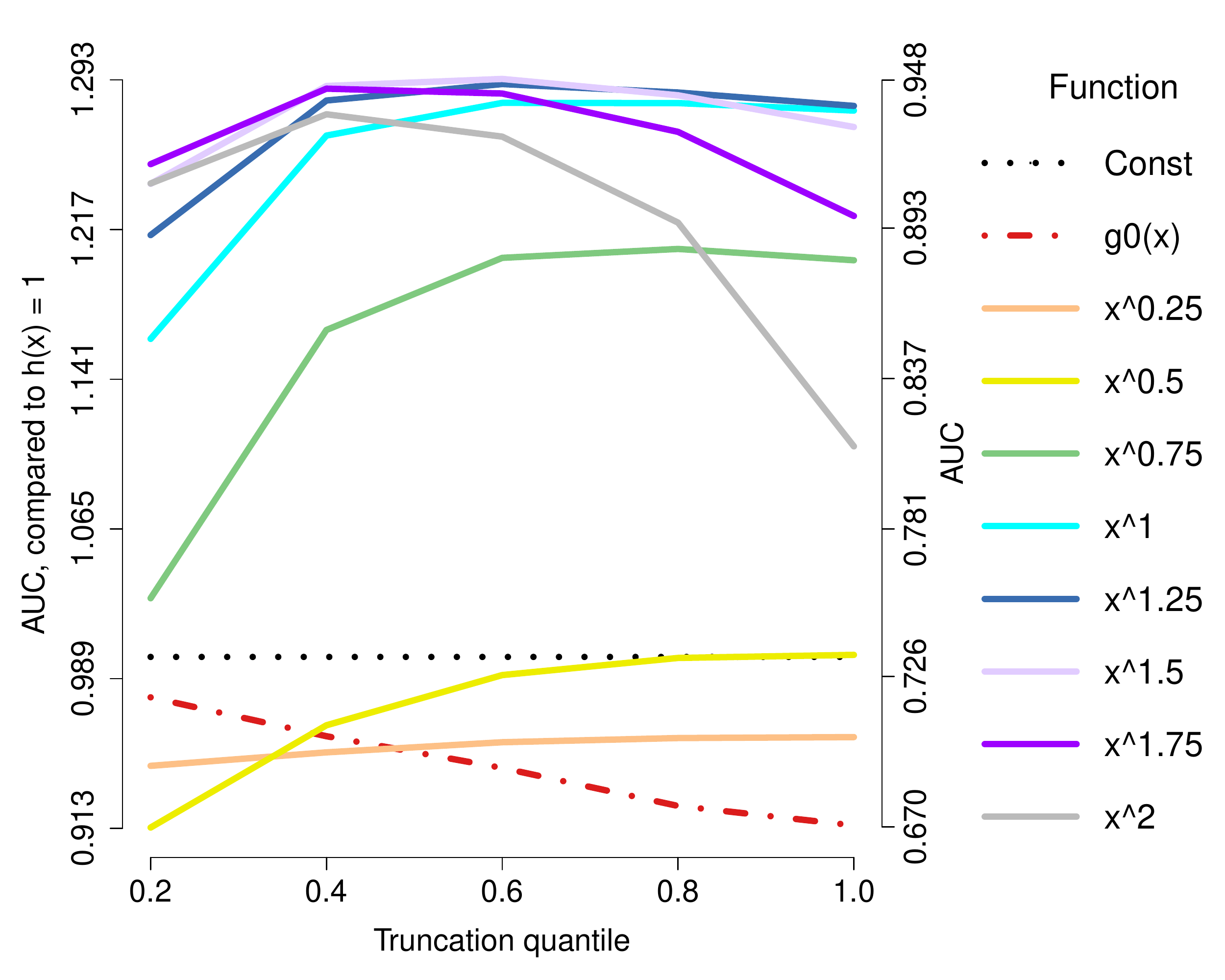}}\hspace{0.1in}}
\subfloat[$n=1000$, $\ell_2^{\complement}$-nn domain]
{\includegraphics[width=0.45\textwidth]{{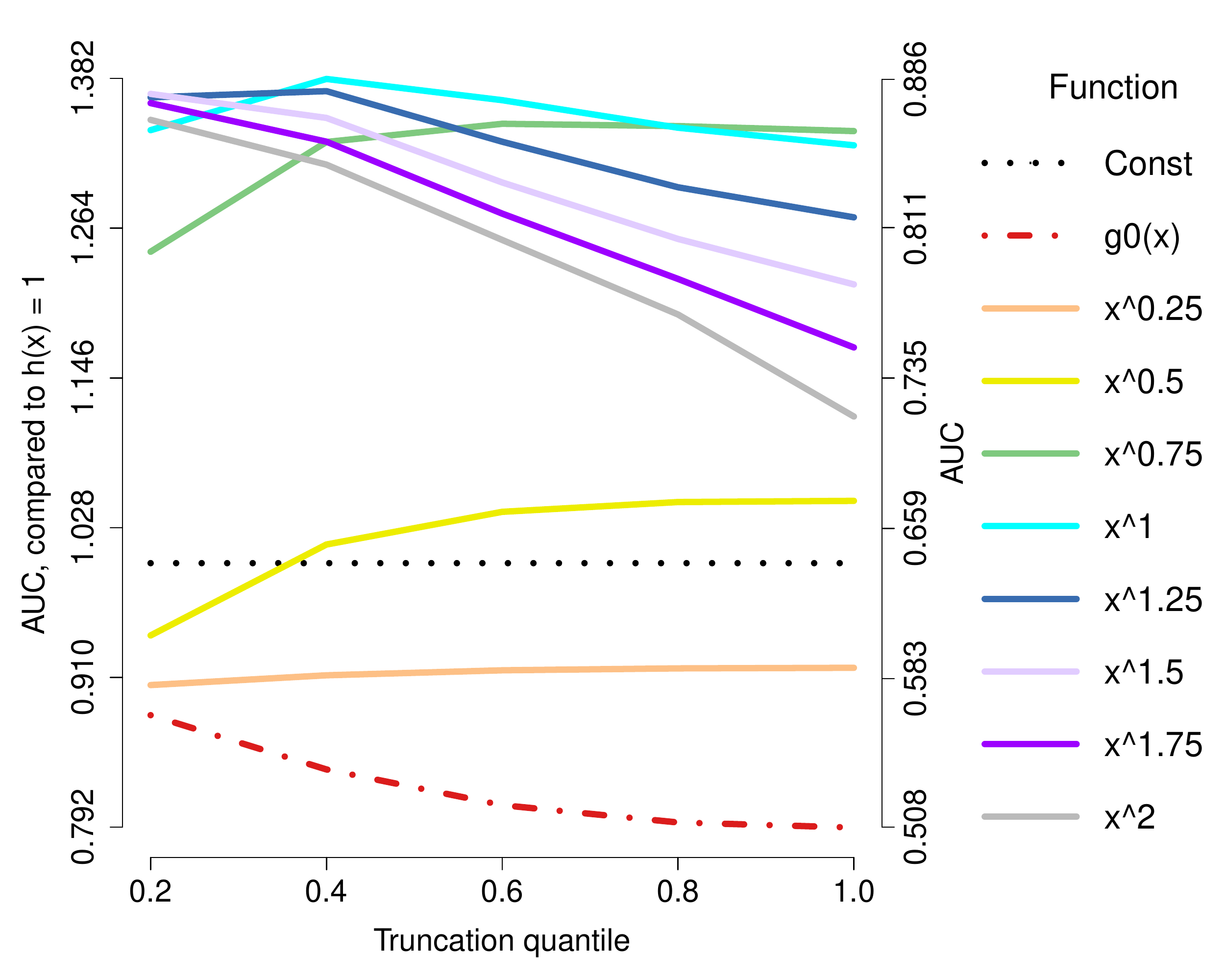}}\hspace{-0.02in}}
\\ \vspace{-0.1in}
\subfloat[$n=80$, unif-nn domain]
{\includegraphics[width=0.45\textwidth]{{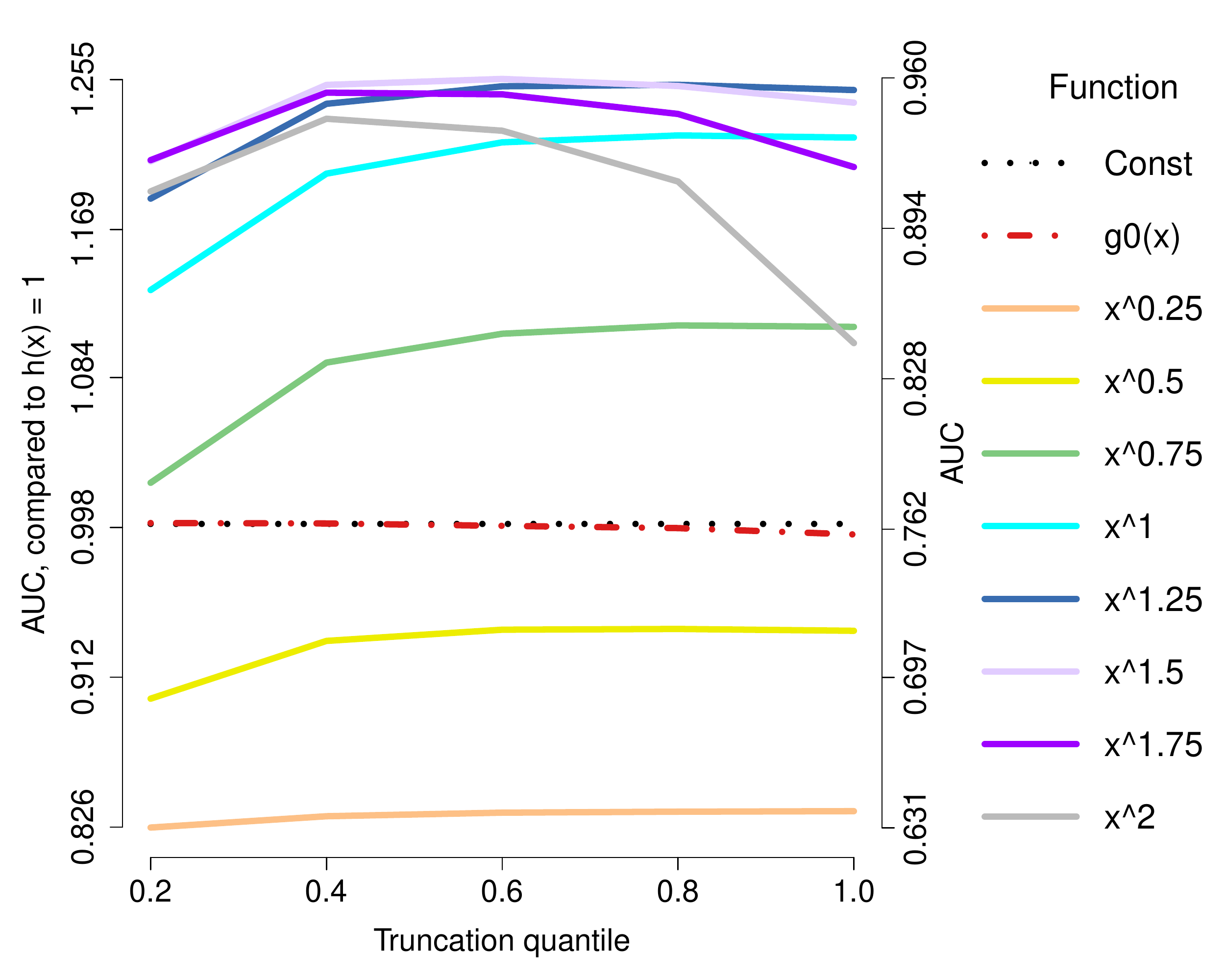}}\hspace{0.1in}}
\subfloat[$n=1000$, unif-nn domain]
{\includegraphics[width=0.45\textwidth]{{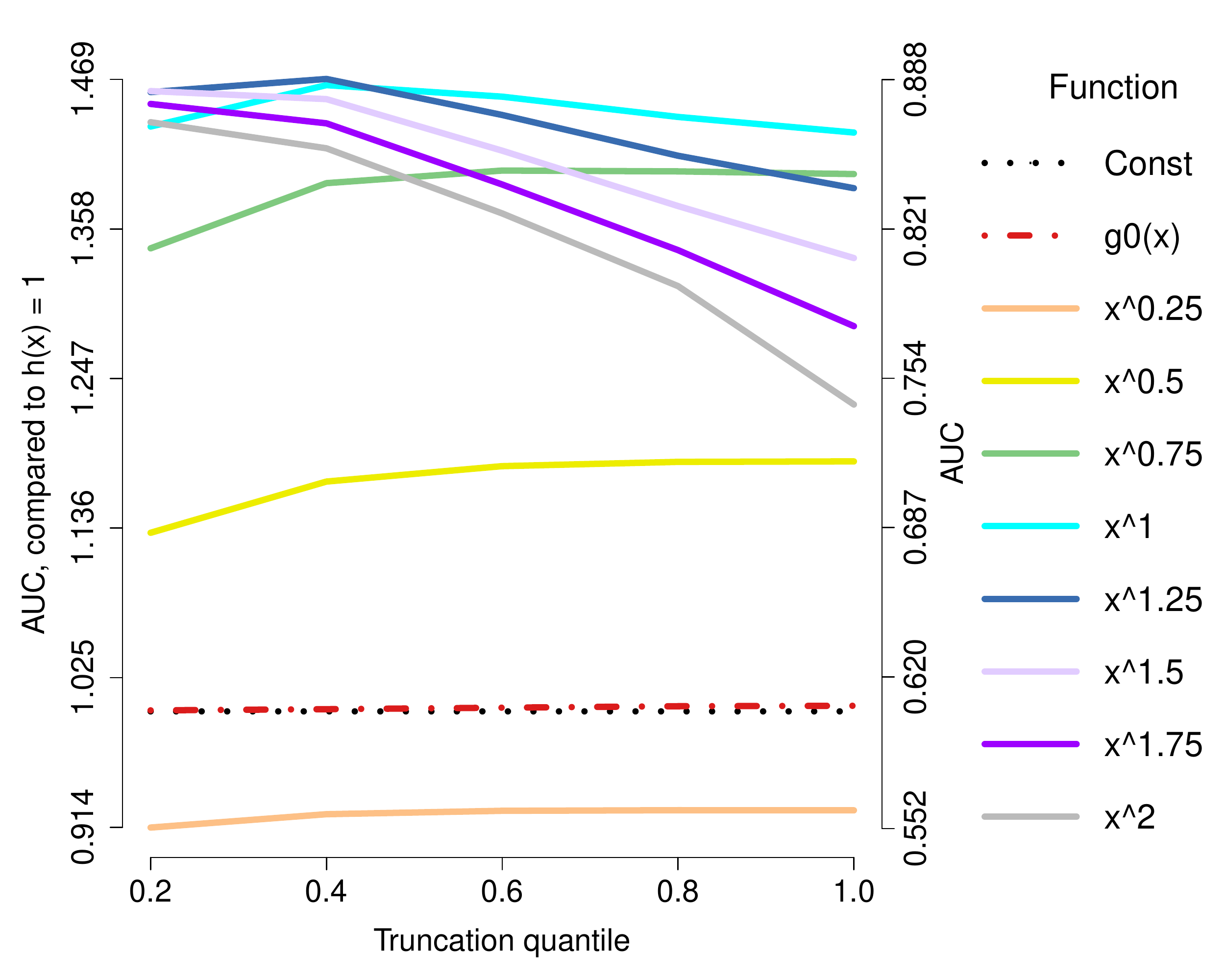}}\hspace{-0.02in}}
\caption{AUCs averaged over 50 trials for support recovery using
  generalized score matching for log models ($a=0$). Each curve
  represents either our extension to $g_0(\boldsymbol{x})$ from
  \citet{liu19} or a choice of power function $h(x)=x^c$.  The $x$
  axes mark the probabilities $\pi$ that determine the truncation points $\boldsymbol{C}$ for the truncated component-wise distances. The colors are sorted by the power $c$.}\label{plot_log}
\end{figure}

\begin{figure}[p]
\centering
\vspace{-0.0in}
\subfloat[$n=80$, $\ell_2$-nn domain]
{\includegraphics[width=0.45\textwidth]{{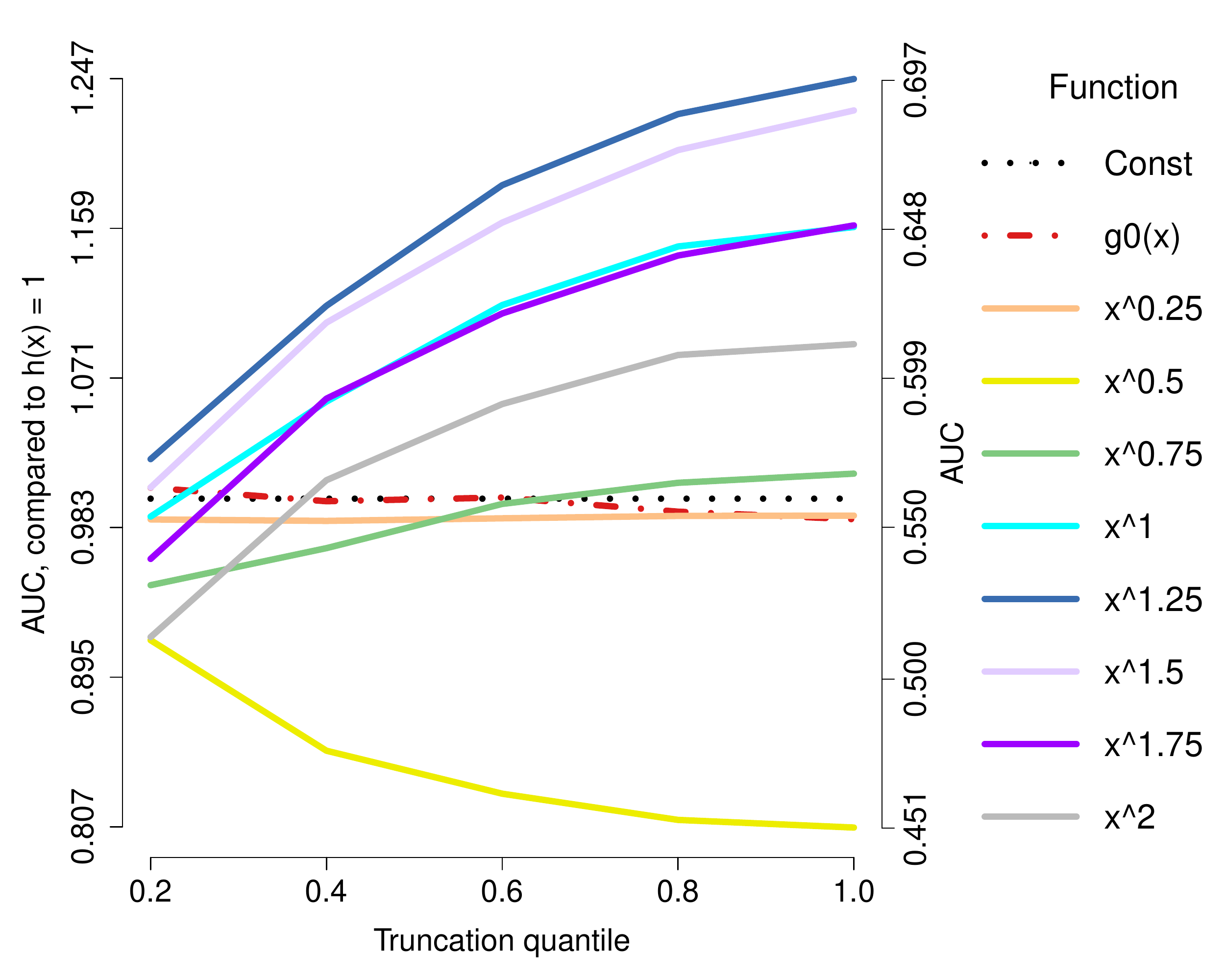}}\hspace{0.1in}}
\subfloat[$n=1000$, $\ell_2$-nn domain]
{\includegraphics[width=0.45\textwidth]{{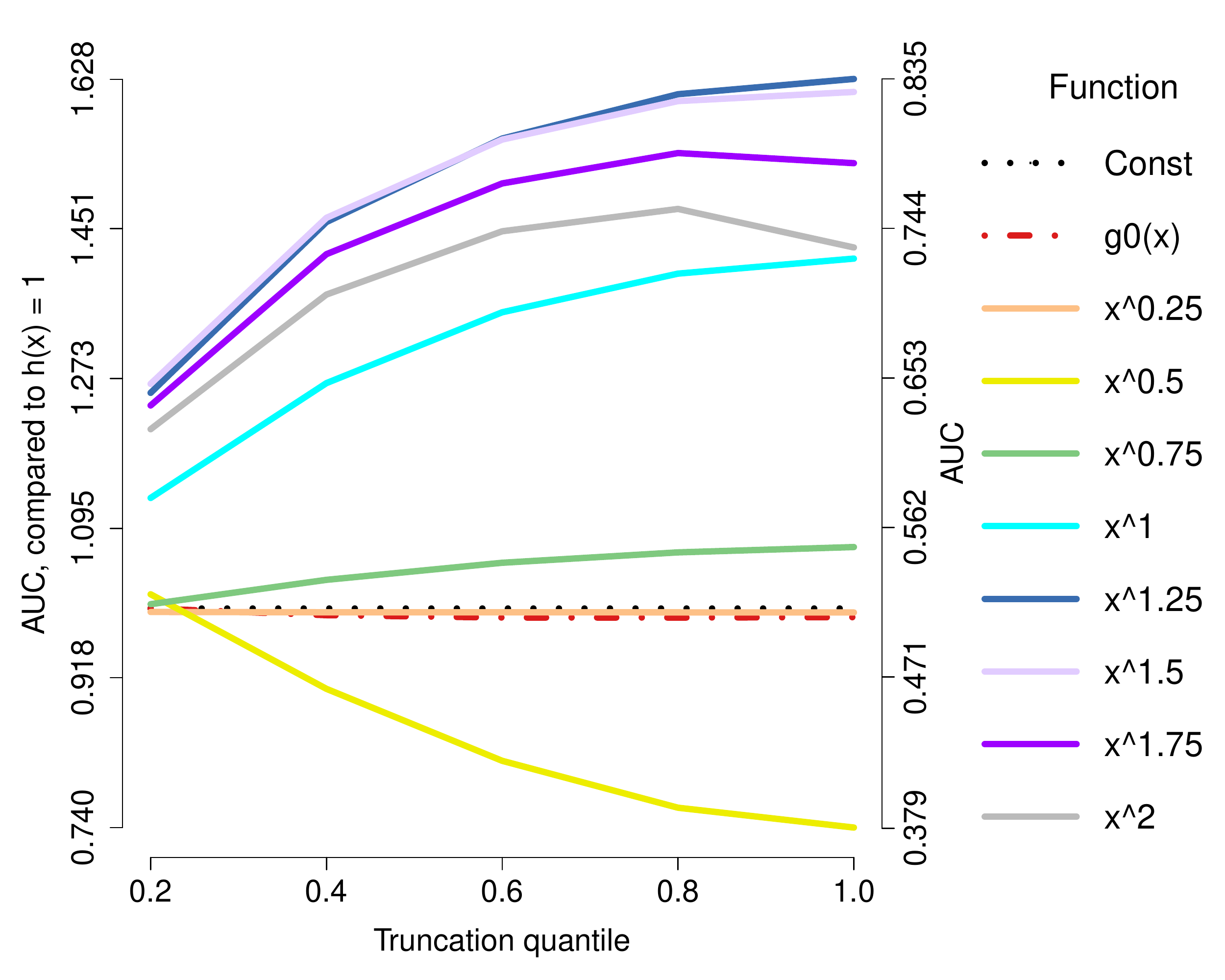}}\hspace{-0.02in}}
\\ \vspace{-0.1in}
\subfloat[$n=80$, $\ell_2^{\complement}$-nn domain]
{\includegraphics[width=0.45\textwidth]{{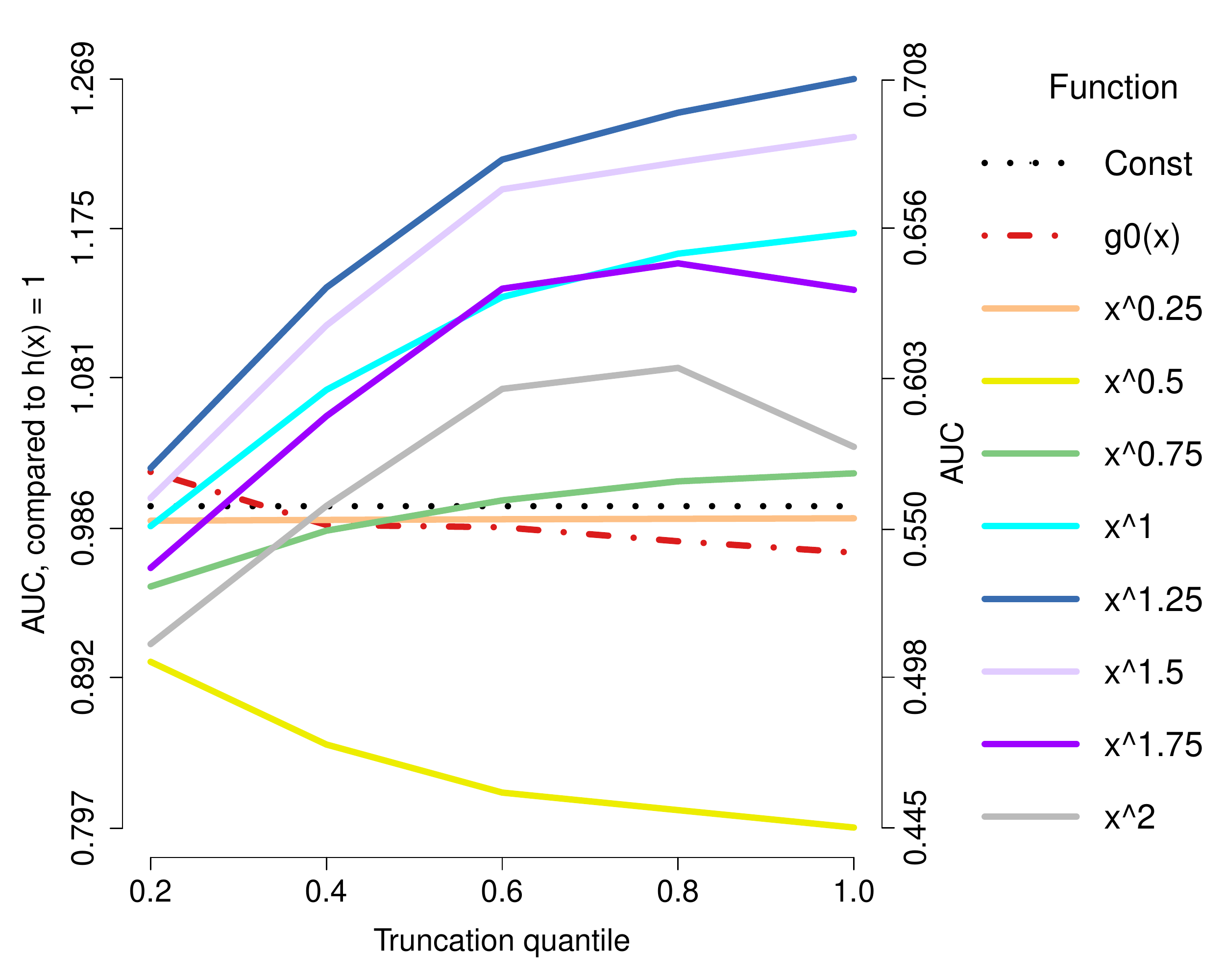}}\hspace{0.1in}}
\subfloat[$n=1000$, $\ell_2^{\complement}$-nn domain]
{\includegraphics[width=0.45\textwidth]{{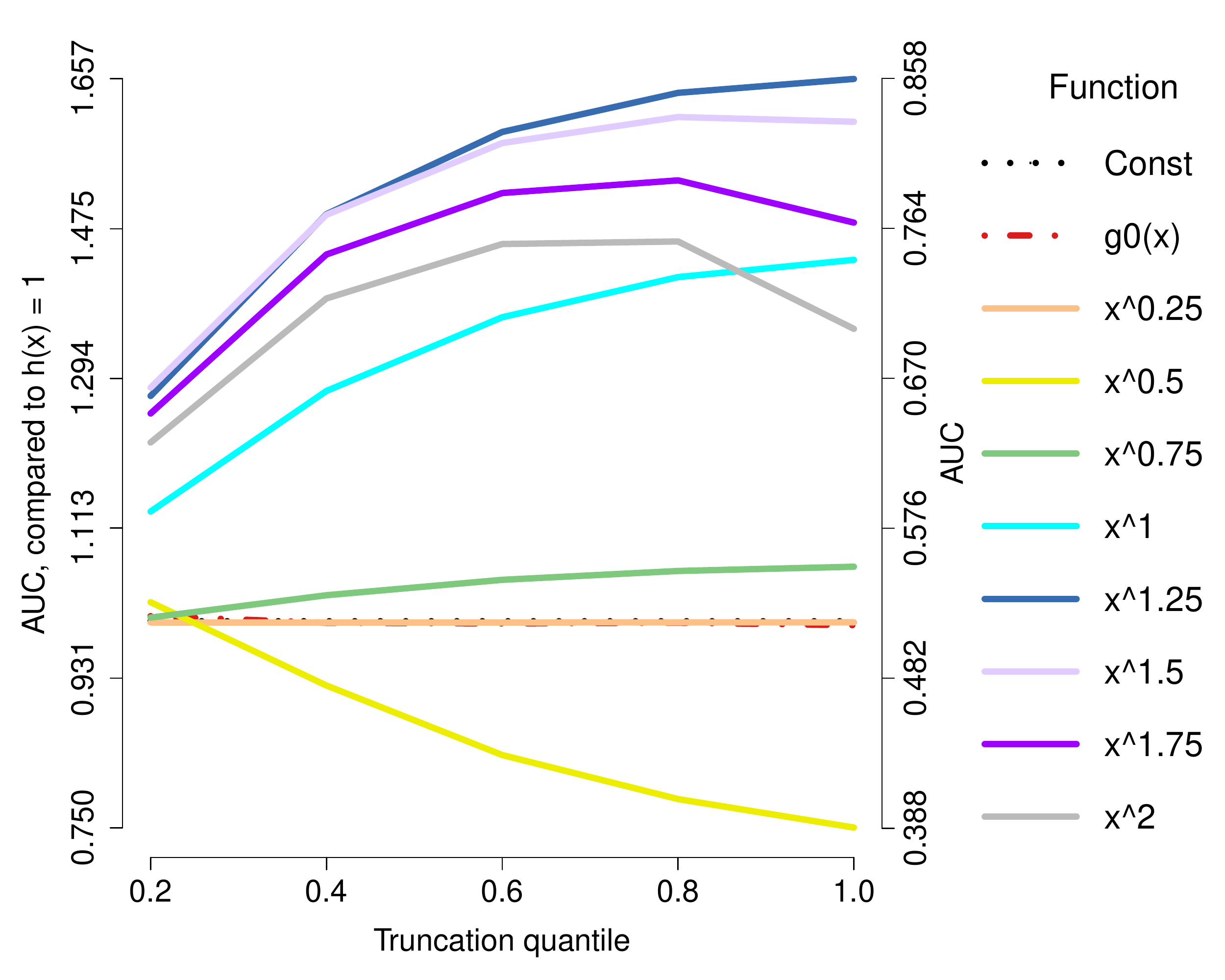}}\hspace{-0.02in}}
\\ \vspace{-0.1in}
\subfloat[$n=80$, unif-nn domain]
{\includegraphics[width=0.45\textwidth]{{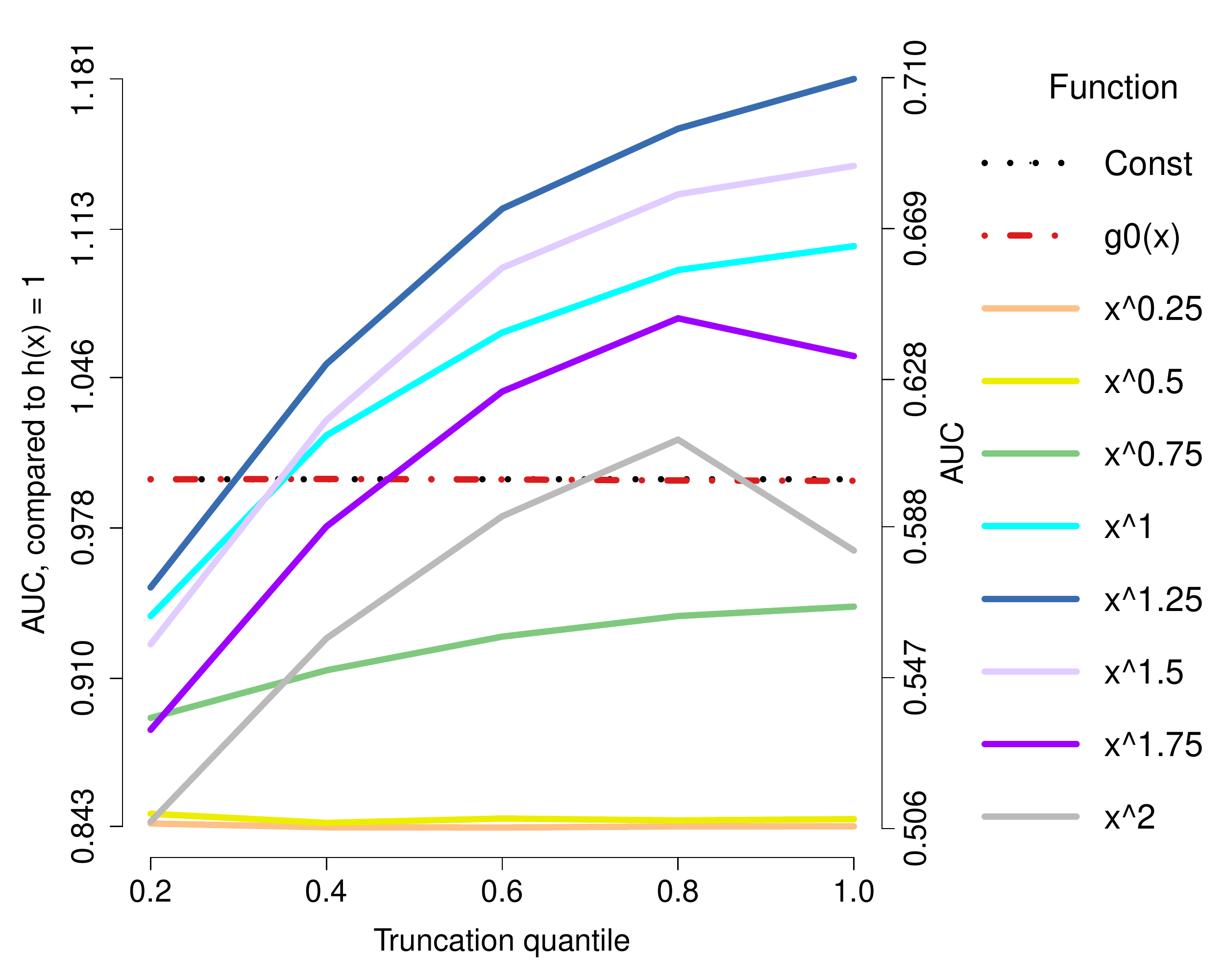}}\hspace{0.1in}}
\subfloat[$n=1000$, unif-nn domain]
{\includegraphics[width=0.45\textwidth]{{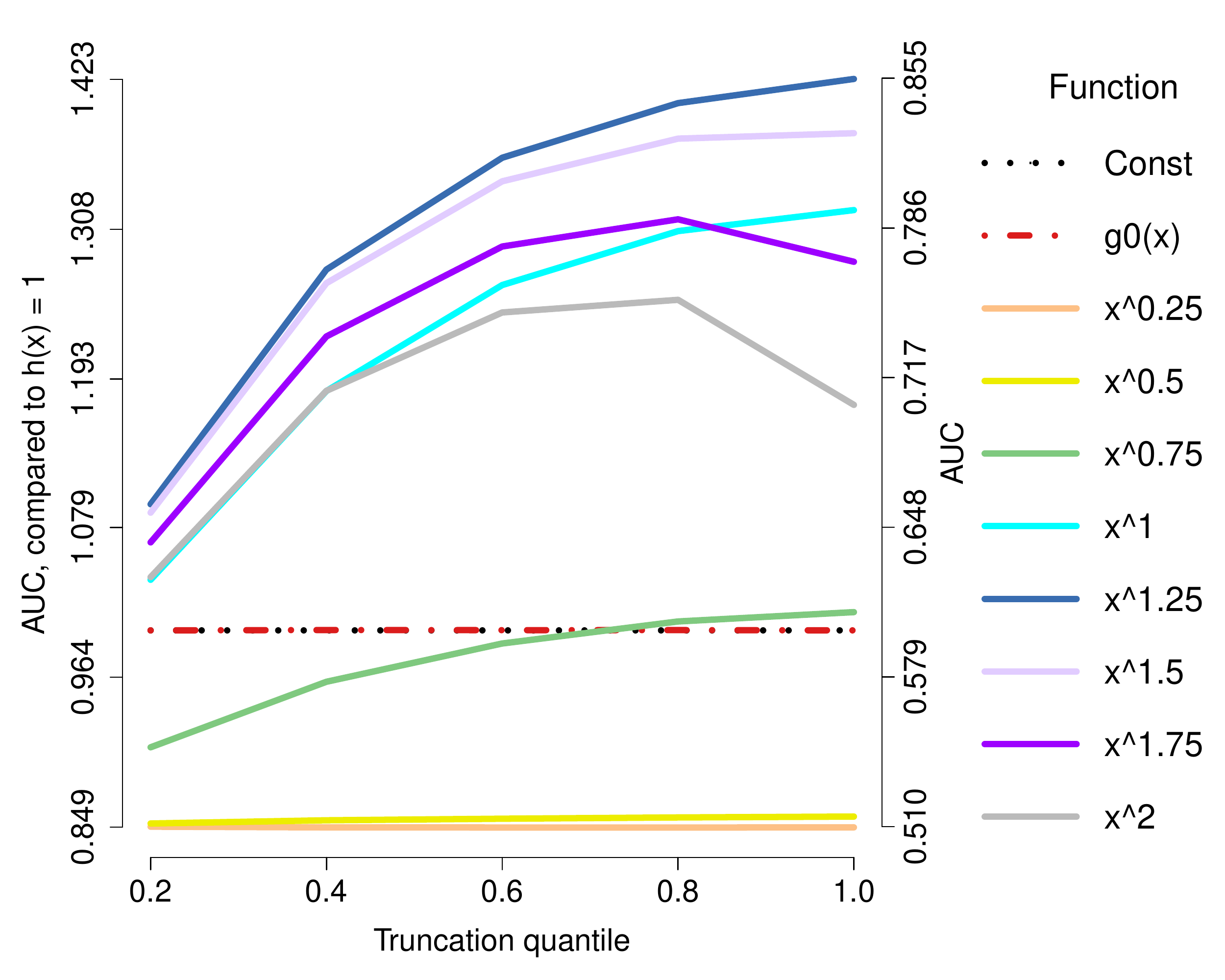}}\hspace{-0.02in}}
\caption{AUCs averaged over 50 trials for support recovery using
  generalized score matching for exponential square-root models
  ($a=1/2$). Each curve represents either our extension to
  $g_0(\boldsymbol{x})$ from \citet{liu19} or a choice of power
  function $h(x)=x^c$.}\label{plot_exp}
\end{figure}

\begin{figure}[p]
\centering
\vspace{-0.0in}
\subfloat[$n=80$, $\ell_2$-nn domain]
{\includegraphics[width=0.49\textwidth]{{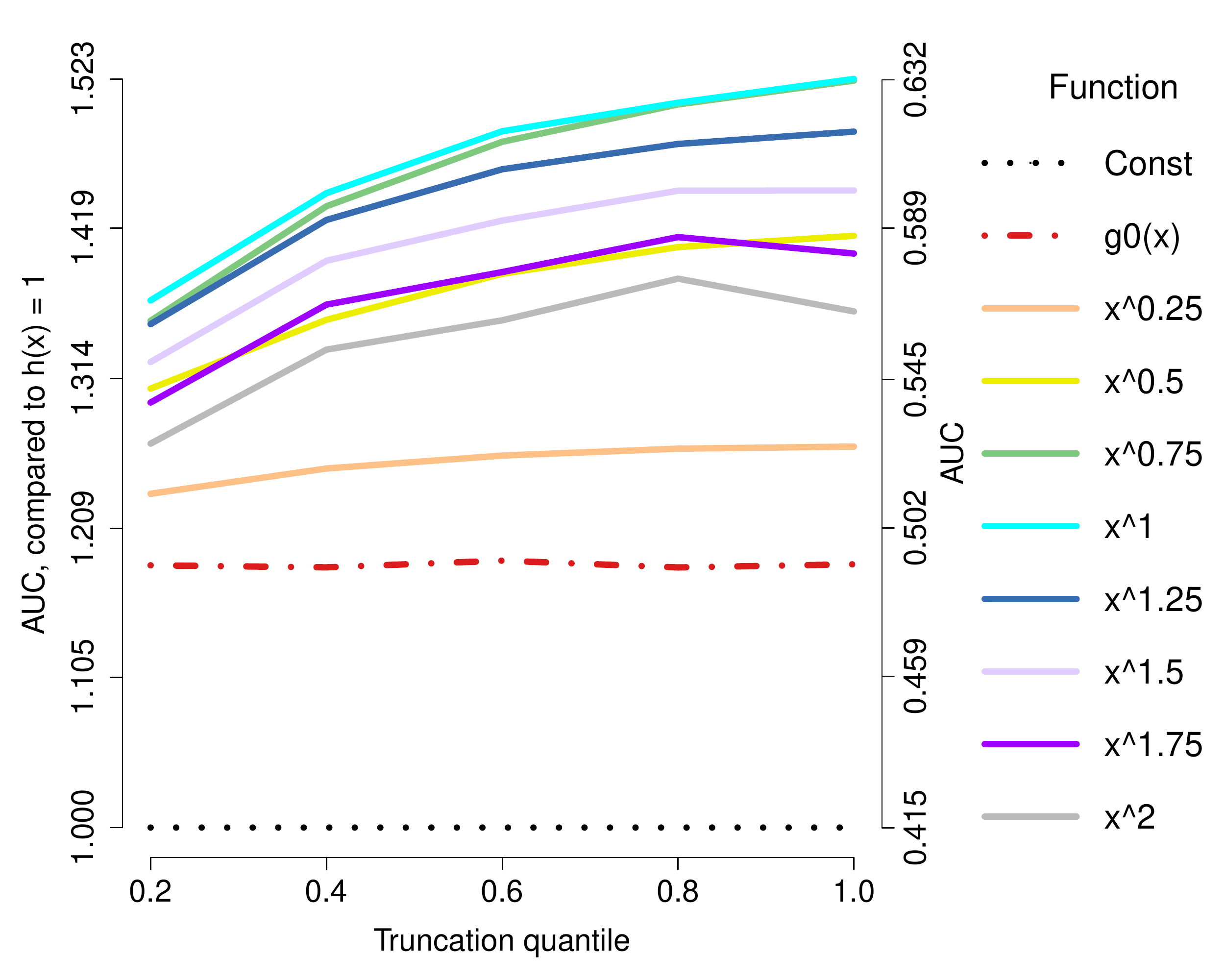}}\hspace{0.1in}}
\subfloat[$n=1000$, $\ell_2$-nn domain]
{\includegraphics[width=0.45\textwidth]{{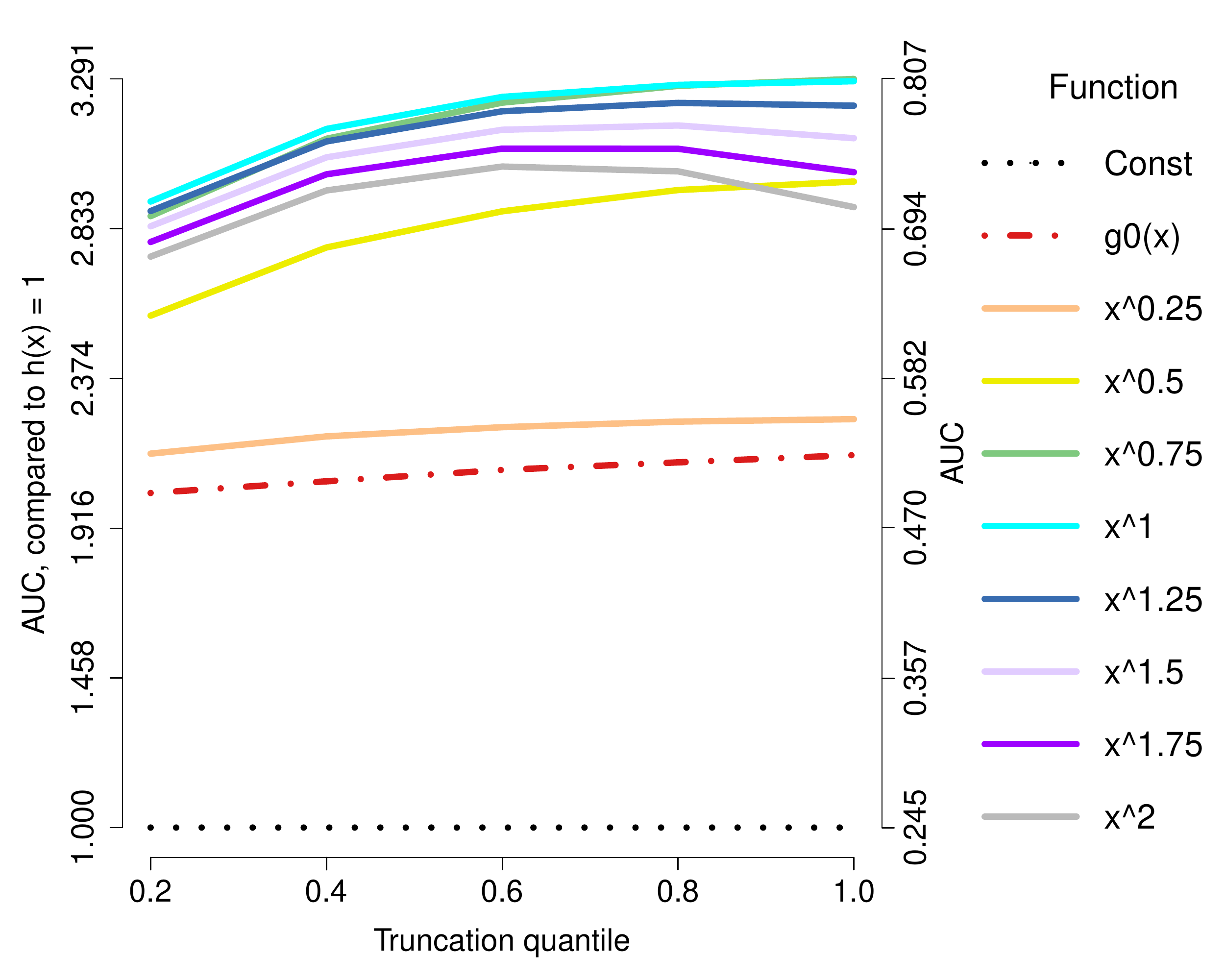}}\hspace{-0.02in}}
\\ \vspace{-0.1in}
\subfloat[$n=80$, $\ell_2^{\complement}$-nn domain]
{\includegraphics[width=0.45\textwidth]{{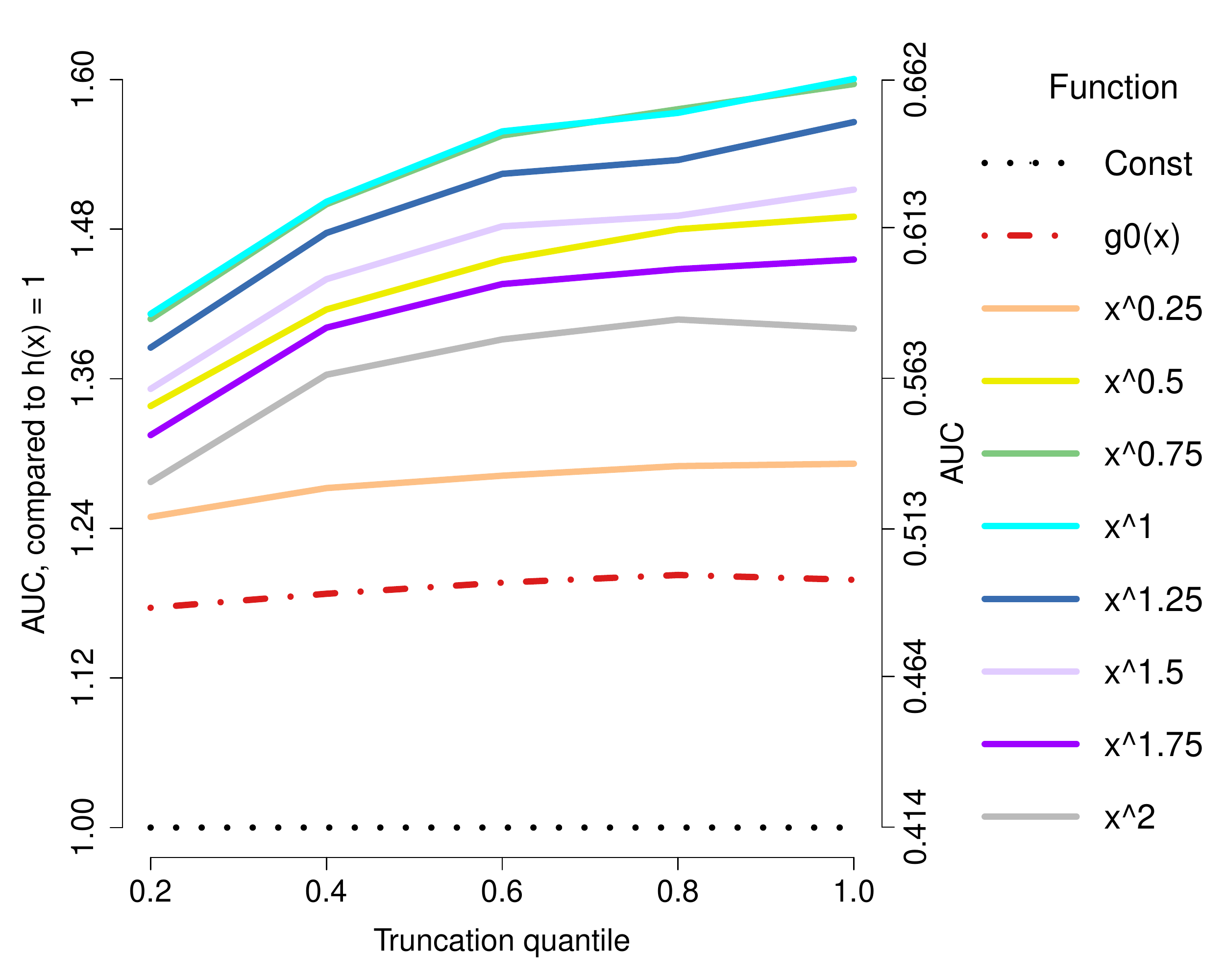}}\hspace{0.1in}}
\subfloat[$n=1000$, $\ell_2^{\complement}$-nn domain]
{\includegraphics[width=0.45\textwidth]{{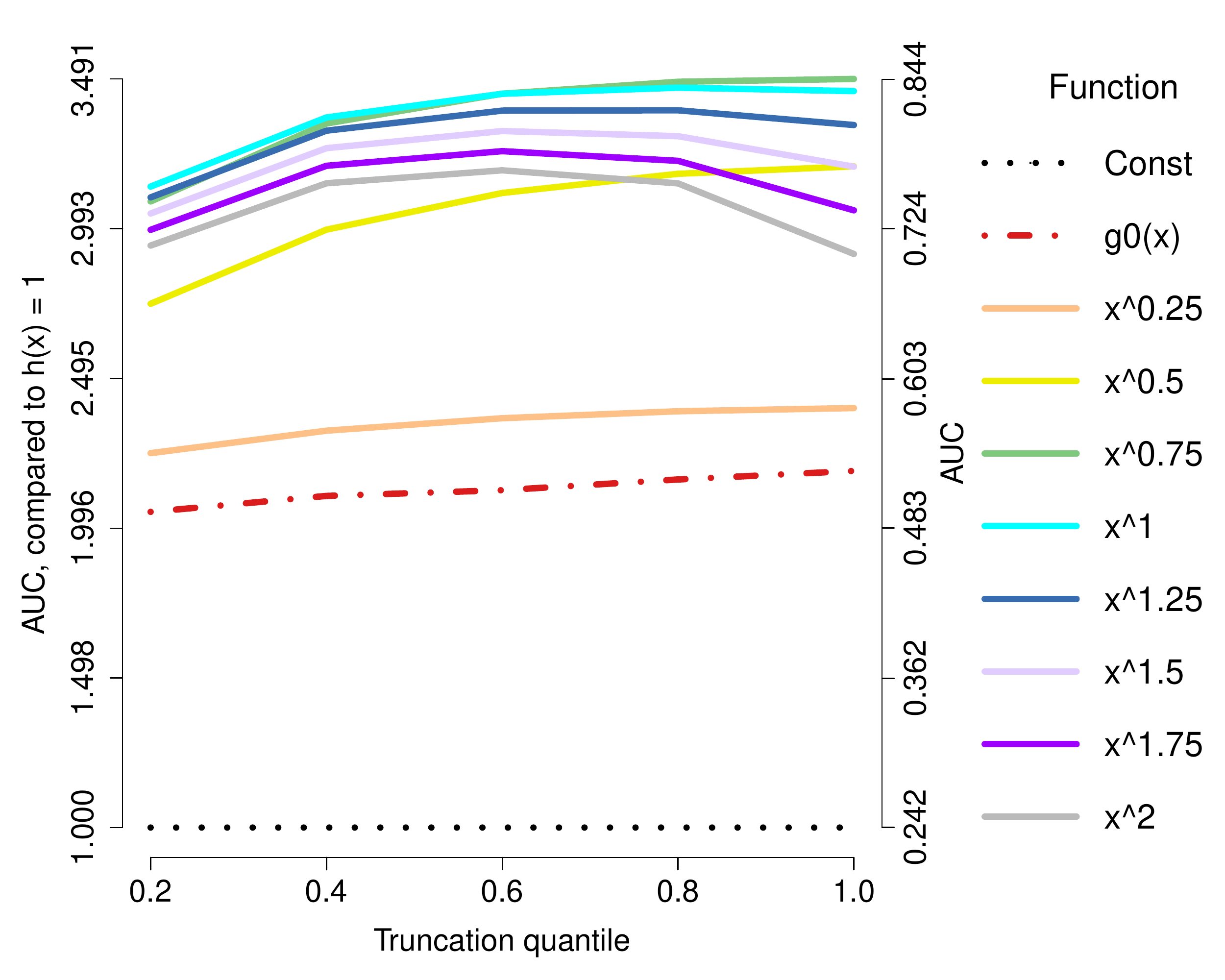}}\hspace{-0.02in}}
\\ \vspace{-0.1in}
\subfloat[$n=80$, unif-nn domain]
{\includegraphics[width=0.45\textwidth]{{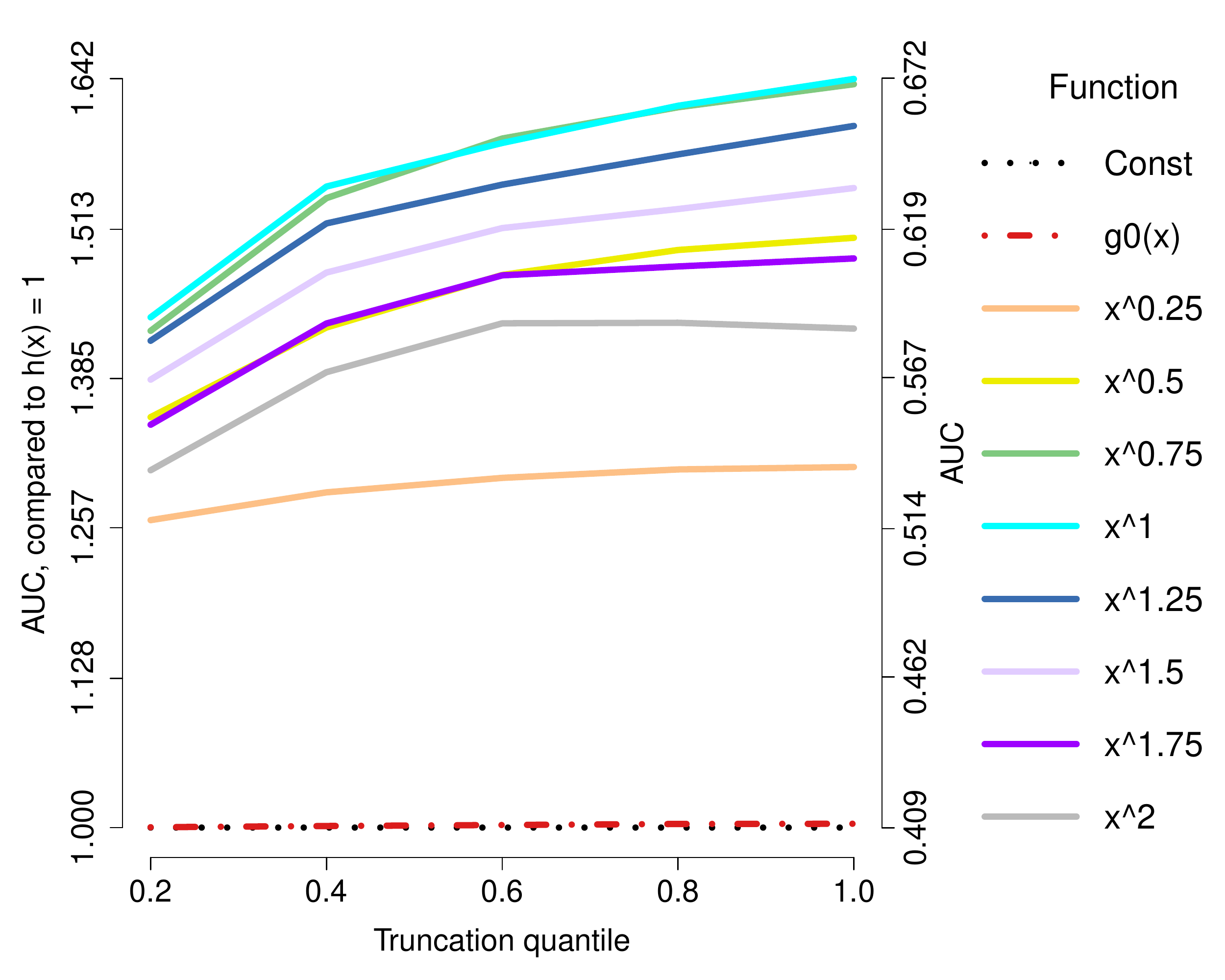}}\hspace{0.1in}}
\subfloat[$n=1000$, unif-nn domain]
{\includegraphics[width=0.45\textwidth]{{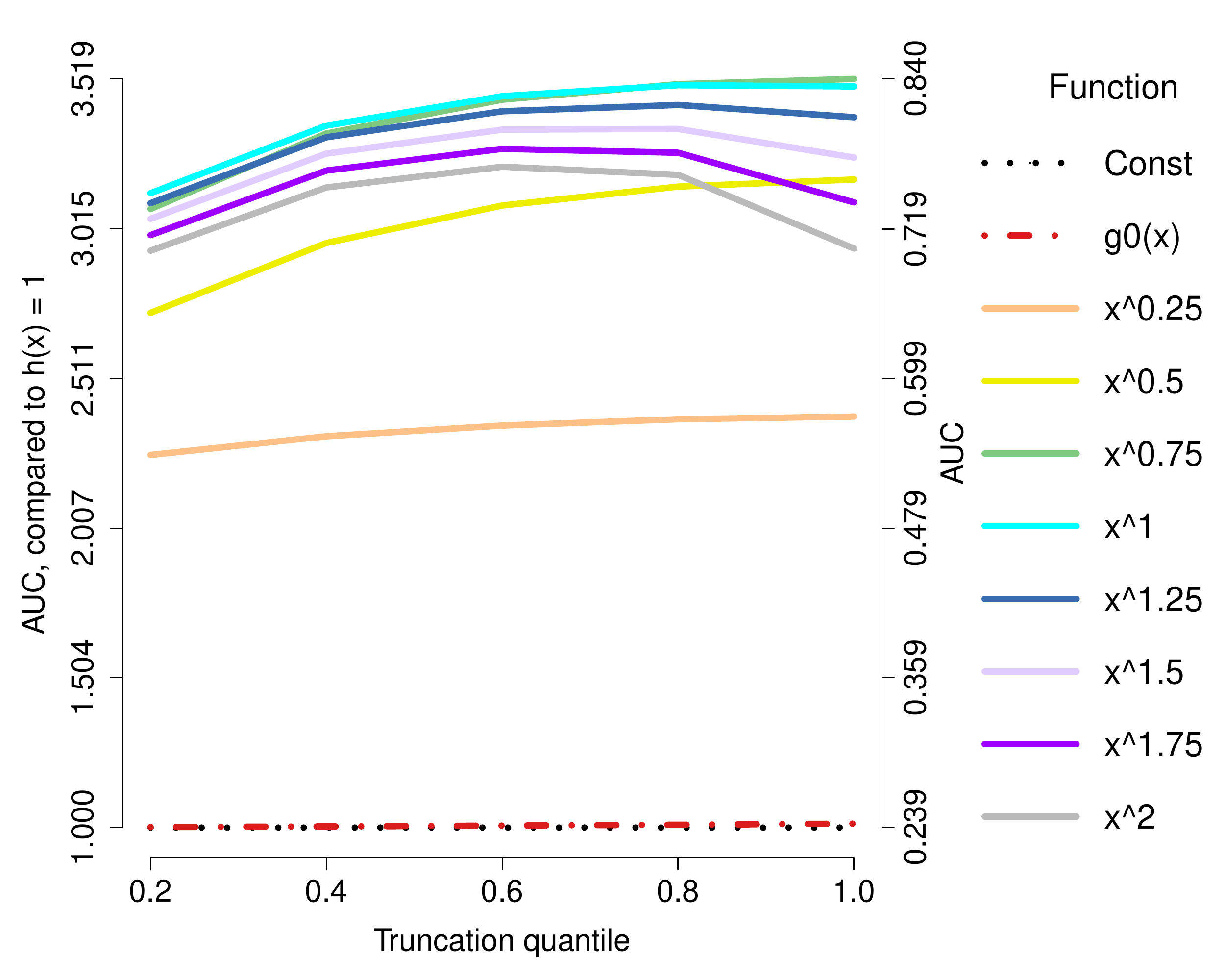}}\hspace{-0.02in}}
\caption{AUCs averaged over 50 trials for support recovery using
  generalized score matching for Gaussian models ($a=1$) on domains
  being subsets of $\mathbb{R}_+^m$.   Each curve represents either our extension to
  $g_0(\boldsymbol{x})$ from \citet{liu19} or a choice of power
  function $h(x)=x^c$. }\label{plot_gaussian_nn}
\end{figure}

\begin{figure}[p]
\centering
\vspace{-0.0in}
\subfloat[$n=80$, $\ell_2$ domain]
{\includegraphics[width=0.45\textwidth]{{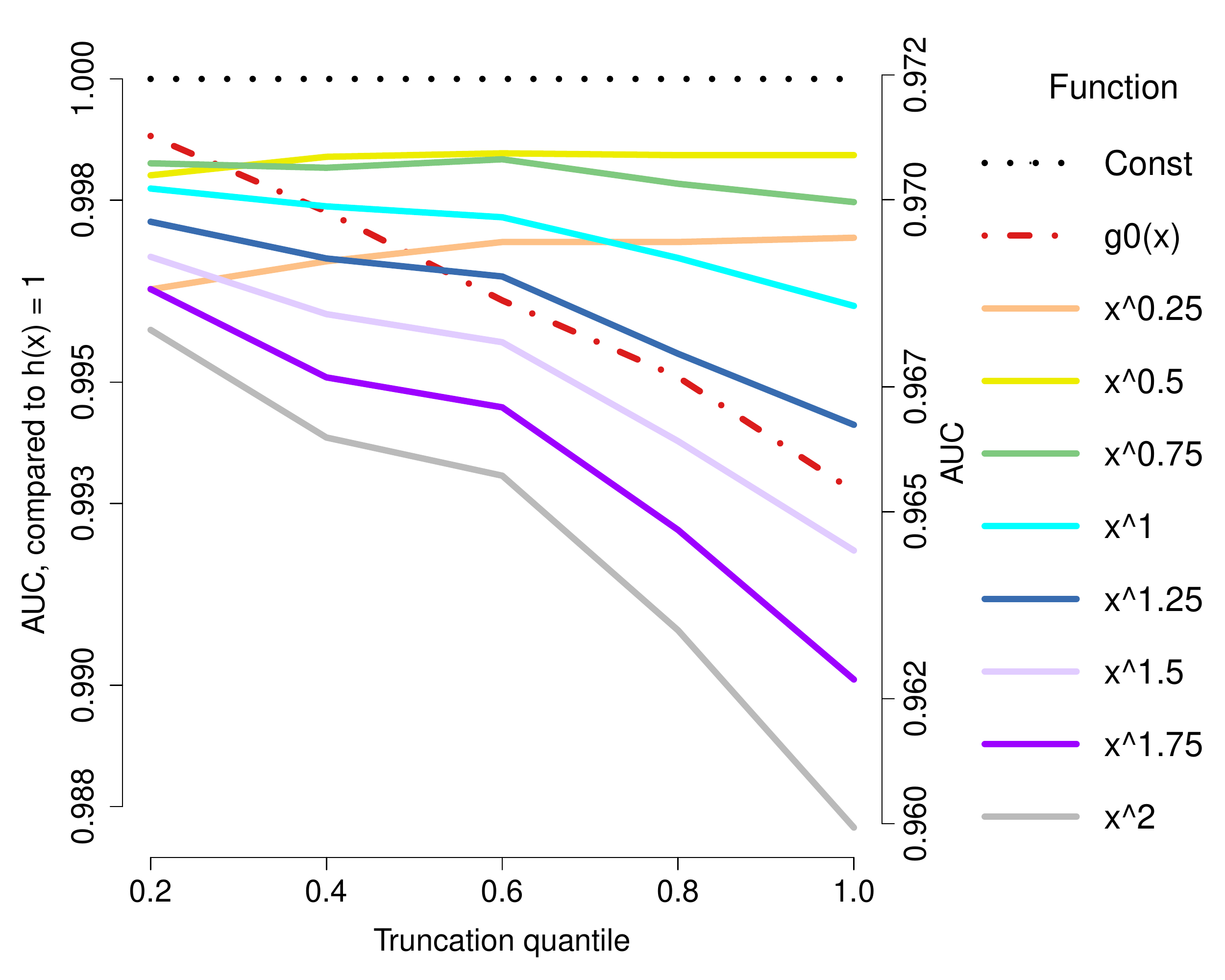}}\hspace{0.1in}}
\subfloat[$n=1000$, $\ell_2$ domain]
{\includegraphics[width=0.45\textwidth]{{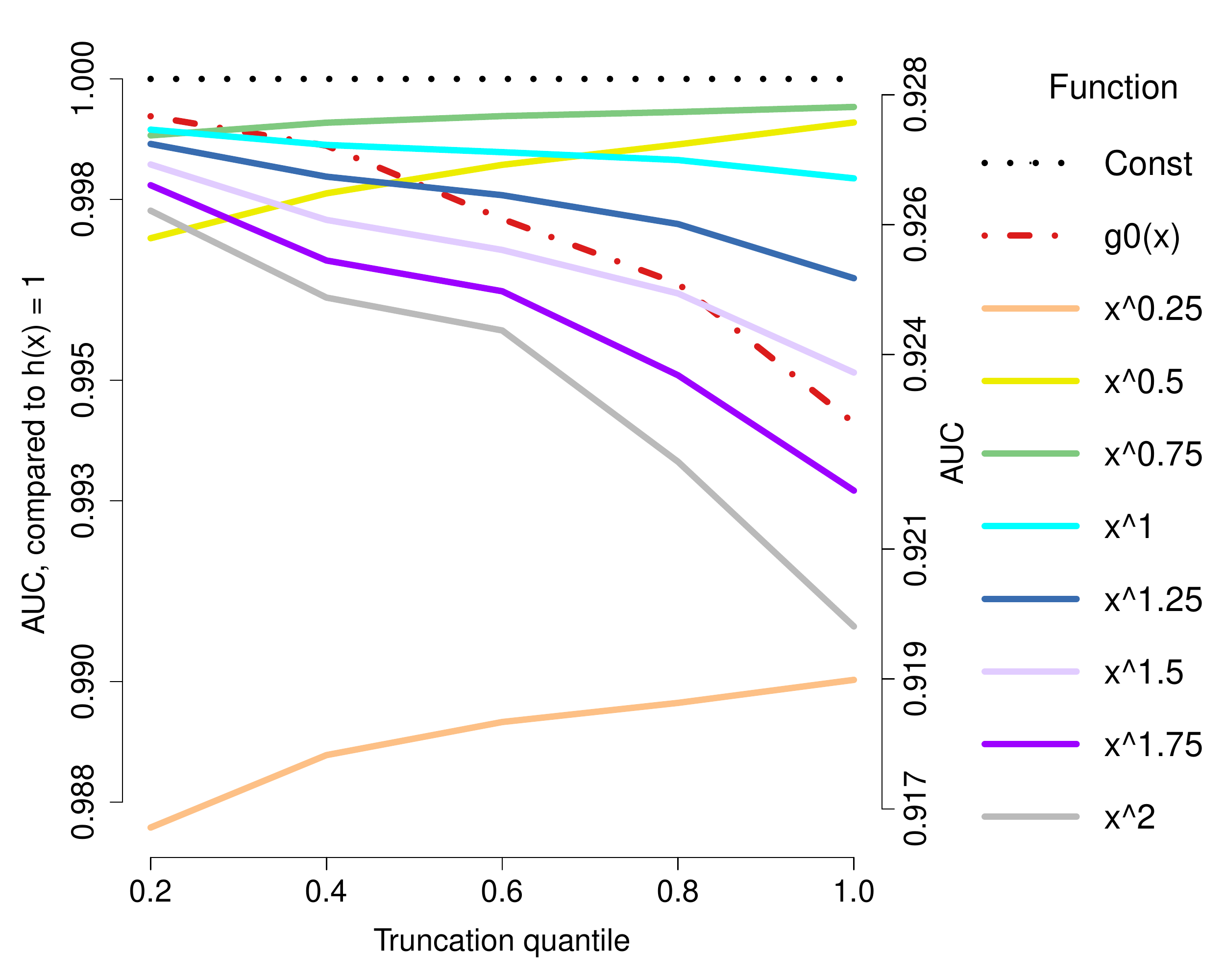}}\hspace{-0.02in}}
\\ \vspace{-0.1in}
\subfloat[$n=80$, $\ell_2^{\complement}$ domain]
{\includegraphics[width=0.45\textwidth]{{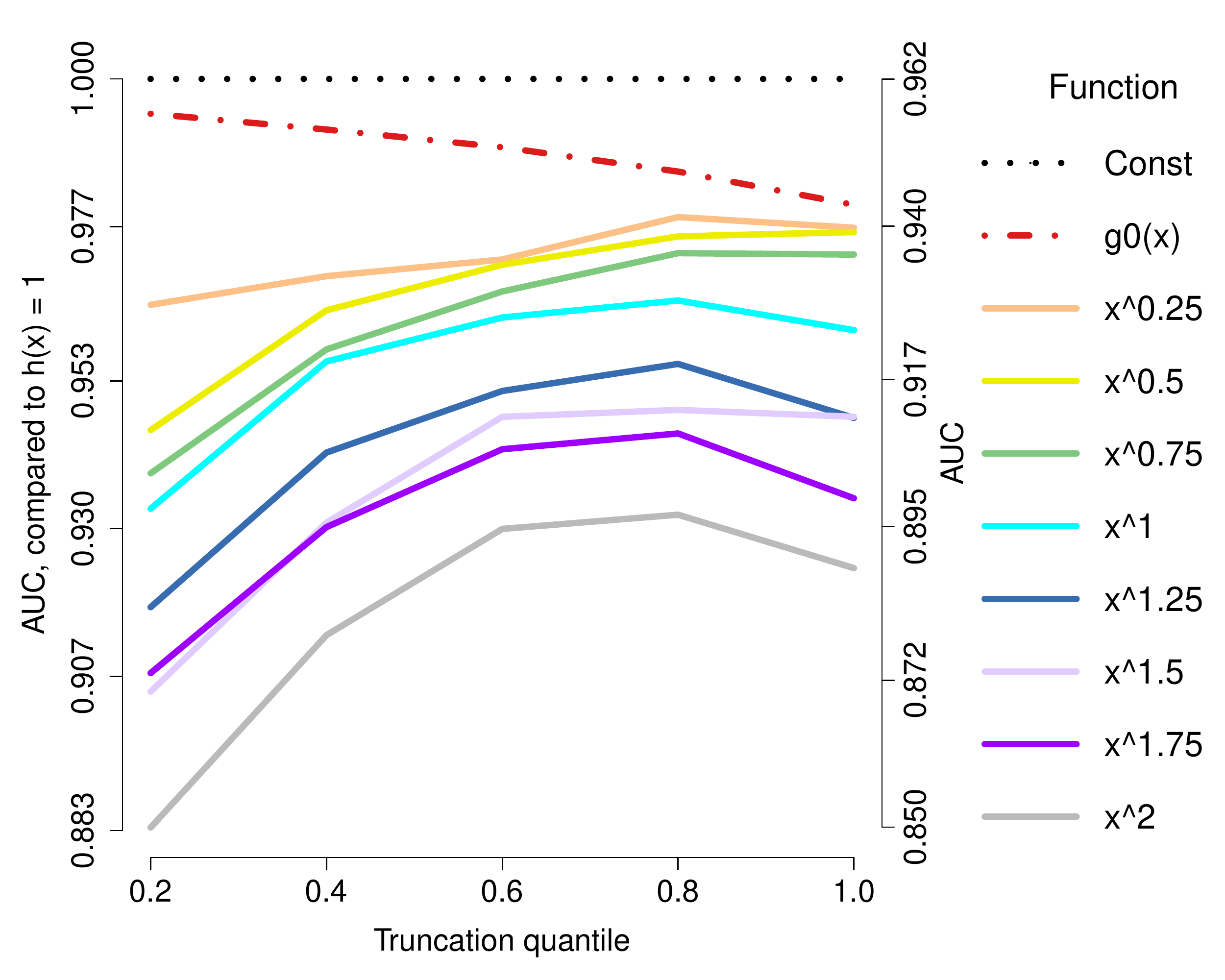}}\hspace{0.1in}}
\subfloat[$n=1000$, $\ell_2^{\complement}$ domain]
{\includegraphics[width=0.45\textwidth]{{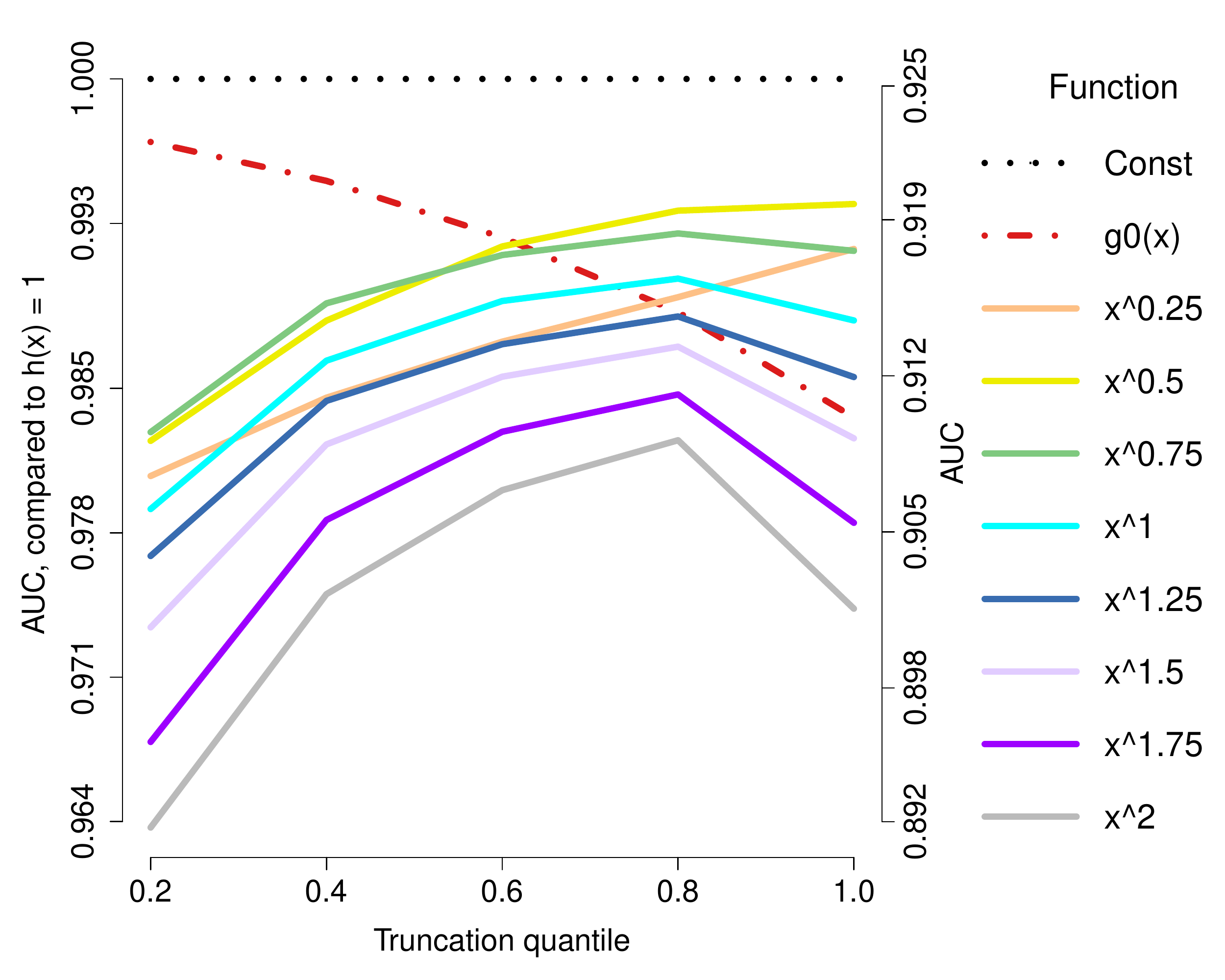}}\hspace{-0.02in}}
\\ \vspace{-0.1in}
\subfloat[$n=80$, unif domain]
{\includegraphics[width=0.45\textwidth]{{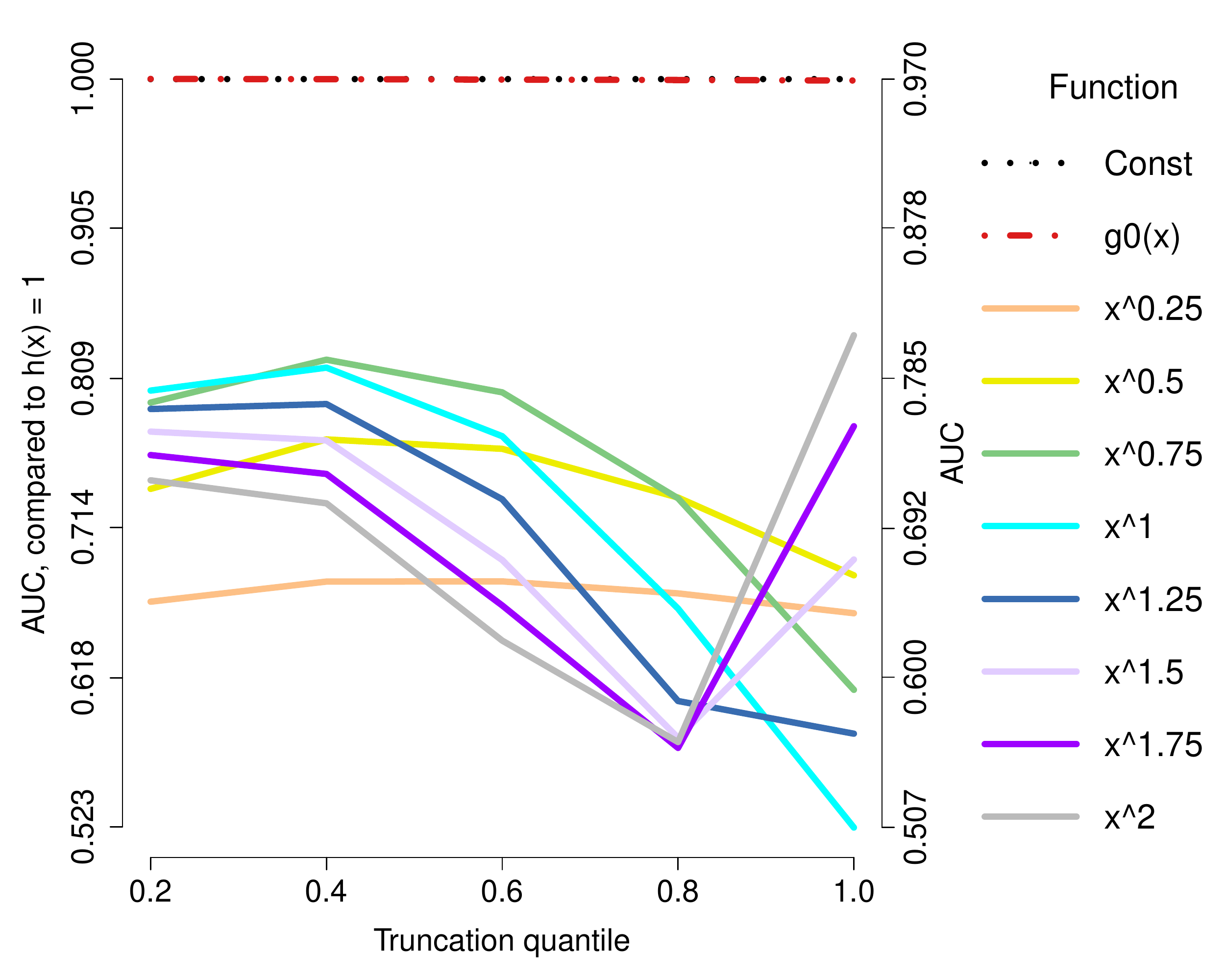}}\hspace{0.1in}}
\subfloat[$n=1000$, unif domain]
{\includegraphics[width=0.45\textwidth]{{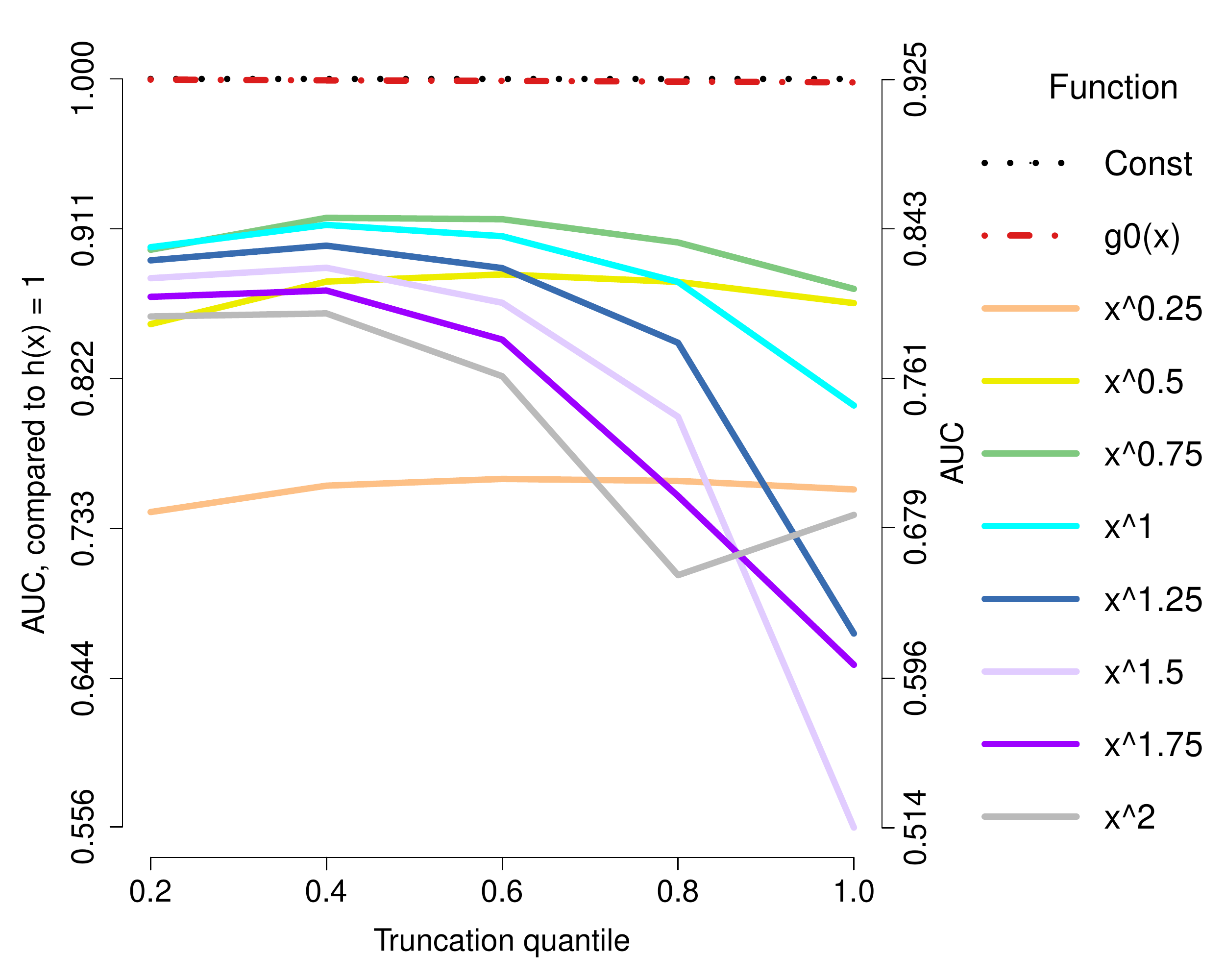}}\hspace{-0.02in}}
\caption{AUCs averaged over 50 trials for support recovery using generalized score matching for Gaussian models ($a=1$) on domains being subsets of $\mathbb{R}^m$ (not restricted to $\mathbb{R}_+^m$). Each curve represents either our extension to
  $g_0(\boldsymbol{x})$ from \citet{liu19} or a choice of power
  function $h(x)=x^c$.}\label{plot_gaussian}
\end{figure}

\section{DNA Methylation Networks}\label{DNA Methylation Data}
We illustrate the use of our generalized score matching for
inference of
conditional independence relations among DNA methylations based on
data for $500$ patients. The dataset contains methylation levels of
CpG islands associated with head and neck cancer from The Cancer
Genome Atlas (TCGA) \citep{wei13}.
Methylation levels are
associated with epigenetic regulation of genes and, according to
\citep{du10}, are commonly reported as
Beta values, a value in $[0,1]$ given by the ratio of the methylated probe
intensity and the sum of methylated and unmethylated probe
intensities, or M values, defined as the base 2-logit of the
Beta values.  Supported on $\mathbb{R}$, M values can be analyzed
using traditional methods, e.g., via Gaussian graphical models.  In
contrast, our new methodology allows direct analysis of Beta values
using generalized score matching for the $a$-$b$ model framework.

We focus on
a subset of CpG sites corresponding to genes known to belong to the
pathway for Thyroid cancer according to the Kyoto Encyclopedia of
Genes and Genomes (KEGG).  Furthermore,  we remove
sites with clearly bimodal methylations,
which we assess 
using the methods from the R package \texttt{Mclust}.
This results in $n=500$ samples and $m=478$ sites belonging to 36 genes. 

When considering M values, we estimate the graph encoding the support of the interaction matrix
(and hence the conditional dependence structure) in a Gaussian model on
$\mathbb{R}^m$, i.e., the $a$-$b$ model with $a=b=1$.  In doing so, we
use the profiled estimator in (\ref{eq_loss_regularized_profiled}),
and choose the upper-bound diagonal multiplier $2-(1+80\sqrt{\log
    m/n})^{-1}$ as suggested in Section 6.2 of \citet{yus19}.
The support being all of $\mathbb{R}^m$ we simply use the original
score matching with
$(\boldsymbol{h}\circ\boldsymbol{\varphi})(\boldsymbol{x})=\mathbf{1}_m$. For
Beta values, we assume a $\log$-$\log$ model ($a=b=0$) on $[0,1]^{m}$,
and use the profiled estimator with the upper-bound diagonal
multiplier $1+\sqrt{(\tau\log m+\log 4)/(2n)}$ as in
(\ref{eq_bounded_nonlog_delta}) with the choice of $\tau=3$. Suggested
by our theory, we use
$\boldsymbol{h}(\boldsymbol{x})=\boldsymbol{x}^2$, and choose the
truncation points in $\boldsymbol{\varphi}$ to be the 40th sample
percentile, as suggested by the simulation results in Figure
\ref{plot_log}.  For our illustration, the $\lambda$ parameter that
defines the $\ell_1$ penalty on $\mathbf{K}$ is chosen so that the
number of edges is equal to 478, the number of sites, following
\citet{lin16} and \citet{yus19}.

The estimated graphs are presented in Figure \ref{data_sites}, where
panel (a) is for Beta values, (c) is for M values,
and (b) shows their common edges, i.e., the intersection graph. The
plots in (a), (b), and (c) exclude isolated nodes and the layout is
optimized for each graph. Figure \ref{data_sites_Beta_layout} in the appendix includes isolated nodes where the layout is optimized for the graph for Beta values. Figure \ref{data_genes} in the appendix shows the graphs in Figure \ref{data_sites} aggregated by the genes associated with the sites. In (a) and (c), red points indicate nodes with degree at least 10. Sites with the highest node degrees are listed in Table \ref{site_table}, where those shared by the two graphs are highlighted in bold.

\begin{figure}[t!]
\centering
\subfloat[Beta values]{\includegraphics[width=0.25\textwidth]{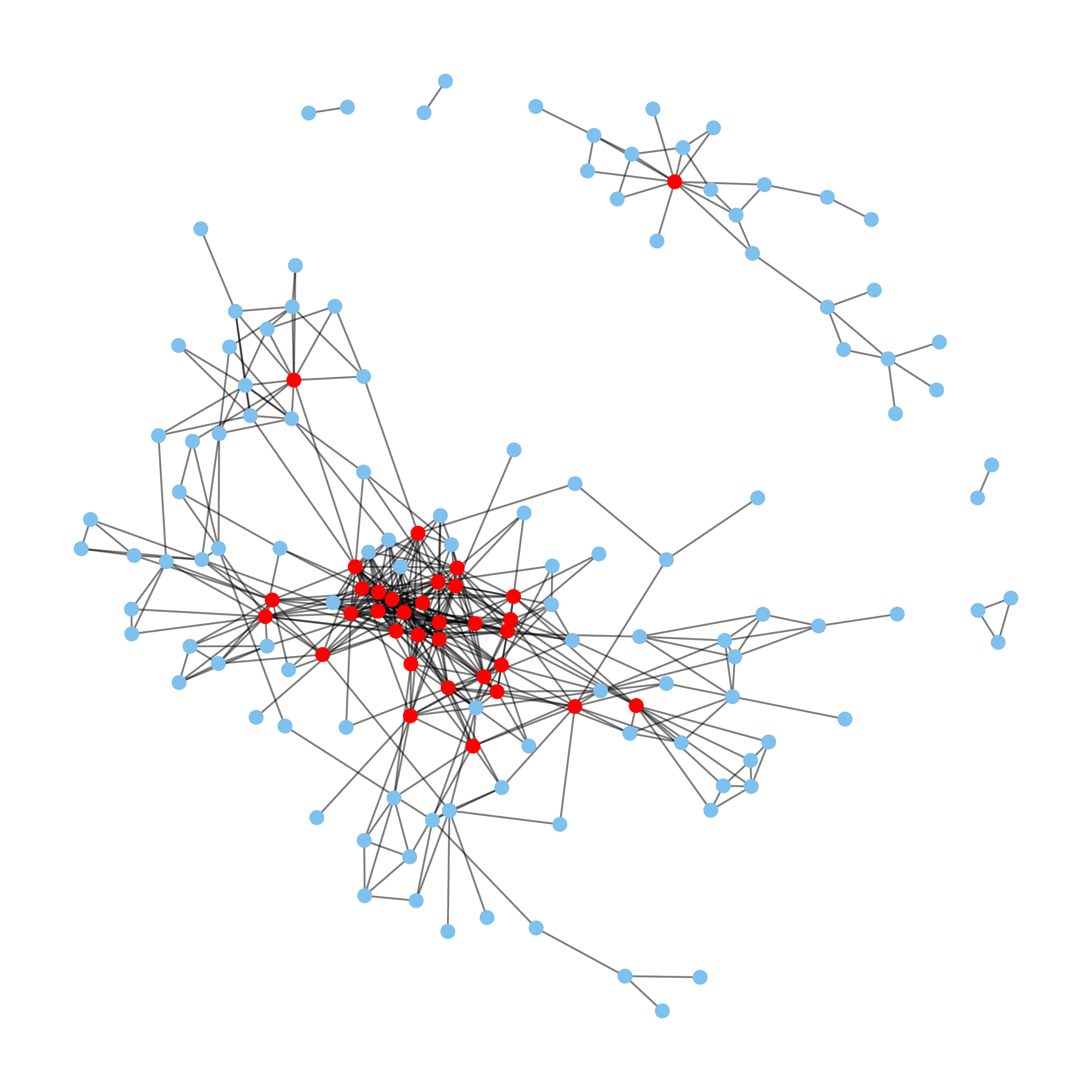}}
\subfloat[Common edges]{\includegraphics[width=0.25\textwidth]{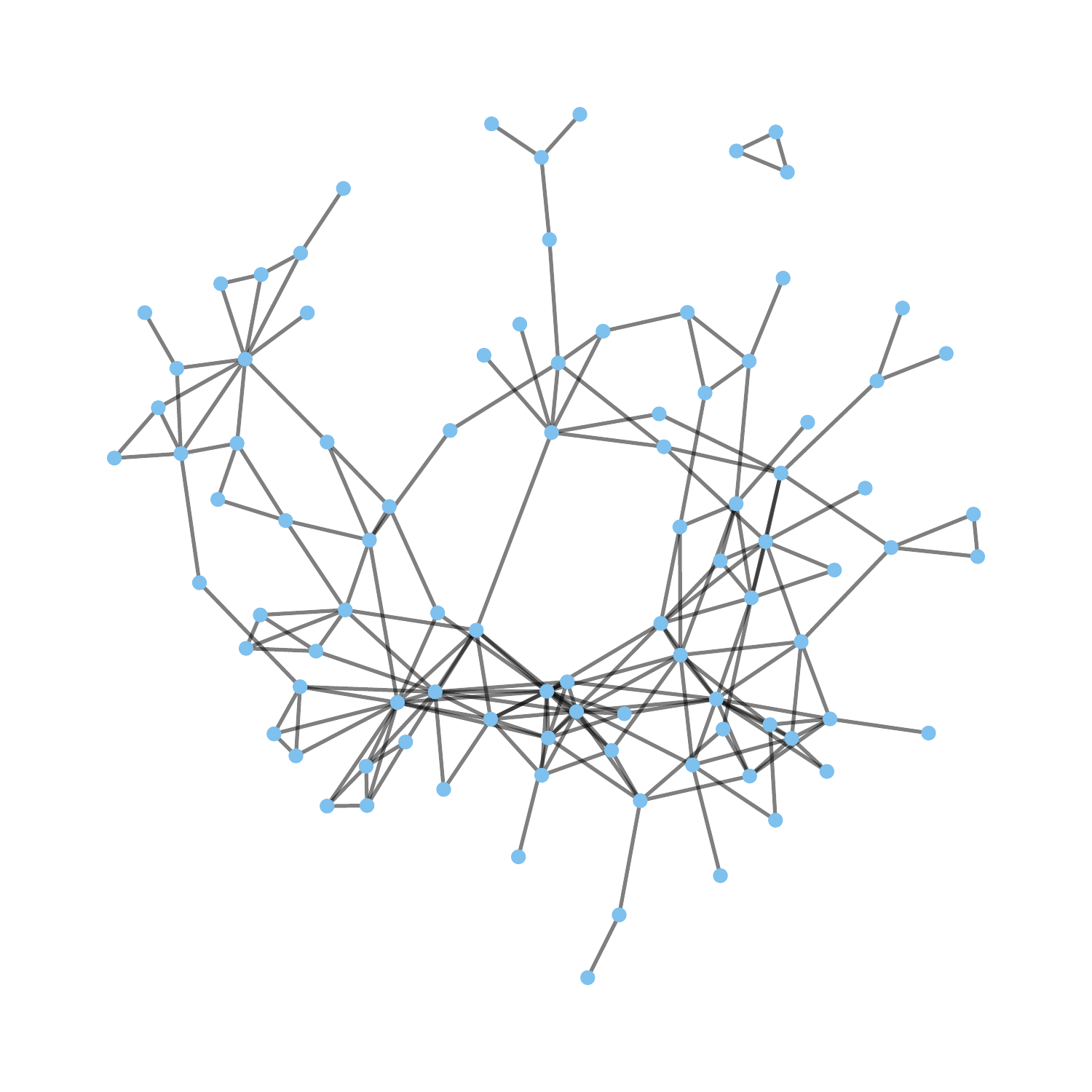}}
\subfloat[M values]{\includegraphics[width=0.25\textwidth]{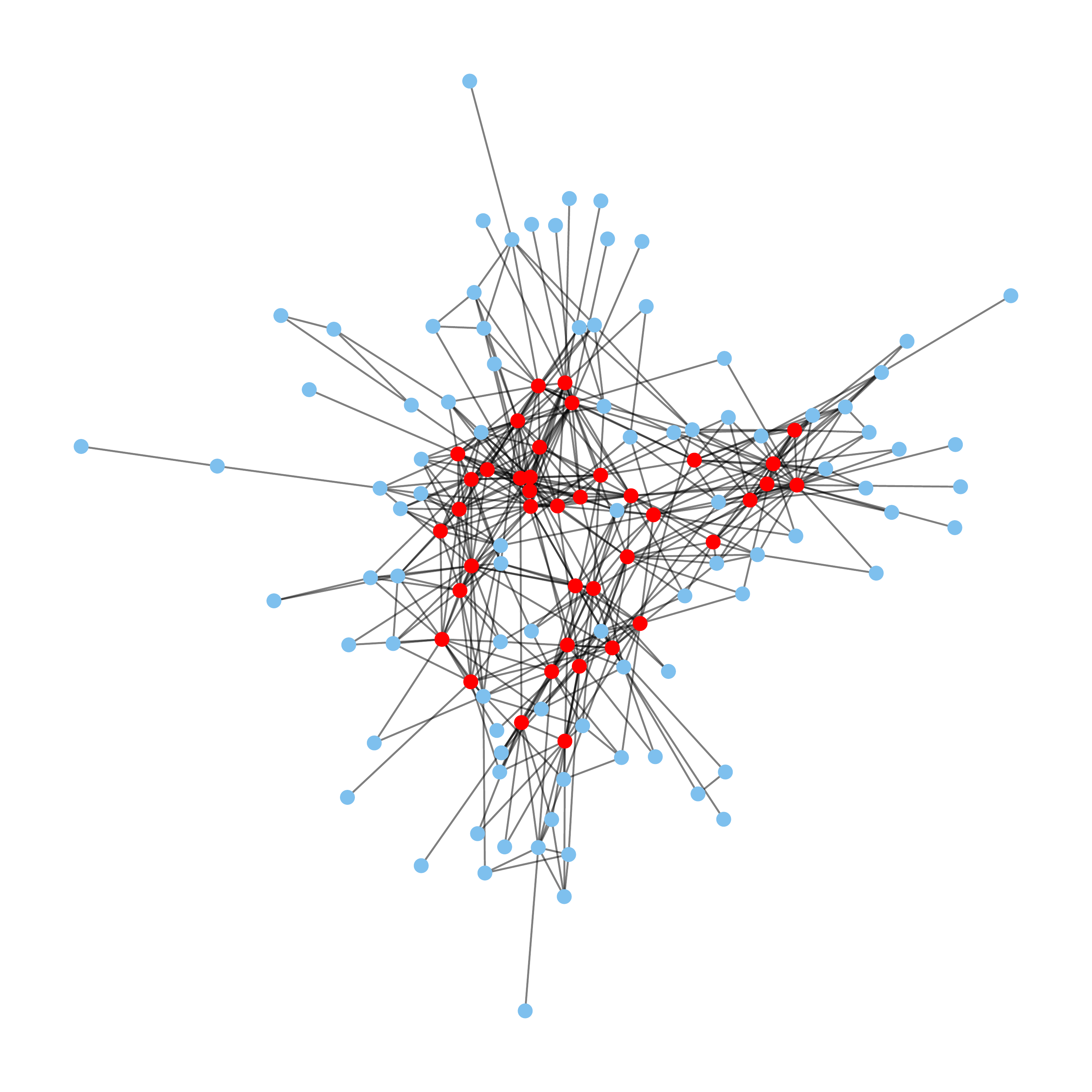}}
\caption{Graphs for CpG sites estimated by regularized generalized score matching estimator using Beta values (a) and M values (c), and their intersection graph (b). 
}\label{data_sites}
\end{figure}

\begin{table}[t!]
\centering
\scalebox{0.75}{
\begin{tabular}{c|c|||c|c}
\hline
Beta values & M values & Beta values & M values
\tabularnewline
\hline
CDH1|4 (28) & RXRB|24 (25)    & LEF1|2 (20) & PAX8|9 (20)\\
\textbf{TCF7L1|18 (22)} & \textbf{MAPK3|8 (22)}     & TCF7L1|13 (20) & TCF7|3 (18) \\
\textbf{RXRA|19 (21)} & PAX8|6 (21)     & CDKN1A|10 (20) & TCF7L1|9 (18) \\ 
RXRA|22 (21) & CCND1|19 (20)     & CDKN1A|6 (19) & \textbf{TCF7L1|18 (18)} \\
RET|22 (21) & RXRA|10 (20)     & \textbf{MAPK3|8 (17)} & TCF7L2|63 (18) \\
RXRB|82 (21) & \textbf{RXRA|19 (20)}     & PAX8|28 (17) & TPM3|12 (18) \\
NTRK1|40 (21) & RXRB|18 (20)    & \textbf{PAX8|29 (17)} & \textbf{PAX8|29 (17)}
\tabularnewline
\hline
\end{tabular}}
\caption{List of sites with the highest node degrees in each estimated graph.}\vspace{-0.1in}
\label{site_table}
\end{table}

We quantify the similarity between the two site graphs (not aggregated) by their Hamming distance and their DeltaCon similarity score \citep{kou13}. The Hamming distance counts the number of edge differences, and thus decreases as two graphs become more similar. Conversely, DeltaCon \citep{kou13} generates a similarity score in $[0,1]$, and the closer the score is to 1, the more similar the two graphs are.

The Hamming distance between the two graphs is 568, which is considerably smaller than 936, the minimal Hamming distance between the graph for Beta values and 10000 randomly generated graphs with the same number of edges, and 940, that value using the graph for M values. 
On the other hand, the DeltaCon similarity score between the two original graphs is 0.114, while the maximal score between the Beta graph and 10000 randomly generated graphs is only 0.0781, while that for the M graph is 0.0761. In Figure~\ref{graph_node_deg}, we compare the distribution of node degrees for both graphs
, with interlaced histogram on the left and Q-Q plot on the right. All
these results suggest that the two estimated graphs are similar to
each other, but that the two analyses also reveal complementary features.

\begin{figure}[t!]
\centering
\subfloat[Beta values]{\includegraphics[scale=0.25]{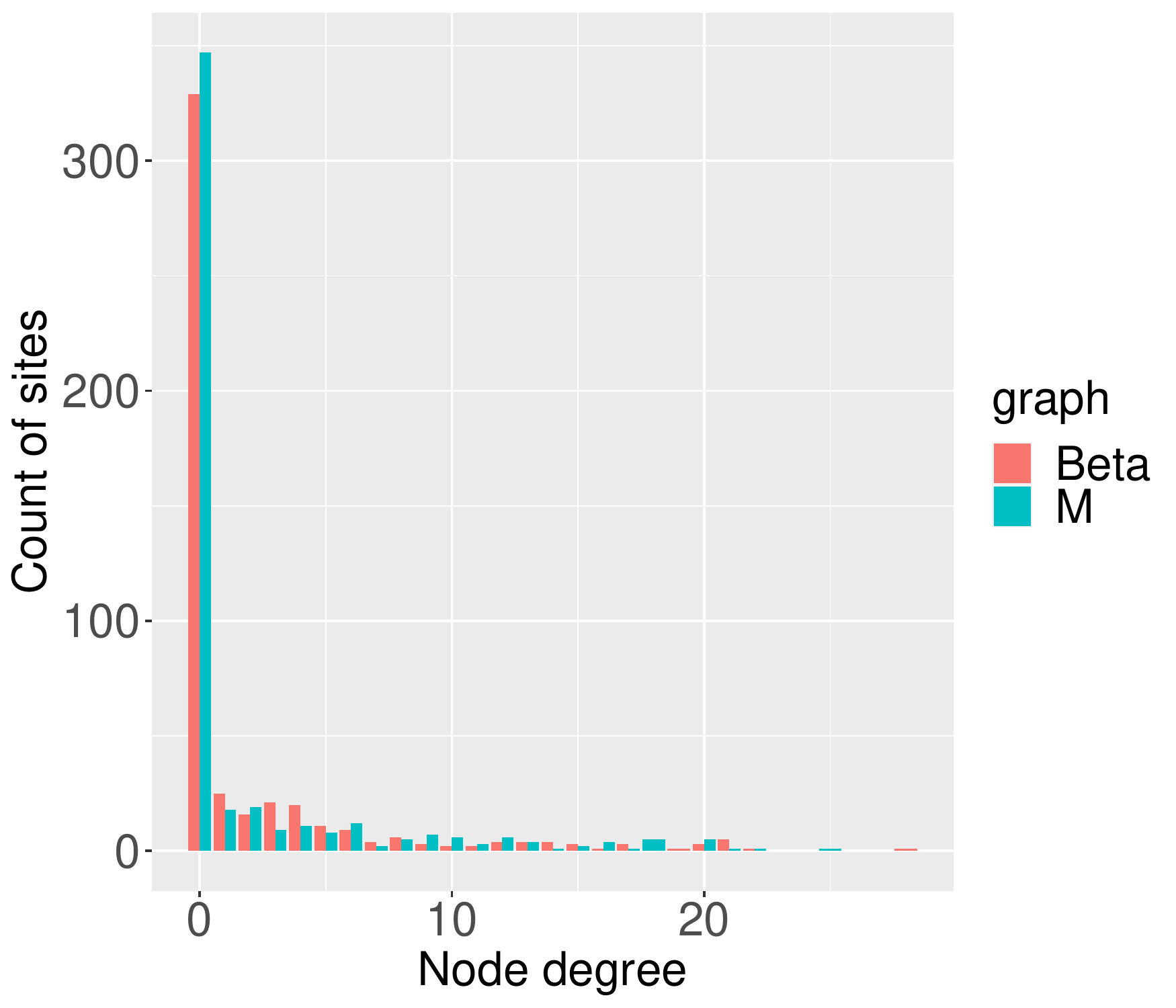}}
\subfloat[Common edges]{\includegraphics[scale=0.25]{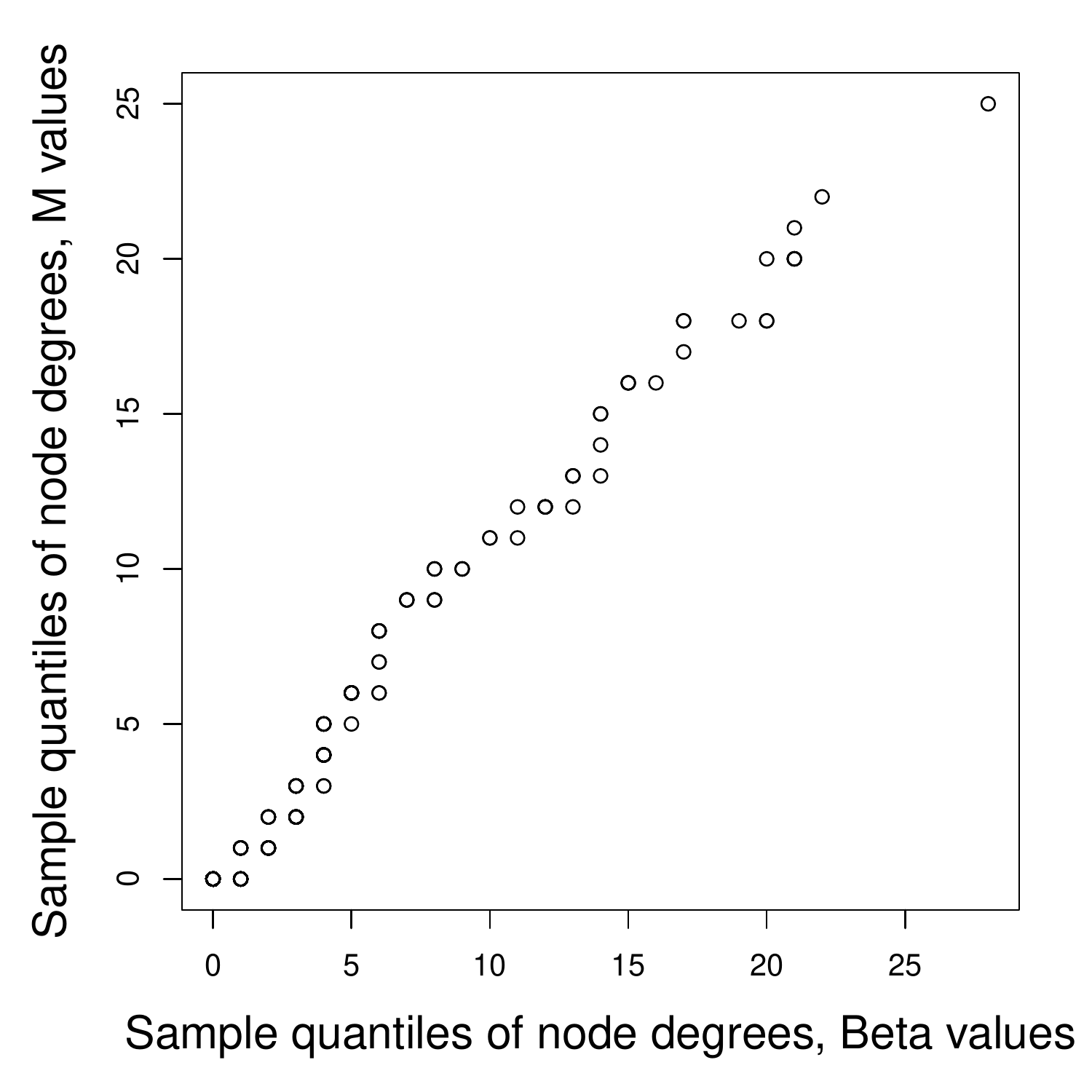}}
\caption{Interlaced histogram (left) and Q-Q plot (right) showing the node degree distributions for both site graphs.}\label{graph_node_deg}
\end{figure}


\section{Conclusion} 
Generalized score matching as proposed in \citet{yus19} is an
extension of the method of \cite{hyv07} that estimates
densities supported on $\mathbb{R}_+^m$ using a loss, in which the log-gradient of the
postulated density, $\nabla\log p(\boldsymbol{x})$, is multiplied
component-wise with a function
$\boldsymbol{h}(\boldsymbol{x})$.
The resulting estimator avoids the often costly calculation of normalizing
constants and has a closed-form
solution in exponential family models.

In this paper, we further extend generalized score matching to be
applicable to more general domains.  Specifically, we allow for
domains $\mathfrak{D}$ that are \emph{component-wise countable union
  of intervals} $\mathfrak{D}$ (see Definition \ref{def_V}).
We accomplish this by composing the function $\boldsymbol{h}$ with a
distance function
$\boldsymbol{\varphi}_{\boldsymbol{C}}=(\varphi_{C_1,1},\ldots,\varphi_{C_m,m}):\mathfrak{D}\to\mathbb{R}_+^m$,
where $\varphi_{C_j,j}(\boldsymbol{x})$ is a truncated distance of $x_j$ to
the boundary of the relevant interval in the section of $\mathfrak{D}$
given by $\boldsymbol{x}_{-j}$.  The resulting loss can again be
approximated by an empirical loss,
which is
quadratic in the canonical parameters for exponential families.  

In our applications we focus on
$a$-$b$ pairwise interaction models supported on domains
$\mathfrak{D}$ with positive Lebesgue measure.  For these models we
give a concrete choice of the function $\boldsymbol{h}$ and 
%
extend the consistency theory for support recovery in \citet{yus19} to
Gaussian models on domains that are finite disjoint unions of convex
sets, and on bounded domains with positive Lebesgue measure, requiring
the sample size to be $n=\Omega(\log m)$. For unbounded domains with
$a>0$, we require an additional multiplicative factor that may weakly
depend on $m$.  Deriving a more explicit requirement on the sample
size would be an interesting topic for future work.
Finally, 
in our simulations we adaptively select the truncation points
$\boldsymbol{C}$ of $\boldsymbol{\varphi}_{\boldsymbol{C}}$ using the
sample quantiles of the untruncated distances.  Developing a method to
choose the best truncation points remains a topic for further
research.

\section*{Acknowledgments}
This work was supported by grants DMS/NIGMS-1561814 from the National Science Foundation (NSF) and R01-GM114029 from the National Institutes of Health (NIH).

\bibliographystyle{plainnat}  
\bibliography{Paper}

\begin{thebibliography}{22}
\providecommand{\natexlab}[1]{#1}
\providecommand{\url}[1]{\texttt{#1}}
\expandafter\ifx\csname urlstyle\endcsname\relax
  \providecommand{\doi}[1]{doi: #1}\else
  \providecommand{\doi}{doi: \begingroup \urlstyle{rm}\Url}\fi

\bibitem[Du et~al.(2010)Du, Zhang, Huang, Jafari, Kibbe, Hou, and Lin]{du10}
Pan Du, Xiao Zhang, Chiang-Ching Huang, Nadereh Jafari, Warren~A Kibbe, Lifang
  Hou, and Simon~M Lin.
\newblock Comparison of beta-value and m-value methods for quantifying
  methylation levels by microarray analysis.
\newblock \emph{BMC Bioinformatics}, 11\penalty0 (1):\penalty0 587, 2010.

\bibitem[Forbes and Lauritzen(2015)]{MR3338335}
Peter G.~M. Forbes and Steffen Lauritzen.
\newblock Linear estimating equations for exponential families with application
  to {G}aussian linear concentration models.
\newblock \emph{Linear Algebra Appl.}, 473:\penalty0 261--283, 2015.

\bibitem[Hyv\"{a}rinen(2005)]{hyv05}
Aapo Hyv\"{a}rinen.
\newblock Estimation of non-normalized statistical models by score matching.
\newblock \emph{J. Mach. Learn. Res.}, 6:\penalty0 695--709, 2005.

\bibitem[Hyv\"{a}rinen(2007)]{hyv07}
Aapo Hyv\"{a}rinen.
\newblock Some extensions of score matching.
\newblock \emph{Comput. Statist. Data Anal.}, 51\penalty0 (5):\penalty0
  2499--2512, 2007.

\bibitem[Inouye et~al.(2016)Inouye, Ravikumar, and Dhillon]{ino16}
David Inouye, Pradeep Ravikumar, and Inderjit Dhillon.
\newblock Square root graphical models: Multivariate generalizations of
  univariate exponential families that permit positive dependencies.
\newblock In \emph{Proceedings of the 33rd International Conference on Machine
  Learning}, volume~48 of \emph{Proceedings of Machine Learning Research},
  pages 2445--2453, 2016.

\bibitem[Janofsky(2015)]{jan15}
Eric Janofsky.
\newblock Exponential series approaches for nonparametric graphical models.
\newblock \emph{arXiv preprint arXiv:1506.03537}, 2015.

\bibitem[Janofsky(2018)]{jan18}
Eric Janofsky.
\newblock Learning high-dimensional graphical models with regularized quadratic
  scoring.
\newblock \emph{arXiv preprint arXiv:1809.05638}, 2018.

\bibitem[Koutra et~al.(2013)Koutra, Vogelstein, and Faloutsos]{kou13}
Danai Koutra, Joshua~T Vogelstein, and Christos Faloutsos.
\newblock Deltacon: A principled massive-graph similarity function.
\newblock In \emph{Proceedings of the 2013 SIAM International Conference on
  Data Mining}, pages 162--170. SIAM, 2013.

\bibitem[Lin et~al.(2016)Lin, Drton, and Shojaie]{lin16}
Lina Lin, Mathias Drton, and Ali Shojaie.
\newblock Estimation of high-dimensional graphical models using regularized
  score matching.
\newblock \emph{Electron. J. Stat.}, 10\penalty0 (1):\penalty0 806--854, 2016.

\bibitem[Liu and Kanamori(2019)]{liu19}
Song Liu and Takafumi Kanamori.
\newblock Estimating density models with complex truncation boundaries.
\newblock \emph{arXiv preprint arXiv:1910.03834}, 2019.

\bibitem[Liu and Luo(2015)]{MR3306432}
Weidong Liu and Xi~Luo.
\newblock Fast and adaptive sparse precision matrix estimation in high
  dimensions.
\newblock \emph{J. Multivariate Anal.}, 135:\penalty0 153--162, 2015.

\bibitem[Maathuis et~al.(2019)Maathuis, Drton, Lauritzen, and
  Wainwright]{maa18}
Marloes Maathuis, Mathias Drton, Steffen Lauritzen, and Martin Wainwright,
  editors.
\newblock \emph{Handbook of graphical models}.
\newblock Chapman \& Hall/CRC Handbooks of Modern Statistical Methods. CRC
  Press, Boca Raton, FL, 2019.

\bibitem[Sun et~al.(2015)Sun, Kolar, and Xu]{NIPS2015_6006}
Siqi Sun, Mladen Kolar, and Jinbo Xu.
\newblock Learning structured densities via infinite dimensional exponential
  families.
\newblock In C.~Cortes, N.~D. Lawrence, D.~D. Lee, M.~Sugiyama, and R.~Garnett,
  editors, \emph{Advances in Neural Information Processing Systems 28}, pages
  2287--2295. Curran Associates, Inc., 2015.

\bibitem[Tan et~al.(2019)Tan, Lu, Zhang, and Liu]{tan19}
Kean~Ming Tan, Junwei Lu, Tong Zhang, and Han Liu.
\newblock Layer-wise learning strategy for nonparametric tensor product
  smoothing spline regression and graphical models.
\newblock \emph{Journal of Machine Learning Research}, 20\penalty0
  (119):\penalty0 1--38, 2019.

\bibitem[Vershynin(2012)]{ver12}
Roman Vershynin.
\newblock Introduction to the non-asymptotic analysis of random matrices.
\newblock In \emph{Compressed Sensing}, pages 210--268. Cambridge Univ. Press,
  Cambridge, 2012.

\bibitem[Wainwright(2019)]{wai19}
Martin~J Wainwright.
\newblock \emph{High-dimensional statistics: A non-asymptotic viewpoint},
  volume~48.
\newblock Cambridge University Press, 2019.

\bibitem[Weinstein et~al.(2013)Weinstein, Collisson, Mills, Shaw, Ozenberger,
  Ellrott, Shmulevich, Sander, Stuart, Network, et~al.]{wei13}
John~N Weinstein, Eric~A Collisson, Gordon~B Mills, Kenna R~Mills Shaw, Brad~A
  Ozenberger, Kyle Ellrott, Ilya Shmulevich, Chris Sander, Joshua~M Stuart,
  Cancer Genome Atlas~Research Network, et~al.
\newblock The cancer genome atlas pan-cancer analysis project.
\newblock \emph{Nature genetics}, 45\penalty0 (10):\penalty0 1113, 2013.

\bibitem[Yu et~al.(2016)Yu, Kolar, and Gupta]{NIPS2016_6530}
Ming Yu, Mladen Kolar, and Varun Gupta.
\newblock Statistical inference for pairwise graphical models using score
  matching.
\newblock In D.~D. Lee, M.~Sugiyama, U.~V. Luxburg, I.~Guyon, and R.~Garnett,
  editors, \emph{Advances in Neural Information Processing Systems 29}, pages
  2829--2837. Curran Associates, Inc., 2016.

\bibitem[Yu et~al.(2019{\natexlab{a}})Yu, Gupta, and Kolar]{yum19}
Ming Yu, Varun Gupta, and Mladen Kolar.
\newblock Simultaneous inference for pairwise graphical models with generalized
  score matching.
\newblock \emph{arXiv preprint arXiv:1905.06261}, 2019{\natexlab{a}}.

\bibitem[Yu et~al.(2018)Yu, Drton, and Shojaie]{yu18}
Shiqing Yu, Mathias Drton, and Ali Shojaie.
\newblock Graphical models for non-negative data using generalized score
  matching.
\newblock In \emph{International Conference on Artificial Intelligence and
  Statistics}, pages 1781--1790, 2018.

\bibitem[Yu et~al.(2019{\natexlab{b}})Yu, Drton, and Shojaie]{yus19}
Shiqing Yu, Mathias Drton, and Ali Shojaie.
\newblock Generalized score matching for non-negative data.
\newblock \emph{Journal of Machine Learning Research}, 20\penalty0
  (76):\penalty0 1--70, 2019{\natexlab{b}}.

\bibitem[Zhang and Zou(2014)]{MR3180660}
Teng Zhang and Hui Zou.
\newblock Sparse precision matrix estimation via lasso penalized {D}-trace
  loss.
\newblock \emph{Biometrika}, 101\penalty0 (1):\penalty0 103--120, 2014.

\end{thebibliography}

\appendix
\section{Additional Plots}\label{Additional Plots}

Below we present additional plots for our simulations in Sections \ref{Numerical Experiments} and \ref{DNA Methylation Data}.
\begin{figure}[!htp]
\centering
\vspace{-0.2in}
\subfloat[$n=80$, $\ell_2$-nn domain]
{\includegraphics[width=0.45\textwidth]{{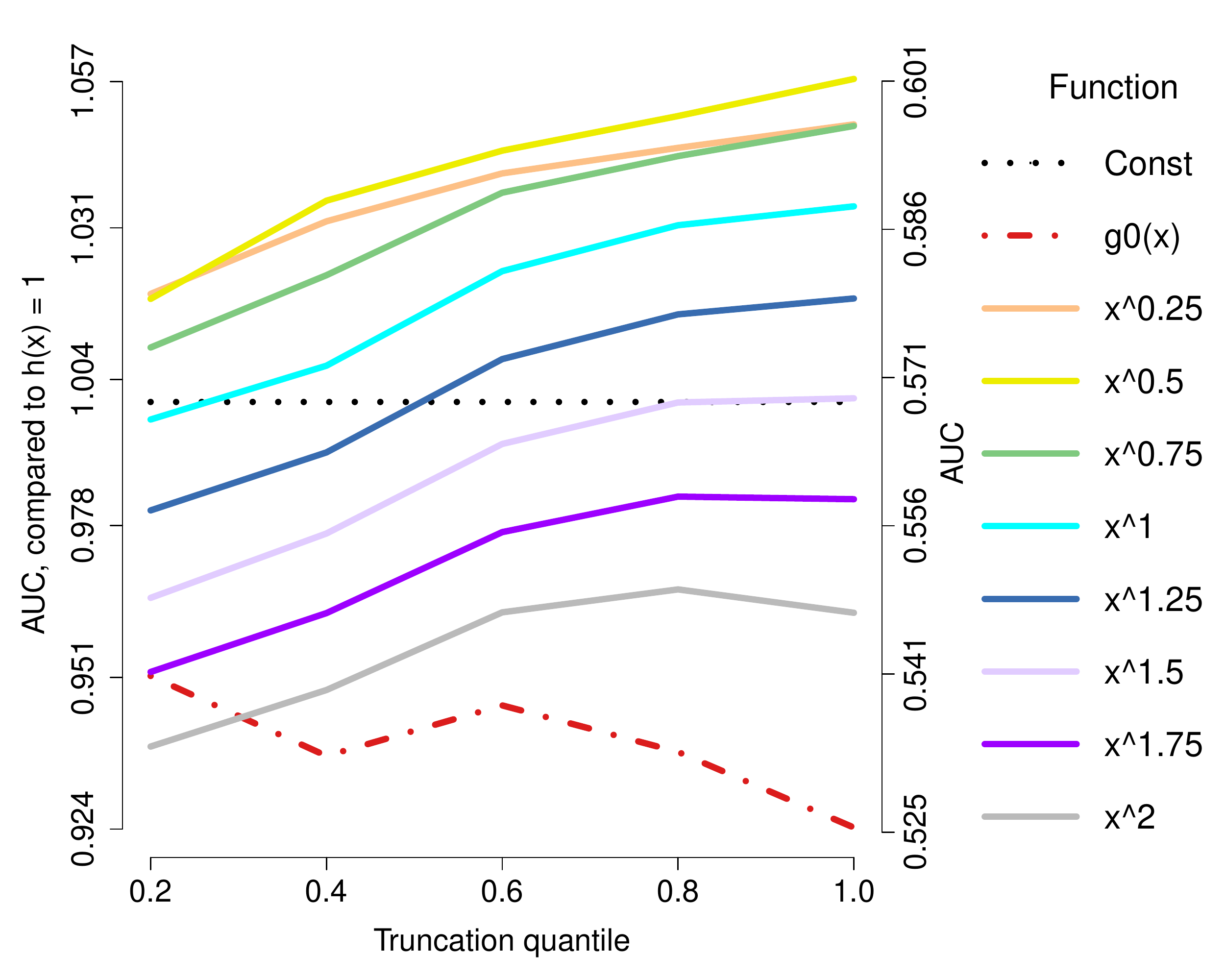}}\hspace{0.1in}}
\subfloat[$n=1000$, $\ell_2$-nn domain]
{\includegraphics[width=0.45\textwidth]{{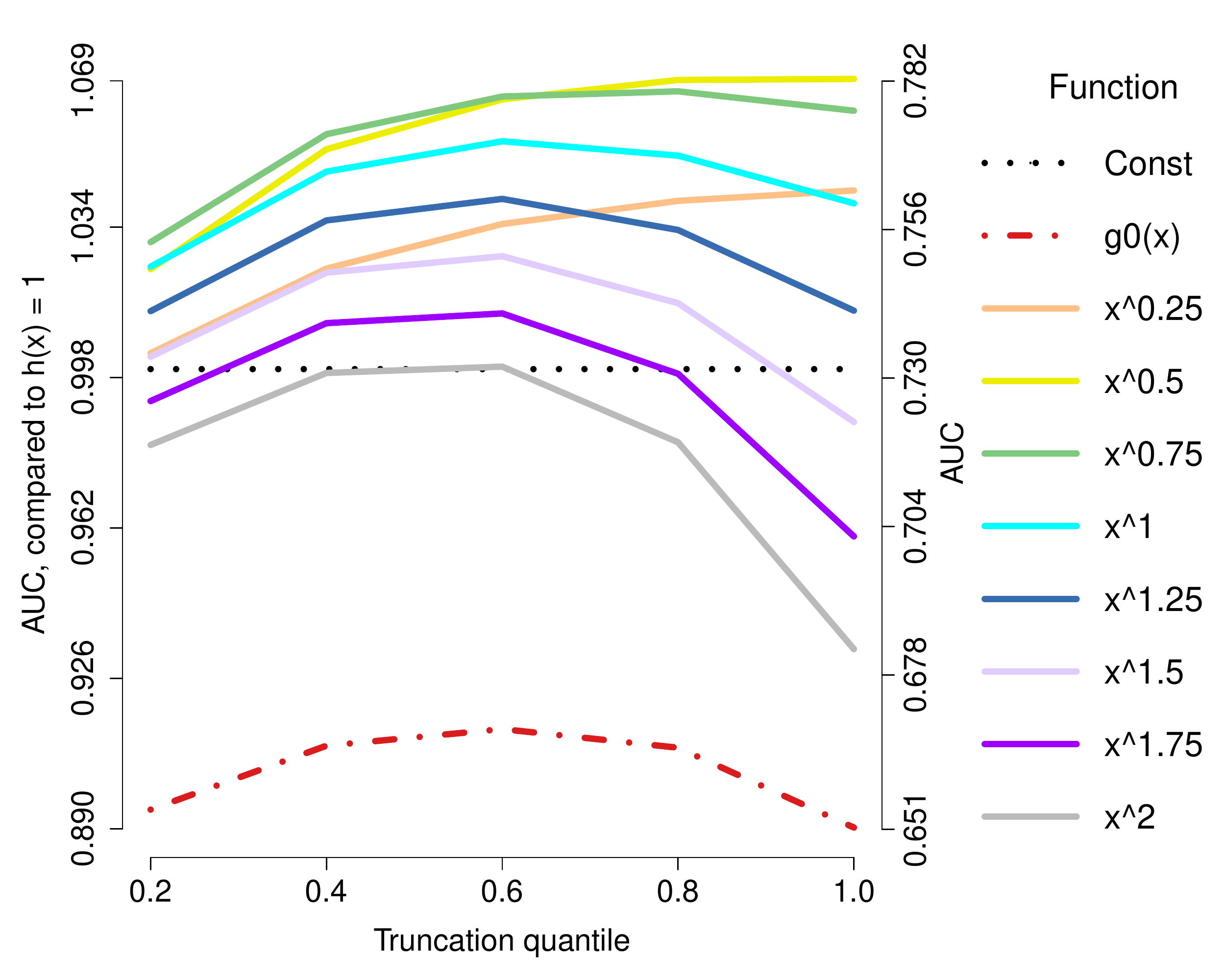}}\hspace{-0.02in}}
\\ \vspace{-0.1in}
\subfloat[$n=80$, $\ell_2^{\complement}$-nn domain]
{\includegraphics[width=0.45\textwidth]{{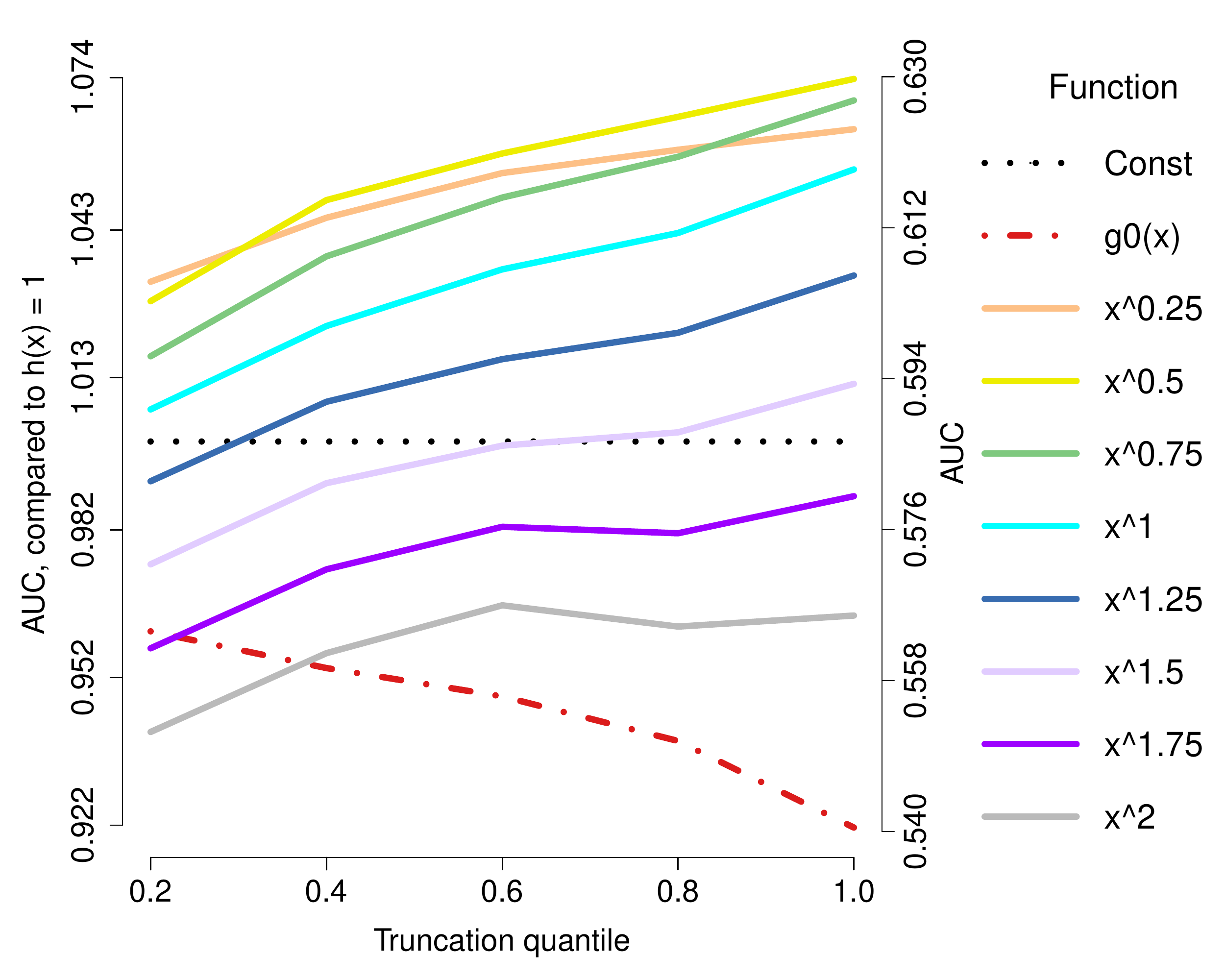}}\hspace{0.1in}}
\subfloat[$n=1000$, $\ell_2^{\complement}$-nn domain]
{\includegraphics[width=0.45\textwidth]{{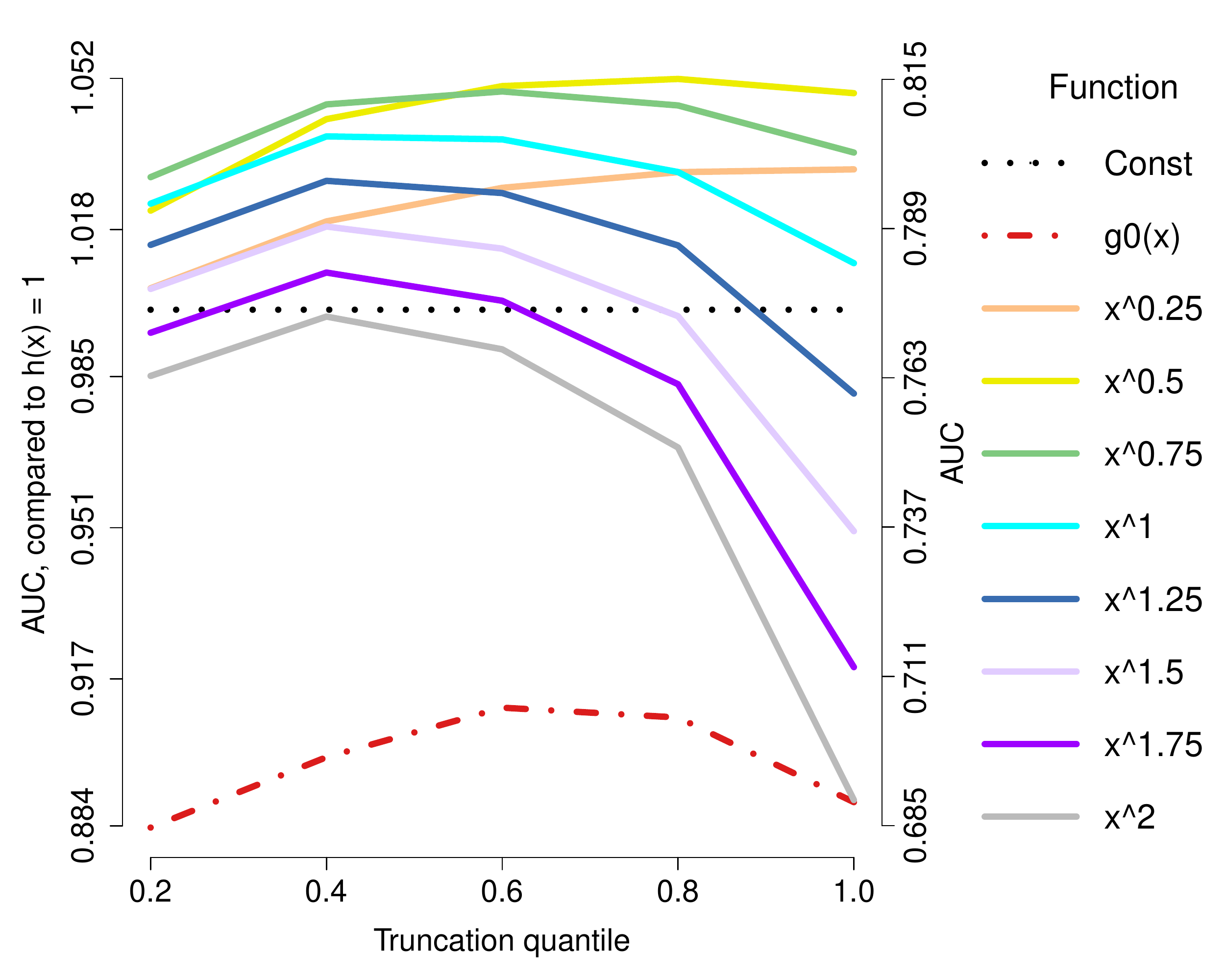}}\hspace{-0.02in}}
\\ \vspace{-0.1in}
\subfloat[$n=80$, unif-nn domain]
{\includegraphics[width=0.45\textwidth]{{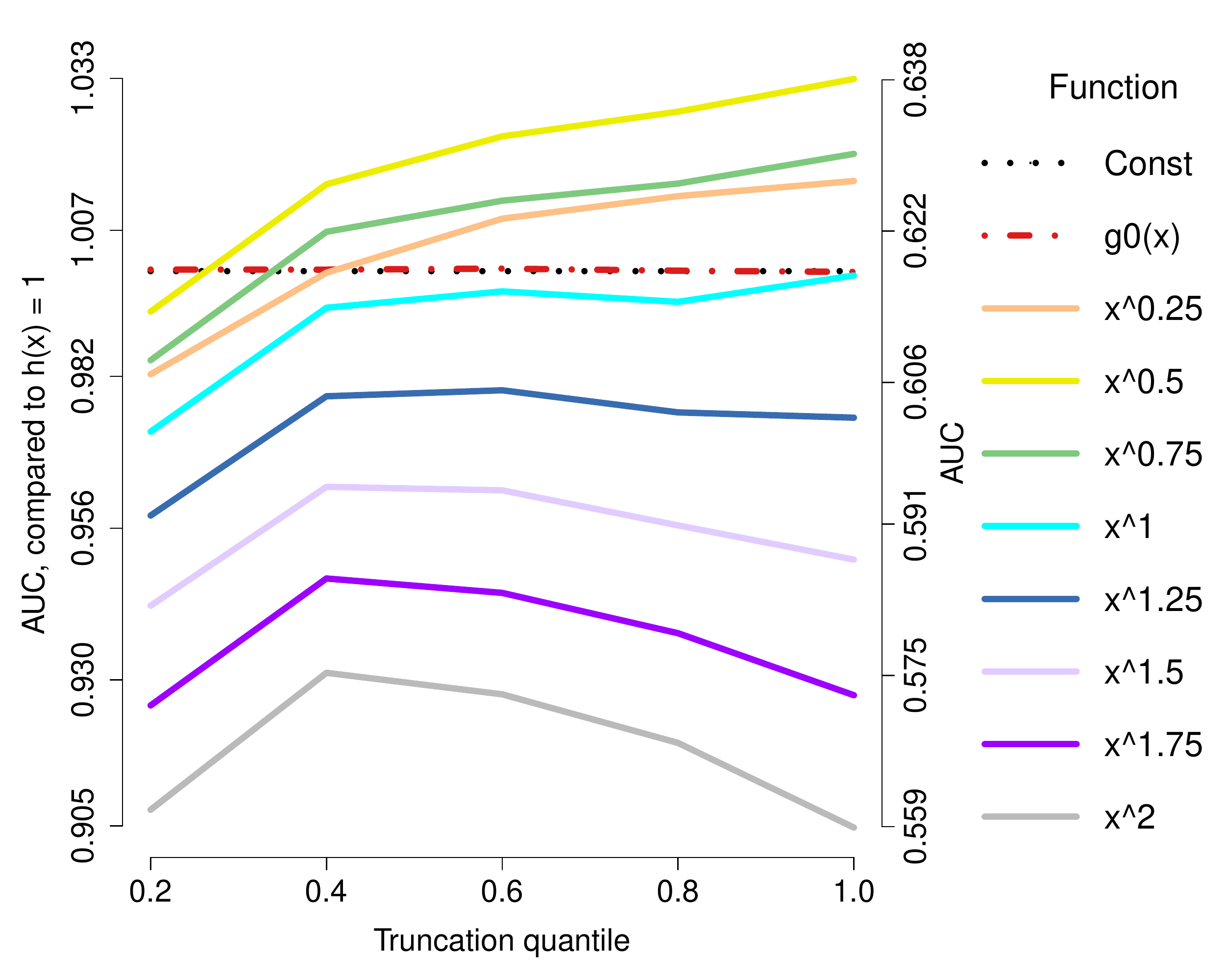}}\hspace{0.1in}}
\subfloat[$n=1000$, unif-nn domain]
{\includegraphics[width=0.45\textwidth]{{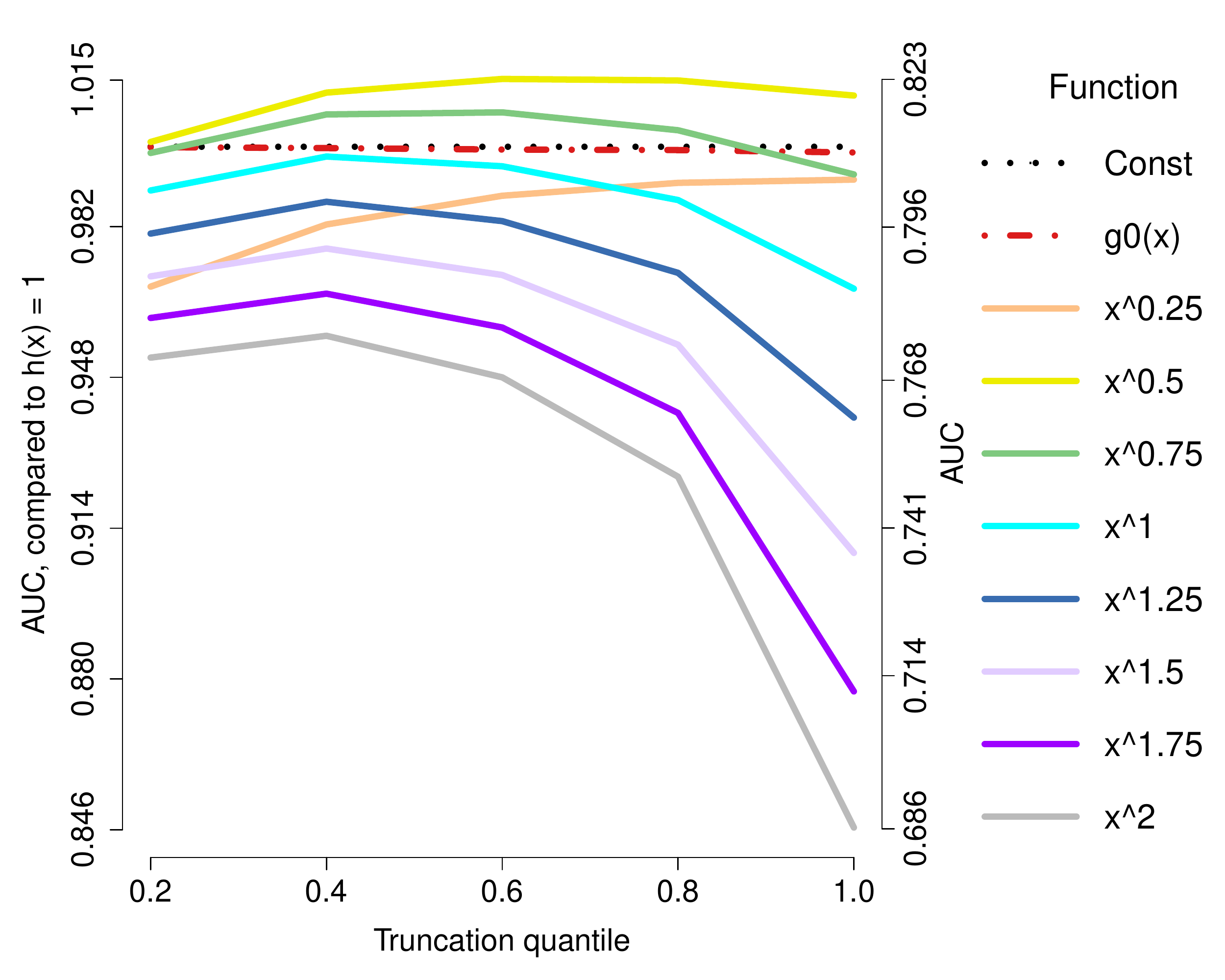}}\hspace{-0.02in}}
\caption{AUCs averaged over 50 trials for support recovery using generalized score matching for the $a=3/2$ models. Each curve represents either our extension to $g_0(\boldsymbol{x})$ from \citet{liu19} or a choice of power function $h(x)=x^c$.}
\vspace{-0.5in}\label{plot_a_3|2}
\end{figure}

\begin{figure}[!htp]
\centering
\vspace{-0.5in}
\subfloat[$n=80$, $\ell_2$-nn domain]
{\includegraphics[width=0.45\textwidth]{{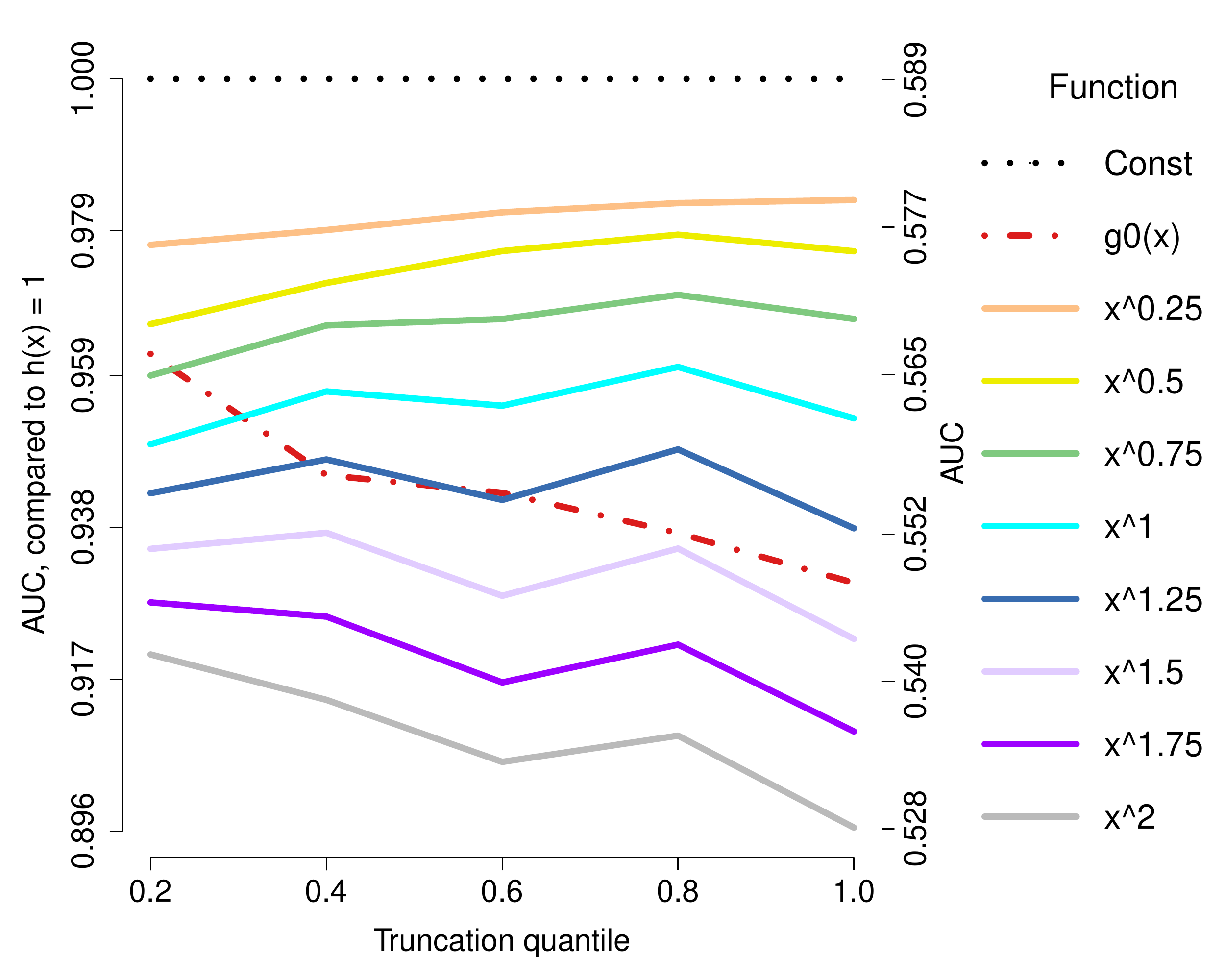}}\hspace{0.1in}}
\subfloat[$n=1000$, $\ell_2$-nn domain]
{\includegraphics[width=0.45\textwidth]{{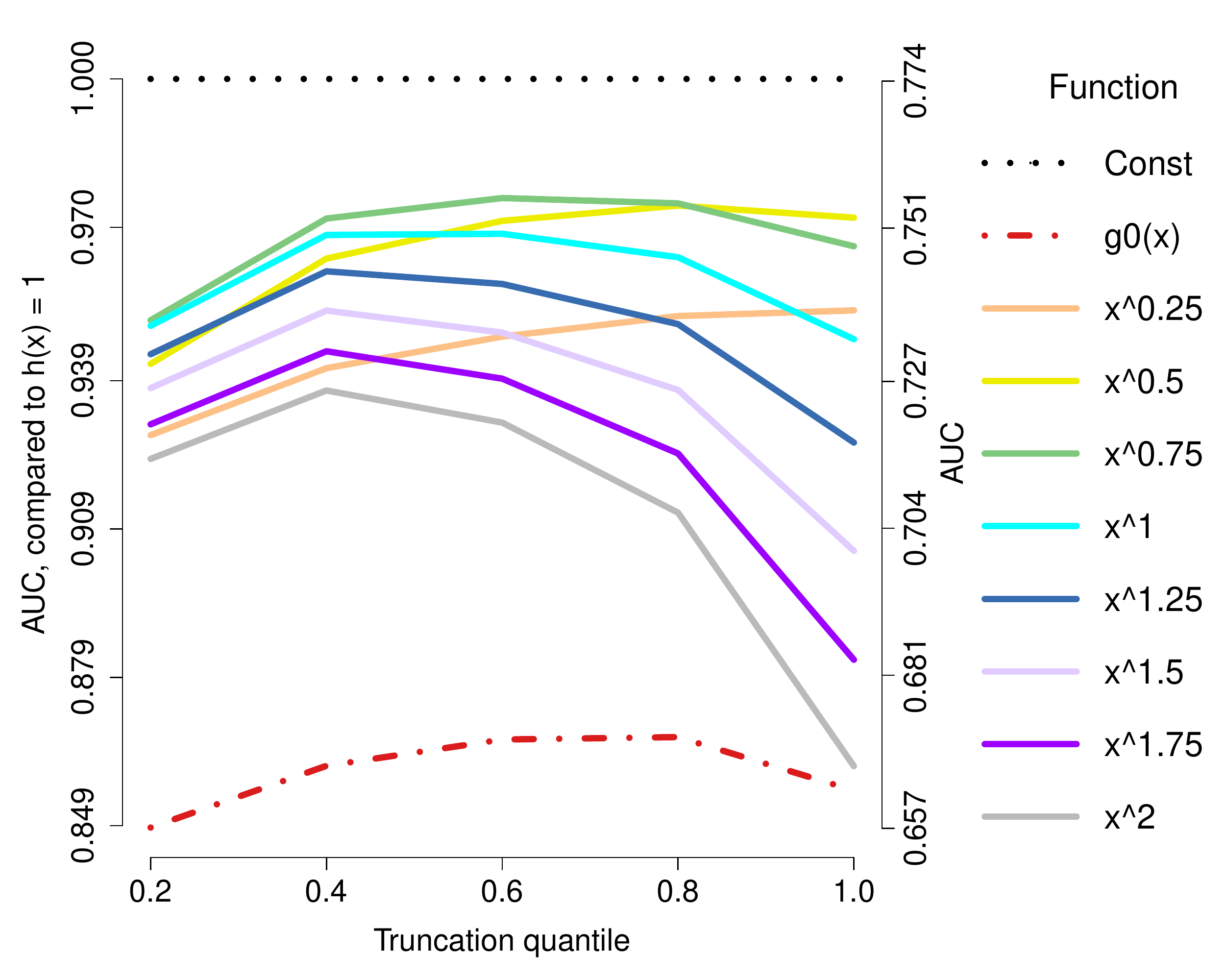}}\hspace{-0.02in}}
\\ \vspace{-0.1in}
\subfloat[$n=80$, $\ell_2^{\complement}$-nn domain]
{\includegraphics[width=0.45\textwidth]{{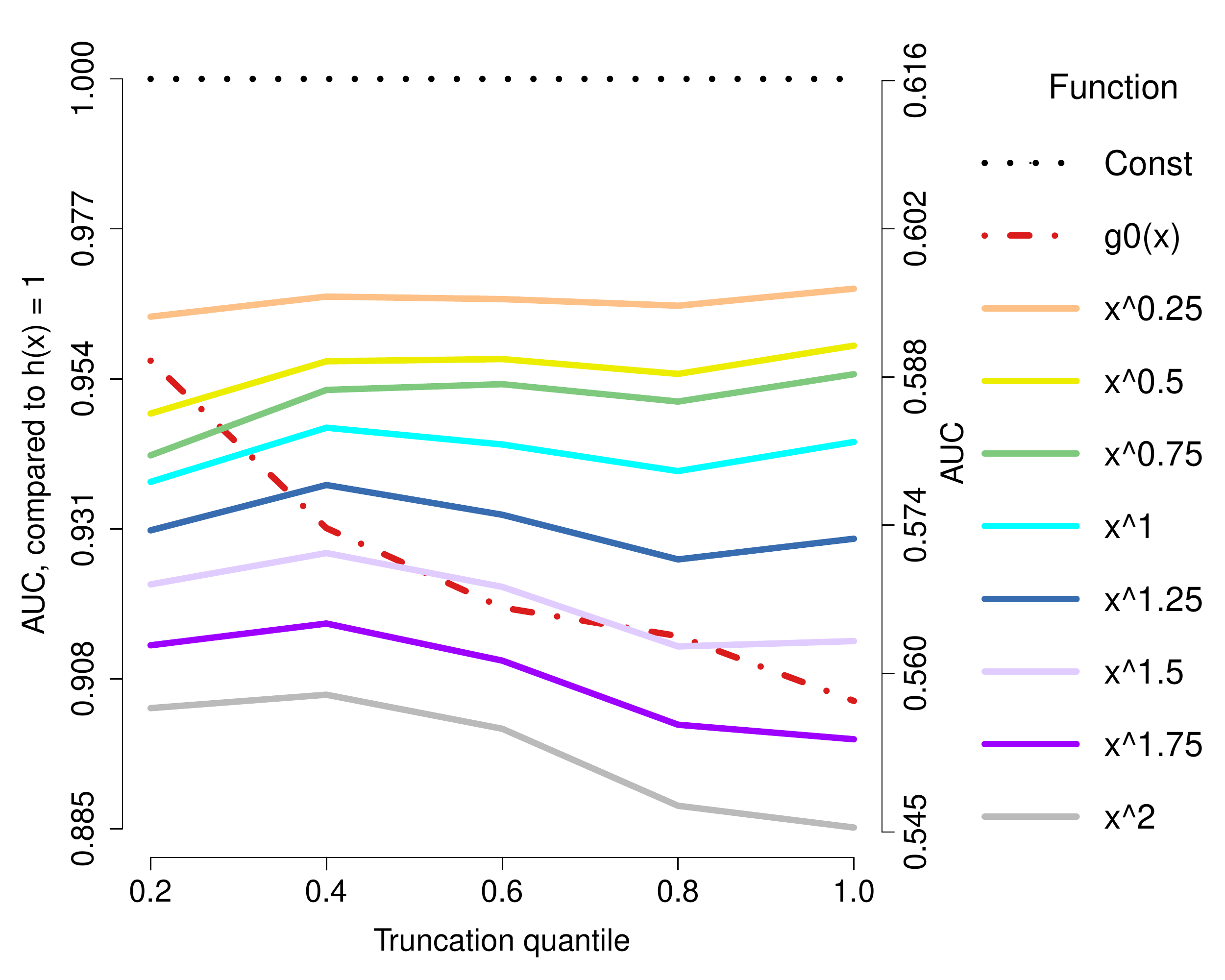}}\hspace{0.1in}}
\subfloat[$n=1000$, $\ell_2^{\complement}$-nn domain]
{\includegraphics[width=0.45\textwidth]{{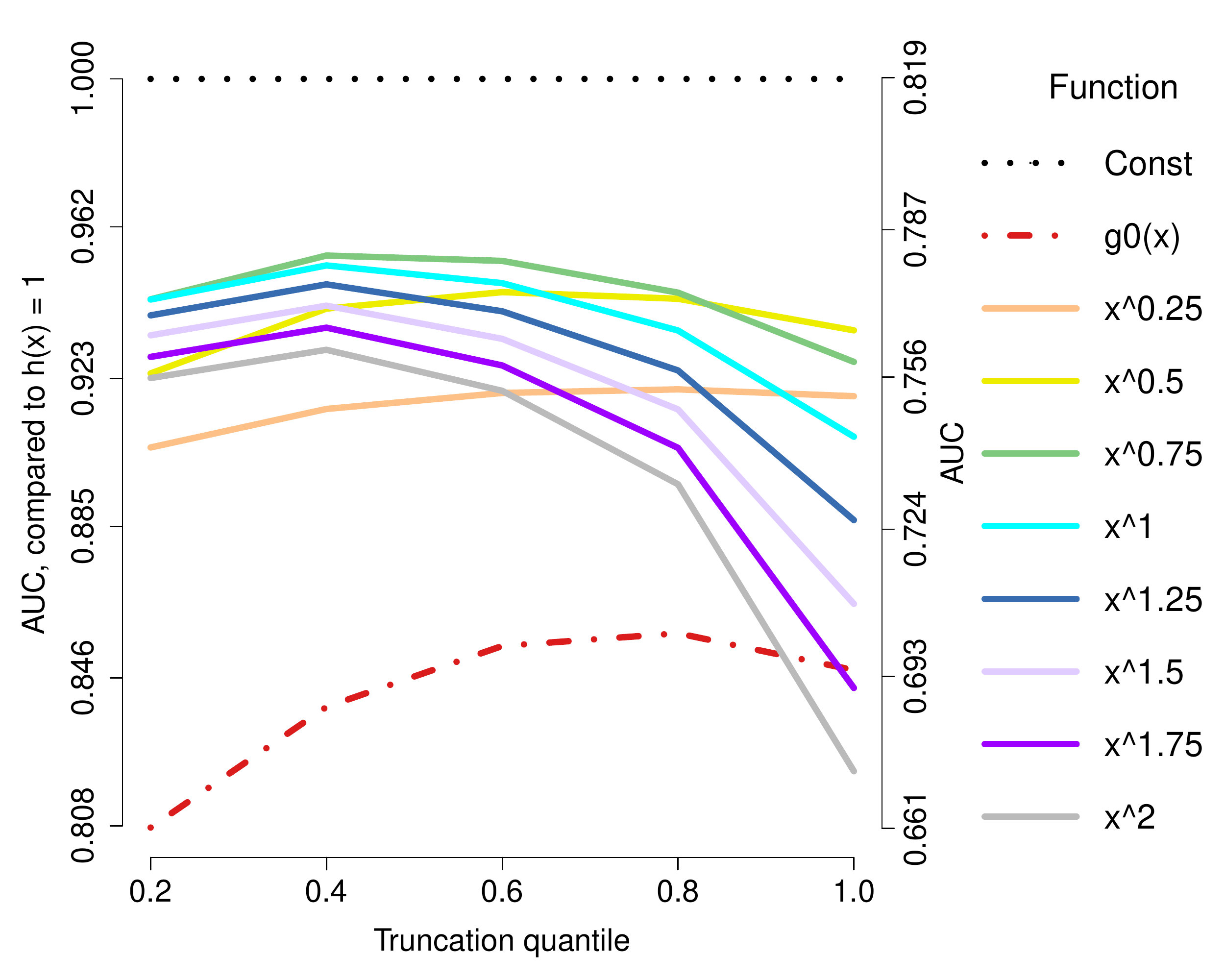}}\hspace{-0.02in}}
\\ \vspace{-0.1in}
\subfloat[$n=80$, unif-nn domain]
{\includegraphics[width=0.45\textwidth]{{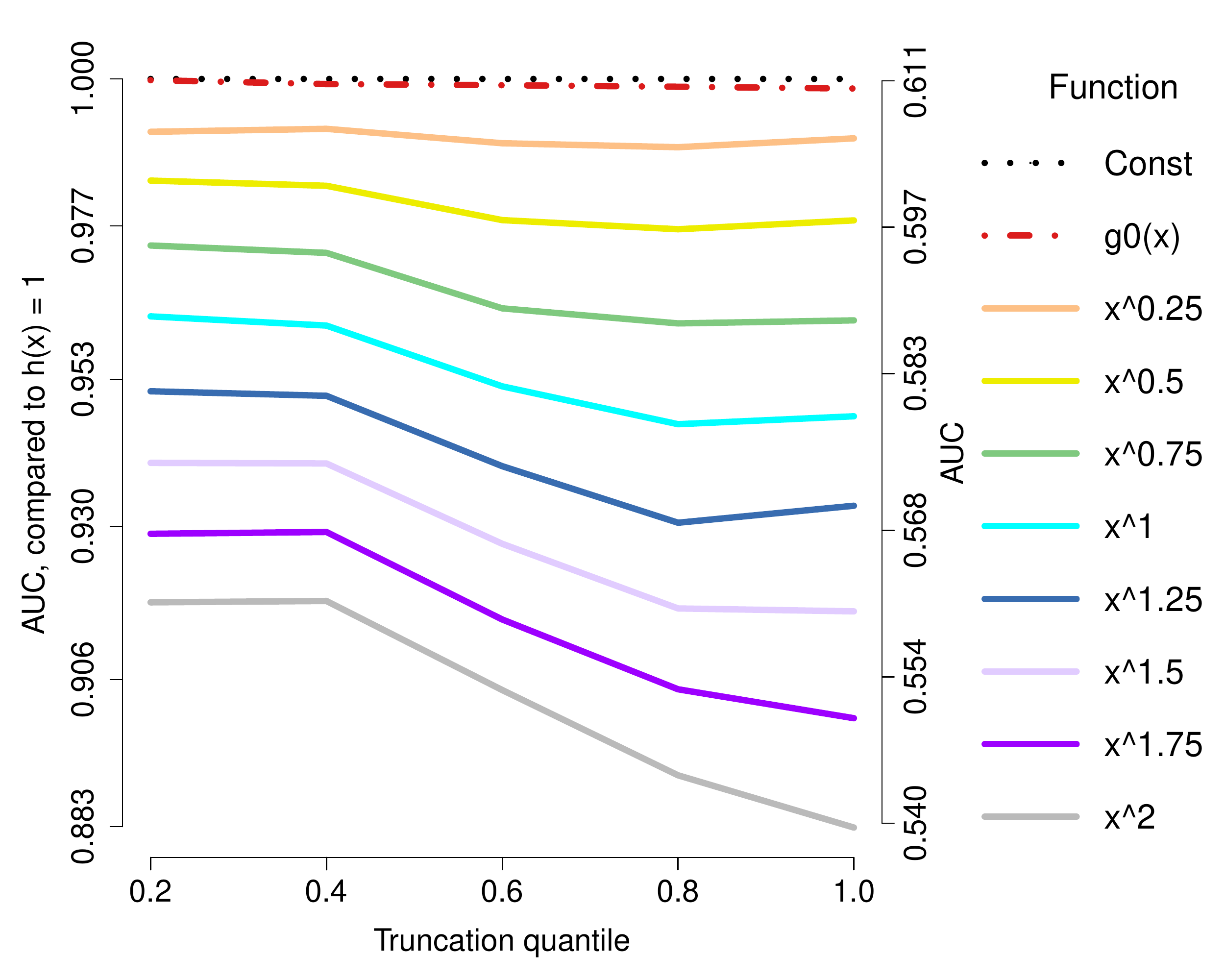}}\hspace{0.1in}}
\subfloat[$n=1000$, unif-nn domain]
{\includegraphics[width=0.45\textwidth]{{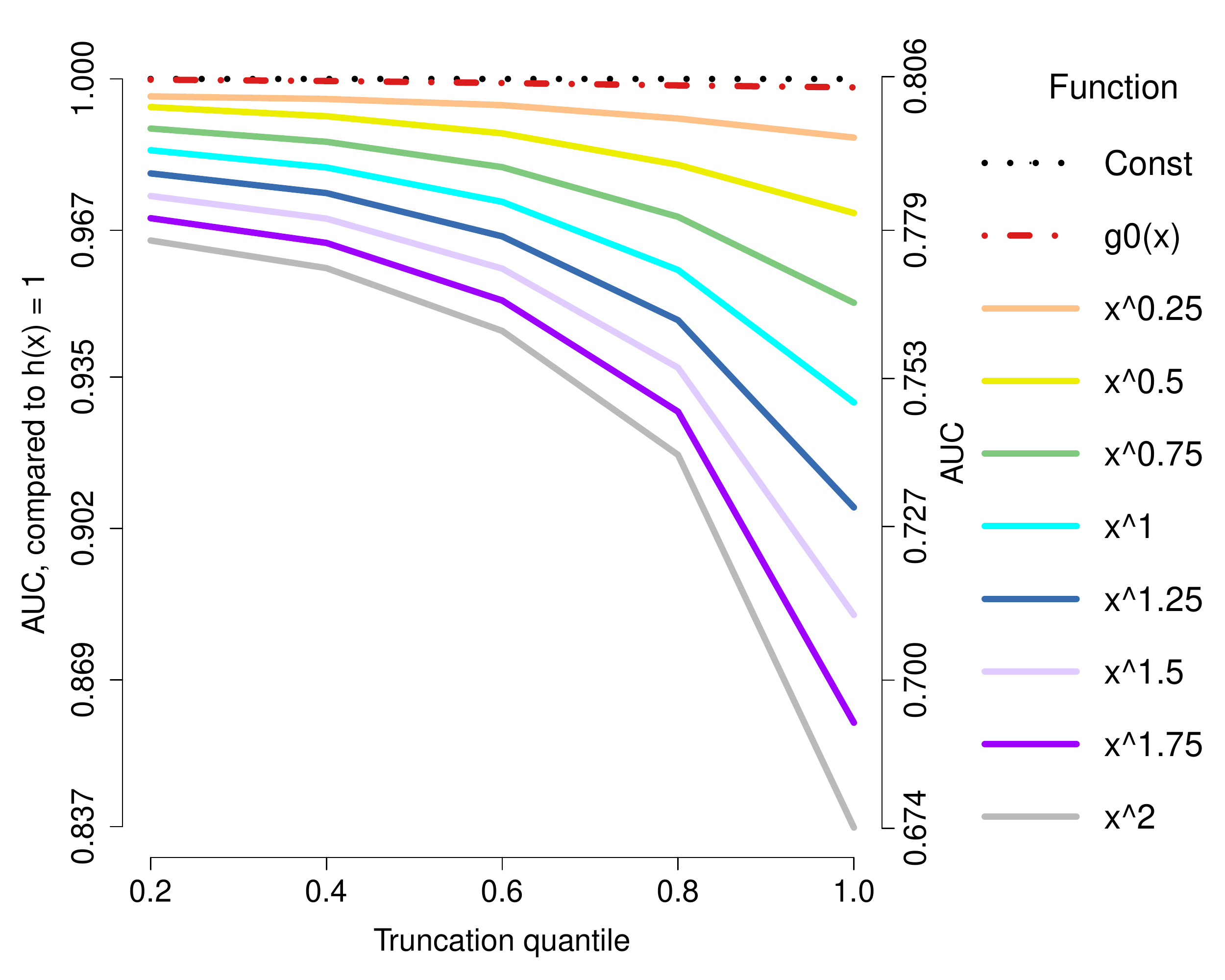}}\hspace{-0.02in}}
\caption{AUCs averaged over 50 trials for support recovery using generalized score matching for the $a=2$ models. Each curve represents either our extension to $g_0(\boldsymbol{x})$ from \citet{liu19} or a choice of power function $h(x)=x^c$.  The $x$ axes mark the probabilities $\pi$ that determine the truncation points $\boldsymbol{C}$ for the truncated component-wise distances. The colors are sorted by the power $c$.}
\label{plot_a_2}
\end{figure}

\begin{figure}[!htp]
\centering
\vspace{-0.5in}
\subfloat[$n=80$, $\ell_2$-nn domain]
{\includegraphics[width=0.45\textwidth]{{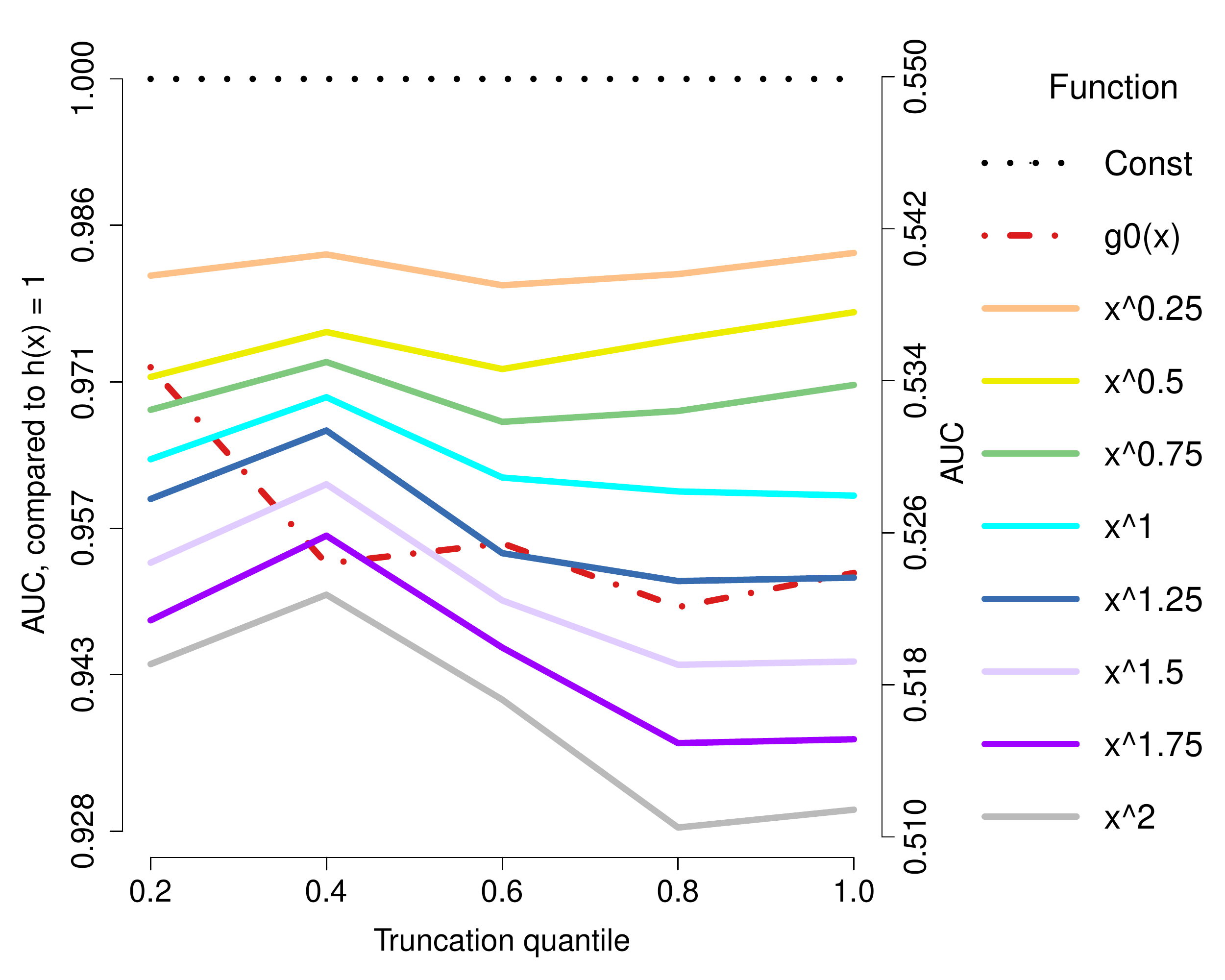}}\hspace{0.1in}}
\subfloat[$n=1000$, $\ell_2$-nn domain]
{\includegraphics[width=0.45\textwidth]{{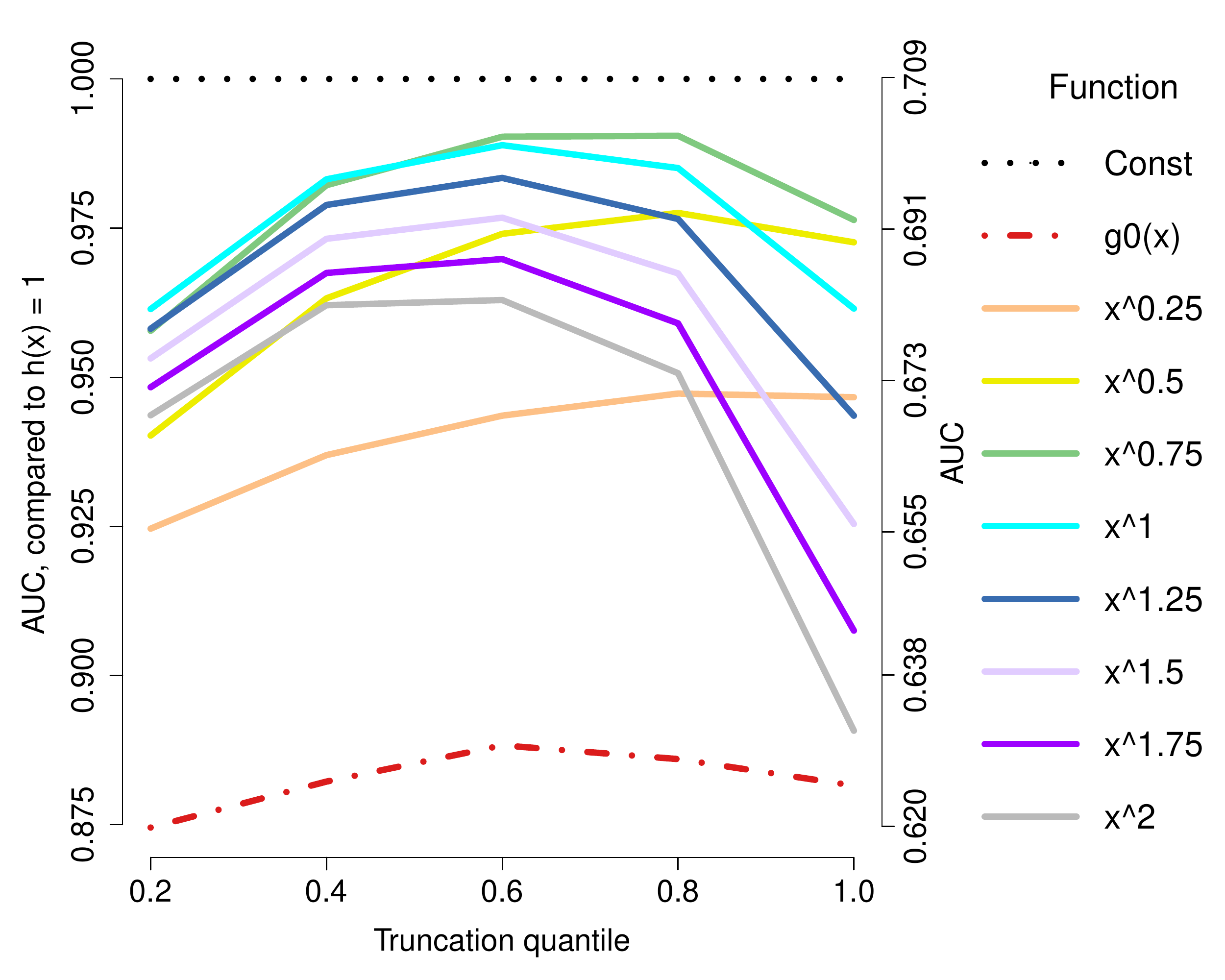}}\hspace{-0.02in}}
\\ \vspace{-0.1in}
\subfloat[$n=80$, $\ell_2^{\complement}$-nn domain]
{\includegraphics[width=0.45\textwidth]{{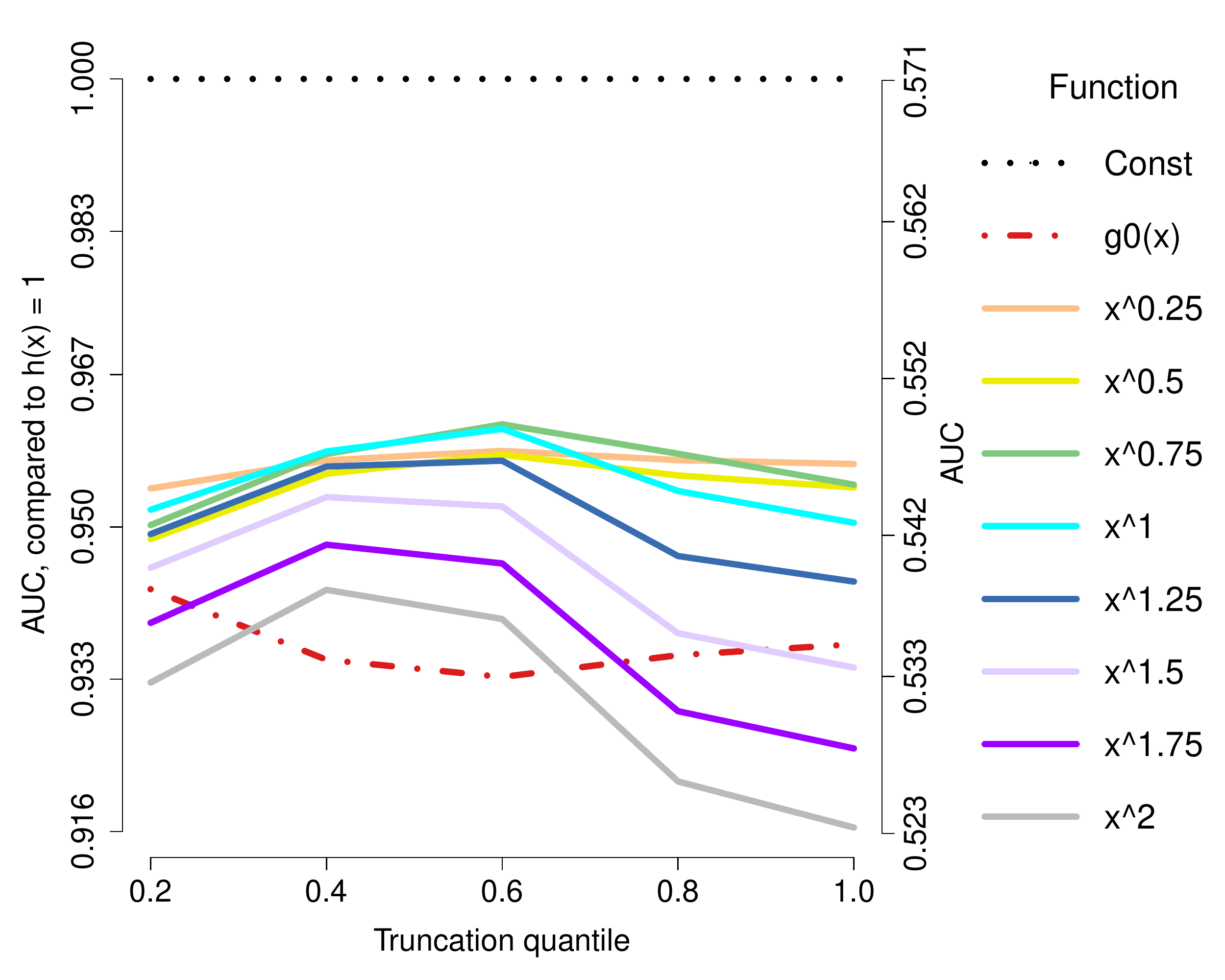}}\hspace{0.1in}}
\subfloat[$n=1000$, $\ell_2^{\complement}$-nn domain]
{\includegraphics[width=0.45\textwidth]{{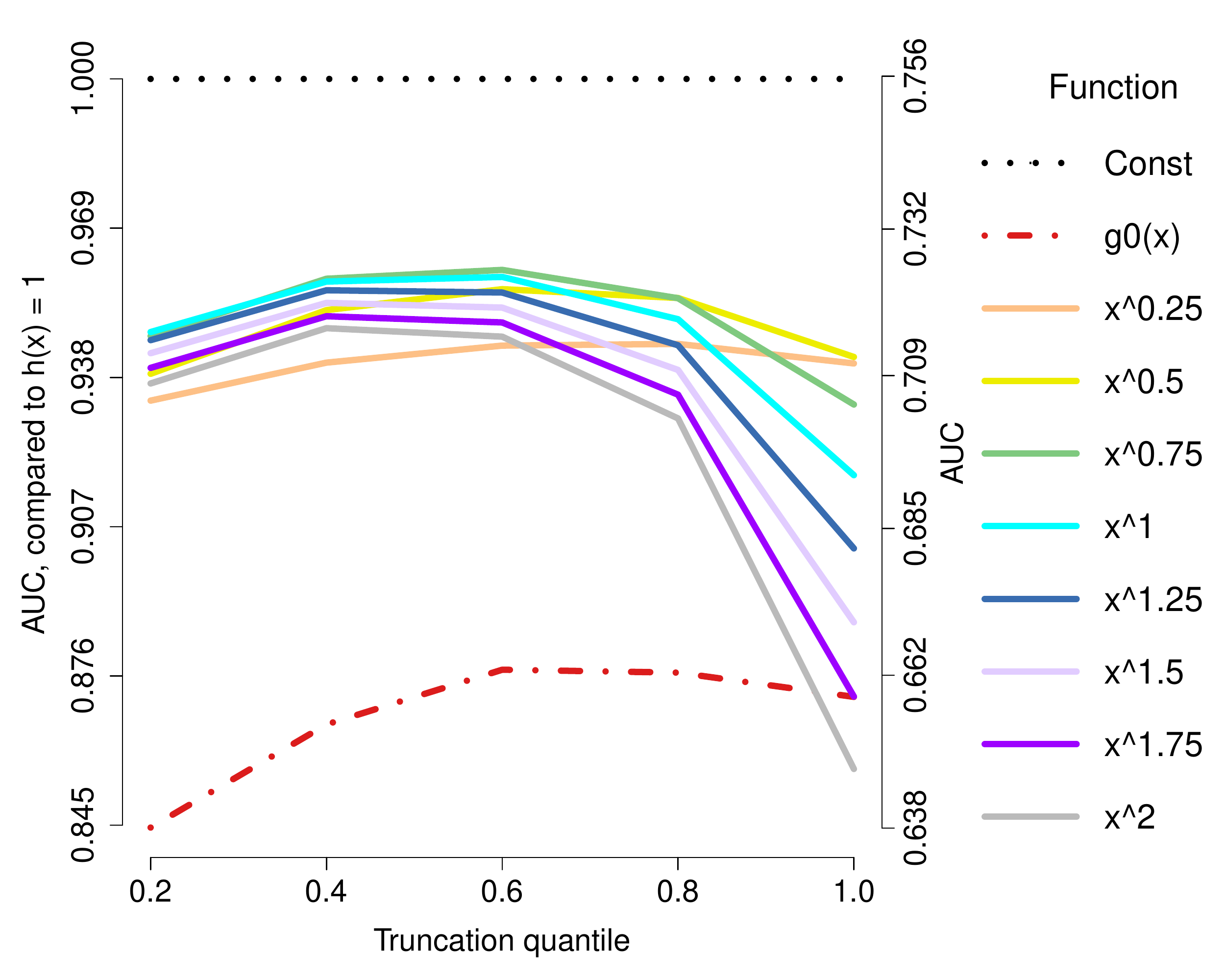}}\hspace{-0.02in}}
\\ \vspace{-0.1in}
\subfloat[$n=80$, unif-nn domain]
{\includegraphics[width=0.45\textwidth]{{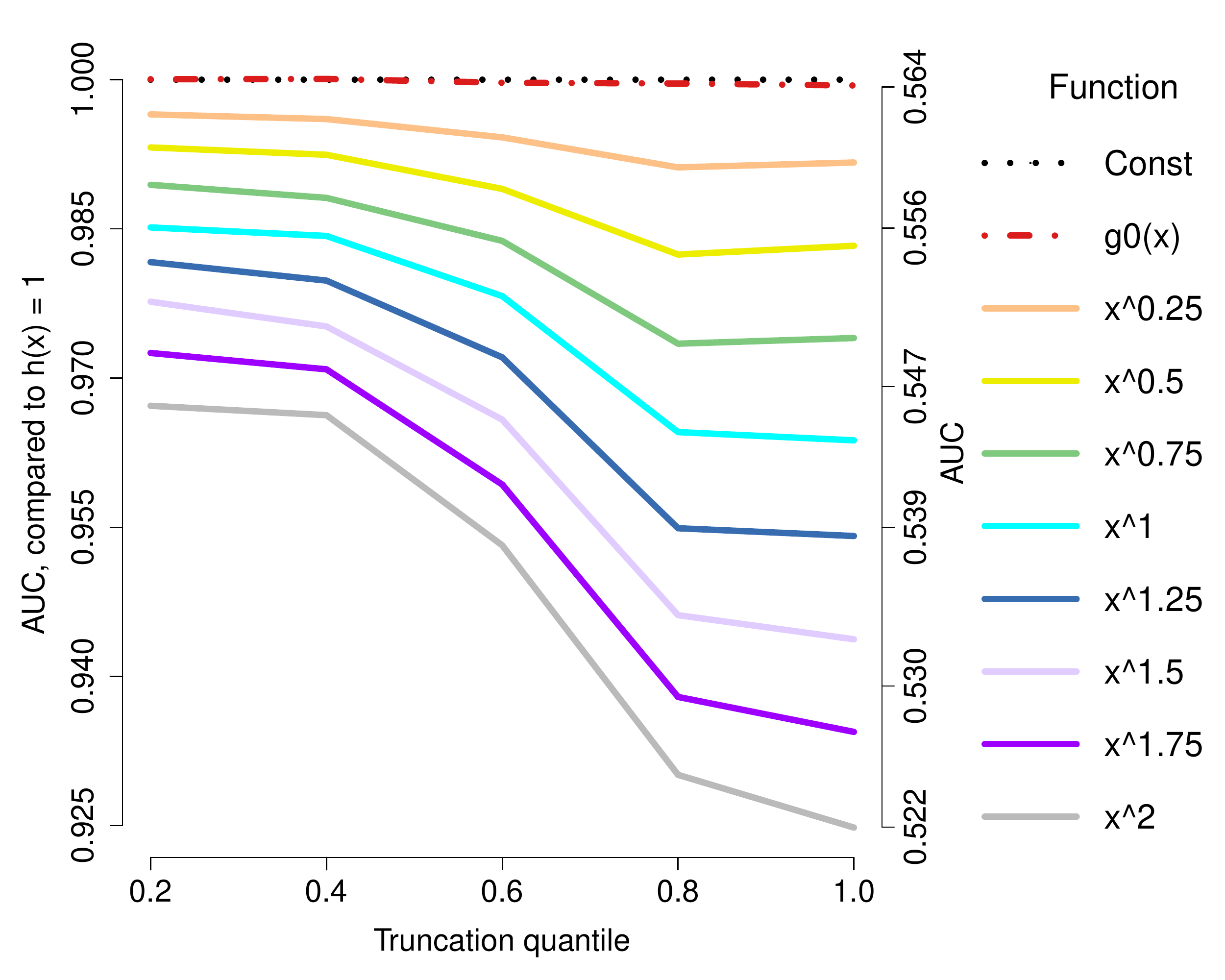}}\hspace{0.1in}}
\subfloat[$n=1000$, unif-nn domain]
{\includegraphics[width=0.45\textwidth]{{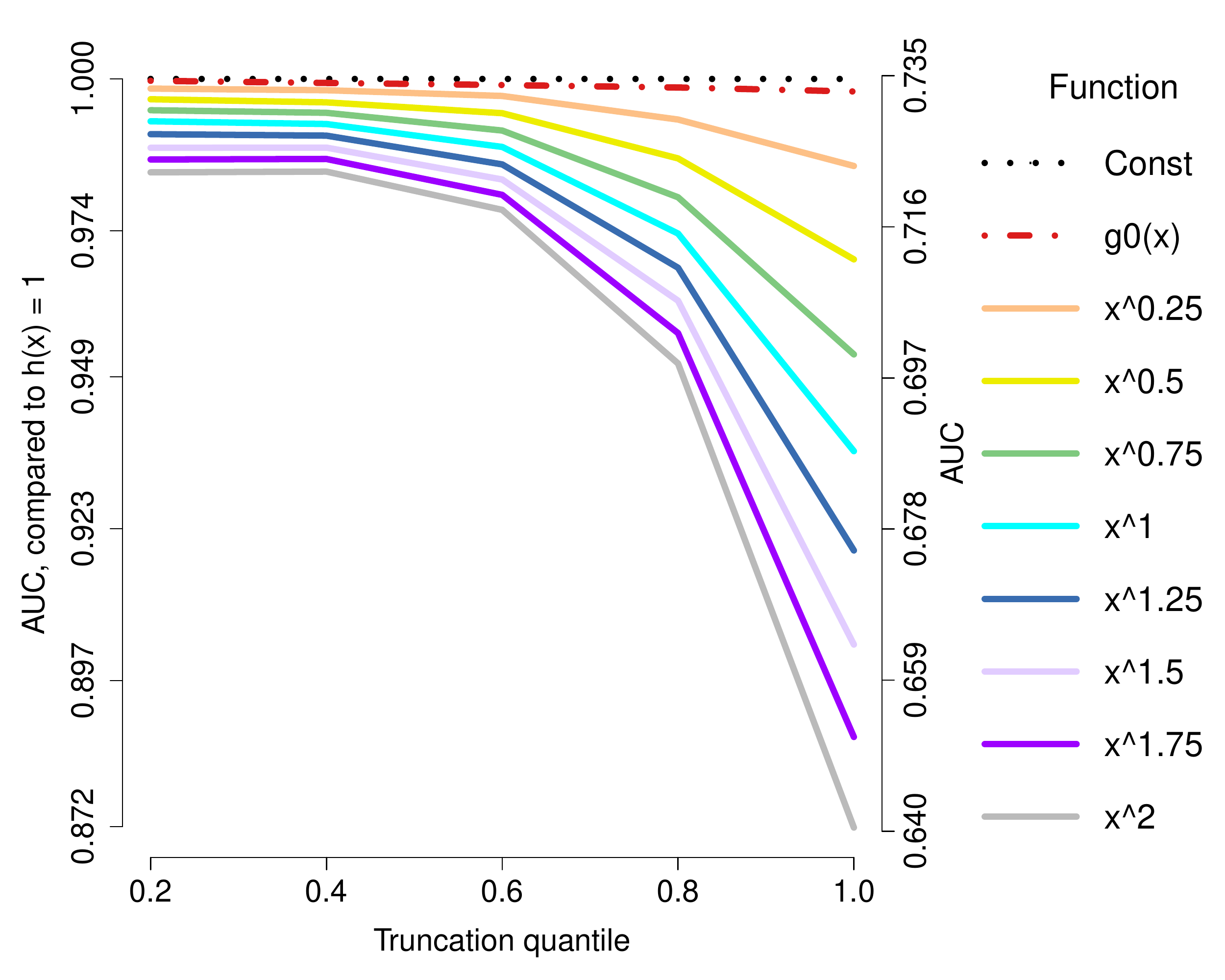}}\hspace{-0.02in}}
\caption{AUCs averaged over 50 trials for support recovery using generalized score matching for the $a=3$ models. Each curve represents either our extension to $g_0(\boldsymbol{x})$ from \citet{liu19} or a choice of power function $h(x)=x^c$.  The $x$ axes mark the probabilities $\pi$ that determine the truncation points $\boldsymbol{C}$ for the truncated component-wise distances. The colors are sorted by the power $c$.}\label{plot_a_3}
\end{figure}

\begin{figure}[t!]
\centering
\subfloat[Beta values]{\includegraphics[width=0.25\textwidth]{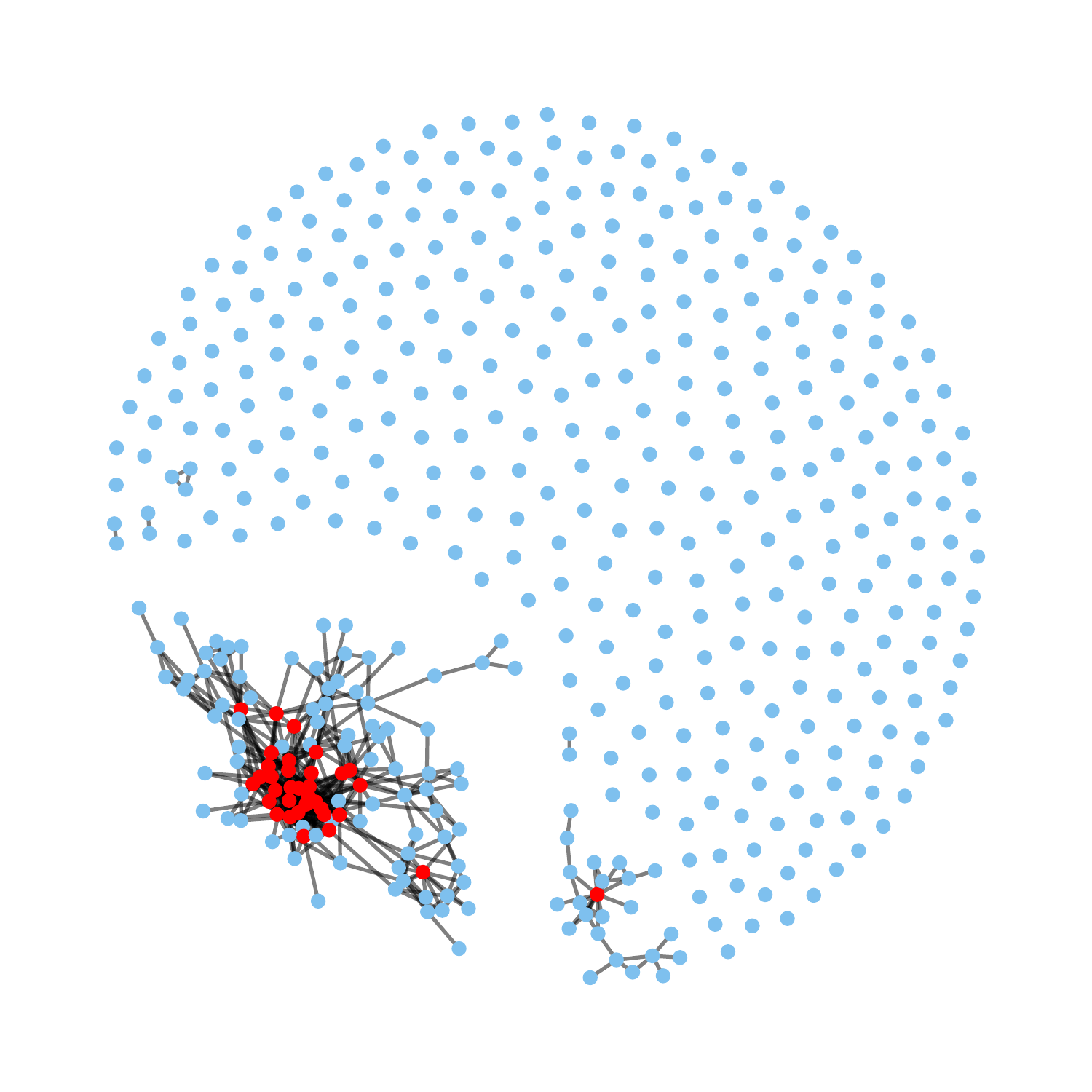}}
\subfloat[Common edges]{\includegraphics[width=0.25\textwidth]{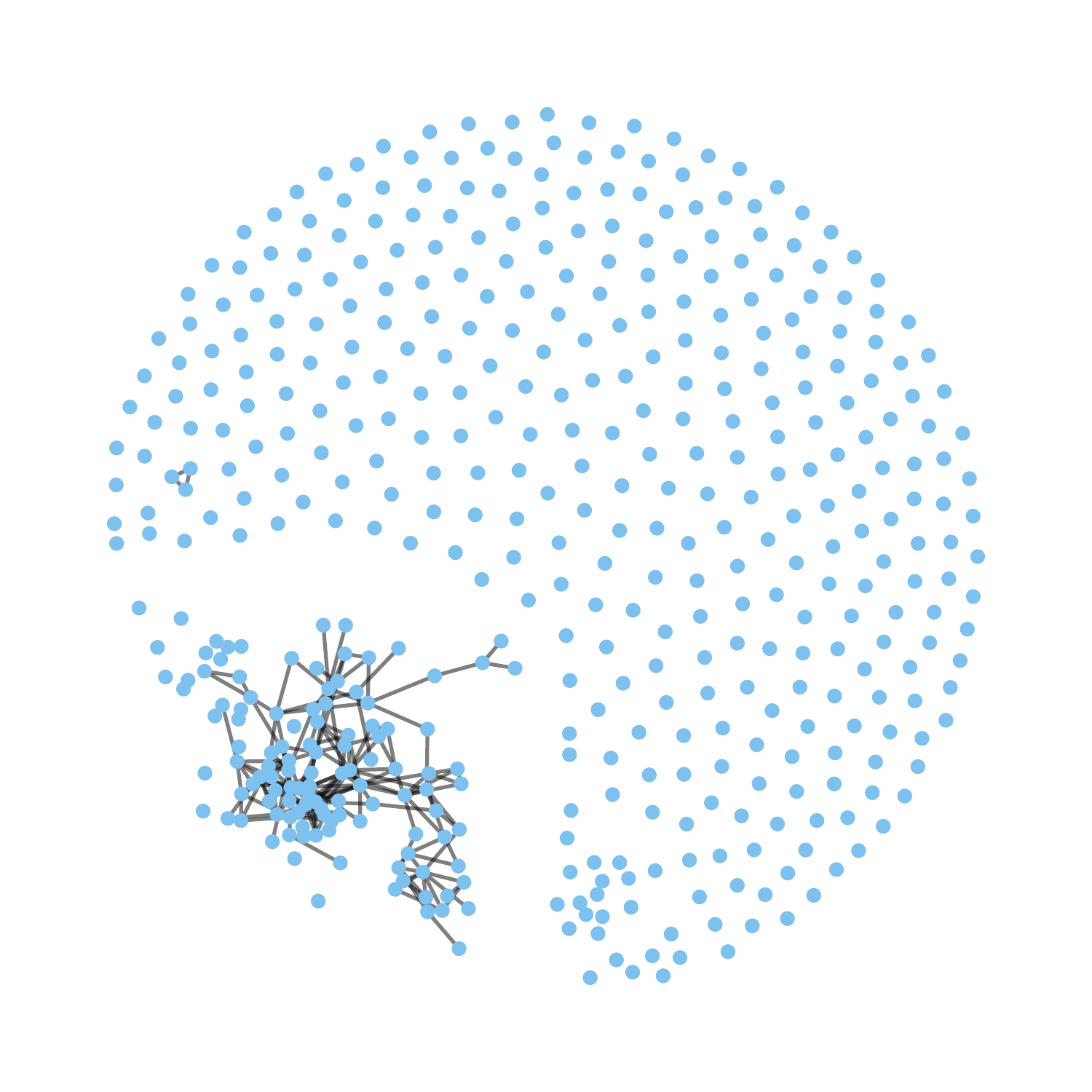}}
\subfloat[M values]{\includegraphics[width=0.25\textwidth]{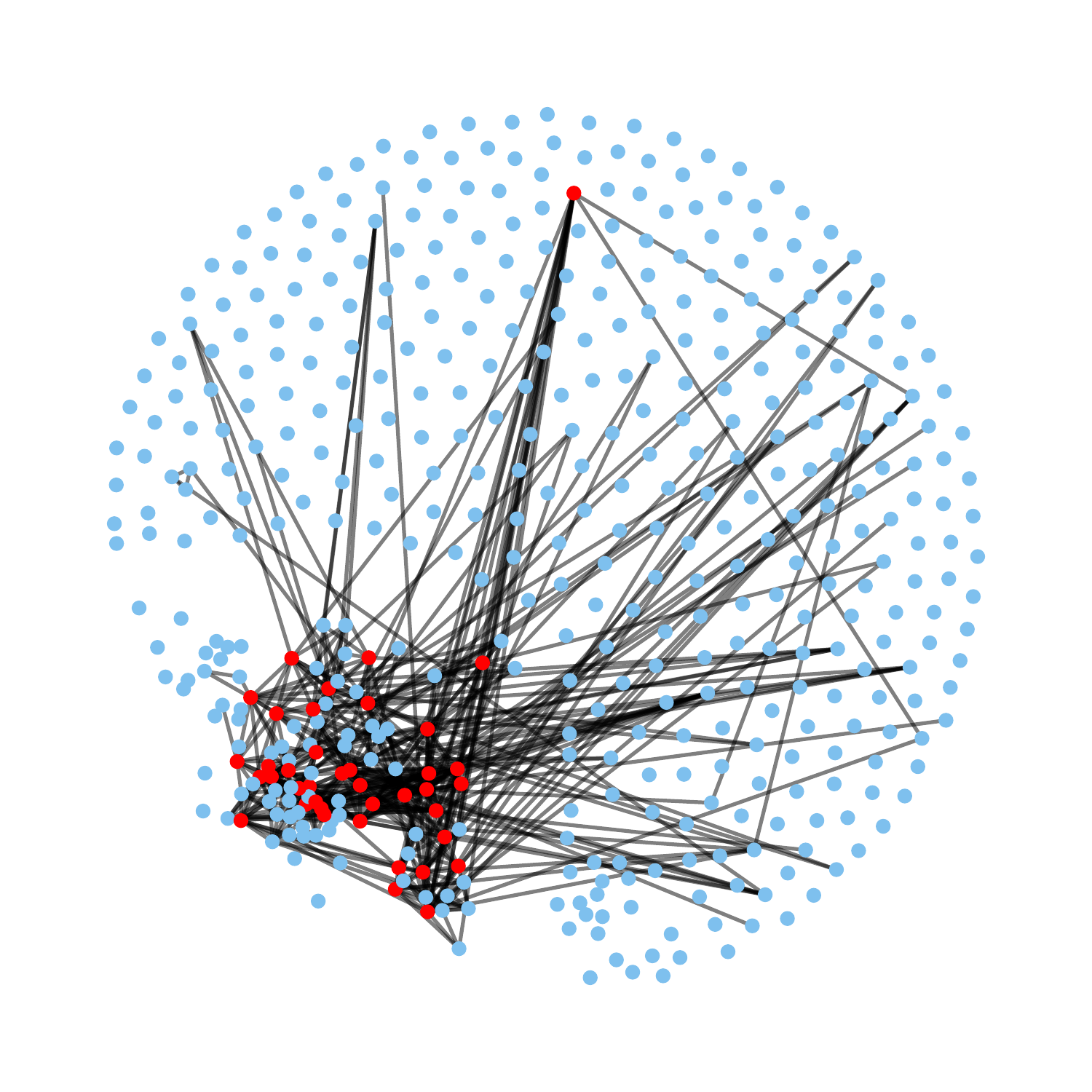}}
\caption{Graphs for CpG sites estimated by regularized generalized score matching estimator using Beta values (a) and M values (c), and their intersection graph (b). 
Isolated nodes are included and the layout is optimized for the graph for Beta values; red nodes have degree at least 10 (``hub nodes'').
}\label{data_sites_Beta_layout}
\centering
\subfloat[Beta values]{\includegraphics[width=0.25\textwidth]{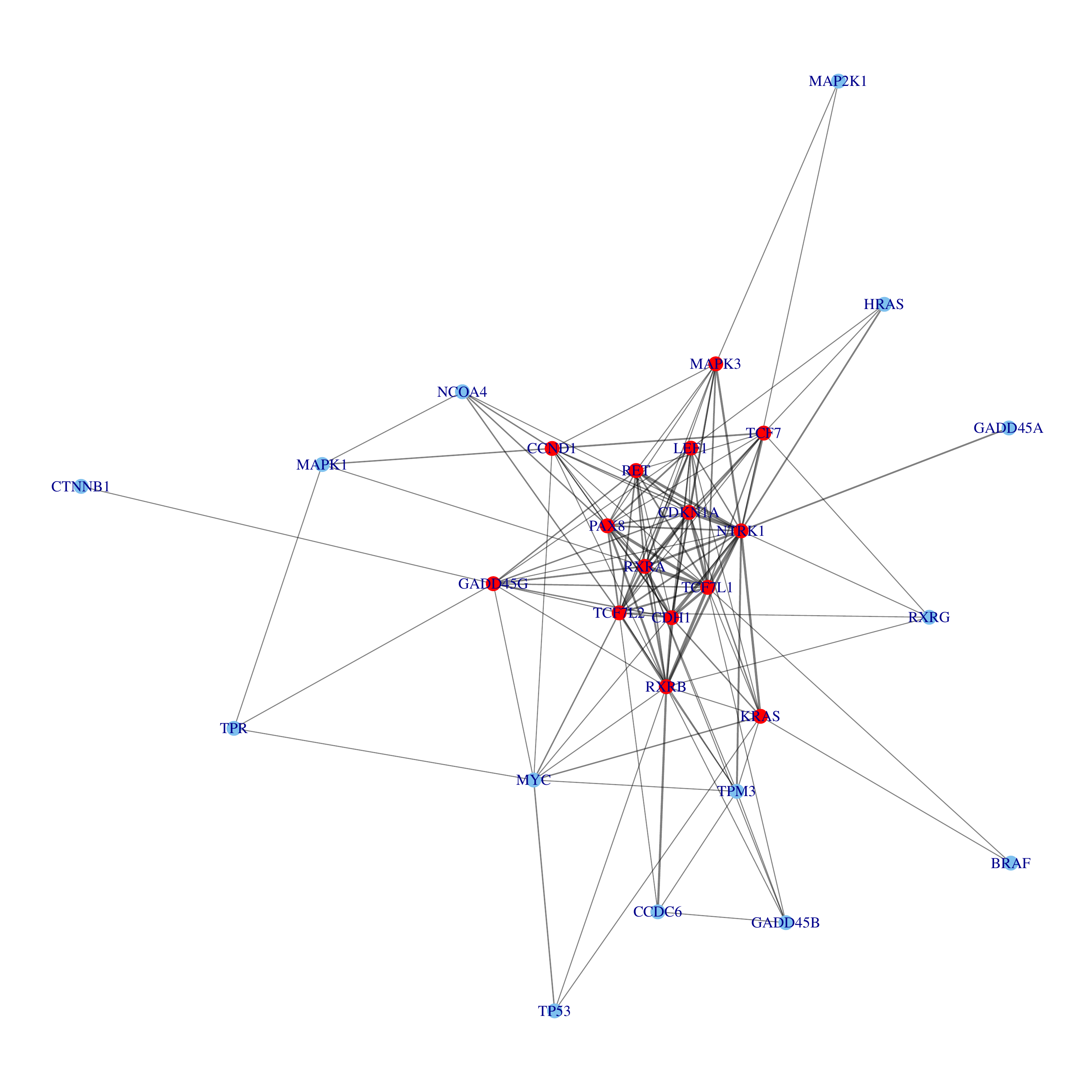}}
\subfloat[Common edges]{\includegraphics[width=0.25\textwidth]{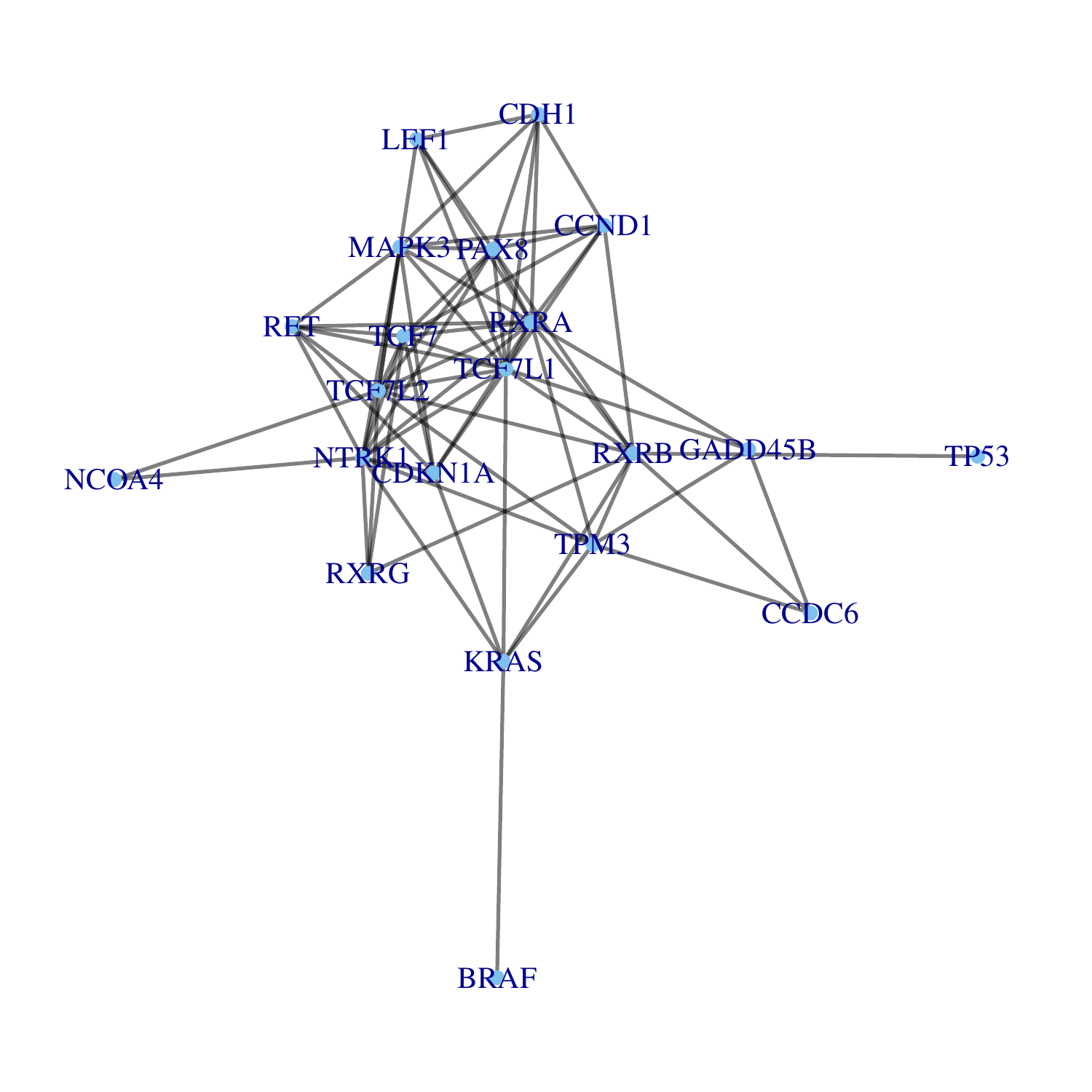}}
\subfloat[M values]{\includegraphics[width=0.25\textwidth]{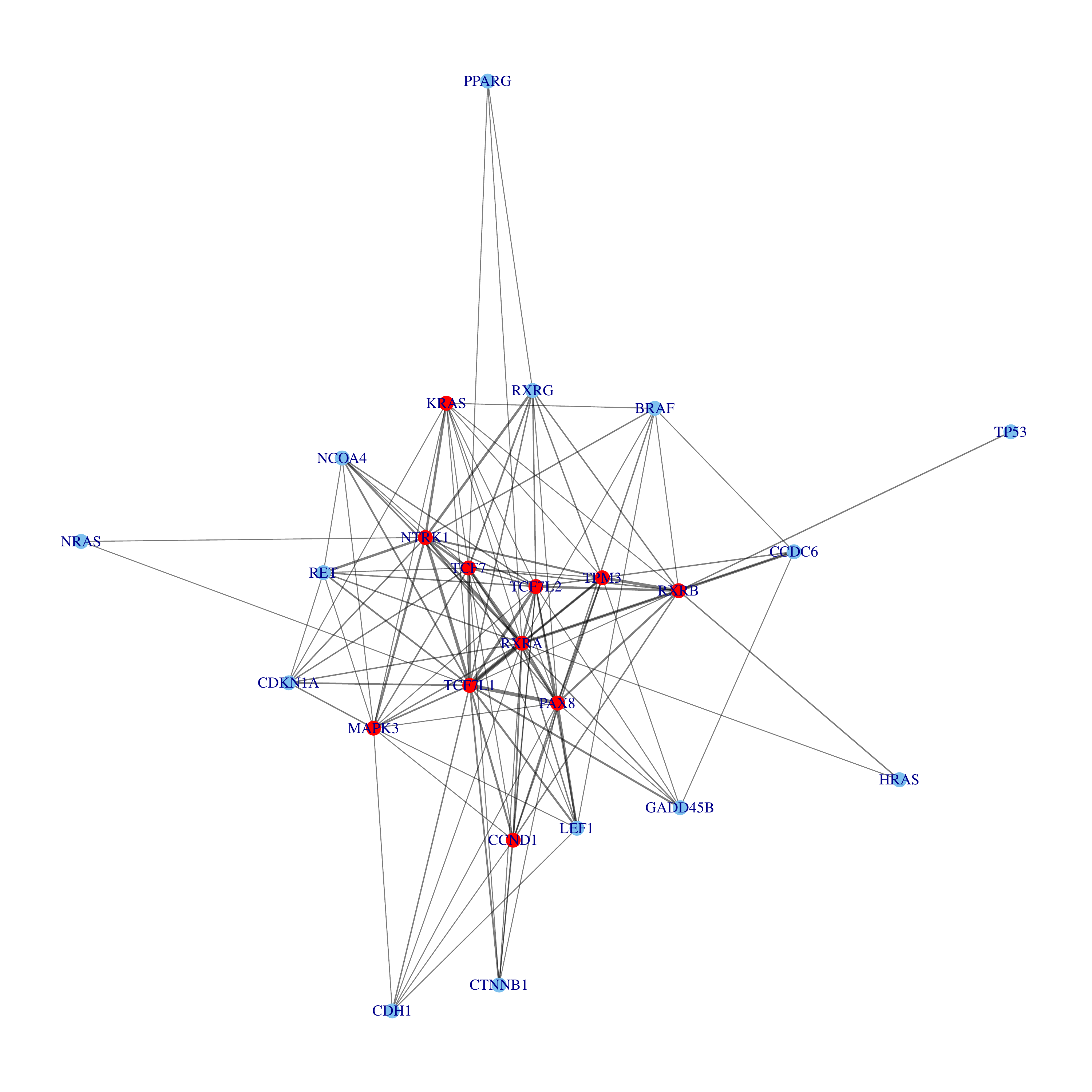}}\\
\subfloat[Beta values]{\includegraphics[width=0.25\textwidth]{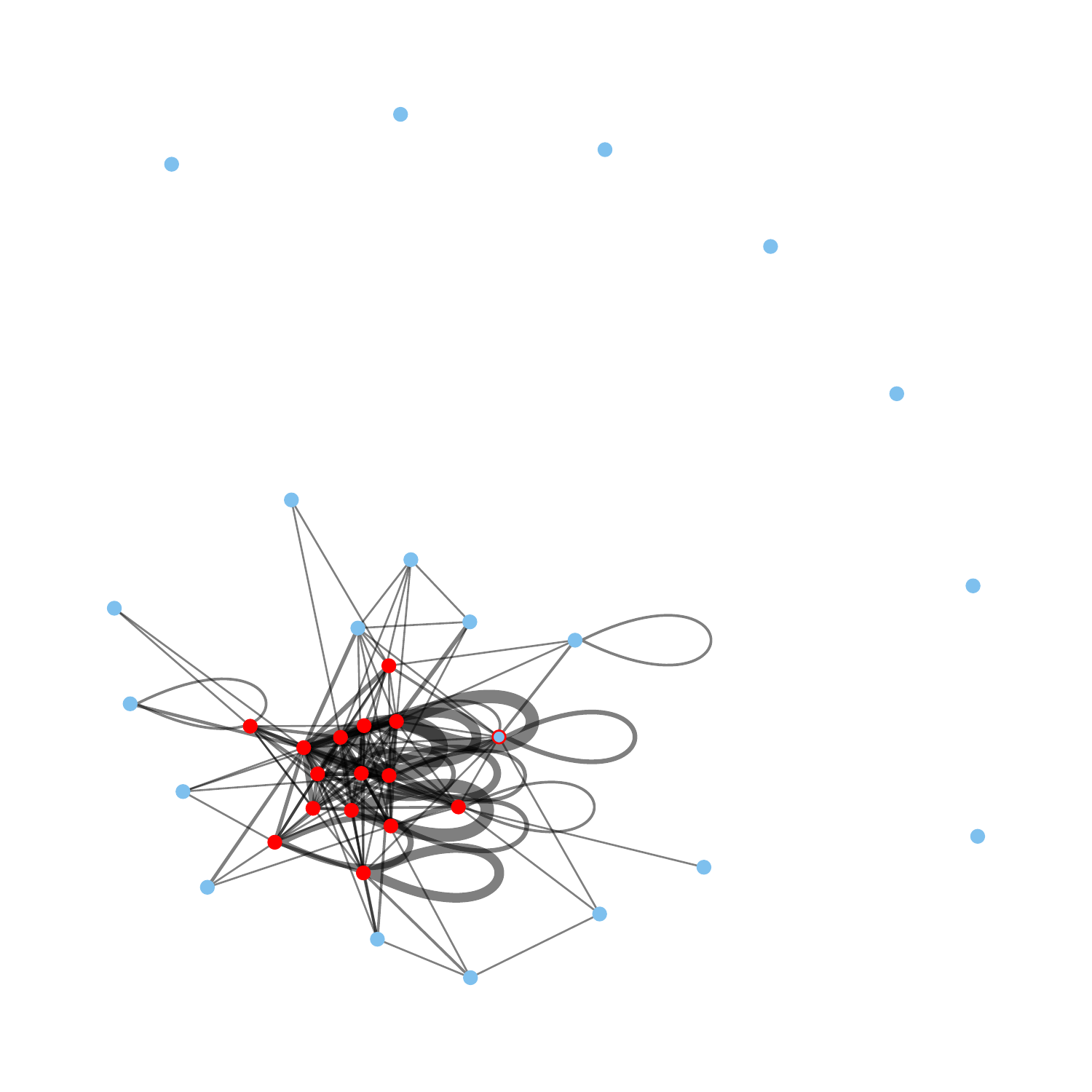}}
\subfloat[Common edges]{\includegraphics[width=0.25\textwidth]{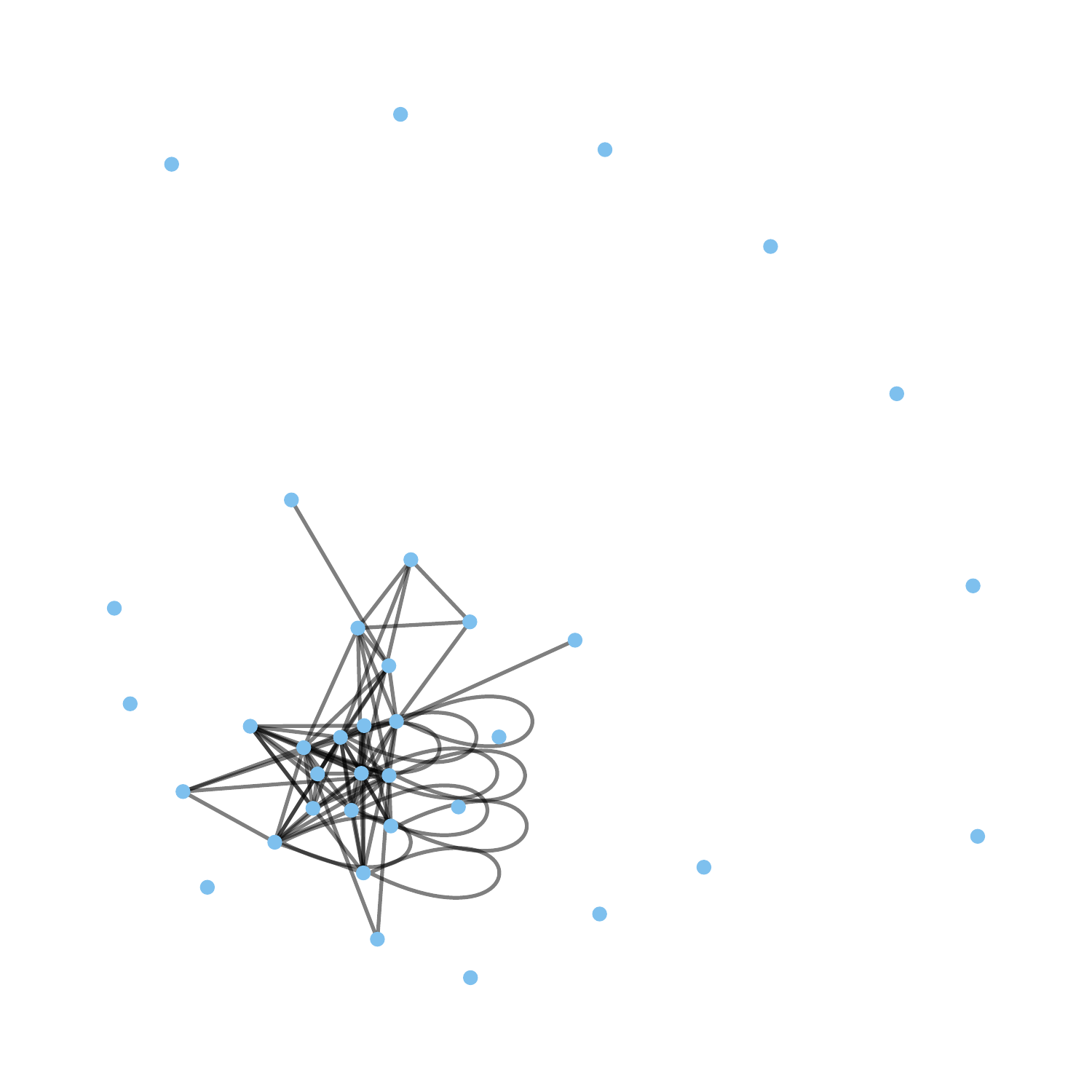}}
\subfloat[M values]{\includegraphics[width=0.25\textwidth]{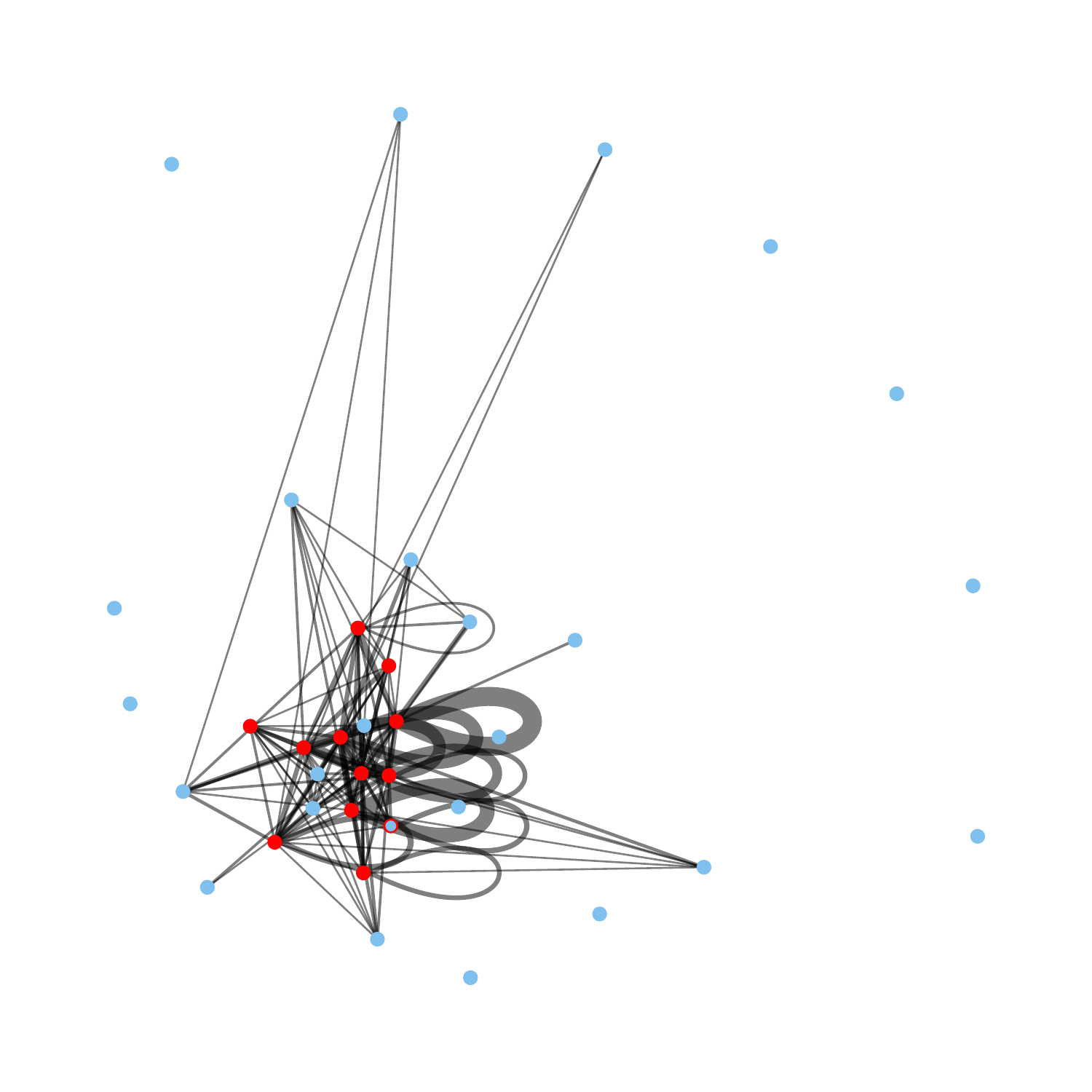}}
\caption{Graphs in Figure \ref{data_sites} (equivalently Figure \ref{data_sites_Beta_layout}) aggregated by the genes associated with the CpG sites, with Beta values (a, d) and M values (c, f), and their intersection graph (b, e). 
Red points indicate genes that are connected to at least 10 other genes in the case of Figure \ref{data_genes}
}\label{data_genes}
\end{figure}

\newpage

\section{Proofs}\label{Proofs}

Before proving Lemma \ref{lem_loss} we first prove the following lemma.
\begin{lemma}\label{lem_abs_con}
Suppose $h_1,\ldots,h_m$ are absolutely continuous in any bounded sub-interval of $\mathbb{R}_+$.  Then for any $j=1,\ldots,m$ and any $\boldsymbol{x}_{-j}\in \mathfrak{S}_{-j,\mathfrak{D}}$, $(h_j\circ\varphi_j)$ is absolutely continuous in $x_j$ in any bounded sub-interval of $\mathfrak{C}_{j,\mathfrak{D}}(\boldsymbol{x}_{-j})$.
\end{lemma}
\begin{proof}[Proof of Lemma \ref{lem_abs_con}]
In the proof we drop the dependency on $\boldsymbol{x}_{-j}$ in notation. By assumption, under Equation \ref{eq_C_union} any bounded sub-interval $[a,b]$ of $\mathfrak{C}_{j,\mathfrak{D}}(\boldsymbol{x}_{-j})$ must be a sub-interval of $[a_{k,j},b_{k,j}]$ for some $k$ (for simplicity we do not differentiate among $[a,b]$, $(a,b]$, $[a,b)$ and $(a,b)$ here).
\begin{enumerate}[(1)]
\item If $a_{k,j}>-\infty$ and $b_{k,j}<+\infty$, denote $C_0\equiv\min\{C_j,(b_{k,j}-a_{k,j})/2\}$ and rewrite
\begin{align*}
&\,(h_j\circ\varphi_j)(\boldsymbol{x})\\
=&\,h_j\left(\min(C_j,x_j-a_{k,j},b_{k,j}-x_j)\right)\\
=&\,h_j(x_j-a_{k,j})\mathds{1}_{x_j\in[a_{k,j},a_{k,j}+C_0]}+h_j(C_j)\mathds{1}_{x_j\in[a_{k,j}+C_0,b_{k,j}-C_0]}+h_j(b_{k,j}-x_j)\mathds{1}_{x_j\in[b_{k,j}-C_0,b_{k,j}]}.
\end{align*}
Then by absolute continuity of $h_j$ in $[a_{k,j},b_{k,j}]$ it is apparent that $(h_j\circ\varphi_j)$ is differentiable in $x_j$ a.e.~with partial derivative
\[h_j'(x_j-a_{k,j})\mathds{1}_{x_j\in[a_{k,j},a_{k,j}+C_0]}-h_j'(b_{k,j}-x_j)\mathds{1}_{x_j\in[b_{k,j}-C_0,b_{k,j}]}.\]
Then by the absolute continuity of $h_j$ again, for $x_j\in[a_{k,j},b_{k,j}]$,
\begin{align*}
&\,\int_{a_{k,j}}^{x_j}\partial_j (h_j\circ\varphi_j)(t_j;\boldsymbol{x}_{-j})\d t_j\\
=&\,h_j(x_j-a_{k,j})\mathds{1}_{x_j\in[a_{k,j},a_{k,j}+C_0]}+h_j(C_j)\mathds{1}_{x_j\in[a_{k,j}+C_0,b_{k,j}-C_0]}+h_j(b_{k,j}-x_j)\mathds{1}_{x_j\in[b_{k,j}-C_0,b_{k,j}]}\\
=&\,(h_j\circ\varphi_j)(\boldsymbol{x}),
\end{align*}
which proves that $(h_j\circ\varphi_j)(\boldsymbol{x})$ is absolutely continuous in $x_j$ in $[a_{k,j},b_{k,j}]$, and hence in $[a,b]\subset[a_{k,j},b_{k,j}]$.

\item If $a_{k,j}>-\infty$ and $b_{k,j}=+\infty$, on $[a,b]$ $(h_j\circ\varphi_j)(\boldsymbol{x})=h_j(\min(C_j,x_j-a_{k,j}))$ is an absolutely continuous function in a linear function of $x_j$ truncated above by $C_j$, and is thus trivially absolutely continuous in $[a,b]$. 

\item If $a_{k,j}=-\infty$ and $b_{k,j}<+\infty$, on $[a,b]$ $(h_j\circ\varphi_j)(\boldsymbol{x})=h_j(\min(C_j,b_{k,j}-x_j))$ is an absolutely continuous function in a linear function of $x_j$ truncated above by $C_j$, and is thus trivially absolutely continuous in $[a,b]$. 

\item If $a_{k,j}=-\infty$ and $b_{k,j}=+\infty$, $(h_j\circ\varphi_j)(\boldsymbol{x})=h_j(C_j)$ is constant and hence trivially absolutely continuous in $[a,b]$.
\end{enumerate}
\end{proof}

\begin{proof}[Proof of Lemma \ref{lem_loss}]
By simple manipulation
\begin{align}
J_{\boldsymbol{h},\boldsymbol{C},\mathfrak{D}}(p)&\equiv\frac{1}{2}\int_{\mathfrak{D}}p_0(\boldsymbol{x})\left\|\nabla\log p(\boldsymbol{x})\odot \left(\boldsymbol{h}\circ\boldsymbol{\varphi}\right)^{1/2}(\boldsymbol{x})-\nabla\log p_0(\boldsymbol{x})\odot \left(\boldsymbol{h}\circ\boldsymbol{\varphi}\right)^{1/2}(\boldsymbol{x})\right\|_2^2\d\boldsymbol{x}\\
&=\frac{1}{2}\sum_{j=1}^m\int_{\mathfrak{D}} p_0(\boldsymbol{x})(h_j\circ\varphi_j)(\boldsymbol{x})\left(\partial_j \log p_0(\boldsymbol{x})-\partial_j \log p(\boldsymbol{x})\right)^2\d\boldsymbol{x}\nonumber\\
&=\frac{1}{2}\sum_{j=1}^m\int_{\mathfrak{D}} p_0(\boldsymbol{x})(h_j\circ\varphi_j)(\boldsymbol{x})\left(\partial_j \log p(\boldsymbol{x})\right)^2\d\boldsymbol{x}\nonumber\\
&\pushright{-\sum_{j=1}^m\int_{\mathfrak{D}} p_0(\boldsymbol{x})(h_j\circ\varphi_j)(\boldsymbol{x})\partial_j \log p_0(\boldsymbol{x})\partial_j\log p(\boldsymbol{x})\d\boldsymbol{x}+\mathrm{const}.}\label{eq_pop_loss}
\end{align}

By (\ref{eq_pop_loss}) it suffices to prove for all $j=1,\dots,m$ that
\begin{multline}\label{eq_proof_adjoint}
\int_{\mathfrak{D}} p_0(\boldsymbol{x})(h_j\circ\varphi_j)(\boldsymbol{x})\partial_j \log p_0(\boldsymbol{x})\partial_j\log p(\boldsymbol{x})\d\boldsymbol{x}
=-\int_{\mathfrak{D}} p_0(\boldsymbol{x})\partial_j\left[(h_j\circ\varphi_j)(\boldsymbol{x})\partial_j\log p(\boldsymbol{x})\right]\d\boldsymbol{x}.
\end{multline}
Since $\int_{\mathfrak{D}} p_0(\boldsymbol{x})\left\|\nabla \log p(\boldsymbol{x})\odot(\boldsymbol{h}\circ\boldsymbol{\varphi})^{1/2}(\boldsymbol{x})\right\|_2^2\d\boldsymbol{x}$ and $\int_{\mathfrak{D}} p_0(\boldsymbol{x})\left\|\nabla \log p_0(\boldsymbol{x})\odot(\boldsymbol{h}\circ\boldsymbol{\varphi})^{1/2}(\boldsymbol{x})\right\|_2^2\d\boldsymbol{x}$ are both finite under assumption, by $|2ab|\leq a^2+b^2$ the integrand in the left-hand side of (\ref{eq_proof_adjoint}) is integrable. Then by Fubini-Tonelli
\begin{align}
&\,\int_{\mathfrak{D}} p_0(\boldsymbol{x})(h_j\circ\varphi_j)(\boldsymbol{x})\partial_j\log p_0(\boldsymbol{x})\partial_j\log p(\boldsymbol{x})\d \boldsymbol{x}\nonumber\\
=&\,\int_{\mathfrak{S}_{-j}}\int_{\mathfrak{C}_j(\boldsymbol{x}_{-j})}\underbrace{(h_j\circ\varphi_j)(\boldsymbol{x})\partial_j p_0(\boldsymbol{x})\partial_j\log p(\boldsymbol{x})}_{\equiv f(\boldsymbol{x})}\d x_j\d \boldsymbol{x}_{-j}\nonumber\\
=&\,\int_{\mathfrak{S}_{-j}}\int_{\mathbb{R}}\mathds{1}_{\mathfrak{C}_j(\boldsymbol{x}_{-j})}(x_j)f(\boldsymbol{x})\d x_j\d\boldsymbol{x}_{-j}\nonumber\\
=&\,\int_{\mathfrak{S}_{-j}}\int_{\mathbb{R}}\left[\sum_{k=1}^{K_j(\boldsymbol{x}_{-j})}\mathds{1}_{[a_{k,j}(\boldsymbol{x}_{-j}),\,b_{k,j}(\boldsymbol{x}_{-j})]}(x_j)\right]f(x_j;\boldsymbol{x}_{-j})\d x_j\d\boldsymbol{x}_{-j}\nonumber\\
=&\,\int_{\mathfrak{S}_{-j}}\left[\sum_{k=1}^{K_j(\boldsymbol{x}_{-j})}\int_{a_{k,j}(\boldsymbol{x}_{-j})}^{b_{k,j}(\boldsymbol{x}_{-j})}f(x_j;\boldsymbol{x}_{-j})\d x_j\right]\d\boldsymbol{x}_{-j}\label{eq_proof_adjoint_segments}
\end{align}
where the interchangeability of integration and (potentially infinite) summation is justified by Fubini-Tonelli again. Then using the decomposition of the domain in (\ref{eq_C_union}) while omitting the dependency of $a_{k,j}$ and $b_{k,j}$ on $\boldsymbol{x}_{-j}$ in notation, for a.e.~$\boldsymbol{x}_{-j}\in \mathfrak{S}_{-j}$ and any $k=1,\ldots, K_j(\boldsymbol{x}_{-j})$ we have
\begin{align*}
&\,\int_{a_{k,j}}^{b_{k,j}}f(\boldsymbol{x})\d x_j\\
=&\,\int_{a_{k,j}}^{b_{k,j}}(h_j\circ\varphi_j)(\boldsymbol{x})\partial_j p_0(\boldsymbol{x})\partial_j \log p(\boldsymbol{x})\d x_j\\
=&\,\lim_{x_j\nearrow b_{k_j}^-}(h_j\circ\varphi_j)(\boldsymbol{x}) p_0(\boldsymbol{x})\partial_j \log p(\boldsymbol{x})-\lim_{x_j\searrow a_{k_j}^+}(h_j\circ\varphi_j)(\boldsymbol{x}) p_0(\boldsymbol{x})\partial_j \log p(\boldsymbol{x})\\
&\pushright{-\int_{a_{k,j}}^{b_{k,j}} p_0(\boldsymbol{x})\partial_j\left[(h_j\circ\varphi_j)(\boldsymbol{x})\partial_j\log p(\boldsymbol{x})\right]\d x_j}\\
=&\,-\int_{a_{k,j}}^{b_{k,j}} p_0(\boldsymbol{x})\partial_j\left[(h_j\circ\varphi_j)(\boldsymbol{x})\partial_j \log p(\boldsymbol{x})\right]\d x_j,
\end{align*}
by integration by parts and by Assumption (A1) on the limits going to 0. The integration by parts is justified by the fundamental theorem of calculus for absolutely continuous functions (Lemma \ref{lem_abs_con}) as well as the product rule (cf.~proof of Lemma 19 in \citet{yus19}). Thus, by going backwards using Fubini-Tonelli twice again, (\ref{eq_proof_adjoint_segments}) becomes
\begin{align*}
&\,\int_{\mathfrak{S}_{-j}}\left\{-\sum_{k=1}^{K_j(\boldsymbol{x}_{-j})}\int_{a_{k,j}(\boldsymbol{x}_{-j})}^{b_{k,j}(\boldsymbol{x}_{-j})}p_0(\boldsymbol{x})\partial_j\left[(h_j\circ\varphi_j)(\boldsymbol{x})\partial_j \log p(\boldsymbol{x})\right]\d x_j\right\}\d\boldsymbol{x}_{-j}\\
=&\,-\int_{\mathfrak{S}_{-j}}\int_{\mathfrak{C}_j(\boldsymbol{x}_{-j})}p_0(\boldsymbol{x})\partial_j\left[(h_j\circ\varphi_j)(\boldsymbol{x})\partial_j\log p(\boldsymbol{x})\right]\d x_j\d\boldsymbol{x}_{-j}\\
=&\,-\int_{\mathfrak{D}}p_0(\boldsymbol{x})\partial_j\left[(h_j\circ\varphi_j)(\boldsymbol{x})\partial_j\log p(\boldsymbol{x})\right]\d\boldsymbol{x},
\end{align*}
proving (\ref{eq_proof_adjoint}).
\end{proof}

\begin{proof}[Proof of Theorem \ref{thm_norm_const}]
Note that the condition ${\boldsymbol{v}^a}^{\top}\mathbf{K}\boldsymbol{v}^a>0$ $\forall\boldsymbol{v}\in\mathfrak{D}\backslash\{\boldsymbol{0}\}$ implies that ${\boldsymbol{v}^a}^{\top}\mathbf{K}\boldsymbol{v}^a>0$ $\forall\boldsymbol{v}\in\mathfrak{D}_+\equiv \{\boldsymbol{v}/\|\boldsymbol{v}\|_2:\boldsymbol{v}\in\mathfrak{D}\backslash\{\boldsymbol{0}\}\}\subseteq\{\boldsymbol{v}\in\mathbb{R}^m:\|\boldsymbol{v}\|_2=1\}\equiv\mathbb{S}^{m-1}$ with $\mathbb{S}^{m-1}$ compact, so 
\begin{align*}
N_{\mathbf{K}}&\equiv\inf_{\boldsymbol{v}\in\mathfrak{D}\backslash\{\boldsymbol{0}\}}{\boldsymbol{v}^a}^{\top}\mathbf{K}\boldsymbol{v}^a/{\boldsymbol{v}^a}^{\top}\boldsymbol{v}^a=\inf_{\boldsymbol{v}\in\mathfrak{D}_+}{\boldsymbol{v}^a}^{\top}\mathbf{K}\boldsymbol{v}^a/{\boldsymbol{v}^a}^{\top}\boldsymbol{v}^a\\
&\geq\inf_{\boldsymbol{v}\in\mathbb{S}^{m-1}}{\boldsymbol{v}^a}^{\top}\mathbf{K}\boldsymbol{v}^a/{\boldsymbol{v}^a}^{\top}\boldsymbol{v}^a>0.
\end{align*}

(1) \emph{Case $a>0$ and $b>0$ (CC1, CC2):}
Since $p$ is bounded everywhere, it is integrable over a bounded $\mathfrak{D}$ (proving (CC1)). Otherwise, assume $\mathfrak{D}$ is unbounded. If either $a$ or $b$ is non-integer, then $\mathfrak{D}\subset\mathbb{R}_+^m$ and a sufficient condition is $\boldsymbol{v}^a\mathbf{K}\boldsymbol{v}^a>0$ $\forall\boldsymbol{v}\in\mathfrak{D}\backslash\{\boldsymbol{0}\}$, and either $\boldsymbol{\eta}^{\top}\boldsymbol{v}^b\leq 0$ $\forall\boldsymbol{v}\in\mathfrak{D}$ or $2a>b>0$, corresponding to (i) and (ii) in the Proof of Theorem 9 in Section A.3 of \citet{yus19}, respectively. If $a$ and $b$ are both integers, $\mathfrak{D}\subset\mathbb{R}^m$ and the same sufficient condition can be implied following the same proof in \citet{yus19}, with integration over $(-\infty,+\infty)$ instead of $(0,+\infty)$. This proves (CC2).

(2) \emph{Case $a>0$ and $b=0$ (CC3):} By definition $\mathfrak{D}\subseteq\mathbb{R}_+^m$. If $\mathfrak{D}$ is bounded, $-\frac{1}{2a}{\boldsymbol{x}^a}^{\top}\mathbf{K}\boldsymbol{x}^a$ as a continuous function is bounded, and so it suffices to bound $\int_{\mathfrak{D}}\exp\left(\boldsymbol{\eta}^{\top}\log(\boldsymbol{x})\right)\d\boldsymbol{x}=\int_{\mathfrak{D}}\prod_{j=1}^mx_j^{\eta_j}\d \boldsymbol{x}\leq\prod_{j=1}^m\int_{\rho_j(\mathfrak{D})}x_j^{\eta_j}\d x_j<+\infty$ if $\eta_j>-1$ for all $j$ such that $0\in\rho_j(\mathfrak{D})$, where for the $\leq$ step we used the fact that $x_j>0$. This proves (CC3) (i).

If $\mathfrak{D}$ is unbounded and ${\boldsymbol{v}^a}^{\top}\mathbf{K}\boldsymbol{v}^a>0$ for all $\boldsymbol{v}\in\mathfrak{D}\backslash\{\boldsymbol{0}\}$, using the fact that $\exp(\cdots)>0$,
\begin{align*}
\int_{\mathfrak{D}}p_{\boldsymbol{\eta},\mathbf{K}}(\boldsymbol{x})\d\boldsymbol{x}&=\int_{\mathfrak{D}}\exp\left(-{\boldsymbol{x}^a}^{\top}\mathbf{K}\boldsymbol{x}^a/(2a)+\boldsymbol{\eta}^{\top}\log(\boldsymbol{x})\right)\d\boldsymbol{x}\\
&\leq\prod_{j=1}^m\int_{\rho_j(\mathfrak{D})}\exp\left(-N_{\mathbf{K}}x_j^{2a}/(2a)+\eta_j\log(x_j)\right)\d x_j.
\end{align*}
Note that the indefinite integral of the last display is 
\[-\frac{1}{2a}x^{1 + \eta_j} \left(\frac{N_{\mathbf{K}}}{2a} x^{2 a}\right)^{-(1 + \eta_j)/(2 a)}\Gamma\left[\frac{1 + \eta_j}{2 a}, \frac{N_{\mathbf{K}} x^{2 a}}{2a}\right]\]
so the definite integral is finite if and only if $\eta_j>-1$ for all $j$ s.t.~$0\in\rho_j(\mathfrak{D})$. This proves (CC3) (ii).

If $\mathfrak{D}$ is unbounded and ${\boldsymbol{v}^a}^{\top}\mathbf{K}\boldsymbol{v}^a\geq 0$ for all $\boldsymbol{v}\in\mathfrak{D}$, then $\int_{\mathfrak{D}}p_{\boldsymbol{\eta},\mathbf{K}}(\boldsymbol{x})\d\boldsymbol{x}\leq\prod_{j=1}^m\int_{\rho_j(\mathfrak{D})}x_j^{\eta_j}\d x_j<\infty$ if $\eta_j>-1$ for all $j$ s.t.~$0\in\rho_j(\mathfrak{D})$ and $\eta_j<-1$ for all $j$ s.t.~$\rho_j(\mathfrak{D})$ is unbounded. This proves (CC3) (iii).

(3) \emph{Case $a=0$, $\mathfrak{D}$ is bounded and $0\not\in\rho_j(\mathfrak{D})$ for all $j$ (CC4):}
If $\mathfrak{D}$ is bounded and $0\not\in\rho_j(\mathfrak{D})$ for all $j$, then $\log(\mathfrak{D})$ is bounded, and since the integrand is continuous and bounded, the integral is finite without any further requirements.

(4) \emph{Case $a=0$ and $b=0$ (CC5):} Assume $\log(\boldsymbol{x})^{\top}\mathbf{K}\log(\boldsymbol{x})>0$ for all $\boldsymbol{x}\in\mathfrak{D}$, then
\begin{align*}
\int_{\mathfrak{D}}p_{\boldsymbol{\eta},\mathbf{K}}(\boldsymbol{x})\d\boldsymbol{x}=&\int_{\mathfrak{D}}\exp\left(-\frac{1}{2}{\log(\boldsymbol{x})}^{\top}\mathbf{K}\log(\boldsymbol{x})+\boldsymbol{\eta}^{\top}\log(\boldsymbol{x})\right)\d\boldsymbol{x}\\
=&\int_{\log(\mathfrak{D})}\exp\left(-\frac{1}{2}{\boldsymbol{x}}^{\top}\mathbf{K}\boldsymbol{x}+(\boldsymbol{\eta}+\mathbf{1}_m)^{\top}\boldsymbol{x}\right)\d\boldsymbol{x}\\
<&\prod_{j=1}^m\int_{\log(\rho_j(\mathfrak{D}))}\exp\left(-N_{\mathbf{K}}x_j^{2}/2+(\eta_j+1)x_j\right)\d x_j\\
<&\prod_{j=1}^m\int_{-\infty}^{\infty}\exp\left(-N_{\mathbf{K}}x_j^{2}/2+(\eta_j+1)x_j\right)\d x_j<+\infty
\end{align*}
since the integrand is proportional to a univariate Gaussian density.

(5) \emph{Case $a=0$ and $b>0$ (CC6, CC7):} Assume $\log(\boldsymbol{x})^{\top}\mathbf{K}\log(\boldsymbol{x})> 0$ for all $\boldsymbol{x}\in\mathfrak{D}$ and $\eta_j\leq 0$ for all $j$ s.t.~$\rho_j(\mathfrak{D})$ is unbounded (from above). Then
\begin{align*}
\int_{\mathfrak{D}}p_{\boldsymbol{\eta},\mathbf{K}}(\boldsymbol{x})\d\boldsymbol{x}=&\int_{\mathfrak{D}}\exp\left(-\frac{1}{2}{\log(\boldsymbol{x})}^{\top}\mathbf{K}\log(\boldsymbol{x})+\boldsymbol{\eta}^{\top}\boldsymbol{x}^b\right)\d\boldsymbol{x}\\
=&\int_{\log(\mathfrak{D})}\exp\left(-\frac{1}{2}\boldsymbol{x}^{\top}\mathbf{K}\boldsymbol{x}+\mathbf{1}_m^{\top}\boldsymbol{x}+\boldsymbol{\eta}^{\top}\exp(b\boldsymbol{x})\right)\d\boldsymbol{x}\\
<&\prod_{j=1}^m\int_{\log(\rho_j(\mathfrak{D}))}\exp\left(-N_{\mathbf{K}}x_j^2/2+x_j+\eta_j\exp(bx_j)\right)\d x_j\\
\leq&\prod_{j=1}^{m}\int_{-\infty}^{\infty}c_j\exp\left(-N_{\mathbf{K}}x_j^2/2+x_j\right)\d  x_j<+\infty,
\end{align*}
where $c_j\equiv 1$ if $\eta_j\leq 0$ or $c_j\equiv\exp\left(\eta_j\left(\sup\rho_j(\mathfrak{D})\right)^b\right)>+\infty$ otherwise. This proves (CC6).

Finally, if $\log(\mathfrak{D})$ is unbounded and $\log(\boldsymbol{x})^{\top}\mathbf{K}\log(\boldsymbol{x})\geq 0$ for all $\boldsymbol{x}\in\mathfrak{D}$, the integral is bounded by
\[\prod_{j=1}^m\int_{\log(\rho_j(\mathfrak{D}))}\exp\left(x_j+\eta_j\exp(bx_j)\right)\d x_j\]
which is finite if and only if $\eta_j<0$ for all $j$ s.t.~$\rho_j(\mathfrak{D})$ is unbounded (from above). This proves (CC7).

\end{proof}

\begin{proof}[Proof of Theorem \ref{thm_A}]
It suffices to consider the case $\mathfrak{D}=\mathbb{R}_+^m$ for general $a$ and $b$ as well as $\mathfrak{D}=\mathbb{R}^m$ for integer $a>0$ and $b>0$ (so that (\ref{eq_interaction_density2}) is well defined on $\mathbb{R}^m$): For (A.1), the irregularities only occur at the boundary points, but with the composition $(h_j\circ\varphi_j)(\boldsymbol{x})$ with $x_j$ approaching any  finite boundary point behaves like $h_j(x_j)$ with $x_j\searrow 0^+$ in $\mathfrak{D}=\mathbb{R}_+^m$, and $(h_j\circ\varphi_j)(\boldsymbol{x})$ with $x_j\to\infty$ behaves like $h_j(x_j)$ with $x_j\to\infty$ in $\mathfrak{D}=\mathbb{R}_+^m$ (or $\mathbb{R}^m$ if applicable). For (A.2), obviously integrability over $\mathfrak{D}$ follows from that over $\mathfrak{D}=\mathbb{R}_+^m$ or $\mathbb{R}^m$. (A.3) is trivially satisfied by a power function $h_j$.

As in the proof of Theorem \ref{thm_norm_const}, $N_{\mathbf{K}}\equiv\inf_{\boldsymbol{v}\in\mathfrak{D}}{\boldsymbol{v}^a}^{\top}\mathbf{K}\boldsymbol{v}^a/{\boldsymbol{v}^a}^{\top}\boldsymbol{v}^a> 0$.

(1) The case for $a>0$ and $b\geq 0$ and $\mathfrak{D}=\mathbb{R}_+^m$ is covered in \citet{yus19}. The proof for the case for $a>0$ and $b>0$ and $\mathfrak{D}=\mathbb{R}^m$ is analogous and omitted.

(2) \emph{Case $a=0$ and $b=0$:} 
\begin{align*}
&\,\left|p_0(\boldsymbol{x})\partial_j\log p(\boldsymbol{x})\right|\\
\propto&\,\exp\left(-\frac{1}{2}\log(\boldsymbol{x})^{\top}\mathbf{K}_0\log(\boldsymbol{x})+\boldsymbol{\eta}_0^{\top}\log(\boldsymbol{x})\right)\\
&\times\left|x_j^{-1}\left(\eta_j-\boldsymbol{\kappa}_{j,-j}^{\top}\log(\boldsymbol{x}_{-j})\right)-\kappa_{jj}x_j^{-1}\log x_j\right|\\
\leq&\,\Big|\left(\eta_j-\boldsymbol{\kappa}_{j,-j}^{\top}\log \boldsymbol{x}_{-j}\right) \exp\left[-N_{\mathbf{K}_0}(\log x_j)^2/2+(\eta_j-1)\log x_j\right]
\\
&\pushright{-\kappa_{jj}\exp\left[-N_{\mathbf{K}_0}(\log x_j)^2/2+(\eta_j-1)\log x_j\right]\log x_j\Big|}\\
&\times \prod_{k\neq m}\exp\left(-N_{\mathbf{K}_0}(\log x_k)^2/2+\eta_j\log x_k\right)\\
\propto&\,\mathcal{O}\left[\exp\left(-N_{\mathbf{K}_0}y_j^2/2+(\eta_j-1)y_j\right)\right]+\mathcal{O}\left[\exp\left(-N_{\mathbf{K}_0}y_j^2/2+(\eta_j-1) y_j\right) y_j\right]
\end{align*}
which apparently vanishes as $x_j\searrow 0^+$ and $x_j\nearrow +\infty$ with $y_j\equiv\log(x_j)$ since it is dominated by a constant times a Gaussian density in $y_j$. Thus, by Proposition \ref{prop_assumptions_power}, (A.1) is satisfied with any $\alpha_j\geq 0$. Likewise, for (A.2),
\begin{align*}
&\,\int_{\mathbb{R}_+^m}p_0(\boldsymbol{x})\left\|\nabla\log p(\boldsymbol{x})\odot (\boldsymbol{h}\circ\boldsymbol{\varphi})^{1/2}(\boldsymbol{x})\right\|_2^2\d\boldsymbol{x}\\
\leq&\,\mathrm{const}\cdot\sum_{j=1}^m\int_{\mathbb{R}_+^m}\prod_{k=1}^m\exp\left[-N_{\mathbf{K}_0}(\log x_k)^2/2+\eta_k\log(x_k)\right]\times \\
&\quad\quad\quad\quad  h_j(x_j)\left[x_j^{-1}\left(\eta_j-\boldsymbol{\kappa}_{j,-j}^{\top}\log(\boldsymbol{x}_{-j})\right)-\kappa_{jj}x_j^{-1}\log x_j\right]^2\d  \boldsymbol{x},
\end{align*}
which can be decomposed into a sum of products of univariate integrals of the form 
\[\mathrm{const}\cdot\exp\left(-N_{\mathbf{K}_0}(\log x_j)^2/2+A\log(x_j)\right)(\log x_j)^B(h_j(x_j))^C\]
with $B=0,1,2$, $C=0,1$, and constants $A$. With $h_j(x_j)=x_j^{\alpha_j}$ for any $\alpha_j\geq 0$ this is bounded by some Gaussian density in $\log x_j$, so $\int_{\mathbb{R}_+^m}p_0(\boldsymbol{x})\|\nabla\log p(\boldsymbol{x})\odot (\boldsymbol{h}\circ\boldsymbol{\varphi})^{1/2}(\boldsymbol{x})\|_2^2\d\boldsymbol{x}<+\infty$. Similarly, we have $\int_{\mathbb{R}_+^m}p_0(\boldsymbol{x})\|[\nabla\log p(\boldsymbol{x})\odot(\boldsymbol{h}\circ\boldsymbol{\varphi})(\boldsymbol{x})]'\|_1\d\boldsymbol{x}<+\infty$ and the proof is omitted.

(3) \emph{Case $a=0$ and $b>0$:} Recall $\rho_j(\mathfrak{D})\equiv\overline{\{x_j:\boldsymbol{x}\in\mathfrak{D}\}}$. Let $\rho_j^*(\mathfrak{D})\equiv\sup\rho_j(\mathfrak{D})$. Since we assume that $\eta_j\leq 0$ for any $j$ such that $\rho_j^*(\mathfrak{D})<+\infty$,
\begin{align*}
&\,p_0(\boldsymbol{x})\partial_j\log p(\boldsymbol{x})\\
\propto&\,\exp\left(-\frac{1}{2}\log(\boldsymbol{x})^{\top}\mathbf{K}_0\log(\boldsymbol{x})+\frac{1}{b}\boldsymbol{\eta}_0^{\top}\boldsymbol{x}^b\right)\\
&\times\left[\eta_j x_j^{b-1}-x_j^{-1}\boldsymbol{\kappa}_{j,-j}^{\top}\log(\boldsymbol{x}_{-j})-\kappa_{jj}x_j^{-1}\log x_j\right]\\
\leq&\,\exp\left(-\frac{1}{2}\log(\boldsymbol{x})^{\top}\mathbf{K}_0\log(\boldsymbol{x})+\frac{1}{b}\sum_{j:\rho_j^*(\mathfrak{D})<+\infty}\eta_{0j}(\rho_j^*(\mathfrak{D}))^b\right)\\
&\times\left[-x_j^{-1}\boldsymbol{\kappa}_{j,-j}^{\top}\log(\boldsymbol{x}_{-j})-\kappa_{jj}x_j^{-1}\log x_j\right]\\
\propto&\,\exp\left(-\frac{1}{2}\log(\boldsymbol{x})^{\top}\mathbf{K}_0\log(\boldsymbol{x})\right)\left[-x_j^{-1}\boldsymbol{\kappa}_{j,-j}^{\top}\log(\boldsymbol{x}_{-j})-\kappa_{jj}x_j^{-1}\log x_j\right]
\end{align*}
is bounded by the corresponding quantity in the $a=b=0$ case with $\boldsymbol{\eta}=\boldsymbol{0}_m$, and (A.1) is thus satisfied. Similarly, the two quantities for (A.2) are bounded by a constant times those in the $a=b=0$ case with $\boldsymbol{\eta}=\boldsymbol{0}_m$ and (A.2) is thus also satisfied.
\end{proof}

\begin{proof}[Proof of Theorem \ref{theorem_bounded_nonlog_ggm}]
It suffices to bound $\boldsymbol{\Gamma}$ and $\boldsymbol{g}$ using their forms in Section \ref{Estimation_general} and apply Theorem 1 in \citet{lin16}. Thus, we first find the bounds of $(h_j\circ\varphi_j)(\boldsymbol{x})x_j^{p_j}x_{k}^{p_k}x_{\ell}^{p_\ell}$ with $h_j(x)=x^{\alpha_j}$, $\alpha_j\geq 0$, $\alpha_j\geq -p_j$, $p_j\in\mathbb{R}$, $p_k\geq 0$, $p_{\ell}\geq 0$ and $x_i\in[u_i,v_i]$ for $i=1,\ldots,m$. Suppose without loss of generality that $j$, $k$, $\ell$ are all different, as 
\[\max_{x_j}f_{j,1}(x_j)\max_{x_j}f_{j,2}(x_j)\max_{x_j}f_{j,3}(x_{j})\geq \max_{x_j}\left(f_{j,1}(x_j)f_{j,2}(x_j)f_{j,3}(x_j)\right)\geq 0\] for any nonnegative functions $f_{j,1}$, $f_{j,2}$, $f_{j,3}$.

As $x_j$ approaches its boundary, $\varphi_j(\boldsymbol{x})\searrow 0^+$ and hence $(h_j\circ\varphi_j)(\boldsymbol{x})x_j^{p_j}\searrow 0^+$ if $\alpha_j>-p_j$. The lower bound 0 for $(h_j\circ\varphi_j)(\boldsymbol{x})x_j^{p_j}x_{k}^{p_k}x_{\ell}^{p_\ell}$ is thus tight enough.

As for the upper bound, the only way for the quantity to be unbounded from above is when $x_j\searrow 0^+$ and $p_j<0$, but as $x_j\searrow 0^+$, $(h_j\circ\varphi_j)(\boldsymbol{x})=x_j^{\alpha_j}$ so this cannot happen with the choice of $\alpha_j\geq -p_j$. Noting that $h_j$ is monotonically increasing, we consider the following cases:
\begin{enumerate}[(1)]
\item Suppose $x_j\geq (u_j+v_j)/2$. Then 
\begin{align*}
(h_j\circ\varphi_j)(\boldsymbol{x})&\leq h_j\left(\min\left\{C_j,v_j-x_j\right\}\right)\\
&\leq h_j(\min\{C_j,(v_j-u_j)/2\})\\
&\leq \min\{C_j^{\alpha_j},(v_j-u_j)^{\alpha_j}/2^{\alpha_j}\},
\end{align*}
and $x_j^{p_j}\leq (u_j+v_j)^{p_j}/2^{p_j}$ if $p_j<0$ or $x_j^{p_j}\leq v_j^{p_j}$ if $p_j\geq 0$. 
\item Suppose $x_j\leq (u_j+v_j)/2$. Then 
\begin{align*}
(h_j\circ\varphi_j)(\boldsymbol{x})x_j^{p_j}&\leq h_j\left(\min\left\{C_j,x_j-u_j\right\}\right)x_j^{p_j}\\
&=\min\{C_j^{\alpha_j}, (x_j-u_j)^{\alpha_j}\}x_j^{p_j}.
\end{align*}
Now let $f(x)=(\min\{C_j,x-u_j\})^{\alpha_j}x^{p_j}$. Then $(\log f(x))'=\alpha_j/(x-u_j)\mathds{1}_{x<u_j+C_j}+p_j/x$. For $x\geq u_j+C_j$ this has the same sign as $p_j$, otherwise it is equal to $((\alpha_j+p_j)x-u_jp_j)/(x(x-u_j))\geq 0$ on $(u_j,v_j)$ since $x>u_j$, $\alpha_j\geq -p_j$ and $\alpha_j\geq 0$.
This implies that if $p_j\geq 0$ or $v_j-u_j\leq 2C_j$, $f$ is increasing on $(u_j,(u_j+v_j)/2)$, and so $(h_j\circ\varphi_j)(\boldsymbol{x})x_j^{p_j}\leq\min\left\{C_j,(v_j-u_j)/2\right\}^{\alpha_j}(u_j+v_j)^{p_j}/2^{p_j}$; otherwise, $f$ is increasing on $(u_j,u_j+C_j)$ and decreasing on $(u_j+C_j,(u_j+v_j)/2)$, so $(h_j\circ\varphi_j)(\boldsymbol{x})x_j^{p_j}\leq C_j^{\alpha_j}(u_j+C_j)^{p_j}$.
\end{enumerate}
Thus, defining 
\begin{align*}
&\,\zeta_j(\alpha_j,p_j)\\
\equiv&\,\begin{cases}\min\left\{C_j,(v_j-u_j)/2\right\}^{\alpha_j}(u_j+v_j)^{p_j}/2^{p_j}, & p_j<0, v_j-u_j\leq 2C_j, \\ 
\min\left\{C_j,(v_j-u_j)/2\right\}^{\alpha_j}(u_j+C_j)^{p_j}, & p_j<0, v_j-u_j> 2C_j, \\
\min\left\{C_j,(v_j-u_j)/2\right\}^{\alpha_j}v_j^{p_j}, & p_j\geq 0,\end{cases}\
\end{align*}
we have $0\leq (h_j\circ\varphi_j)(\boldsymbol{x})x_j^{p_j}x_{k}^{p_k}x_{\ell}^{p_\ell}\leq\zeta_j(\alpha_j,p_j)v_k^{p_k}v_{\ell}^{p_{\ell}}$. Now assume additionally that $\alpha_j\geq \max\{1,1-p_j\}$, then by $h_j'(x)=\alpha_j x_j^{\alpha_j-1}$,  $0\leq\partial_j(h_j\circ\varphi_j)(\boldsymbol{x})x_j^{p_j}x_k^{p_k}\leq\alpha_j\zeta_j(\alpha_j-1,p_j)v_k^{p_k}$. 

First assume $a>0$. Then assuming $\alpha_1,\ldots,\alpha_m\geq\max\{1,2-2a,2-2b,1-a,2-a,2-b\}=\max\{1,2-a,2-b\}$, using the form of $\boldsymbol{\Gamma}$ and $\boldsymbol{g}$ in Section \ref{Estimation_general}, for all $j,k,\ell$ we have 
\begin{align*}
0&\leq\gamma_{j,k,\ell}(\mathbf{x})\leq \varsigma_{\boldsymbol{\Gamma}}\equiv\max\limits_{j,k =1,\ldots,m}\max\{\zeta_j(\alpha_j,2a-2)v_k^{2a},\,\zeta_j(\alpha_j,2b-2)\}
\end{align*}
and 
\begin{align*}
0&\leq g_{j,k}(\mathbf{x})\\
&\leq \varsigma_{\boldsymbol{g}}\equiv\max\limits_{j,k =1,\ldots,m}\max\{\alpha_j\zeta_j(\alpha_j-1,a-1)v_k^{a}+|a-1|\zeta_j(\alpha_j,a-2)v_k^{a}+a\zeta_j(\alpha_j,2a-2),\\
&\pushright{\alpha_j\zeta_j(\alpha_j-1,b-1)+|b-1|\zeta_j(\alpha_j,b-2)\}}.
\end{align*}

Then by Hoeffding's inequality,
\begin{align}
\mathbb{P}\left(\max_{j,k,\ell}\left|\gamma_{j,k,\ell}-\mathbb{E}_0\gamma_{j,k,\ell}\right|\geq\epsilon_1/2\right)&\leq 2\exp\left(-n\epsilon_1^2/(2\varsigma_{\boldsymbol{\Gamma}}^2)\right)\label{proof_bounded_gamma},\\
\mathbb{P}\left(\max_{j,k}\left|g_{j,k}-\mathbb{E}_0g_{j,k}\right|\geq\epsilon_2\right)&\leq 2\exp\left(-2n\epsilon_2^2/\varsigma_{\boldsymbol{g}}^2\right).\label{proof_bounded_g}
\end{align}
Let $\epsilon_1\equiv\varsigma_{\boldsymbol{\Gamma}}\sqrt{2(\log m^{\tau}+\log 4)/n}$ and $\epsilon_2\equiv \varsigma_{\boldsymbol{g}}\sqrt{(\log m^{\tau}+\log 4)/(2n)}$. With the choice of $\delta\leq 1+\sqrt{(\log m^{\tau}+\log 4)/(2n)}$ and using the fact that $0\leq\max_{j,k,\ell}\gamma_{j,k,\ell}\leq \varsigma_{\boldsymbol{\Gamma}}= \epsilon_1/(2\delta-2)$, (\ref{proof_bounded_gamma}) and (\ref{proof_bounded_g}) imply that
\begin{align}
&\,\mathbb{P}\left(\max_{j,k,\ell}\left|\delta\gamma_{j,k,\ell}-\mathbb{E}_0\gamma_{j,k,\ell}\right|\geq\epsilon_1\right)\\
\leq&\,\mathbb{P}\left(\max_{j,k,\ell}\left|\gamma_{j,k,\ell}-\mathbb{E}_0\gamma_{j,k,\ell}\right|+(\delta-1)\max_{j,k,\ell}\gamma_{j,k,\ell}\geq\epsilon_1\right)\nonumber\\
\leq&\,\mathbb{P}\left(\max_{j,k,\ell}\left|\gamma_{j,k,\ell}-\mathbb{E}_0\gamma_{j,k,\ell}\right|\geq\epsilon_1/2\right)\leq m^{-\tau}/2,
\end{align}\vspace{-0.2in}
\begin{align}
\mathbb{P}\left(\max_{j,k}\left|g_{j,k}-\mathbb{E}_0g_{j,k}\right|\geq\epsilon_2\right)&\leq m^{-\tau}/2.
\end{align}
The results then follow by applying Theorem 1 in \citet{lin16}.

In the case where $a=0$, and $u_k>0$ for all $k$,
\begin{align*}
|(h_j\circ\varphi_j)(\boldsymbol{x})x_j^{p_j}\log(x_{k})\log(x_{\ell})|
\leq\zeta_j(\alpha_j,p_j)\cdot\max\{|\log(u_k)\log(u_{\ell})|,|\log(v_k)\log(v_{\ell})|\}
\end{align*}
and everything else follows similarly as for $a>0$.
\end{proof}

\begin{proof}[Proof of Lemma \ref{theorem_subexp_positive_measure}]
We show that $X_j^{2a}$ for $a>0$ or $\log X_j$ for $a=0$ is sub-exponential by showing its moment-generating function is finite. Then the sub-exponentiality follows from Theorem 2.13 of \citet{wai19}.

First consider the case where $a=0$. In Corollary \ref{cor_norm_const}, we only require $\mathbf{K}$ to be positive definite without any restrictions on $\boldsymbol{\eta}$, and thus for any $t\in\mathbb{R}$, $\mathbb{E}_0\exp(t\log X_j)$ is the inverse normalizing constant for the model with parameters $\mathbf{K}_0$ and $\boldsymbol{\eta}_0+t\boldsymbol{e}_j$, where $\boldsymbol{e}_j$ is the vector with the $j$-th coordinate equal to $1$ and the rest equal to $0$, and is thus finite.

Next, consider $a>0$. Corollary \ref{cor_norm_const} requires $\mathbf{K}_0$ to be positive definite, and in addition $\boldsymbol{\eta}_0\succ-\mathbf{1}_m$ if $b=0$. Then, again writing $\boldsymbol{x}^0/0=\log\boldsymbol{x}$ for the $b$ part,
\begin{align*}
p_0(\boldsymbol{x})\exp\left(t x_j^{2a}\right)&\propto\exp\left(-\frac{1}{2a}{\boldsymbol{x}^a}^{\top}\mathbf{K}_0\boldsymbol{x}^a+\frac{1}{b}\boldsymbol{\eta}_0^{\top}\boldsymbol{x}^b+t x_j^{2a}\right)\\
&\leq\exp\left(\sum_{k=1}^m\left(\left(-\lambda_{\min}\left(\mathbf{K}_0\right)+2a t\mathds{1}_{k=j}\right)x_k^{2a}/(2a)+\eta_{0,k}x_k^b/b\right)\right),
\end{align*}
a constant times the density for parameters $\mathrm{diag}\left(\lambda_{\min}\left(\mathbf{K}_0\right)\mathbf{1}_m-2at\boldsymbol{e}_j\right)$ and $\boldsymbol{\eta}_0$. Thus, for $t\in\left(-\infty,\lambda_{\min}\left(\mathbf{K}_0\right)/(2a)\right)\ni 0$, $\mathbb{E}_0\exp\left(t X_j^{2a}\right)$ is finite.
\end{proof}

\begin{proof}[Proof of Corollary \ref{corollary_apos_unbounded}]
Let the sub-exponential norm of $X_j^{2a}$ be $\left\|X_j^{2a}\right\|_{\psi_1}\equiv\sup_{q\geq 1}(\mathbb{E}_0|X_j|^{2aq})^{1/q}/q$, then by Lemma 21.6) of \citet{yus19} or Corollary 5.17 of \citet{ver12},
\[\mathbb{P}\left(\left|X_j^{2a}-\mathbb{E}_0 X_j^{2a}\right|\geq\epsilon_{3,j}\right)\leq\exp\left(-\min\left(\frac{\epsilon_3^2}{8e^2\left\| X_j^{2a}\right\|_{\psi_1}^2},\frac{\epsilon_3}{4e\left\|X_j^{2a}\right\|_{\psi_1}}\right)\right).\]
Letting 
\begin{multline*}
\epsilon_{3,j}\equiv\max\Big\{2\sqrt{2}e\left\| X_j^{2a}\right\|_{\psi_1}\sqrt{\log 3+\log n+\tau\log m+\log \left(m-\left|\rho_{\mathfrak{D}}^*\right|\right)},\\
4e\left\| X_j^{2a}\right\|_{\psi_1}\left(\log 3+\log n+\tau\log m+\log \left(m-\left|\rho_{\mathfrak{D}}^*\right|\right)\right)\Big\},
\end{multline*}
then $\max\left\{\mathbb{E}_0 X_j^{2a}-\epsilon_{3,j},0\right\}^{1/(2a)}\leq{X_j^{(i)}}\leq\left(\mathbb{E}_0 X_j^{2a}+\epsilon_{3,j}\right)^{1/(2a)}$ for all $j\not\in\rho_{\mathfrak{D}}^*$ and $i=1,\ldots,n$ with probability at least $1-1/(3m^{\tau})$. The rest follows as in the proof of Theorem \ref{theorem_bounded_nonlog_ggm}.
\end{proof}

\end{document}